\documentclass[reqno]{amsart}

\usepackage{amsmath,mathtools}
\usepackage{graphicx}
\usepackage[dvipsnames]{xcolor}
\usepackage{chngcntr}
\usepackage[normalem]{ulem}
\usepackage{amsmath,amsfonts, amssymb,amsfonts,amsthm}
\numberwithin{equation}{section}

\usepackage{tikz-cd}
\usetikzlibrary{shapes,arrows,intersections}
\usetikzlibrary{matrix,fit,calc,trees,positioning,arrows,chains,shapes.geometric,shapes,angles,quotes}

\usepackage{pgfplots}
\usepackage{bm}
\usepgfplotslibrary{fillbetween}
\usetikzlibrary{backgrounds}
\usetikzlibrary{patterns}
\usetikzlibrary{scopes}
\usetikzlibrary{shapes,arrows}
\usetikzlibrary{plotmarks,intersections}
\usetikzlibrary{calc}

\usepackage{amssymb}

\usepackage{graphicx}
\usepackage{xcolor}
\usepackage{mathtools}

\usepackage[cmtip,all]{xy}
\usepackage{amssymb}

\numberwithin{equation}{section}
\newtheorem{thm}{Theorem}

\newtheorem{lem}[thm]{Lemma}

\newtheorem{cor}[thm]{Corollary}
\newtheorem{prop}[thm]{Proposition}
\theoremstyle{definition}

\newtheorem*{defn}{Definition}

\newtheorem{problem2}{Loss to Gain Moment Rate  Property}
\theoremstyle{remark}
\newtheorem{remark}{\bf Remark}

\newcommand{{\kc}}{{\bf k_c}}
\newcommand{{\kb}}{{\bf k_b}}
\newcommand{{\klb}}{{\bf k_{lb}}}
\newcommand{{\kgamma}}{{\bf k_{\gamma}}}

\newcommand{\kmax}{{\bf{k}_{\text{\tiny mx}}}}
\newcommand{\kmin}{{\bf{k}_{\text{\tiny mn}}}}

\newcommand{\di}{d}

\newcommand{\ir}{\int_{\R^\di}}

\newcommand{\ph}{\varphi}
\renewcommand{\th}{\vartheta}

\newcommand{\sign}{\text{sign}}

\newcommand{\C}{{\mathcal C}}

\newcommand{\R}{{\mathbb R}}
\newcommand{\vprime}{{}'\!v}
\newcommand{\uprime}{{}'\!u}
\newcommand{\fprime}{{}'\!f}

\newcommand{\gprime}{{}'\!g}

\def\nn{{\noindent}}

\def\ep{\varepsilon}
\def\eps{\varepsilon}
\def\A{{\mathcal A}}

\def\B{{\mathcal B}}
\def\C{{\mathcal C}}

\def\real{{\mathbb R}}

\def\loc{\operatorname{loc}}

\def\b{\beta}

\def\e{\varepsilon}

\def\m{\mu}

\def\s{\sigma}

\def\A{\mathcal{A}}

\def\~{\approx}

\def\ee{\in}

\def\be{\begin{equation}}
\def\ee{\end{equation}}

\def\cH{{\cal H}}

\def\b{\beta}

\def\e{\varepsilon}

\def\m{\mu}

\def\s{\sigma}

\def\A{\mathcal{A}}

\def\~{\approx}

\def\ee{\in}

\def\S{{\mathcal S}}
\def\eps{\varepsilon}

\def\b{{\bf b}}
\def\s{{\bf s}}

\def\ttm{{\tt \bf m}}

\def\nn{\noindent}

\def\real{{\mathbb R}}
\def\E{{\mathcal E}}

\def\cH{{\mathcal H}}

\def\RR{\hbox{{\rm I}\kern-.2em\hbox{\rm R}}}
\def\pRR{\hbox{{\tiny \rm I}\kern-.1em\hbox{{\tiny \rm R}}}}
\def\nn{\noindent}

\def\kbar{\bar k}

\def\smallcircledV{{\scriptscriptstyle{\bigcirc\kern
-5.66pt\vee\kern 1.33pt}}}

\begin{document}
\title[Boltzman equation, coerciveness, exponential tails ans Lebesgue integrability ]
{The Boltzmann equation for hard potentials with integrable angular transition: Coerciveness, exponential tails rates, and Lebesgue integrability}

\author{Ricardo J. Alonso} 
\curraddr{{\bf Ricardo J. Alonso}; Texas A\&M, Science Department, Education City, Doha, Qatar;   and Departamento de Matemática, PUC-Rio,
Rio de Janeiro, Brasil.}
\email{ricardo.alonso@qatar.tamu.edu}

\author{Irene M. Gamba} 
\curraddr{{\bf Irene. M. Gamba}; Department of Mathematics and Oden Institute for Computational Science and Engineering, The University of Texas at Austin.}
\email{gamba@math.utexas.edu}
\thanks{}



%
%

\begin{abstract}
This manuscript  focus on an extensive survey with new techniques on the problem of  solving the Boltzmann flow by bringing a unified approach to the Cauchy problem to homogeneous kinetic equations with Boltzmann-like collision operators under integrability assumption of the scattering profile in the particle-particle interaction mechanism.  The work focuses on the relevant hard potential case where the solution properties are studied with a modern take.  While many of the discussed results can be found the literature spread over several papers along the years, we bring a complete program that includes a new approach to the existence and uniqueness theorem securing the well-posedness theory,  to moments estimates, and integrability propagation for the homogeneous Boltzmann flow.   In particular,  a detailed calculation of classical polynomial  moments upper bounds as function of the coerciveness  is described, which characterized the rate of exponential moments obtained by summability of the polynomial ones. In addition, a  proof of  uniform propagation of $L^\infty$ regularity under general integrable scattering kernels is performed.  Along the way, constants  appearing in estimates are carefully calculated,  improving most of  previous exiting results in the literature.
For the non expert reader we also include a general discussion of the basic elements of the Boltzmann model and important key results for the understanding of the mathematical discussion of the equation and include an extensive  set of references that enrich  and motivate  further discussions on Boltzmann flows for broader gas modeling configuration such as  gas mixtures systems, polyatomic gases, multilinear collisional forms such as, ternary or quartic, to  those derived from symmetry braking  quantum mean field theories or weak turbulence models from  spectral energy  waves in classical fluid.  \end{abstract}

\maketitle

\tableofcontents

\vfill\newpage
\section{Fundamentals of binary elastic collisional flows}

\subsection{Introduction}

The notion of statistical mechanics and physics   models  were independently introduced by Ludwig Boltzmann and Charles Maxwell during last quarter of the nineteen century motivated  by the search of finding rigorous questions and answers  to the   theory of thermodynamics of a gas. 
The resulting statistical viewpoint essentially   expresses a gas evolution  model  as the chaotic motion of huge numbers of molecules rushing here and there at large speeds interacting, or colliding, and rebounding 
according to  a given law, being of elementary mechanics nature or other attributes from probabilistic considerations.
These models nowadays have applications beyond statistical mechanics and physics. Then, without loss of generalization, the $N$ velocity states may also be interpreted as the space of attributes associated to the model on which a probability density evolves in time as the attribute change values. That is the case in recent applications that range from opinion dynamics, wealth distribution as much as multi agent dynamics operating with given interactions laws.
   However,  we will focus in this manuscript mostly in the original in the classical framework for particle physics, revisiting a framework developed in the last forty years and introducing new techniques and improved results that unified the  the theory of the Boltzmann equation in the space homogeneous setting that will enable   clear improvements in the space inhomogeneous case under specific boundary conditions in bounded domains or at infinity. 
   
   The space homogeneous Boltzmann equation is model of a statistical flow. It  describes the evolution of a probability density  according to a law given by   a non-local multi linear operator acting on the probability density, the so called collision operator. Such operator  determines the evolution of the   probability density evaluated at different states of velocities, or attributes,  interacting  according to the basic laws of thermodynamics in physical sciences.    
In classical thermodynamics and mechanics pressure is identify as the mechanical effect of the impact of the moving 
molecules as they strike against a solid wall  or a region of different molecular properties.  For the starting of the derivation of collisional theory by the classical Boltzmann flow,  molecules are assumed to be 
modeled  by spheres that move according to the laws of classical elastic or inelastic mechanics. 
If no external forces (such as gravity) are assumed to act on the molecules, each of them will move along a 
straight line unless they strike another sphere or the wall. In the classical mechanics interaction, these are called  billiard models. 

The {\em Boltzmann Transport Equation} (BTE) models an attempt to describe the an evolution process from a 
statistical point of view: {\em ``it determines the state of the gas using some statistical mechanics laws from knowing the statistical mechanical initial state"}.  In this sense, the  Boltzmann Transport Equation is an {statistical  flow  model  whose solution describe the evolution of a probability density of finding a particle at a given position, with a given velocity, at a given time}. 

\subsection{The Boltzmann Transport Equation}
The model for a rarefied mono atomic gases under binary elastic interactions, originally introduced  by L. Boltzmann in 1872 \cite{Boltzmann}, and analytically studied by many authors including Carleman, Wild, Chapman, Cowling, Cercignani, Illner, Pulvirenti  \cite{Ca,wild-1950,CC70,Cl,CIP} to name a few,  is written as follows
\begin{equation}\label{BTE}
{\frac{\partial f}{\partial t} + v\cdot \nabla_x f
= a^2 N \int_{\real^3\times \mathbb{S}^{2}_{+}}(\fprime \fprime_* - ff_*) B(|u|, \hat u\cdot\eta ) \,dv_*d\eta}
\end{equation}
with
\begin{equation}\label{Elastic Interaction}
\vprime=v - (u\cdot\eta)\, \eta ,\ \ \vprime_*=v_* + (u\cdot\eta)\, \eta ,\ \ \quad u=v-v_* 
\end{equation}

\nn for $\eta\in\mathbb{S}^2$ a vector in the direction of $v-\vprime=\vprime_*-v_*$; 
where $f(x,v,t)$ describes the evolution of the probability distribution
function (pdf) of finding a particle at position $x$, velocity $v$ and time $t$,
for a billiard model of hard spheres and prescribed initial probability
$f(x,v,0) = f_0 (x_0,v_0)$. In addition, the notation for the pdf on pre collisional velocities is $\fprime:=f( x,\vprime,t)$ and $\fprime_*:=f( x,\vprime_*,t)$. The corresponding to the ones on post-collisional velocities are simply $f:=f( x,v,t)$ and $f_*:=f( x,v_*,t)$.  \\

 Rigorous mathematical derivations of the Boltzmann homogeneous  flow was first developed by Sznitman \cite{sznitman84} by probabilistic methods.
  In  the space inhomogeneous setting, the  derivations of the Boltzmann flow model under consideration that have been discussed in  several manuscripts  the past twenty years, including the work of  Pulvirenti, Saffirio,  Simonella \cite{PulSafSimonella-2014},  Gallagher, Saint Raymond and Texier \cite {Gal-SRay-Texier} among many other groups, are  based in the formal derivation by Landford.   All these results, following strict mathematical rigor, involve  arguments  short time $t$ with short time validity up to  the average of the first collision time.  However, the classical space inhomogeneous Boltzmann initial valued problem, has beed solved  global in time, first for short range potentials  with initial data near vacuum following the pioneer work  by Kaniel and Shimbrot, and Illner and Shimbrot \cite{kaniel-shinbrot, illner-shinbrot},  and later by for long range potentials by  Ukai and Asanov, Bellomo, Toscani and Palczewski,  \cite{ukai-asano-82, bellomo-toscani, toscani-ARMA-86}. Toscani  \cite{toscani-ARMA-88, pal-tosc-JMP89} address the first existence for large data, near a fixed Maxwellian state. Hamdache, in \cite{hamdache,hamdache1}, solved the near vacuum problem  by fixed point theorem's arguments; and Goudon \cite{goudon-97} extended the Hamdache techniches for long range potentials, to data near a fixed Maxwellian state. A couple of decades later,   Glassey \cite{glassey} found global solutions for the relativistic Boltzmann flow near vacuum; followed by the work of Alonso and Gamba \cite{AG-JSP09, AG-JSP09-R} motivated by the work $W^{m,p}_k(\R^d)$ from Alonso and Carneiro and Alonso Carneiro and Gamba \cite{AC, ACG-cmp10},  revisited the Kaniel-Shimbrot and Illner-Shimbrot \cite{kaniel-shinbrot, illner-shinbrot} iteration by a modified approach and improved global global in time estimates not valid for near vacuum  data for short rage potentials, but enlarged the initial data set subject to different Maxwellian distribution with closeness in $L^\infty$-exponentially weighted, and including the $L^p$-differentiability.   Most recently, the problem of longtime behavior of the solution near vacuum was address for long-range potentials by Bardos, Gamba,  Golse and Levermore \cite{BGGL-2016} showing  that, in the vicinity of global Maxwellians with finite mass, the dispersion due to the advection operator quenches the dissipative effect of the Boltzmann collision integral,   resulting in a typical scattering mechanisms where the  large time limit of solutions of the Boltzmann flow  is given by noninteracting, freely transported states.  

 The initial value problem associated to the space homogeneous equation \eqref{BTE} in the space homogeneous setting has been proven to be globally well-posed under general physically conditions and rather broad initial data in size.  Furthermore, the model enjoys properties such as propagation and generation of polynomial and exponential statistical moments or observables, as well as propagation of Sobolev norms of any order.  These particular properties of higher Lebesgue integrability and Sobolev regularity propagation has been the attention of several studies, perhaps starting with Carleman \cite{Ca57} and continued with Arkeryd \cite{Arkeryd72, Arkeryd-Linfty}, Gustaffson \cite{Gustaff}, Lions \cite{LionsIII}, and Toscani and Villani \cite{TV}. 

In the last few years, the advance of these theories have enable new techniques involving  the notion of angular averaging lemmas that permitted further studies of lower estimates in the search of coerciveness, summability of  moments to obtained the propagation and generation of exponential moments, as well as exponential wighted $L^p_\ell(\R^d)$ estimates for $p\in[1,\infty]$ and $\ell>2.$

 In general,   moment estimates  is a fundamental tool that enables the existence and uniqueness theory  associated to the Cauchy problem for the the spatially homogeneous Boltzmann equation. One can roughly separate the problem characterizing the transition probability  measure rates (or collision kernels) on whether  depending of the following characteristics: they may depend on potential rates associated intramolecular potential  modeled by functions of the relative  speed associated  binary interaction, and on the angular transition measure that quantifies the transition rate with respect to the deflection angle before or after the interactions, as envisioned by H. Grad \cite{Grad_1958}.
 That is, these transition rates may or may not depend on the relative speed of the local colliding velocities as much as in the integrability condition of the angular part of such interactions. 
 
 Historically, these problems started to be addressed, starting  before the mid twenty's century, starting by the work of Carleman, Wild  and later  Arkeryd  and Elmroth \cite{Ca, Ca57,Arkeryd72, Arkeryd-Linfty,  Elmroth, Povzner},  for  transition rates with  bounded  intramolecular rates and bounded angular transition functions.   Di Blassio \cite{diblassio}contributed by finding the proof of uniqueness to the Boltzmann flow for hard potentials, but just having a few moments bounded. The estimates needed to obtain these results were under the assumption that the angular transition function is bounded.      

Around the same time,  the  pioneering work of Bobylev \cite{Bo76,B88} fully solved the problem for transition probability rates independent of the intramolecular potentials commonly referred as to {\em Maxwell Molecule}  or  {\em Maxwell type models}, where the Boltzmann flow can be fully solved in Fourier space with the $C^\infty(\R^d)$ topology associated to generating master equations to the evolutions  probabilistic or  stochastic flows, with profound discussions on convergence rates to the Boltzmann-Maxwell statistical equilibrium state, as shown in the work of  Toscani and Villani, and Carlen, Carvalho and Lu and  \cite{TV, Carlen} and references therein.
 
The problem of hard or soft intramolecular potential interactions, that is the higher the relative speed  the more or less collision rates, respectively, started to be addressed on whether   the integrability of the angular transition function was bounded or not. 
 
 By the last decade of the twentieth century the Boltzmann flow model enters in maturity, starting by  the work of Desvilletes \cite{Desvillettes}  who proved interesting point that  it can propagate any moments, and shortly after,  Wennberg \cite{WennbergMP} showed that moments can be generated if the initially data has bounded entropy in addition to bounded  mass and energy.  In the same year of Wennberg's work, a groundbreaking result of Bobylev \cite{boby97} showed that the generated moments are summable paving the way to the understanding of propagation and generation of $L^1(\R^d)$ exponentially weighted estimates in several frameworks such us inelastic theories for hard spheres by  
 Gamba Panferov and Bobylev, Alonso and Lods for hard spheres in \cite{ BGP04, Alonso-Lods-SIMA10, ALods} and later with Panferov, Gamba and Villani  \cite{GPV09} showed the extension of this estimates to hard potentials and show the technique extends to angular transition functions in $L^1(\S^d)$ as well establishing that $L^\infty(\R^d)$-Maxwellian weighted estimates for  solutions as constructed by Arkeryd \cite{Arkeryd-Linfty}, later  extended by Bobylev and Gamba \cite{BG-17-maxw-bounds} to the Maxwell type of interactions.
 These results inspired the work of Mouhot \cite{mouhot06} on generation of  $L^1(\R^d)$ exponentially weighted estimates,  Alonso and Gamba \cite{AG-JMPA08} solutions derivatives,  and  Canizo and Mouhot, in  collaboration with the authors, \cite{ AlonsoCGM}, to study their summability by showing the convergence of moments's partial sums. 
 
  Further studies of these results were extended to hard potentials with the classical Grad angular non-cut-off condition  \cite{LuMouhot, TAGP,  PC-T18}  valid for type of solutions constructed by Morimoto, Wang and Yang \cite{MWY16}. In addition  Gamba, Pavlovic, and Taskovic \cite{Gamba-Pavlovic-Taskovic-SIMA2019} showed pointwise exponential estimates for conditional solutions and  recent results of Fournier \cite{Fournier-21} further clarifies and improves $L^1(\R^d)$ exponentially weighted estimates, both paper under  angular non-cut-off conditions.

 These manuscripts also expands and improves the existing results of properties such as the propagation of higher Lebesgue integrability,  initiated by  the work Gustaffson,  and later by Wennberg  in \cite{Gustaff, WennbergLp}, obtaining estimates for  the collision operator in  Banach spaces   $L^r(\R^d)$  for $1 \leq r < \infty$. Lions \cite{LionsIII}  introduced the breakthrough concept of {\em gain of integrabilty}, by meaning that it is possible to control   the gain operator  estimates for hard potentials in  $L^r(\R^d)$ weighted spaces having lower growth when compared  to the corresponding one  obtained by  the loss operator. This property, which still requires a bounded angular transition rate functions,  may be viewed as a smoothing effect obtained  under  special conditions of the potential and angular transition rates or collision kernels.  Later by Toscani,  Villani and Mouhot \cite{TV, MV04} used these Lions' results to obtain not only the propagation and Sobolev regularity in the classical sense for the space homogenous Boltzmann flows.  In the last decade, Carneiro and the authors revisited the work of  Gustaffson \cite{ACG-cmp10, ACG-cmp10} that inspired the work of the authors \cite{AG-krm11} to obtained significant improvements removing conditions from the work of Lions on the transition probability kernels, but also quantified the rate of integrability estimates
by parameters  depending on the Cauchy problem data and   the $L^r(\R^d)$ norm of the initial data.

 Thus, the new results presented in the current  manuscript provides a unified set of  fundamental  technical improvements and corrections  related to the collision operator estimation,  removing many of the constrains  angular transition and potential rate growth extended now to  any integrable angular transition function and any hard potential behavior, bounded at zero and with a growth rate corresponding to supper hard spheres.    Such improvements include a new complete proof of existence of unique solutions in $C^1(0,T, L^1_k(\R^d))$ enabling   explicit computations of  coercive constants and their dependence of the model parameters, and obtaining a priori estimates on the collision operator  enabling the implementation  of a proof theorem for Ordinary Differential flows in Banach spaces for  existence and uniqueness of  solutions propagation and generation of  on moments (or expectations) associated to the solution $f(v,t)$ of the Boltzmann flow, globally in time. In addition we developed bound, explicitly as function of the coercive constants. In addition a new proof  for the summability of moments globally in time   is presented,  securing the the propagation and generation of  exponential tails whose rate and exponent order are explicit and written in terms of the problem data. In particular one can easily characterize   the exponential rates by the Cauchy Problem  data.  The existence and uniqueness theorem follows the lines of the proposed approach by Bressan \cite{bressan}, in an unpublished note, after being completed by a complete proof  for the collisional integral satisfying the  sub tangent condition property  in Lemma~\ref{prop_subtan}. 
   
  The second part of the manuscript focus on the applications of tools developed for the $L^p_\ell(\R^d)$  inspired by the estimates from  the authors developed in \cite{ACG-cmp10, ACG-cmp10} revisit  Young's inequalities and Carleman integral representation in obtaining explicit tares of gain of integrability estimates.  New results, such us the lower bound  Proposition \ref{lem-cancel}, depending inly on the Banach space norm and the choice of potential rates in the transition probability kernel,  enables  $L^p_\ell(\R^d)$-norm propagation  for general transition probabilities   rates for arbitrary hard potentials and just angular intregability in the scattering direction with an explicit coercive rate.  These results yield,  most significantly, not only the generalization of Arkeryd's pioneering work on the propagation of $L^{\infty}$-integrability property for very constrained conditions of the potential and angular scattering to  such rather general scattering form  achieved in a very elementary manner to any  $L^{\infty}_\ell$, but also     exhibits how the coercive constant,  calculated in the first part giving global in time control to propagation and generation of  $L^1_{2\ell}(\R^d)$,  regulates the rate of exponential weights for any  $L^p_\ell(\R^d)$. As a consequence, the  control  of spectral gap estimates can be explicitly characterized as a function of the coerciveness, as it  will be shown in a future work by the authors.

 A remarkable point of this manuscript is that entropy is not needed for any of these results, since the proof of  lower bounds only require statistical moment considerations just from the Banach space without invoking an a priori estimate to a solution of the Boltzmann flow, as seen in Lemma \ref{lblemma}.


 This manuscript is lay down as follows.  
 Section~\ref{elementary-properties} focus on Elementary properties of the Boltzmann Equation,   Section~\ref{hardpotsection} states the Cauchy problem for Maxwell and hard potentials  in Theorem~\ref{CauchyProblem}, and  develops lower bounds in  Theorems~\ref{lblemma},  a priori estimates  on moments of the collisional integral   Theorem~\ref{mom-coll-op},   and Lebesgue polynomial  moments  propagation  and generation  of global in time bounds fully characterized by the coerciveness of the problem data Theorem~\ref{propagation-generation} that enable    the proof of  existence and uniqueness Theorem~\ref{CauchyProblem} in Section~\ref{existence-uniqueness}.
 Section~\ref{exponential-tails} considerably    improves the Propagation and generation of exponential moments with a detailed characterization of the exponential rates.
The following two Sections~\ref{convo-ine+gain int} and  Section~\ref {Lr-propagation} focus on 
 novel  proofs of  convolution inequalities and gain of integrability properties for the collision operator, and develop en new energy method to show  $L^r_\ell$  propagation theory for any  $r\in[1,\infty]$, with $\ell>2^+.$

The manuscript culminates with  a short Section~\ref{fine properties} addressing a revision on fine properties of the collision operator, followed by 
 an Appendix ~\ref{appendix}  included  in order 
to reproduced the proof of 
Theorem~\ref{Theorem_ODE} following the same lines for  solving  ODEs in Banach spaces as  proposed by A. Bressan in  an unpublished manuscript posted in \cite{bressan}.  

In the last few years,  these techniques and there results have been recently extended to the case of vector value solutions to gas mixture system for  particle with disparate masses, and for the scalar model,  polyatomic gases, in the works of Gamba, Pavic-Colic and De la Canal \cite{IG-P-C-mixtures, IG-P-C-poly-2020, DLC-IG-PC-mixturesL1}  ,  by Strain and Yu for relativistic particles, and later Strain and  Taskovic 
  \cite{ Strain-Yu-14, Strain-Taskovic-2017} for the the relativistic Fokker Plank  Landau flow.  
The $L^p_\ell(\R^d)$ vector valued solutions for to gas mixture system is being addressed in \cite{DLC-IG-PC-mixturesLp} and for an arbitrary system   for of monoatomic and polyatomic gases is being addressed  in  \cite{Alonso-gamba-Pcolic-22}.
 
  In addition several kinetic collisional frameworks have implemented the summability of moments estimates to calculate  $L^1(\R^d)$ polynomial and exponentially weighted estimates, like in the work  of 
Alonso, Bagland and Lods \cite{ABCL:2016} implemented these techniques to construct solutions to one dimensional Boltzmann flow modeling polymers;   Alonso, Gamba and Tran \cite{Alonso-Gamba-Tran} in the solution of system of the Quantum Boltzmann gas coupled to the Bose Einstein condensed  system at cold temperatures, as well as Gamba, Smith and Tran \cite{AST-M3A-20} applied to the estimates for the kinetic wave equation derived by Luov and Zacharov for water waves applied to stratified flows in deep oceans. 
Most recently,  the Illner, Kaniel and  Shimbrot  iteration schemes have been extended   to global in time solutions to  binary-ternary Boltzmann flows  with initial data near vacuum,  while the global in time solutions for its space homogeneous problems for large data with  propagating and generating  moments estimates  resulted in lovable problems whose solutions also exhibit global in time $L^1(\R^d)$-polynomial and exponential moments, as shown by Ampatzoglou, Gamba, Pavlovic and Taskovi\v{c} \cite{Ampatzoglou_G_P_T-22, Ampatzoglou_G_P_T-23}. 

These type of propagation of  moments estimates were also developed for the space homogeneous Fokker Plank Landau equation for hard potentials by Desvillettes and Villani \cite{Des-Vill-CPDE2000-part1, Des-Vill-CPDE2000-part2} in their studies of $L^1_\ell(\R^d)$-norm propagation properties as much as higher $H^\alpha(\R^d)$ polynomial and exponential  weighted regularity enabling the spectral analysis of the Landau flow for hard potentials.

Finally, in a unique application,  $L^1(\R^d)$ and  $H^\alpha(\R^d)$  polynomial and exponential weighted propagation estimates have been very instrumental to developed error estimates for conservative spectral schemes for both the numerical error, consistency and stability to numerical solutions of the homogeneous Boltzmann flow by and Tharkabhushaman and authors  \cite{AlonsoGThar}, and also by  Pennie and one of the authors  for the homogeneous Landau flows \cite{Pennie-Gamba-2020}.

\section{Elementary Properties of the Boltzmann Equation}\label{elementary-properties}

\subsection{Elastic binary (particle-particle) interaction:}\label{items elastic}
The standard glossary associated to the collision of
two particles located at positions in space centered  at points $x$ and $x_*$ in $\R^d$ with velocities  $v$ and $v_*$ encountering a collisional interaction and leaving with post collisional velocities $v'$ and $v'_*$, consists on

${\bf 1.}$  Their relative velocities are 
\begin{equation*}
u:= v - v_* \  \ \text{  \it before, \ \ \ and }\
u':= v' - v'_*  \ \  \text{\it  after the binary interaction.} 
\end{equation*}

${\bf 2.}$\  The unitary vector $\eta:=\frac{x-x*}{|x-x_* |}\in\mathbb{S}^2$ is the {\bf impact direction}, and $\frac{u\cdot \eta}{|u|} = \cos(\pi-\phi)=-\cos\phi$, with {\it impact angle} $\phi$.
\smallskip%

${\bf 3.}$ \ The unitary vector  $\sigma:= \frac{u'}{|u|}\in\mathbb{S}^2$ is the {\bf scattering direction} given by 
 the {\em post-collisional relative velocity} $u'$ with respect to an elastic collision, from which $|u|=|u'|$ naturally follows.
\smallskip

${\bf 4.}$  \ Consequently, the  identity $\cos \theta= \frac{u\cdot \sigma}{|u|} = \frac{u\cdot u'}{|u|^2}$ defines the {\it scattering angle} $\theta$.   The relationship  between impact and scattering angles is $2\phi = \pi - \theta.$ 
\smallskip

%
In addition,  when  the angles $\phi=0$ or $\theta=\pi$,  the interaction is called a {\it head-on} or {\it knock on} collision relevant for hard to infinite (short range) intramolecular  potentials. When $\phi=\pi/2$ or $\theta=0$, it is  
 called a {\it grazing} or {\em glancing} collision. The later one is relevant for soft (long range) interactions such as in the modeling of Coulomb interactions.

\subsection{Post-collisional velocities using the impact direction}
From items 1 and 2 from above, the elastic binary interaction laws are fully determined by local conservation of momentum and kinetic energy following from collisional specular reflection laws with respect to the impact direction $\eta$, that is 
\begin{equation}\label{eq:6.1}		
{u\cdot \tau = u' \cdot \tau}\ \text{ and }\
{u'\cdot\eta = -u\cdot \eta} \quad \text{with}\ \ \tau\cdot\eta=0\,, \qquad \text{implying}\quad u' = u - 2( u\cdot\eta) \eta\,,
\end{equation}
and the equivalent relations for the pairs $\{v', v'_*\}$ determined by  by their pre interactions ones $\{v,v_*\}$ by 
\begin{equation}\label{eq:6.3}		
\begin{array}{l}
v' = v - \left( (v - v_*) \cdot\eta\right) \eta  \ =\  v - (u\cdot \eta)  \eta\\
\noalign{\vskip6pt}
v'_* = v_* + \left( (v - v_*) \cdot\eta\right) \eta  \ =\  v_* +(u\cdot \eta)  \eta \,.
\end{array}
\end{equation}


\nn Elastic binary interactions mean that relations \eqref{eq:6.1} and \eqref{eq:6.3} are
{\bf reversible}

\subsection{Post-collisional velocities using the scattering direction and the notation $u^{\pm}$.}  Equivalently, the relation between pre to post collisional velocities can also be expressed in the center of mass $V=\frac{v+v_*}2$ and relative velocity $u=v-v_*$ framework through the scattering direction $\sigma$, defined in item 3 above which in light of the latter equation in \eqref{eq:6.1} reads 
\begin{equation*}
\sigma = \widehat{u} - 2( \widehat{u}\cdot\eta) \eta \in \mathbb{S}^{2}\,,
\end{equation*}
by the interaction law
\begin{equation}\label{eq:9.3-Vu}		
\begin{array}{l}
v' \ =\  V +\frac12 |u|\sigma  \ =\  v + \frac12(-u+|u| \sigma) =: v +u^- = v_* + u^+\, ,\\
\noalign{\vskip6pt}
v'_* \ = \ V - \frac12 |u|\sigma   \ =\  v_* +  \frac12(u-|u| \sigma)  =: v_* - u^- = v- u^+\,  ,
\end{array}
\end{equation}
after introducing the notation   $u^{\pm} = \frac{ |u|\sigma \pm u}{2}$. Therefore, setting the vectors $u^-=v'-v$, while $u^+=v- v'_*$, clearly  $u^{-}+u^{+}=u'=|u|\sigma$. It follows that  $|u|=|u'|$ and  clearly  $u^{-}\cdot u^{+}=0$  by direct computations.\\

\nn We left to the readers to check that  these  elastic binary constitutive relations  are  equivalent to the following {\em local conservation} identities
\begin{align}\label{eq:8}
  v + v_* &=  v' + v'_* \ ,  \qquad \qquad \text{ center of mass  or  local momentum}  \\ 
  |v|^2 + |v_*|^2 &= |v'|^2 + |v'_*|^2\ , \qquad \text{ local energy.}  \nonumber
\end{align}
making indistinguishable the `pre' from the `post' interaction at the particle velocity level.  \\

\begin{figure}
	\begin{center}
		\begin{tikzpicture}[scale=2.5]
		\node (v) at (1.93,0.52) [preaction={shade, ball color=violet},
		,circle,scale=1.5] {};
		\node (vs) at (-1.93,-0.52) [preaction={shade, ball color=white!70!gray},circle,scale=1.5, ,
		] {};
		\node (vp) at (1.414,-1.414)
		 [preaction={shade,ball color=violet},circle,scale=1.5, ,pattern=crosshatch dots, pattern color=black] {};
		\node (vps) at (-1.414,1.414) 
		[preaction={shade,ball color=white!70!gray}, circle,scale=1.5,pattern=crosshatch dots, pattern color=black ] {};
		\node (V) at (0,0)	 [circle,scale=1, shade,ball color=green] {};
		\node (start) at (0,-3)  [circle,scale=0.1, gray] {};
		
		\begin{scope}[>=latex]
		\draw[->,line width=0.4ex,gray!80!black] (V) -- (0.6071,-0.6071) node[near end, above,black] {$\sigma$};
		\draw[->,line width=0.2ex] (start) --  (vs) node[left, pos=1.07] {$v_*$};
		\end{scope}
		\begin{scope}[>=latex]
		\draw[->,line width=0.4ex,gray!80!black] ( V) -- ([turn]345:0.8cm) node[midway, above=-0.7cm, black] {$\eta$};
  		\draw    pic["$\theta$", draw=red, ->,  angle eccentricity=1.2,  angle radius=1cm]{angle=vp--V--v};
		\end{scope}
		\begin{scope}[>=latex, on background layer]
		\draw[densely dotted,color=green!50!black!50!white,line width=0.2ex] (0,0) circle (2cm);
		\draw[densely dotted,->,color=green!30!black!70!white,line width=0.2ex] (vps) -- (vp) node[densely dotted, color=green!30!				black,pos=0.25,above right=-0.2cm,line width=0.2ex] {};
		\draw[->,color=violet!80!white, line width=0.2ex] (start) -- (v)  node[pos=1.0,right=0.4cm] {\color{violet}$v$};
		\draw[->,line width=0.2ex] (vs) -- (v) node[near end,sloped, black,above left] {$u$};
		\draw[densely dashed,->,color=green!80!black,line width=0.2ex] (start) -- (V) node[color=green!50!black, above=0.03cm] {$V\!=\!			\frac{v\!+\!v_*}{2}$};
		\draw[->,color=blue!60!white!40!,line width=0.2ex]  (vps)--(vp) node[color=green!30!black, right=-1.2cm, near start, sloped, above] 			{\color{blue}$\!\!\! \! \ \ \ \  u'\!=\!|u|\sigma$};
		\draw[densely dotted,->,color=violet!50!black!50!white,line width=0.2ex] (start) -- (vp) node[color=violet!50!black,below right,line 		width=0.2ex] {$ v'\!=\!V + \frac{1}{2} \left|u\right| \sigma$};
		\draw[densely dotted,->,color=black!45!white,line width=0.2ex] (start) -- (vps) node[color=black!90!white, above=0.1cm, left=0.32cm] 
		{$ V \!-\! \frac{1}{2} \left|u \right| \sigma\!=\!v'_*$};
		\draw[->,color=red,line width=0.2ex] (vp) -- (v) node[
		midway, sloped, below]{\color{red} {\scriptsize\color{red}$ v\!-\!v'\!=\!-u^-\!:=\! -\frac12 (-u\!+\!\left|u\right| \sigma  \!) $}};
		\draw[->,color=red,line width=0.2ex] (vps) -- (vs) node[
		midway, sloped, above] {\scriptsize\color{red} {\color{red}$\tiny\ v_*\!-\!v'_*\!=\! u^-\!:=\!\frac12 (-u\!+\!\left|u\right| \sigma  \!) $}};
		\draw[->,color=red,line width=0.2ex] (v) -- (vps) node[midway, sloped, above]
		{\color{red} {\scriptsize \color{red}$\qquad v_*\!-\!v\! =-u^+\!:= \! -\frac12(u \!+\!|u|\sigma)$ }};
		\draw[->,color=red,line width=0.2ex] (vs) -- (vp) node[pos=0.01cm, below right=0.08cm,  
		 sloped, above] {\scriptsize\color{red} {\color{red}$\qquad \ \ u^+\!:= \!\frac12 (u \!+\! |u|\sigma)$}};
		\end{scope}	
		\end{tikzpicture}
	\end{center}
	\caption{elastic interaction diagram and coordinates	} 
	\label{coll-diagram}
\end{figure}
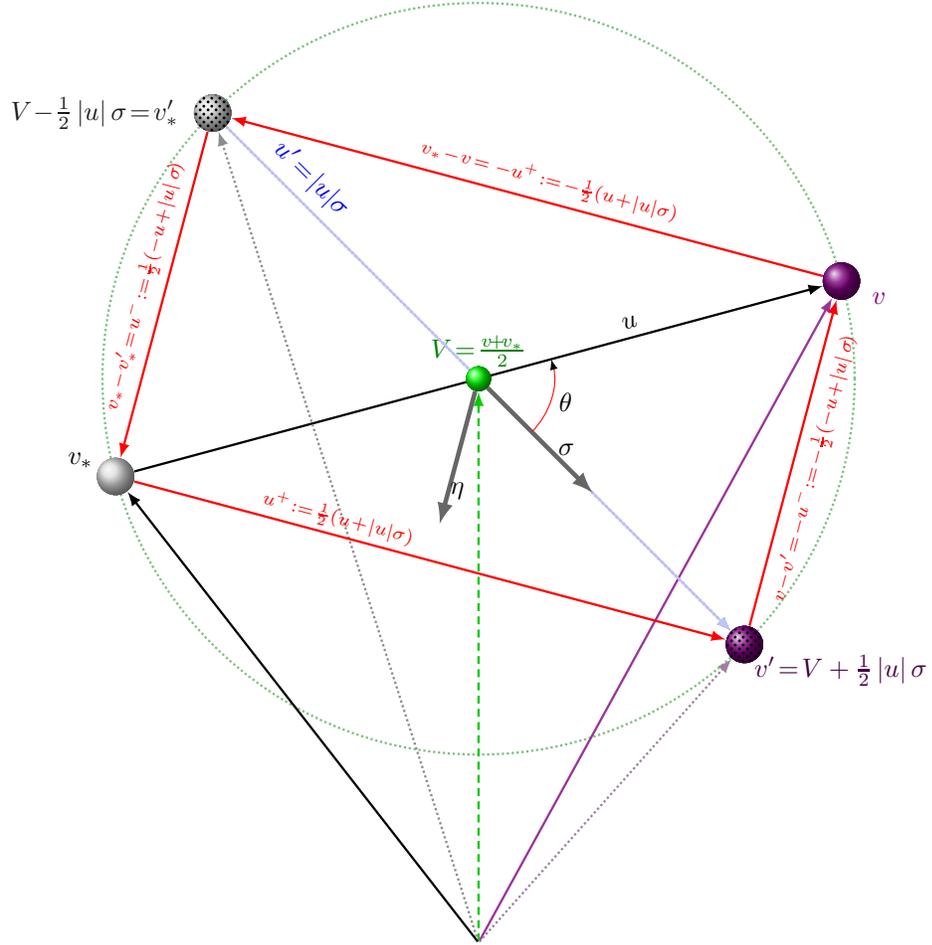
%
%
%
This wording induces the notion of time associated to collisional events, in terms of what is the conformation of  velocity states pairs,   before or after the interaction.
In particular,  while the notation exchange is admissible for binary interacting  pairs
$\{v', v'_*\}$ by $\{\vprime,\vprime_*\}$ for as long the collision is {reversible} or elastic, it will  soon become transparent that  this time  reversibility will be lost  at the level of  Boltzmann equation due  to the Molecular Chaos Assumption which distinguishes time. 
Because of properties to be studied soon, this previous comment prompts the immediate quite useful representation given in Section~\ref{pre-notation} on the distinction in the notation of {\it pre} and {\it post} interacting velocities though out this manuscript.\\

\subsection{Time reference notation}\label{pre-notation}
Relations \eqref{eq:6.3} express the post collisional velocities in terms of pre collisional velocities.  It is possible to reverse such presentation to express the pre collisional velocities $\{\vprime,\vprime_*\}$ and its relative velocity $\uprime$ in terms of the post collisional velocities $\{v,v_*\}$ and its relative velocity $u$.  In the impact direction $\eta$ framework the expression for the relative velocity is given, in light of \eqref{eq:6.1}, by the relations 
\begin{equation}\label{eq:9.1}
{u\cdot\tau = \uprime\cdot\tau}\,,\qquad
{\uprime\cdot\eta = -u\cdot\eta}\,, \quad \text{ and }\quad \uprime = u - 2(u\cdot \eta)  \eta\,.   
\end{equation}
Now, using equations \eqref{eq:6.3} and \eqref{eq:9.1}, the pair $\{\vprime,\vprime_*\}$ can be written as
\begin{equation}\label{eq:9.3}		
\begin{array}{l}
\vprime \ =\  v +  ( \uprime \cdot \eta )\eta   \ =\  v - ( u \cdot \eta)  \eta\\
\noalign{\vskip6pt}
\vprime_*\  = \  v_* + ( \uprime \cdot \eta )\eta   \ =\  v_* - (u\cdot \eta)  \eta . 
\end{array}
\end{equation}
Comparing expressions \eqref{eq:6.3} and \eqref{eq:9.3} one observes the reversibility of the collisional laws, being the only difference the conceptual interpretation of the pair $\{v,v_*\}$; for the former represents pre collisional velocities and for the later the pos collisional velocities.\\

Consequently, the pre collisional velocities representation in terms of the post collisional velocities in the scattering direction $\sigma$ framework is identical to \eqref{eq:9.3-Vu}, namely for $u^{\pm} = \frac{u \pm |u|\sigma }{2}$
\begin{equation}\label{eq:9.3-Vu-post}		
\vprime \ = \ v \,+ u^-  ,\qquad \vprime_* \ =\  \ v \, - u^+ \,,\qquad \text{so}  \ \ \uprime :=\vprime-\vprime_*= -|u|\sigma
\end{equation}
The difference lies in the interpretation of the $\sigma$ direction is reversed with respect of the direction of  $\uprime$,  or simply $'\!\sigma=-\widehat{\uprime}$ in light of the last equation in \eqref{eq:9.1}.

\medskip

While this observation may be irrelevant in the definition of local  elastic interactions, {\em it actually is very significant for the derivation of the Boltzmann equation  as it sets the stability properties for the Boltzmann flow}.  In addition, this notation is naturally relevant for the study of the inelastic kinetic theory, where   {\em local inelastic interaction are not reversible,} as shown in \cite{GPV04}



\smallskip

\vspace{8pt}
\noindent The following lemma rigorously calculates the Jacobian transformation for  elastic interactions  case, and its determinant value, in the impact direction $\eta$ framework, which is
a crucial point in the derivation of the Boltzmann equation.
\bigskip

\begin{lem}\label{lemma1}
 \noindent  After an elastic interaction, the determinants of the Jacobian of the variables transformation for the elastic interactions for a pair $\{v,v_*\}$ viewed as either before or after the interaction, are both the same for fixed $\eta\in\mathbb{S}^{2}$.  That is $ \det{J_{(v',v'_*)/(v,v_*)}}\ =\ -1 \ =\ 
 \det{J_{(\vprime,\vprime_*)/(v,v_*)}}$.
\end{lem}
\begin{proof}
In order to show this property,   we use the exchange of coordinates defined in 
(\ref{eq:6.1} - \ref{eq:6.3}) for  $(v_,v_*)$ before the interaction, or (\ref{eq:9.1} - \ref{eq:9.3})  after the interaction.  Indeed starting from  \eqref{eq:6.3} for a fixed unitary vector $\eta$, we can subtract the equations to obtain for $u=v - v_*$ and $v_*$, that
\begin{equation}\label{eq:10.1.1}		
u' = u- 2\left( u \cdot\eta\right) \eta\,,\hspace{1cm} v'_* = v_*  +(u\cdot \eta)  \eta\,.
\end{equation}

\medskip

\nn Since  $J_{(u',v'_{*})/(u,v_{*})} = J_{(v',v'_{*})/(v,v_{*})}$, it suffices to find the Jacobian for the linear transformation \eqref{eq:10.1.1}.  This transformation can be simply written as

{\large \begin{equation}\label{eq:10.1.2}
\begin{bmatrix}
u' \\
v'_*
\end{bmatrix} \ = \
\begin{bmatrix}
\mathcal{R}_{\eta} &0 \\
\mathcal{P}_{\eta}  & \text{1}_{3\times 3}  \\
\end{bmatrix} \, . 
\begin{bmatrix}
u\\
v_* 
\end{bmatrix} \, ,
\end{equation}}
where $\mathcal{R}_{\eta}$ is the simple reflection along $\eta$ and $\mathcal{P}_{\eta}$ the simple projection along $\eta$.  In particular, 
{\large
\begin{equation}\label{eq:10.1.2}
\det J_{_{{ (u',v'_*)}/{(u,v_*)} } }\ =  \ \det \mathcal{R}_{\eta}  \ = \  -1 \, .
\end{equation}  }

\nn  Due to the nature of elastic interaction laws, a similar proof (left for the reader) shows 
  {\large \begin{equation}\label{eq:10.1.2}
 \det J_{_{{(\vprime,\vprime_*)}/{(v,v_*)} } }\ =      J_{_{ {(\uprime,\vprime_*)}/{(u,v_*)} } }   \ =\  -1 \, .
\end{equation} 
}
\end{proof}

\subsection{The collision kernels}
Transition probability rates quantifying the scattering rates of  collisional transfer between before and after the interactions as  a function of the relative speeds and   impact or scattering directions,  are often referred  collision kernels taking the form
\begin{equation}\label{coll-ker}                             
 B\left(|u|, \frac{|u\cdot\eta|}{|u|}\right) 
\end{equation}
where $|u|$ the relative speed determined by the intramolecular potentials,  and the angular transition function dependence on $\hat u\cdot\eta$ the impact (or  $\hat u\cdot\sigma$ scattering) angle determined by the strength of the collision with respect to the angular dependence of such interaction, with $\hat u$ denoting the renormalized unit vector in the direction of $u$

\medskip

\nn It is common  to assumed that, for instance, $ B(|u|, \hat u\cdot\eta)$ takes the form 
\begin{equation}\label{coll-ker-2}                             
 B(|u|, \hat u\cdot\eta)= \Phi_\gamma(|u|) \, b(\hat u\cdot\eta)\, ,
\end{equation}
where $\Phi(|u|)$ is called the  potential calculated from intramolecular potential laws for binary interactions,  and  $b(\hat u\cdot\eta)$ the angular transition probability or  differential cross section.
They are usually assumed to be power laws but this is not necessary. The growth  of $\Phi(|u|)$ and  integrability  of $b$ conditions are the most important properties to consider since they sensibly affect the fundamental characteristics of solutions to the equation.  Typical collision kernels and designations found in the literature are, for example,  given by a potential transition rate function $\Phi_\gamma(|u|)$  to be non-negative and to satisfy the following  $\gamma$-order homogeneity condition
$c_{\Phi}  |u|^\gamma \le  \Phi_\gamma (|u|)  \le C_{\Phi} |u|^\gamma,$, with $c_{\Phi}$ and $C_{\Phi}$  positive constants, independent of the relative velocity  $u$. The potential rate $\gamma$ is a crucial characteristic of, not only the physics modeling of the interaction, but also the profound impact in the analytical properties of the Boltzmann flow.  
 
These rates are commonly classify by referring to these potential rates as to  {\it hard sphere model} if $ \gamma=1$;   as {\it hard potentials model}  if $0< \gamma<1$; as {\it  Maxwell type interactions model}  if $ \gamma=0$ and $ c_{\Phi}= C_{\Phi}$ constant;
as  {\it soft potentials model}  if $-d< \gamma<0$;  and  as { \it Coulomb potential model}  if  $\gamma^{-d}$, in which case the Boltzmann model fails to be well posed problem if  the angular transition $b(\hat u\cdot\eta)\in L^1(\mathbb {S}^{d-1}).$\\

\subsection{Scattering direction representation for the collision transition probability or cross section} 
The angular coordinates exchange from the impact direction $\eta$ to the scattering direction $\sigma$ is carefully described as follows.  Let $u= v-v_*$ the relative velocity associated to an elastic  interaction and let $P$ be the orthogonal plane to $\widehat{u}$.  Taking spherical coordinates to represent the coordinate system $\{\widehat{u},P\}$ generating $\mathbb{R}^{d}$,  set $\{r , \phi,\ep_1,\ldots,\ep_{d-2}\}$, where $r>0$ is the radial coordinate, $\phi \in [0,\pi)$ the polar angle (being the zenith defined by $\widehat{u}$), and $\{\ep_1,\ldots,\ep_{u-2}\}$ the $d-2$ azimuthal angular variables associated to $P$, then,
%
 since the scattering direction  $\sigma$ is unit vector in the direction deflecting from specular reflection of $u$ in the plane orthogonal to the $\eta$ direction,  it is clear that $|u| \sigma = u - 2(u\cdot \eta)\eta = u'_{elastic}$ is the specular reflection relation,  or, equivalently, 
$\sigma = \hat{u} - 2(\hat{u}\cdot\eta)\eta\,.$

Gathering all the angular information, the Boltzmann collisional integral is modeled 
in terms of  the scattering direction  $\sigma$ instead of  the impact direction $\eta$, as a consequence it is needed to  perform the exchange of coordinates $(\eta, v_*) \to (\sigma,v_*)$ for  both $\eta$ and $ \sigma$ in $\mathbb{S}^{d-1}$,  proceeding as in the following steps.

First, starting form  item 2 in subsection~\ref{items elastic}, and using that $
\hat u\cdot \widehat P=0$,  the impact $\phi$ and scattering $\theta$ angles are related through the identities 
\begin{equation}\label{cos theta sigma}
\cos\phi = -\hat{u}\cdot\eta\quad\text{and }\quad \cos\theta = \hat{u}\cdot\sigma = 1-2|\hat{u}\cdot\eta|^2= 1- 2\cos^2\phi\,.
\end{equation}  

\medskip

  From the second identity in \eqref{cos theta sigma}, clearly the cosine of the impact angle is 
\begin{equation}\label{scat-direc.2}    
\cos\phi = \big|\hat u\cdot \eta\big| \ = \ \sqrt{ \frac{1 - {\hat u\cdot\sigma} }{2} }  \, ,
\end{equation}
 and the unitary impact $\eta$  is represented  by the decomposition in the impact spherical  framework 
$\eta = -\cos(\phi)\widehat{u} + \sin(\phi)\widehat{P}$.

%
%

Next,  Under this framework,
  the differential of area element of  impact direction $\eta\in \mathbb{S}^{d-1}$  
as a function of the reference  coordinates $(\phi, \ep)$ is given by 
\begin{equation}\label{angle.2}
{\rm d}\eta =  \sin^{d-2} (\phi)\, {\rm d}\phi\,{\rm d}\ep\,, 
\end{equation}

 Correspondingly, the differential of area element of  scattering direction $\sigma$  
as a function of the impact direction $\eta$, in the reference  coordinates $(\phi, \ep)$, is given by 
\begin{equation}\label{angle.3}
{\rm d}\sigma = 2 \sin^{d-2} (2\phi)\, {\rm d}\phi\,{\rm d}\ep\,. 
\end{equation}

As a consequence of the relations \eqref{angle.2} and \eqref{angle.3}, 
\begin{equation}\label{angle.4}
{\rm d}\sigma = \frac{2 \sin^{d-2}(2\phi)}{\sin^{d-2}(\phi) } \, {\rm d}\eta = 2^{d-1} \cos^{d-2}(\phi) \, {\rm d}\eta \, . 
\end{equation}
%
%
 %

\medskip

\nn In particular, in the case   of hard spheres collisional symmetric kernel in three dimensions $B(|u|, {\hat u\cdot\eta} ) = (u\cdot\eta)_+$, relation \eqref{angle.4} yields the identity $(u\cdot\eta)_+  {\rm d}\eta\, =  |u| \cos\phi\,  {\rm d}\eta\, =  \frac{|u|}4\, {\rm d}\sigma$.  That is,
for any integrable function $g$ on the sphere $\mathbb{S}^{2}$ 
\begin{equation}\label{angular.1}
 \int_{\mathbb{S}^{2}} \ (u\cdot\eta)_+\, g\big((u\cdot\eta)\eta\big) {\rm d}\eta =
 \frac{ |u|}4 \int_{\mathbb{S}^{2}_{+}} g\Big( \frac{u-|u|\sigma}{2} \Big) {\rm d}\sigma  \, .
\end{equation}
For more general kernels in $d$-dimension, the exchange of angular coordinates becomes
\begin{equation}\label{angular.2}\begin{split}
\int_{\mathbb{S}^{d-1}}   g\big((u\cdot\eta)\eta\big)\, B_{\eta} (|u|, (\widehat u\cdot\eta)_{+} ) {\rm d}\eta = \int_{\mathbb{S}^{d-1}} g\Big( \frac{u-|u|\sigma}{2} \Big) \, B_{\sigma}(|u|,\widehat{u}\cdot \sigma )\, {\rm d}\sigma \,,
\end{split}\end{equation}
where, from relation \eqref{angle.4}  and \eqref{scat-direc.2} in $d$-dimensions
\begin{equation}\label{angular.3}
B_\sigma (|u|,{\widehat u\cdot \sigma} ) \ =\ \frac{1}{2^{d/2}} \big( 1 - \widehat{u}\cdot\sigma \big)^{-\frac{d-2}{2}}\,B_\eta \Big( |u|, \sqrt{\tfrac{1-\widehat{u}\cdot\sigma}{2}}\Big)\,.
\end{equation}

\bigskip

\subsection{Integrability of the angular transition} 
The Grad's cut-off assumption \cite{Grad_1958} is simply the condition that the transition probability  kernel $B(|u|,\widehat{u}\cdot\sigma)\in L^\infty(\mathbb{S}^{d-1})$, that is  the angular transition density being bounded in $\widehat{u}\cdot\sigma$ with $\sigma\in \mathbb{S}^{d-1}$. While some literature has referred as to Grad's cut-off assumption having an  integrable the angular part, this manuscript will not use the {\it cut-off} wording.   
More precisely, this   work extends all known results for the existence, uniqueness and long time stability  to the weaker condition of assuming  that the angular transition density  is integrable, exactly meaning $B(|u|,\widehat{u}\cdot\sigma)\in L^1(\mathbb{S}^{d-1}, d\sigma), $ without any constraints on the pointwise behavior of  $B(\cdot,y)$ with respect to $y$.  That is, 
\begin{align}\label{grad-cut-off.1}
&\int_{\mathbb{S}^{d-1}} B\big(|u|,{\hat u\cdot \sigma}\big) \,{\rm d}\sigma = \omega_{d-2}   \int_{-\pi}^\pi B (|u|,\cos(\theta))\sin\theta^{{d-2}} \,{\rm d}\theta  \ =\ \omega_{n-2}   \int_{-1}^1 B (|u|,z)({1-z^2})^{\frac{d-3}2} \,{\rm d}z  
\end{align}
is an a.e. finite function of $|u|$, with $\omega_{d-2}$ the surface measure of the $d-2$ dimensional sphere.  In general, the integrand on \eqref{grad-cut-off.1} is not bounded.  In the cases when $B(|u|,\cos(\theta)) \sim |u|^\gamma \,  b(\cos(\theta))$
the integrability assumption simply becomes
$\int_0^\pi b(\cos(\theta)) \sin^{d-2} \theta\, {\rm d}\theta\  < \infty.$

The techniques proposed in this manuscript to overcome the  lack of boundedness of  $b(\hat u\cdot \sigma)$ at  pointwise behavior in 
$\hat u\cdot \sigma$, is overcome by decompositions qualifying the amount of  mass of the function $b(\hat u\cdot \sigma)$ at points of singularities.

This section is discuss under the assumption  that there, is at least, a positive solution $f$ of the elastic Boltzmann equation that has the regularity and decay at infinity properties that are needed to perform all the operations that follow.
 There are several fundamental  properties of the {\em elastic collision integral operator} $Q(f,f)((x,v,t)$ when focusing on it as a function of  velocity space $v$.  A more general form of the Boltzmann equation is given by, for $d\ge 2$
\begin{equation}\label{eq:elementary2}
{Q(f,f)(v) = \int_{\real^d} \int_{\mathbb{S}^{d-1}} \Big( \frac1{{J}}\fprime\fprime_* - ff_* \Big)  \,B(|u|, \hat u\cdot\eta)
d\eta dv_*} \, .
\end{equation}
where  $u=v-v_*$ is the relative molecular velocity for an interacting pair  $\vprime = v- \eta (\eta\cdot u)$, $\vprime_* = v_* + \eta(\eta\cdot u)$, respectively, whose exchange of molecular  velocity  coordinates transformation from pre to post binary interaction is given by the Jacobian  ${{J}}= | \det J_{(v,v_*)/(\vprime,\vprime_*)} |$.   In addition,  if $B(|u|, \hat u\cdot\eta)\in L^1(\mathcal S^{d-1}, d \sigma)$ then 
\begin{equation}\label{eq:elementary2}
Q(f,f)(v) = \int_{\real^d\times \mathbb{S}^{d-1}} \frac1{J}f(\vprime)f(\vprime_*) \,B(|u|, \hat u\cdot\eta)
d\eta dv_* \ -\ f(v)  \nu(v),
\end{equation}
where $\nu(v) =   \int_{\real^d\times \mathbb{S}^{d-1}}f(v_*)  \,B(|u|, \hat u\cdot\eta)
d\eta dv_* $ is the collision frequency.  Of course, in the case of elastic interactions Lemma \ref{lemma1} shows that $J=1$.

\subsection{Weak or Maxwell formulation}\label{weak-maxell-form}
The weak or Maxwell formulation of the Boltzmann collision operator is a fundamental tool in the basic  analysis of the dynamics of the Boltzmann flow.  In fact moment estimates are obtained by averaging properties of the collisional integral when multiplied and integrated by a test function in $v$-space.

\smallskip
\noindent
We start introducing the symmetrized form of the collisional integral \eqref{BTE} is defined as
\begin{equation*}\label{eq:sym-4.1}
Q_s (f,g) (v):=\frac12 \iiint_{\real^d\times\times \mathbb{S}^{d-1}}
\left[ \frac1J\left( \fprime\gprime_* + \gprime\fprime_* \right) - fg_* - f_*g\right] \,   B(|u|, \hat u\cdot\eta )\,{\rm d}\eta {\rm d}v_* \,,
\end{equation*}
with, clearly,  $Q_s(f,f)=Q(f,f)$.  The objective is to use the interchange of velocity and angular coordinates to yield an integral formulation where the exchange of molecular velocity laws only appears in the test function.  Thus, let $\varphi (v)$ be a scalar value function $\varphi(v) :\real^d \to\real$, it holds that
\begin{equation}\label{eq:4.1}  	
\begin{split}
\int_{\real^d} & Q_s (f,g) (v) \varphi (v)\, {\rm d}v\\
& = \frac12 \iiint_{\real^d\times\real^d \times \mathbb{S}^{d-1}}
\left[ \frac1J\left( \fprime\gprime_* + \gprime\fprime_* \right) - fg_* - f_*g\right]  \varphi (v)\,   B(|u|, \hat u\cdot\eta )\,{\rm d}\eta {\rm d}v_* {\rm d}v\,. 
\end{split}
\end{equation}
\begin{thm}\label{Maxwellwf} The weak or Maxwell formulation of the collisional integral, written in symmetric form is given by
\begin{equation}\label{maxwell-weak}  	
\begin{split}
\int_{\real^d}  & Q_s (f,g) (v) \varphi (v)\, {\rm d}v\\
& = \frac14 \iiint_{\real^d\times\real^d \times \mathbb{S}^{d-1}}
\left( fg_* + f_*g\right)  \big( \varphi (v')+\varphi (v'_*)- \varphi (v)-\varphi (v_*) \big) B(|u|, \hat u\cdot\eta ) \,{\rm d}\eta {\rm d}v_* {\rm d}v\,,
\end{split}
\end{equation} 
where  now $v'$ and  $v'_*$ are the {\bf post-collisional velocities} to $v$ and $v_*$, respectively. 
\end{thm}
\begin{proof} Note that $J\, {\rm d}\vprime_* {\rm d}\vprime = {\rm d}v_* {\rm d}v$.  Therefore, for the terms related to the gain collisional part it holds that
\begin{equation}
\begin{split}\label{eq:4.2}  	
A_1 :&= \frac12  \iiint_{\real^d\times\real^d \times \mathbb{S}^{d-1}}\frac1J \big(  \fprime_*\,\gprime + \gprime_*\,\fprime\big)\varphi (v) B(|u|, \hat u\cdot\eta )\,{\rm d}\eta {\rm d}v_* {\rm d}v\\
&=\frac12  \iiint_{\real^d\times\real^d \times \mathbb{S}^{d-1}} \big(  \fprime_*\,\gprime + \gprime_*\,\fprime\big)\varphi (\vprime -  ( \uprime \cdot \eta )\eta) B(|'\!u|, -\widehat {'\!u}\cdot\eta ){\rm d}\eta {\rm d}\vprime_* {\rm d}\vprime\,.
\end{split}
\end{equation}
We used, for the second equality, \eqref{eq:9.3} and \eqref{eq:9.1} to express $v$ and $u\cdot\eta$ in terms of the pre collisional velocities and the fact that $|u|=|'\!u|$ due to the elastic nature of the interaction.  Consequently, performing the change $\eta\rightarrow-\eta$ it follows that
\begin{equation}
\begin{split}\label{eq:4.21}  	
A_1 &=\frac12  \iiint_{\real^d\times\real^d \times \mathbb{S}^{d-1}} \big(  \fprime_*\,\gprime + \gprime_*\,\fprime\big)\varphi (\vprime -  ( \uprime \cdot \eta )\eta) B(|'\!u|, \widehat {'\!u}\cdot\eta ){\rm d}\eta {\rm d}\vprime_* {\rm d}\vprime\\
&=\frac12  \iiint_{\real^d\times\real^d \times \mathbb{S}^{d-1}} \big(  f_*\,g + g_*\,f\big)\varphi (v') B(|u|, \widehat {u}\cdot\eta ){\rm d}\eta {\rm d}v_* {\rm d}v\,,
\end{split}
\end{equation}
where in the latter we simply renamed the pre collisional variables as $(v,v_*)$.  Of course, $v'$ is the post collisional velocity defined in \eqref{eq:6.3}.  Additionally, performing the interchange $v\leftrightarrow v_* $ which produces the interchange $\vprime\leftrightarrow \vprime_* $ and, then, the change $\eta\rightarrow-\eta$ it follows that
\begin{equation}
\begin{split}\label{eq:4.22}  	
A_1 &=\frac12  \iiint_{\real^d\times\real^d \times \mathbb{S}^{d-1}} \big(  f_*\,g + g_*\,f\big)\varphi (v'_*) B(|u|, \widehat {u}\cdot\eta ){\rm d}\eta {\rm d}v_* {\rm d}v\,.
\end{split}
\end{equation}
Then, adding \eqref{eq:4.21} and \eqref{eq:4.22} it holds that
\begin{equation}\label{eq:4.23}
A_1 =  \frac14\iiint_{\real^d\times\real^d \times \mathbb{S}^{d-1}} \big(  f_*\,g + g_*\,f\big)\big(\varphi(v') + \varphi (v'_*)\big) B(|u|, \widehat {u}\cdot\eta ){\rm d}\eta {\rm d}v_* {\rm d}v\,.
\end{equation}
For the term related to the loss part one simply uses the interchange $v\leftrightarrow v_*$ and, then, the change $\eta\rightarrow-\eta$ to obtain that
\begin{equation}
\begin{split}\label{eq:4.24}  	
A_2 :&=\frac12  \iiint_{\real^d\times\real^d \times \mathbb{S}^{d-1}} \big(  f_*\,g + g_*\,f\big)\varphi (v) B(|u|, \widehat {u}\cdot\eta ){\rm d}\eta {\rm d}v_* {\rm d}v\\
&=\frac12  \iiint_{\real^d\times\real^d \times \mathbb{S}^{d-1}} \big(  f_*\,g + g_*\,f\big)\varphi (v_*) B(|u|, \widehat {u}\cdot\eta ){\rm d}\eta {\rm d}v_* {\rm d}v\,.
\end{split}
\end{equation}
Then, adding the two lines in \eqref{eq:4.24} one has that
\begin{equation}\label{eq:4.25}  	
A_2 =\frac14  \iiint_{\real^d\times\real^d \times \mathbb{S}^{d-1}} \big(  f_*\,g + g_*\,f\big)\big(\varphi (v) + \varphi(v_*)\big) B(|u|, \widehat {u}\cdot\eta ){\rm d}\eta {\rm d}v_* {\rm d}v\,.
\end{equation} 
The result follows after subtracting \eqref{eq:4.25} from \eqref{eq:4.23}.
\end{proof}
The following corollary is a direct consequence for the quadratic case $g=f$.
\begin{cor}\label{coro4}   {\bf  Weak form or Maxwell form  of the collisional integral}
\begin{align}
\begin{split}\label{eq:maxwell}
\int_{\real^3} Q(f,f) \varphi (v) \,{\rm d}v
&= \frac12\iint_{\real^d\times\real^d }
ff_*  \left(\int_{\mathbb{S}^{d-1}} \big(\varphi (v') + \varphi (v'_*) - \varphi (v) - \varphi (v_*)\big)
B(|u|, \hat u\cdot\eta ) {\rm d}\eta \right) {\rm d}v_* {\rm d}v\\
&= \frac12\iint_{\real^d\times\real^d }f(v)f(v-u) G_{i}(v,u) {\rm d}u {\rm d}v\, ,  \quad i=1,2, \quad \text{with} \\
G_{1}(v,u) :&=  \int_{\mathbb{S}^{d-1}} \big( \varphi (v') + \varphi (v' - u') - \varphi (v) - \varphi (v-u) \big)
B(|u|, \hat u\cdot\eta ) {\rm d}\eta\,,\\
G_{2}(v,u):&= 2\int_{\mathbb{S}^{d-1}} \big( \varphi (v') - \varphi (v) \big)
B(|u|, \hat u\cdot\eta ) {\rm d}\eta\,.
\end{split}
\end{align}
\end{cor}
\begin{remark} Corollay \ref{coro4} exhibit a {\sl  weighted  double convolution form} for the weak form of the collisional integral.
This structure is due to the bilinear (or multilinear) mixing associated to these operators. 
Such structures are   actually  very convenient forms for analytical and computational purposes related to the Boltzmann equation.
\end{remark}
Something very unique associated to elastic interactions is the micro irreversibility of collisions, that is, $v' = \vprime$ and $v'_* = \vprime_*$.  A consequence of such irreversibility and Theorem \ref{Maxwellwf} is the following symmetric weak form which consequently produce the dissipation of entropy for the Boltzmann model.
\begin{cor}\label{coro5}{\bf  Symmetric weak form or Maxwell form of the collisional integral}
\begin{align*}	
\int_{\real^d} Q(f,f) \varphi(v){\rm d}v
& = \frac14 \int_{\real^d\times\real^d\times \mathbb{S}^{d-1}}
(f'f'_* \!-\! ff_*) \big(\varphi(v)\!+\varphi(v_*) -  \varphi(v'_*) - \varphi(v'_*) \big)B(|u|, \hat u\cdot\eta ){\rm d}v_*{\rm d}v{\rm d}\eta\,.
\end{align*}
\end{cor}
\begin{proof}
Due to irreversibility $\varphi (v') + \varphi (v'_*) = \varphi (\vprime) + \varphi (\vprime_*)$.  Recall also that $J=1$ due to elasticity, consequently one has that
\begin{align*}
\int_{\real^d}& Q(f,f) \varphi(v){\rm d}v\\
&=\frac12\iint_{\real^d\times\real^d\times\mathbb{S}^{d-1} }
ff_* \big(\varphi (v') + \varphi (v'_*) - \varphi (v) - \varphi (v_*)\big)
B(|u|, \hat u\cdot\eta ) {\rm d}\eta {\rm d}v_* {\rm d}v\\
&=\frac12\iint_{\real^d\times\real^d\times\mathbb{S}^{d-1} }
ff_* \big( \varphi (\vprime) + \varphi (\vprime_*) - \varphi (v) - \varphi (v_*)\big)
B(|u|, \hat u\cdot\eta ) {\rm d}\eta {\rm d}v_* {\rm d}v\\
&= \frac12\iint_{\real^d\times\real^d\times\mathbb{S}^{d-1} }
f'f'_* \big( \varphi (v) + \varphi (v_*) - \varphi (v') - \varphi (v'_*)\big)
B(|u|, \hat u\cdot\eta ) {\rm d}\eta {\rm d}v_* {\rm d}v\,.
\end{align*}
The result follows adding the first and last equalities.
\end{proof}

\subsection{Collision invariants  of $Q(f,f)$ }

It is very easy to see  from \eqref{eq:maxwell}  that 
testing the elastic collisional integral with any linear combination of test function 
  $\varphi(v)\in \text{span}\{ 1, v, |v|^2 \}$,  with  arbitrary coefficients $A,B$ and $C$, the elastic interacting pairs $(v,v_{*})$ and $(v',v'_{*})$ satisfying
\begin{equation}\label{elastic-col}
\begin{aligned}
v+v_* \  & = \ v'+v'_*   \qquad \ \ \   \text{local momentum conservation} \\
|v|^2 +|v_*|^2 \ & = \  |v'|^2+|v'_*|^2     \qquad \text{local energy conservation,}   
\end{aligned}
\end{equation}
cleary, by Corollary \ref{coro4}, the following equivalence hold
\begin{equation}\label{eq:collision2}		
\int_{\real^d} Q(f,f) \varphi (v) \,{\rm d}v = 0\quad\text{ if and only if }
\quad {\varphi (v'_*) + \varphi (v') = \varphi (v) + \varphi (v_*)}, 
\quad \text{a.e. in} (v,v_*).
\end{equation}
In particular, any $\varphi (v)$ satisfying \eqref{eq:collision2} for any vectors $(v,v_{*})$ and $(v',v'_{*})$ satisfying \eqref{elastic-col},  is referred to as a {\em collision invariant}.
The three invariants of the collisional integral are referred as conservation of mass, momentum and kinetic energy conditions, which  follows from the transport equation when integrating in velocity space and in physical space under boundary conditions without external sources of energy. \\

Next, the {\it the Boltzmann Theorem} resolves the question on whether it is possible to show that  elastic collisional law is enough to show that, if $\int_{\real^d} Q(f,f) \varphi (v) \,{\rm d}v=0$  for any multiplier  $\varphi(v)$,  then  $\varphi(v) = A +B\cdot v + C|v|^2$ for $v\in\real^3$, under the assumptions that  $f$ is assumed non-negative, so $ff_* > 0$ in $\real^{2d} \times \mathbb{S}^{d-1}$,  and the transition probability or collision cross section $ B(|u|, u \cdot \eta)=0$ has positive  measure zero in $\real^d \times \mathbb{S}^{d-1}.$  \\

\begin{thm}[The Boltzmann Theorem]\label{BoltzTheo}
The function  $\varphi(v)$ is a collision invariant if and only if
\begin{equation}\varphi (v) = \varphi (v) = A+\vec B\cdot v +  C|v|^2,\qquad \text{or all}\ v\in\real^3\,.
\end{equation}
\end{thm}

This statement was first proven by   Boltzmann in 1890 for  $\varphi$ twice differentiable, Cercignani and Arkeryd 1970 for 
$\varphi$ in $L^1_{\loc}$, and  Wennberg in 1992 for $\varphi$ just a distribution on $v$. We show here a different proof.

\begin{proof}
Let any distribution $\varphi \in \mathcal{D}^1 (\real^d)$ 
be a collision invariant satisfying the relation  \eqref{eq:collision2}
for all pairs $(v,v_*)$ and $(v' ,v'_*)$ generating an elastic interaction
\begin{equation}\label{B2}
v+v_* = v' + v'_* \qquad\text{ and }\qquad
|v|^2 + |v_*|^2 = |v'|^2 + |v'_*|^2, 
\end{equation}
%
%

Take any arbitrary $V\in \real^d$, and consider any two arbitrary vectors ${\bf v_1}$ and ${\bf v_2}$ with {\em the same magnitude}. 
These vectors can be thought to be taking an arbitrary center of mass  $V$, and two arbitrary relative velocities ${\bf v_1}:=u$ and ${\bf v_2}:=|u|\sigma$, respectively, as it will be shown next.

 Set the interacting pairs     $(v,v_*)$ and $(v' ,v'_*)$ defined  the arbitrary vectors, $V$, ${\bf v_1}$ and ${\bf v_2}$ as defined by the choice   the center of mass and relative velocity coordinates \eqref{eq:9.3-Vu},  that is  
\begin{align}\label{four-point}
(v, v_*)&= (V+{\frac{\bf v_1}2} ,V-{\frac{\bf v_1}2}) = (V+{\frac{u}2} ,V-{\frac{u}2})  \qquad \qquad\ \mathrm{ and} \nonumber \\ 
(v',v'_*)&= (V+\frac{\bf v_2}2, V-{\frac{\bf v_2}2}) = (V+{\frac{|u|\sigma}2} ,V-{\frac{|u|\sigma}2}),
\end{align}
which clearly satisfy the local conservation properties relations in \eqref{B2}, since, using  $|{\bf v_1}|=|{\bf v_2}|$,
\begin{align*}\label{ce-rv-coord}
\frac{v+v_*}2 &= \frac{v' + v'_* }2 =V,    \ \  \text{conservation of center of mass} \\
  |v|^2 + |v_*|^2 &= \ 2|V|^{2} + \frac{|{\bf v_1}|^{2}}{2} =|v'|^2 + |v'_*|^2,  
  \ \ \text{energy identity}, 
\end{align*}
Moreover,   $v'-v=\frac12 ({\bf v_2}-{\bf v_1})= \frac12(|u|\sigma -u) = -u^- $ and $v'_*-v_*=-\frac12 ({\bf v_2}-{\bf v_1})= -\frac12(|u|\sigma -u)=u^-.$\\

Next, let any function $\varphi (v)$  in $C^2 (\real^d)$ to 
satisfy the relation \eqref{eq:collision2} on the pairs $(v, v_*)$ and $(v', v'_*)$  given by \eqref{four-point}, respectively. Then, subtract 
$2\varphi (w)$ on both sides of identity \eqref{eq:collision2} and divide by $|v|^2$  the left hand side, and  by $|v_*|^2$ the right hand side, respectively (as  both arbitrary vectors $v$ and $v_*$ have the same magnitude).  The resulting identity approximates  the linear wave equation by taking second variations by finite differences 
\begin{equation}\label{fpapprox}
\frac1{|{\bf v_1}|^2} \left\{ \varphi \Big( V+\frac{\bf v_1}2 \Big)
+ \varphi \Big( V-\frac{\bf v_1}2\Big) - 2\varphi (V)\right\}
= \frac1{|{\bf v_2}|^2}\left\{ \varphi \Big( V+\frac{\bf v_2}2 \Big)
+ \varphi \Big( V-\frac{\bf v_2}2\Big) -2\varphi (V)\right\},
\end{equation}
defined for any arbitrary vectors $v$ and $v_*$ with the same magnitude $|\bf v_1|= |\bf v_2|$.  Therefore, invoking the $v_i$-directional Taylor expansions, for $i=1,2$,  and taking the limit as $|\bf v_1| = |\bf v_2| \to 0^+$,  yields the identity
\begin{equation*}
\langle \mathcal{H}(\varphi)(V)\widehat{\bf v_1},\widehat{\bf v_1}\rangle
= \langle \mathcal{H}(\varphi)(V)\widehat{\bf v_2},\widehat{\bf v_2}\rangle,
\end{equation*}
where  $\mathcal{H}(\varphi)$ is the Hessian matrix of $\varphi$ and $\widehat{y}$ is the  unit vectors in the direction of $y$.  This property suffices to show that the second derivatives of $\varphi$, in any direction, must be constant.  Indeed, choosing
\begin{equation*}
\widehat{\bf v_1} = \frac{e_{i}+e_{j}}{\sqrt{2}}\,, \qquad  \widehat{\bf v_2} = \frac{e_{i} - e_{j}}{\sqrt{2}},\quad  \text{implying}\   \ \widehat{\bf v_1}\cdot \widehat{\bf v_2} =0, \  \ \text{for any}\ \ j\neq i,  \ 1\le i,j\le d.
\end{equation*}
It follows
\begin{equation*}
\big(\partial^{2}_{ij}\varphi\big)(V) = \langle \mathcal{H}(\varphi)(V)e_{i},e_{j}\rangle = 0\, \quad\text{and}\quad 
\big(\partial^{2}_{ii}\varphi\big)(V) = \langle \mathcal{H}(\varphi)(V)e_{I},e_{I}\rangle=\langle \mathcal{H}(\varphi)(V)e_{j},e_{j}\rangle = \big(\partial^{2}_{jj}\varphi\big)(V)\,, 
\end{equation*}
and since  the vector $w\in\mathbb{R}^{d}$ is arbitrary this leads to $\partial^{2}_{ij}\varphi = 0$ and $\partial^{2}_{ii}\varphi = \partial^{2}_{jj}\varphi$ for any $i\neq j$, means that mixed derivatives vanish, and   imply $\varphi(v)$ separates as
$$ 
\varphi(v) = \sum^{d}_{i=1} f_{i}(v^{i})\,,\qquad \text{on any arbitrary}\ \  v=(v^{1},v^{2},\cdots,v^{d})\,,
$$
for some scalar functions $f_{i}:\mathbb{R}\rightarrow\mathbb{R}$. 

 Furthermore, since second pure derivatives are equal, yields
$f''_{1}(v^{1}) = f''_{2}(v^{2}) = \cdots = f''_{d}(v^{d}) = \text{constant}$.  That is, $f_{i}(v^{i}) = a^{i} + b^{i}v^{i} + c^{i} (v^{i})^{2}$, for some real numbers $a^{i},\,b^{i},\,c^{i}$.  Or equivalently,
$$\varphi(v) = \frac{C|v|^2}2 + \vec B \cdot v + A\  \text{ for any }\ v\in\real^d\ .$$
This proves the statement of the theorem for any $\varphi \in C^2 (\real^d)$.

\medskip

To complete the proof take $\varphi \in \mathcal{D}^1 (\real^d)$ a collision invariant and use the mollification $\varphi *\rho_\ep$, where $\rho_\ep \in C^\infty_o$-mollifier.  Then, since the mass and energy conserving pre and post collisional pairs $(v, v_*)$ and $(v', v'_*)$, defined in \eqref{four-point}, remain  the mass and energy conserving under any translation $\tau_z \varphi(v)=\varphi(v-z)$, then the  chosen collision invariant $\varphi(v)\in \mathcal{D}^1 (\real^d)$ satisfies
\begin{align*}
(\varphi * \rho_\ep) (v) +  (\varphi * \rho_\ep)(v_*) &= \int_{\mathbb{R}^{d}}\left(\tau_z \varphi(v) + \tau_z \varphi(v_*)\right) \rho_{\ep}(z) {\rm d}z \\
&= \int_{\mathbb{R}^{d}} \left(\tau_z \varphi(v') + \tau_z \varphi(v'_*)\right)  \rho_{\ep}(z) {\rm d}z = (\varphi * \rho_\ep) (v') +  (\varphi * \rho_\ep)(v'_*)\,.
\end{align*}

 Consequently, $(\varphi *\rho_\ep)$ satisfies the five-point approximating rule \eqref{fpapprox} for the couple of pairs defined in \eqref{four-point},  implying 
$$(\varphi *\rho_\ep) (v) = \frac{C^{\ep}}{2} |v|^2 +\vec B^{\ep}\cdot v +A^{\ep}\,, $$
Hence, $\lim_{\ep\downarrow 0} (\varphi *\rho_\ep) (v) =\varphi$ induces the limit s$A^{\ep},\,\vec B^{\ep},\, C^{\ep}$ to some $A,\,\vec B,\,C$, and so 
$\varphi (v) = \frac{C}2 |v|^2 +\vec B\cdot v +A\,.$
This last statement completes the proof.
\end{proof}

\subsection{The Boltzmann inequality and the $\cH$-Theorem}

From the weak-Maxwell formulation we can  obtain the following fundamental results for the conservative (elastic) Boltzmann collisional form
\begin{thm}[The Boltzmann inequality and the $\cH$-Theorem]
If $f$ is non-negative and $Q(f,f)\,\log f$ is integrable and
$\log f = \varphi (v)$, then 
\begin{equation}\label{eq:boltzmann}
\begin{aligned}
&{\rm i)} \quad\text{Boltzmann Inequality: }\quad
\int_{\real^d} Q(f,f)\log f\,{\rm d}v \le 0 \, , \quad \text{and}\\
&{\rm ii)}\quad \text{The $\cH$-Theorem:}\quad \ \ \displaystyle \int_{\real^d}Q(f,f)\, \log f \, {\rm d}v  = 0 \quad   \text{ iff } \quad f(v) = \exp (A+B\cdot v - C|v|^2)\,.
\end{aligned}
\end{equation}
\end{thm}

\begin{proof}
For $\varphi (v) = \log f(v)$ as a test function, from the symmetric weak Maxwell  
form Corollary~\ref{coro4}
\begin{equation}\label{eq:10.1}		
\begin{split}
\int_{\real^d} Q(f,f) \log f(v){\rm d}v
& = \frac14 \int_{\real^d\times\real^d\times \mathbb{S}^{d-1}}\mkern-48mu
(f'f'_* \!-\! ff_*) [\log f\!+\!\log f_* \!-\! \log f'_* \!-\!\log f']B(|u|, \hat u\cdot\eta ){\rm d}v_*{\rm d}v{\rm d}\eta\\
\noalign{\vskip6pt}
&= \frac14\int_{\real^d\times\real^d\times \mathbb{S}^{d-1}}\mkern-48mu
(f'f'_* - ff_*) (\log ff_* - \log f'_* f') B(|u|, \hat u\cdot\eta )\,{\rm d}v_*\,{\rm d}v\,{\rm d}\eta\\
\noalign{\vskip6pt}
& =: - \mathcal{D} (f) \le 0\qquad \text{\em Entropy dissipation rate}\,.
\end{split}
\end{equation}
The sign in last inequality holds because the $\log z$ is monotone increasing function, satisfying 
$(\log z-\log y) (z-y) > 0\; \text{ for any }\; z \ne y>0\; \text{ and }\;  (\log z-\log y) (z-y) = 0\; \text{ iff }\; z= y\,.$

Therefore, if there is an probability density functions $M(v)$ that nullifies the entropy dissipation rate, then it must satisfy
$\mathcal{D}(M) = 0$ if , and only if  $M'M'_* = MM_*$
and hence,  the Boltzmann Theorem~\ref{BoltzTheo} yields 
\begin{equation*}
  {\log M(v) = A+B\cdot v +C |v|^2}     \text{ or, equivalently }\     {M(v) = \exp (A+B\cdot v +  C|v|^2)}\,.
\end{equation*}

Finally, since any solution of the Boltzmann equation is a probability density in all  $v \in \R^d$, then its integral must be finite, 
and hence the constant $C<0$.  That means  equilibrium solution  $M(v):=f_{eq}(v)$  must be a Gaussian distribution in $v$-space.
\end{proof}
In particular, setting $\vec B= {\bf v}, \  C = -\frac{m}{2k_BT},\ \exp (A) = \left(\frac{m}{2\pi k_BT} \right)^{3/2},$
the equilibrium state $M(v)$ becomes
\begin{equation}\label{eq:boltzmann-maxwell}
M(v)=M_{m,{\bf v}, T}(v)= \left( \frac{m}{2\pi k_BT}\right)^{3/2}
e^{-\frac{m|v -{\bf v}|^2}{2k_BT}}  \qquad
\text{Boltzmann-Maxwell molecular distribution}\,,
\end{equation}
that is also  expected  to be a stable stationary state for the single particle probability density function that solves the Boltzmann equation.   This state  $M_{m(x,t),{\bf v}(x,t), T(x,t)}(v)$ is often called a  {\em Local Maxwellians} or {\em  Local Equilibrium Statistical States (L-ESS)} for  the case of the spatially inhomogeneous problem. 

Additionally, when these  states,  solve both the transport and the collision part, that is
\begin{equation}
\partial_t {M} + v\cdot \nabla_x {M} = 0 \ \ \ \text{and}  \ \ \ Q({M},{M})=0 
\end{equation}
they are called  {\em Global Maxwellians} or  {\em  Global Equilibrium Statistical States (G-ESS)}.  If the parameters  $\bf u$ and $T$ are time independent  (may still depend on $x$-space) the  state $f_{eq}$  is referred as {\em Stationary Statistical Equilibrium  States (SESS)}.  {\em SESS} are stable and the H-Theorem is crucial for proving such stability as much as the functional structure of the solution space and the prescribed spatial boundary conditions.   In addition the {\em Global Maxwellians} are natural  {\em similarity} solutions   and can be proved to be attractors of the flow, even when they still are time dependent.  In such case one can referred to them as {\em stable Global Equilibrium Statistical States}.

\subsubsection{Entropy functionals and thermalization} Define the Boltzmann entropy functional as
\begin{equation}\label{eq:Htheorem.5.1}
\cH (f) = \int_{\real^d} f\, \log f\,{ d}v\,,
\end{equation}
and the relative Boltzmann entropy functional, defines for an equilibrium Maxwellian $M(v)$, meaning the for any $f(v)\in L^1_2(\mathbb{R}^d)\cap L\log L)(\mathbb{R}^d)$, such that $\int_{\mathbb{R}^d} f(v,t) \log M(v) dv= \int_{\mathbb{R}^d} f_0(v) \log M(v) dv$
\begin{equation}\label{eq:Htheorem.5.2}
\cH (f\, | \,M)  :=    \int_{\real^d} \tfrac{f}{M(v)} \log\Big(\tfrac{f}{M}\Big) M(v) dv =\int_{\real^d} f\, \log\Big(\tfrac{f(v,t)}{M(v,t)}\Big)\,{\rm d}v =  \int_{\real^d}f \log f dv -\int_{\mathbb{R}^d} f_0(v) \log M(v) dv.
\end{equation}
Then, from  calculation \eqref{eq:10.1} and the $\mathcal{H}$-theorem, it follows
\begin{equation}\label{eq:Htheorem.6}
\frac{{\rm d}}{{\rm d}t}\cH (f\,|\,M)  = - \mathcal{D} (f) \le 0\,,
\end{equation}
where the entropy dissipation functional was given in the calculation \eqref{eq:10.1}.

\begin{remark}
This last theorem stresses the {\em irreversible in time} nature of the
Boltzmann equation, even though we started with an $N$-particle collisional
system with {\em time reversible elastic collisions}.
\end{remark}

The reason for this ``loss of reversibility'' lies on the molecular chaos
assumption: ``velocities of the particles {\em about to collide} are
uncorrelated.''
combined that we model the gain operator as acting on the pre-collisional
velocities.
Looking back at the Boltzmann inequality \eqref{eq:10.1}
that presets the sign of $(\log z-\log y) (z-y)$ expression that defines
$\mathcal{D}(f).$\\

{
This preliminary section concludes with the anticipation of new results sprouting from the work in this manuscript. 
More precisely, the addressing of the question of   whether is  
possible to use \eqref{eq:Htheorem.6} to conclude, at least for solutions $f(v,t)$ of the homogeneous Boltzmann problem, that the long time thermalization limit  
$\lim_{t\to\infty^+ }f(v,t) = M_{\rho, \rho u,\rho e} (v)$ occurs 
occurs in some functional sense, and there this is a Banach space that not only the limit exits but also it is possible to establish the decay rate to zero of the $f(v,t) - M_{\rho, \rho u,\rho e} (v)$ in a related functional sense.
While  question has been addressed en the last decades, the following statements offer an anticipation of how to improve existing results.

}
\smallskip

%
%
%
%

%
%

Starting from the key step in solving this problem, discussed and proven in \cite{TV, Villani}, is the inequality
\begin{equation}\label{entropy-dissipation-control}
c_f\,\mathcal{H}(f|M)^{\alpha}\leq \mathcal{D}(f)
\end{equation}
with $M$ the thermal equilibrium associated to $f$ and for some explicit constants $c_f>0$ and $\alpha>1$ depending on properties of the solution $f$ itself.  Of particular importance is the appearance of higher moments, discussed in a later section, since the tail of a solution plays a central role in the equilibrium convergence rate which is related to $\alpha$.  Interestingly the aforementioned inequality is valid for soft and hard potentials under similar regularity properties of $f$, but in practice soft potential models have much ticker tails which can produce very slow thermalisation rates \cite{Carlen}.  We refer to \cite{Villani} for a extensive discussion of this inequality including an illuminating short draft of the proof.\\

Continuing with the discussion, from \eqref{eq:Htheorem.6} and \eqref{entropy-dissipation-control} we are led to the inequality
\begin{equation*}
\frac{{\rm d}}{{\rm d}t}\cH (f|M)   + c_{f}\,\cH (f|M)^{\alpha} \le 0\,,
\end{equation*}
which in turn leads to
\begin{equation}\label{eq:Htheorem.16}
\cH (f(t)|M) \leq \frac{\cH (f_0|M)}{\big(1+(\alpha-1)\cH (f_0|M)^{\alpha-1}\,c_f\,t\big)^{\frac{1}{\alpha-1}}}\,.
\end{equation}
A elegant application of the {\em Cszisar-Kullback inequality}, developed for information theory, shows that
\begin{equation}\label{eq:Htheorem.17}
\| f - M\|_{L^{1}(\real^d)} \leq \int_{\real^d} |f-M|\,dv \le \sqrt{2\,\cH (f \mid M)}\,. 
\end{equation}
Therefore, from \eqref{eq:Htheorem.16} it follow that
\begin{equation*}
\| f(t) - M\|^{2}_{L^{1}(\real^d)}  \ \leq \frac{2\,\cH (f_0|M)}{\big(1+(\alpha-1)\cH (f_0|M)^{\alpha-1}\,c_f\,t\big)^{\frac{1}{\alpha-1}}}\,.
\end{equation*}
So $f(v,t) \xrightarrow[t\to\infty]{} M(v)$ strongly in $L^1(\real^d)$ at algebraic rate $\frac{1/2}{\alpha-1}$.

\section{The Cauchy problem for Maxwell and hard potentials}\label{hardpotsection}

A natural space to solve the homogeneous Boltzmann equation is the space of probability measures.  In particular, we develop the theory in the space of integrable functions, whose observables, or moments, may be bounded. 
from the functional viewpoint, we recall  classical weighted Banach space $L^1_\ell(\real^d)$ associated to statistical processes describing the evolution weighted integrable functions,  defined by 
\begin{equation}\label{k-norm}
L^{1}_{\ell}(\real^d)\  = \  \Big\{ \ g \; \text{measurable} \ : \  \|g\|_{L^1_\ell} \  = \ \int_{\real^d} |g(v)| \langle v \rangle^{\ell}\,dv \ < \ \infty\Big\}\,, 
\end{equation}
any $\ell\geq0$. 

As usual, we referred as the Lesbegue weight $\langle v \rangle = \sqrt{1+|v|^{2}}$.  Note that $L^1_{\ell_2} \hookrightarrow L^1_{\ell_1}$, since  $\|f\|_{L^1_{\ell_1}} \le \|f\|_{L^1_{\ell_{2}}} $, whenever  $\ell_2 \ge \ell_1$.  

A closely related concept is the one of $k^{th}$-Lebesgue moment or observable of $g(v,t)$
\begin{equation}\label{k-moment}
m_{k}[g](t):=\int_{\real^d}\ g(v,t)\,\langle v \rangle^{2k}\ dv\,, \qquad   k \in  \overline{\mathbb{R}^+}  \,.
\end{equation}
Note that $\|g\|_{L^1_{2k}}\equiv m_{k}[g]$ whenever the scalar valued function takes only positive values.
 
We present a revision of an approach first introduced in the unpublished notes  by A. Bressan  \cite{bressan} in the context of classical Boltzmann  equations for elastic binary collisions for hard spheres in three  dimensions with angular cross section.  The approach is based on abstract ODE theory and it is used here to solve the Cauchy problem for the  Boltzmann equation for Maxwell type of interactions as well as for hard potentials for initial data with finite mass and energy, that is,
\begin{equation}\label{energy}
0\leq f_0\,,\qquad m_{1}(f_0)<\infty\,.
\end{equation}

Thus, the associated classical space homogeneous Boltzmann collisional model for elastic binary interactions induces an evolution flow in the space $L^{1}_{\ell}(\real^d)$ by setting the identity
\begin{equation}\label{k-moment-identity}
\frac{d}{dt}m_{k}[f]:=\int_{\real^d}\ Q(f,f)(v,t)\,\langle v \rangle^{2k}\ dv\,,\qquad k\geq0\,, 
\end{equation}
where the right had side may or may not  be positive.

Therefore, we focus first in an existence theory for the Cauchy problem in a general framework that allows for solution to the Boltzmann flow by solving an Ordinary Differential Inequality in the Banach space$L^{1}_{\ell}(\real^d)$, namely the following theorem will be rigorously  addressed 
after recalling the weak formulation of the Boltzmann collision operator
\begin{equation}\label{bina-weak2}
\begin{split}
\frac{d}{dt} \int_{\R^N}  f(v,t) \varphi (v)\, dv\,
=  \int_{\R^d \times \R^d}  f(v,t)\, f(v_*,t) 
G_{\varphi}(v_*,v)  dv_*\, dv\, ,
\end{split}
\end{equation}
for any test function $\varphi(v,t)$, where  the binary weights function $G_{\varphi}(v_*,v) $ is determined by the $\sigma$-average of the $\mathbb{S}^{d-1}$ sphere, decomposed in a possitive and negativeparts.
\begin{equation}\label{weightuv2}
\begin{split}
&G_{\varphi}(v_*,v)  := \left(G_{\varphi}^+ - G_{\varphi}^-\right)(v_*,v):=\int_{\mathbb{S}^{d-1}} \big(\varphi (v') + \varphi (v'_*) - \left(\varphi (v) + \varphi (v_*)\right) \big) B(|v-v_*|, \sigma) d\sigma\\
 \noalign{\vskip8pt} 
&v'=v - \tfrac12(|u|\sigma-u)\, , \qquad v'_*= v_* +\tfrac12(|u|\sigma-u)\, ,  \qquad u=v-v_*\, ,
\end{split}
\end{equation}
where the transition probability  is assumed to take the  form $B(|v-v_*|, \sigma) = \Phi(u)\,b(\hat{u}\cdot\sigma).$  The angular part of the scattering angular kernel $b:=b(\hat{u}\cdot\sigma)$ is  assumed to be integrable with respect to the measure $d\sigma$. Conditions on the potential transition rate function $\Phi_\gamma(|u|)$ are assumed to be non-negative and to satisfy the following  $\gamma$-order homogeneity condition
\begin{align}\label{pot-Phi_1}
c_{\Phi}  |u|^\gamma \le  \Phi_\gamma (|u|)  \le C_{\Phi} |u|^\gamma\, ,
\end{align}

\begin{thm}[\bf{Cauchy problem for the homogeneous binary Boltzmann equation}] \label{CauchyProblem} 
Let $\gamma\in [0,2]$, $b(\hat{u}\cdot\sigma, d\sigma)$ integrable, and initial data $f_0$ with finite mass and energy $\|f_0\|_{L^1_2}<\infty.$   

If $\gamma\in(0,2]$ and the initial data  $\|f_0\|_{L^1_{2+\eps}}<\infty$ for any arbitrary $\epsilon>0$, then there exists a unique nonnegative function
\begin{equation*}
f(v,t) \in \mathcal{C} \big( [0,\infty); L^{1}_{2}( \mathbb{R}^d ) \big) \cap  \mathcal{C}^1\big( (0,\infty);L^{1}_{\ell}(\mathbb{R}^{d})\big)\,,\quad \forall\ 
 \ell\ge 2,
\end{equation*}
solving  
\begin{equation}\label{collision3}
\left\{\begin{array}{l}
f_t = Q(f,f) =  \int_{\R^d \times S^{d-1}} \big( \fprime\, \fprime_*  - f\, f_*)\, \Phi(|v-v_*|)\, b( \hat u\cdot \sigma) \, d\sigma dv_* \,, \\ \noalign{\vskip6pt}
f(v,0) = f_0 (v)\,.
\end{array}\right.
\end{equation}
Furthermore, the conservation laws hold
\begin{equation}\label{collision3.2}
\int f(v,t)\text{d}v = \int f_0(v)\text{d}v,\quad \int f(v,t)v\text{d}v = \int f_0(v)v\text{d}v,\quad \int f(v,t)|v|^{2}\text{d}v =\int f_0(v)|v|^{2}\text{d}v\,.
\end{equation}
 In the case $\gamma\in[0,2]$,   $\|f_0\|_{L^1_{\ell}}<\infty$, then    $f(v,t) \in \mathcal{C} \big( [0,\infty); L^{1}_{\ell}( \mathbb{R}^d ) \big) \cap  \mathcal{C}^1\big( (0,\infty);L^{1}_{\ell}(\mathbb{R}^{d})\big)$ for all  
 $\ell\ge 2.$
\end{thm}

The success of such global existence solely relies in two fundamental sufficient conditions. One is the upper control of the positive part associated to the collision operator $k^{th}$-moment in the right hand side of identity \eqref{k-moment-identity} by a linear or constant form in the solution $k^{th}$-moment, whose constant parameters may only depend of the conserved constant quantities  $m_0[f]$ and $m_1[f_0]$, mass and Lebesgue energy respectively.  This upper control is obtained firstly by an energy identity that splits the local energy conservation in a convex combination of the potential energy and the center of mass; and secondly by calculating the $\sigma$-averages of the positive part of $G^+_{\langle v\rangle^\ell}= \int_{\mathbb{S}^{d-1}} (\langle v\rangle^\ell+\langle v_*\rangle^\ell)b( \hat u\cdot \sigma) d\sigma$, which    decays in $\ell$, proportionally to  $\ell^{th}$-Lebesgue energy. These estimates, inspired in the work of Bobylev \cite{boby97},  extend the classical Povzner estimates to a large larger class of angular transition functions $b(\hat{u}\cdot\sigma)$.  Such upper estimates are now referred  by Angular Averaging Lemma for the control of the positive contributions to the weight  $G_{\langle v\rangle^k}$. 

The second  fundamental sufficient condition is a lower bound.  We look into two possible venues. Under the posed assumption on the angular averaging of $b(\hat{u}\cdot\sigma)$ being finite, the first one simple  looks into a  lower pointwise estimate of the negative contribution $G^-_{\langle v\rangle^\ell}= \|b\|_{L^1(\mathbb{S}^{d-1})}\, (\langle v\rangle^\ell+\langle v_*\rangle^\ell)$. 
This estimate is sufficient to prove Theorem~\ref{CauchyProblem}. 
However, to work out the propagation of the $L^p(\real^d)$ theory $1\le p\le\infty$ we need a stronger version of the the lower bound, which can be obtain by a functional inequality for elements in $L^1(\real^d)$ that, under suitable conditions  controls the convolution of such an element by $\langle v\rangle^\gamma$, by a finite constant proportional to  $\langle v\rangle^\gamma$, for $0<\gamma\leq 2$. This is all worked out in  the Lower Bound Lemma~\ref{lblemma}. 

These estimates also enable the control of propagation and generation of exponential tails.

Further, we will show that, if  the initial data satisfies   $ \| f_0\|^1_{\ell} \cap   \| f_0\|^p_{\ell}<\infty$, with $1<p\leq\infty$, then this property propagates in time.
This result requires an extra moment requirement $m_{1+\epsilon}(f_0)<\infty$, for some $\epsilon>0$ to secure a lower bound that warrants   a coerciveness property  needed to show the $  \| f\|^p_{\ell}(t)$ boundedness for all time (i.e. propagation of initial `regularity'.)  This is the  crucial Lemma~\ref{lblemma}  shown in  in subsection~\ref{lower-bounds}.

\smallskip

While the wellposedness theory of the Boltzmann equation has been addressed by several authors \cite{Ca57, Arkeryd72}, with the most general result is given in \cite{Misch-Wenn}, where the initial data is assumed only with finite mass and energy,  their arguments are based on the moment propagation and generation theory \cite{Elmroth, Desvillettes, WennbergMP,boby97}   that is based on an estimate of integration on each $|v'|^\ell$ and  
$|v'_*|^\ell$ weights  evaluated in the postcollisional molecular velocities proportional to the angular transition  element $b( \hat u\cdot \sigma) d\sigma$.
All these results assumed that not only the function $b( \hat u\cdot \sigma)$ was bounded but also needed to introduced  a cut-off near $\hat u\cdot \sigma=0$.    Such estimates are referred as Povzner lemma and was originally introduced by Povzner \cite{Povzner} in 1960s and the show that such angular integration is proportional to $(|v| + |v_*|)^\ell$-power.    
  That estimate yields the propagation and generation of moments for which an iteration map, constructed by Arkeryd in \cite{Arkeryd72} that yields compactness sufficiently enough to construct a unique solution with $2^+$-moments. 
     
In 1997 a groundbreaking technique was introduced by Bobylev   \cite{boby97}, applied to the hard sphere model in three  dimensions with a constant angular function $b( \hat u\cdot \sigma)$. This new approach proposes that uses the  local energy conservation  to write the integration of postcollisional molecular velocities in the center of mass-relative velocity framework.  That way the integration  readily proves a control of the decay rate of  such angular averaging for the case of hard spheres in three dimensions with a constant angular cross section. 
An important step to obtain evolution estimates for moments is a fundamental lemma that controls the function $G_{\varphi}(v_*,v)$  Later, it was extended in \cite{BGP04} for inelastic collisions.  More extension related to the generality of $b$ where made later on in \cite{GPV09,Alonso-Lods-SIMA10}.  Finally, in \cite{LuMouhot,TAGP} the lemma was extended to non cutoff scattering.\\

\subsection{Lower bound estimates}\label{lower-bounds} 
A functional estimate for lower bound of convolution function with concave power law functions can be obtain  by the following lower bound estimate, that will become essential later on the calculation of propagation of  $L^p$-Banach norms for any $1<p\leq \infty$.  This is a fundamental result related to a lower control for the convolution operator associated to the collision frequency corresponding to the negative contributions of the moments associated to the collision operator $Q(f,f)(v)$.  This a new result, developed for probability densities and independent of being solution of the Boltzmann equation. A first notion of this lower bound was introduced by the authors for the analysis of stability and error estimates for numerical approximations to the Boltzmann equation by conservative spectral schemes as shown in \cite{AlonsoGThar}, and it is reproduced here in a more analytical framework. Modification of this results were employed in to obtain  existence results and qualitative results  for the one dimensional dissipative Boltzmann type equation for polymers \cite{ABCL:2016}, and more  recently used in the well-posedness of of solution of a gas mixture modeled by a system of Boltzmann equation with disparate masses \cite{IG-P-C-mixtures}  and for polyatomic gas taking into account the exchange of internal energies \cite{IG-P-C-poly-2020} \\

\subsubsection{Pointwise estimates}
 Conditions on the potential transition rate function $\Phi_\gamma(|u|)$ are assumed to be non-negative and to satisfy the following  $\gamma$-order homogeneity condition
\begin{align}\label{pot-Phi_1}
c_{\Phi}  |u|^\gamma \le  \Phi_\gamma (|u|)  \le C_{\Phi} |u|^\gamma\, ,
\end{align}
for which,   invoking convex and concave inequalities for non negative scalar quantities the term  $|u|^\gamma $ can be  estimated by the Lebesgue weights $ \langle v\rangle^{\gamma}:= (1+|v|^2)^{\gamma/2}$ from above and below, respectively,   as follows
\begin{equation}\label{pot-Phi_2.0}
      C_{\gamma} \langle v\rangle^{\gamma} - \langle v_{*}\rangle^{\gamma} \leq  \Phi_\gamma(|u|)  \leq \tilde C_{\gamma} \left( \langle v\rangle^{\gamma} + \langle v_{*}\rangle^{\gamma}\right)
\end{equation}
with the lower and upper bound constants $C_{\gamma} =c_{\Phi} \min\{ 2^{-\gamma/2}, 2^{(2-\gamma)/2}\}$ and 
 $\tilde C_{\gamma} =
C_{\Phi} \max\{ 2^{\gamma/2}, 2^{\gamma/2-1}\}$, respectively,  with $c_{\Phi}\leq C_{\Phi}. $   

Since in this work we focus on values of $\gamma \in(0,2]$, we only consider the potential part of the transition probability  for   $C_{\gamma} =c_{\Phi}  2^{-\gamma/2}$ and $\tilde C_{\gamma}= C_{\Phi} 2^{\gamma/2}$  to be 
\begin{align}\label{pot-Phi_2}
     c_{\Phi}  (2^{-\gamma/2}\, \langle v\rangle^{\gamma} - \langle v_{*}\rangle^{\gamma}) \leq  \Phi_\gamma(|u|)  
     \leq  C_{\Phi} 2^{\gamma/2}\left( \langle v\rangle^{\gamma} + \langle v_{*}\rangle^{\gamma}\right)\, .
\end{align}

In particular, the  lower constant bound $c_{\Phi}$   is part of the evaluation of the coercivity property associated to the proof of existence of global in time Cauchy problem solution of the classical elastic Boltzmann equation for binary elastic interactions in the Banach space $C(0,\infty, L^1_k(\mathbb{R}^d) )$, $k>2$, as much as the characterization of propagation and generation of exponential high energy tails   associated to  solutions in $ L^1(\mathbb{R}^d)$.

\subsubsection{A functional  lower bound estimate for convolution of functions in $L^1_k(\mathbb{R}^d)$ spaces with concave power law functions}\label{coercive-2}
This is a functional  estimate in $L^1_\ell(\mathbb{R}^d)$ Banach spaces that provides a fundamental  necessary control   to obtain  a coercive estimate that enables  global in time estimates  to the Boltzmann flow in  $L^p_\ell(\mathbb{R}^d)$, for $1< p \le \infty$. 
It consists in finding a lower bound inequality  for 
 any probability density function   $f(x,\cdot,t) \in L^1_\ell(\mathbb{R}^d)$,  in  $v\in \mathbb{R}^d$ and $\ell>2^+$,  convolved with a concave potential   function $\Phi(v)$ satisfying a condition like  \eqref{pot-Phi_1} by   the $\gamma$-Lebesgue polynomial weight proportional to  a  constant that depends on the semi-norms  $\dot{L}_\ell^1(\mathbb{R}^d)$, for $\ell=0,1,2$ and $2+\beta$, for any $\beta>0$. The proof, as it will be shown next, is independent of the the function being a solution of the Boltzmann flow,  controls from below the collision frequency term $\nu(f)(v,t)$  associated to a potential rate   $\Phi(|u|)$, i.e. the convolution of any function $f(v,t)$.

{
\medskip
 \begin{remark}
This result actually convolves the transition probability rate, or collision kernel, with a probability density satisfying the condition of the Lower Bound Lemma~\ref{lblemma}.   
It must be noted that this  Lemma was not used for the existence theory. It is not needed since the transition probability can be  estimated pointwise  using \eqref{pot-Phi_1} and \eqref{pot-Phi_2}.  However the reader may have noticed that the existence of solutions  relied of the moment estimates of the collisional integral shown in Theorem\ref{mom-coll-op},  proven  in this Section~\ref{hardpotsection}, and the application an existence Theorem~\ref{Theorem_ODE}  and Lemma~\ref{Lemma_ODE-extension2} making use of the   upper solutions of the moments ordinary inequalities  as global in times bounds found in   Theorem~\ref{propagation-generation} to obtain solutions in $ \mathcal{C}\big([0,\infty); L^{1}_{2}(\mathbb{R}^{d}) \big)\ \cap\ \mathcal{C}^{1}\big((0,\infty); L^{1}_{\ell}(\mathbb{R}^{d}) \big),$ with $2\le \ell<2(1+\gamma), $ solving  the Cauchy problem with initial data $f_0\in L^{1}_{2+\epsilon}(\mathbb{R}^{d}).$ 
In particular the strategy developed in this manuscript, produces a solution with 
lower upper global bounds for large times than  the one obtained from the classical upper global bound estimates obtained by Bernoulli's equations, as written in \cite{Desvillettes} and \cite{WennbergMP}.    Thus,  while this lower bound is not strictly necessary for the argument given in the proof of Theorem~\ref{thm:exp-moment-gen},  it considerably simplifies its proof without sparing the explicit characterization of the exponential rates, which is one of the goals of this review.
  
  Nevertheless, we found the Lower Bound Lemma necessary for study of  $L^{2}$-integrability propagation, that   allows for  the considerably improvement and completeness  properties of  $L^p_\ell(\real^d)$ propagation norms  proof,  for $p\in (0,\infty]$,  from the original argunent developed in \cite{MV04}.
\end{remark}
}
\begin{lem}[\bf{Lower Bound Lemma}]\label{lblemma}
Fix $\gamma\in(0,2]$, and assume $0\leq  f:=f(\cdot, v)$ satisfies
\begin{equation*}
C\ge \max\left\{ \int_{\mathbb{R}^{d}}f\, dv, \, \int_{\mathbb{R}^{d}}f \,|v|^{2}dv  \right\}  \geq 
\min\left\{ \int_{\mathbb{R}^{d}}f\, dv, \, \int_{\mathbb{R}^{d}}f\,|v|^{2}dv  \right\}   \ge  c >0\,,\quad \int_{\mathbb{R}^{d}}f\,v\,dv=0\,,
\end{equation*}
for some positive constants $C$ and $c$.  Assume also that for some $\beta>0$
\begin{equation*}
0<\int_{\mathbb{R}^{d}}f\, |v|^{2+\beta}dv\leq B_{\beta}. 
\end{equation*}
\bigskip
\nn Then, there exists $c_{lb}:=c_{lb}(B,C,c, \beta,\gamma)>0$ such that
\begin{equation}\label{lower-bound}
\big(f \ast_v|\cdot|^{\gamma}\big)(v)\geq c_{lb}\langle v \rangle^{\gamma}\,.
\end{equation}
with $0<c_{lb}$ explicitly defined by 
\begin{equation}\label{c-lb}
c_{lb}    := \frac {c}{4}\, 
\left(2\left(\frac{2^{2/\beta}}c \max\{C,B_{\beta}\}\langle \left(4C/c\right)^{\frac{1}{\gamma}} \rangle^{2+\beta} \right)^{\frac1\beta} \right)^{-\frac{1}{2-\gamma}  } \left( 1+\big(4C/c\big)^{\frac{2}{\gamma}}\right)^{-\frac{\gamma}2} 
\end{equation}
\end{lem}

\begin{proof}
The case $\gamma=0$ is trivial, thus, assume $\gamma\in(0,2]$.  Take $r>0$, $v\in B(0,r)\subset{\mathbb{R}^{d}}$, the open ball centered at the origin and radius $r$. 
Note that for any $R>0$,
\begin{align}\label{lbe1}
\begin{split}
\int_{\{|v-w|\leq R\}}f(\cdot,w)|v-w|^2&\text{d}w=\int_{\mathbb{R}^d}f(\cdot,w)|v-w|^2\text{d}w-\int_{\{|v-w|\geq R\}}f(\cdot,w)|v-w|^2\text{d}w\\
&\geq c\,\langle v \rangle^{2} - \int_{\{|v-w|\geq R\}}f(\cdot,w)|v-w|^2\text{d}w\\
&\geq c\, \langle v \rangle^{2} - \frac{1}{R^{\beta}}\int_{\{|v-w|\geq R\}}f(\cdot,w)|v-w|^{2+\beta}\text{d}w\,.
\end{split}
\end{align}

Since $|v| < r$, then 
\begin{equation*}
\int_{\{|v-w|\geq R\}}f(\cdot,w)|v-w|^{2+\beta}\text{d}w\leq 2^{1+\beta}\max\{C,B_{\beta}\}\langle v \rangle^{2+\beta}\leq 2^{1+\beta}\max\{C,B_{\beta}\}\langle r \rangle^{2+\beta}.
\end{equation*}

Next, taking    $R=R(r,c, C ,\beta)$ sufficiently large to satisfy \\
\begin{equation}\label{R-clb}
\frac{2^{1+\beta}}{R^{\beta}}\max\{C,B_{\beta}\}\langle r \rangle^{2+\beta}\leq \frac{c}{2}\,, 
\end{equation}\\
\nn or, equivalently,\\
\begin{equation}\label{largeR}
2\left(\frac{2^{2/\beta}}c \max\{C,B_{\beta}\}\langle r \rangle^{2+\beta} \right)^{\frac1\beta} \leq {R}\,, 
\end{equation}\\
\nn it follows from estimate \eqref{lbe1} that\\
\begin{equation}\label{lbe2}
\int_{\{|v-w|\leq R\}}f(\cdot,w)|v-w|^2\text{d}w\geq c\, \langle v \rangle^2 - \frac{2^{1+\beta}}{R^{\beta}}\max\{C,B_{\beta}\}\langle r \rangle^{2+\beta}\geq \frac{c}{2}\,, \quad \forall\; v \in B(0,r)\,.
\end{equation}
Hence,  one can also infer from \eqref{lbe2} that, for any $\gamma\in(0,2]$ and $R$ satisfying \eqref{largeR}, it follows\\
\begin{align}\label{lbe3}
\begin{split}
\int_{\mathbb{R}^d}f(\cdot,w)|v-&w|^{\gamma}\text{d}w  \geq \int_{\{|v-w|\leq R\}}f(\cdot,w)|v-w|^{\gamma}\text{d}w\\
&\geq\frac{1}{R^{2-\gamma}}\int_{\{|v-w|\leq R\}}f(\cdot,w)|v-w|^{2}\text{d}w\geq\frac{c}{2R^{2-\gamma}}\,, \quad \forall\; v \in B(0,r)\,.
\end{split}
\end{align}
Additionally, for any $v\in\mathbb{R}^{d}$ and $\gamma\in(0,2]$,
\begin{align}\label{lbe4}
\int_{\mathbb{R}^d}f(\cdot,w)|v-w|^{\gamma}\text{d}w \geq \int_{\mathbb{R}^d}f(\cdot,w)\big(\tfrac{1}{2}|v|^{\gamma}-|w|^{\gamma}\big)\text{d}w\geq \big(\tfrac{c}{2}|v|^\gamma - C\big),\qquad \,.
\end{align}
As a consequence of \eqref{lbe3}  and \eqref{lbe4}, 
\begin{equation*}
\int_{\mathbb{R}^d}f(\cdot,w)|v-w|^{\gamma}\text{d}w\geq\Big(\frac{c}{2R(r)^{2-\gamma}}\,\textbf{1}_{B(0,r)}(v)+\big(\tfrac{c}{2}|v|^\gamma - C\big)\,\textbf{1}_{B(0,r)^c}(v) \Big).
\end{equation*}

Choosing $r:=r_{*}=\big(4C/c)^{\frac{1}{\gamma}}\geq1$ one ensures that $\frac{c}{2}|v|^\gamma - C\geq \frac{c}{4}|v|^{\gamma}$ for any $|v|\geq r$.  Then, combining with the election of $R=R(r_{*})\geq1$ from \eqref{largeR},
\begin{align*}
\int_{\mathbb{R}^d}f(\cdot,w)|v-w|^{\gamma}\text{d}w\geq\Big(\frac{c}{2R(r_{*})^{2-\gamma}}\,\textbf{1}_{B(0,r_{*})}(v)+\frac{c}{4}\,|v|^\gamma\,\textbf{1}_{B(0,r_{*})^c}(v) \Big)\\
\geq \frac {c}{4\,R^{2-\gamma}(r_*)}\Big( \textbf{1}_{B(0,r_{*})}(v) + |v|^\gamma\,\textbf{1}_{B(0,r_{*})^c}(v) \Big)\geq \frac {c\,\langle v \rangle^{\gamma}}{4\,R^{2-\gamma}(r_*)\langle r_{*}\rangle^{\gamma}}\,.
\end{align*}

Thus, the lower bound \eqref{lower-bound} holds with $c_{lb} $ as 
\begin{equation*}
c_{lb}  = \frac {c}{4\,R^{2-\gamma}(r_*)\langle r_{*}\rangle^{\gamma}}\,
\end{equation*}

Therefore, using that $r_{*}=\big(4C/c)^{\frac{1}{\gamma}}$ and  invoking \eqref{R-clb} to calculate $R(r_*)$, the exact value of the lower bound constant $c_{lb}$ representation \eqref{c-lb}  is obtained.
\end{proof}

\smallskip

\begin{remark}\label{poincare}
This lower estimate is a functional inequality that control a convolution with a probability density  from below. 
In this context, the constant $c_{lb}$  can be viewed as the analog to the Poincare constant fundamental to obtain  the coerciveness property needed to develop an  existence and uniqueness property in a suitable classical Sobolev  space.
\end{remark}

\smallskip
 
 \subsection{moment-estimates for binary interactions}
We start by proving the following  identity, that can be found in \cite{boby97},  giving a relation between local energies written  in the scattering angle framework  for center of mass and relative velocity coordinates from \eqref{weightuv2}.
 
 \noindent
The following elementary estimate naturally follows  from the binomial identity  

\begin{lem}\label{lem:binom}
The following estimate holds
\begin{equation}\label{povz18}
(z+y)^{q}-z^{q}-y^{q}  \le q2^{q-3}\big(z^{q-1}y + y^{q-1}z\big)\,,\qquad\text{for all} \quad  2 \le q\in\mathbb{R}.
\end{equation}
\end{lem}

\begin{proof}
Note the identity 
\begin{equation*}
(z+y)^{q}-z^{q}-y^{q} = \int_0^z \int_0^y q(q-1) (t+\tau)^{q - 2} \,d\tau\,dt\,.
\end{equation*}
For any $q\geq2$, the function $|\cdot|^{q-2}$ is convex, therefore,
\begin{equation*}
(t+\tau)^{ q - 2} \leq 2^{q-3}\big( t^{q-2} + \tau^{q-2} \big)\,.
\end{equation*}
Thus, we obtain
\begin{equation*}
(z+y)^{q}-z^{q}-y^{q} \le 
2^{q-3}\int_0^z \int_0^y q(q - 1) \big( t^{ q - 2 } + \tau^{ q - 2 } \big) \,d\tau\,dt =
q2^{(q-3)}\big(zy^{ q - 1 } + z^{ q - 1 }y \big)\,.
\end{equation*} 
\end{proof}

In addition, in order to   study of  $sk$-moments  of the solution to the Boltzmann equation with $0<s\le1$, it is also useful to invoke the following Lemma whose proof can be found in  reference \cite{BGP04}, Lemma~2. The statement is transcribe here, adjusted to the current notation  for the ease of the reader. 

\begin{lem}\label{BGP04-bino-sum} Let the {\em floor function} associated to a  real number, be denoted by $
\lfloor q \rfloor := \mbox{integer part of} \; q.$ Then, for any $1<n\in \mathbb{N}$ and $s\in (0,1]$, let  $1<sn\in\mathbb{R}$  and $J_{sn}= \lfloor \frac{sn}2\rfloor$,  the Newton's generalized  binomial expansion upper estimate for binomial forms with real valued exponents, is given by 
\begin{equation}\label{binsum}
\sum^{J_{sn}-1}_{j=1} { sn \choose j}\left(x^{sn-j}y^{j} + x^{j} y^{sn-j}\right)
\le (x+y)^{sn} - x^{sn} -y^{sn}\! \le \! \sum^{J_{sn}}_{j=1} { sn \choose j}\left(x^{sn-j}y^{j} + x^{j} y^{sn-j}\right) 
\end{equation}
where the binomial coefficients for non-integer $sn$ are
defined as
\begin{equation}\label{binsum-coeff}
 { sn \choose j}:=\frac{sn(sn-1) \dots (sn-j + 1)}{j!}, \ j\ge 1; \quad { sn \choose 0}:=1.
\end{equation}
In addition,  the case when $sn$ is an odd integer the second inequality  in \eqref{binsum}
becomes an equality, which coincides with the binomial expansion of $(x+y)^{sn}.$
\end{lem}

A natural consequence from the two inequalities in \eqref{binsum} are the following two estimates
 \begin{align}\label{binsum-2}
&\sum^{J_{sn}}_{j=0} { sn \choose j}\left(x^{sn-j}y^{j} + x^{j} y^{sn-j}\right) \le K (x+y)^{sk},  \qquad \text{for} \ 0<k\le2\nonumber \\
\text{and}\qquad&\ \\
 &\sum^{J_{sn}}_{j=1} { sn \choose j}\left(x^{sn-j}y^{j} + x^{j} y^{sn-j}\right) < \! 2^{J_{sn}}\left(x^{sn-j}y^{j} + x^{j} y^{sn-j}\right). \nonumber
\end{align}

Lemma~\ref{BGP04-bino-sum} and its follow up estimates \eqref{binsum-2} provide a  natural tool to study the $sn^{th}$-moments of collisional operator form  $Q(f,f)(v,t)$ describing a binary interaction framework, that arises
in the calculation  of  upper bounds  terms from  its positive contribution, as shown in the sequel.

The following energy identity is crucial for the development of moment estimates by means of the scattering angular averaging Theorem~\ref{BGP-JSP04}.  It is obtain writing the local energy  conservation written for the Lebesgue weight functions in the post collisional velocities, that is  $ \langle v' \rangle$  and $ \langle v'_* \rangle$   by means  of  center of mass and relative velocity coordinates expressed in the scattering direction $\sigma$.

\begin{lem}[\bf{Energy identity in the scattering  direction coordinates}]\label{PovznerI}   
For any $(v',v'_*)$ and $(v,v_*)$ satisfying \eqref{weightuv2}, let  $E:= \langle v \rangle^2 + \langle v_*\rangle^2$,  and 
$\psi(x)$ be  any function defined in $\R^+$, 
then, the following identity holds
\begin{equation}\label{Pov1}
\begin{split}
&\psi\big(\langle v'\rangle^2\big)   +  \psi\big(\langle v'_* \rangle^2\big)  = \psi\bigg( E\Big( \tfrac{1-\xi\, \hat V\cdot\sigma}2\Big)\bigg) +  \psi\bigg( E\Big( \tfrac{1+\xi\, \hat V\cdot\sigma}2\Big)\bigg)\,,\\
\noalign{\vskip6pt}
&\ \text{ with }\ \ \  \xi = \frac{2|u|\, |V|}{E}\in [0,1]  \qquad \text{and}  \qquad  \hat V=\frac{V}{|V|}\,.
\end{split}
\end{equation}
\end{lem}

\begin{proof} Setting the center of mass $V=\frac{v+v_{*}}{2}$ and relative velocity $u=v-v_*$, we recall from \eqref{weightuv2} that
\begin{equation*}
v' = V-\frac12 |u| \sigma \quad \ \ \text{and } \quad 
v'_* = V+\frac12 |u| \sigma\,.
\end{equation*}
As a consequence, we can compute as
\begin{equation*}
\begin{split}
\langle v' \rangle^{2} = 1+|v'|^2=1+\Big| V-\frac12|u| \sigma\Big|^2 & = 1+|V|^2 - |u| V\cdot \sigma +\frac{1}{4} |u|^2  = \frac{E}{2} - |u||V|\hat{V}\cdot\sigma\, ,\\
\langle v'_{*} \rangle^{2} = 1+|v'_*|^2=1+\Big| V+\frac12|u| \sigma\Big|^2 & = 1+|V|^2 + |u| V\cdot\sigma+\frac{1}{4} |u|^2 =  \frac{E}{2} + |u||V|\hat{V}\cdot\sigma\, .
\end{split}
\end{equation*}
Therefore,
\begin{equation*}
\langle v' \rangle^{2} = E\bigg(\frac{1 - \xi \, \hat{V}\cdot\sigma}{2}\bigg)\,,\qquad \langle v'_{*} \rangle^{2} = E\bigg(\frac{1 + \xi\, \hat{V}\cdot\sigma}{2}\bigg)\,,\qquad \text{where}\quad \xi:=\frac{2\,|u||V|}{E}\,.
\end{equation*}
In addition, note that
\begin{equation*}
0\leq \xi=\frac{2|u|\, |V|}{E} \le \frac{|v|^{2} + |v_{*}|^{2}}{E}\leq 1\,,
\end{equation*}
which proves the lemma.
\end{proof}
\subsection{Angular averaging regularity} This is a profound result, independent of the solution of the  Cauchy problem for the Boltzmann flow posed in Theorem~\ref{CauchyProblem}. It  provides the basic component  not only a coercive condition  that yields global stability of solutions but also the control of high energy tails or exponential moments to be discussed in Subsection~\ref{exponential-tails} .
 It has its roots in the work of  Povzner \cite{Povzner}, Desvillettes \cite{Desvillettes} and Wennberg\cite{WennbergMP}, but it was the ground breaking work of Bobylev \cite{boby97}  who introduced,  by means of the energy identity \eqref{Pov1}, the concept of angular averaging  that produces a broad and natural argument to  control the integral of the sum of convex weight functions evaluated in post-collisional velocities  proportional to the  angular transition function $\b(\hat u\cdot\sigma) \in L^1(\mathbb{S}^{d-1} )$, without any need to omit point of singularities.  While the original work in \cite{boby97} was developed  for the case of hard potentials an three dimensions (i.e. $\gamma=1$ and $\b(\hat u\cdot\sigma)$ being constant), we 
  present in this review  
ideas techniques   that are slight modifications of those  from \cite{BGP04,GPV09,Alonso-Lods-SIMA10}. These new  angular averaging approaches have been recently developed for a system of multi-component gas mixtures \cite{IG-P-C-mixtures}, and  more recently for a broad class of cross sections and transition probabilities for the case of polyatomic gases \cite{IG-P-C-poly-2020}.  
\begin{thm}[\bf{ Angular Averaging Lemma}]\label{BGP-JSP04}
Assume that the collision cross-section splits into $B(|u|, \hat u\cdot\sigma)= \Phi(|u|) \, b(\hat u\cdot\sigma)$ and let $\psi_{r} (x) = \langle x\rangle^{2r}$, with any real value $r$.  Then, there exists a positive constant $\mu_r$, referred by the {\em $r^{th}$-contractive factor}, such that  the positive component of the weight function  \eqref{weightuv2} in the weak formulation of the Boltzmann equation satisfies
\begin{equation}\label{povzner0}
G^+_{\psi_{r}}(v,v_*)  =  \int_{ \mathbb{S}^{d-1} } \big( \psi_{r} (v') + \psi_{r} (v'_*) \big)
B(|u|, \hat u\cdot\sigma)\, d\sigma \le \mu_{r}\,\Phi(|u|)\,\Big( \langle v\rangle^{2} + \langle v_{*}\rangle^{2} \Big)^{r}
\end{equation}
where 
\begin{equation}\label{mu_r}
\mu_r  := \sup_{\{\hat{V}\in \mathbb{S}^{d-1} ,\hat{u} \in \mathbb{S}^{d-1}\}}\int_{ \mathbb{S}^{d-1} }  \bigg(\left| \tfrac{1 - \hat V\cdot\sigma}2 \right|^{r} +  \left| \tfrac{1 + \hat V\cdot\sigma}2 \right|^{r} \bigg)  b(\hat u\cdot\sigma)\, d\sigma \le \|b \|_{L^1(\mathbb S^{d-1})}\,, \   \ r>1\,.
\end{equation}
$\mu_{1}=\|b \|_{L^1(\mathbb S^{d-1})}$, and $\mu_r\searrow0$ strictly as $1<r \nearrow \infty.$

Furthermore, if $b\in L^{p}(\mathbb{S}^{d-1})$, with $1<p\leq \infty$, the following estimate holds
\begin{equation}\label{povzner0.1}
\mu_r\leq 2^{(2p-1)/p} \|b\|_{L^p(\mathbb{S}^{d-1})}|\mathbb{S}^{d-2}|^{(p-1)/p}   \bigg(\frac{p-1}{p(1+r)-1}\bigg)^{(p-1)/p}  =\mathcal{O}(r^{-\frac{p-1}p}) \,.
\end{equation} 
\end{thm}

\begin{remark}\label{coercive-factor}  The factor $\mu_r$   generates  constants $0< \|b \|_{L^1(\mathbb S^{d-1})}-\mu_r < 1 $, for all $r>1$,  for the weak form  of Boltzmann flow when tested by the polynomial forms $\psi_{r}$. In particular, the main result from this Theorem~\ref{BGP-JSP04}, is  to show a priori global in time $L^1_k$ estimates to unique solutions of the Boltzmann  initial value problem, where  factor   $\|b \|_{L^1(\mathbb S^{d-1})}-\mu_r >0 $ is associated to coerciveness for the dynamics of the Boltzmann equation in a suitable Banach space and norm. See Remark~\ref{coercive-remark}  for further comments in coerciveness.
\end{remark}

\ 
\begin{proof} 
Using Lemma \ref{PovznerI} and the homogeneity of $\psi_{k}$, it follows that the positive part of the weight function \eqref{weightuv2} is controlled by  
\begin{equation*}
\begin{split}
G^+_{\psi_{k}}(v_*,v) & = \Phi(|u|) \,  \int_{ \mathbb{S}^{d-1} }\big( \psi_{r}(v')   +  \psi_{r}(v'_*)  \big) b(\hat u\cdot\sigma)\, d\sigma \\
&= \Phi(|u|)\,\Big(\langle v \rangle^{2} + \langle v_{*}\rangle^{2}\Big)^{r} \, \int_{ \mathbb{S}^{d-1} }  \bigg(\left| \tfrac{1-\xi \hat V\cdot\sigma}2 \right|^{r} +  \left| \tfrac{1+\xi \hat V\cdot\sigma}2\right|^{r}\bigg)  b(\hat u\cdot\sigma)\, d\sigma\,.
\end{split}
\end{equation*}
{Noting that  the derivative of the integral termr $\bigg(\left| \tfrac{1-\xi \hat V\cdot\sigma}2 \right|^{r} +  \left| \tfrac{1+\xi \hat V\cdot\sigma}2\right|^{r}\bigg)$ containing  $\hat V\cdot\sigma$  are non-negative, as their factor $\pm \xi$ is bounded by unity,   then}
\begin{equation}\label{Pov2}
\left| \tfrac{1-\xi \hat V\cdot\sigma}2 \right|^{r} +  \left| \tfrac{1+\xi \hat V\cdot\sigma}2 \right|^{r} \leq \left| \tfrac{1 - \hat V\cdot\sigma}2 \right|^{r} +  \left| \tfrac{1 + \hat V\cdot\sigma}2 \right|^{r}\,.
\end{equation}
Therefore,
\begin{align}\label{Pov0}
G^+_{\psi_{r}}(v_*,v) &\leq \Phi(|u|)\,\Big( \langle v\rangle^{2} + \langle v_{*}\rangle^{2}\Big)^{r} \int_{ \mathbb{S}^{d-1} }  \bigg(\left| \tfrac{1 - \hat V\cdot\sigma}2 \right|^{r} +  \left| \tfrac{1 + \hat V\cdot\sigma}2 \right|^{r}\bigg)  b(\hat u\cdot\sigma)\, d\sigma\\
&\leq \mu_r\,\Phi(|u|)\,\Big( \langle v \rangle^{2} + \langle v_{*} \rangle^{2}\Big)^{r}\,,\qquad r>1\,.\nonumber
\end{align} 
Thus, we define the {\sl $r^{th}$-contractive factor } 
\begin{equation*}
\mu_r  := \sup_{\{\hat{V}\in \mathbb{S}^{d-1} ,\hat{u} \in \mathbb{S}^{d-1}\}}\int_{ \mathbb{S}^{d-1} }  \bigg(\left| \tfrac{1 - \hat V\cdot\sigma}2 \right|^{r} +  \left| \tfrac{1 + \hat V\cdot\sigma}2 \right|^{r} \bigg)  b(\hat u\cdot\sigma)\, d\sigma < \|b \|_{L^1(\mathbb S^{d-1})}\,,\qquad  \text{for any} \ r>1\,.
\end{equation*}
%
Clearly,
\begin{equation*}
\mu_{1} =  \sup_{\{\hat{V}\in \mathbb{S}^{d-1} ,\hat{u} \in \mathbb{S}^{d-1}\}} \int_{ \mathbb{S}^{d-1} }  \bigg(\left| \tfrac{1 - \hat V\cdot\sigma}2 \right| +  \left| \tfrac{1 + \hat V\cdot\sigma}2 \right|\bigg)  b(\hat u\cdot\sigma)\, d\sigma = \|b \|_{L^1(\mathbb S^{d-1})} 
\end{equation*}
In addition, for all $r>1$, the function
\begin{equation*}
\mu_r(\hat{V},\hat{u}):= \int_{ \mathbb{S}^{d-1} }  \bigg(\left| \tfrac{1 - \hat V\cdot\sigma}2 \right|^{r} +  \left| \tfrac{1 + \hat V\cdot\sigma}2 \right|^{r}\bigg)  b(\hat u\cdot\sigma)\, d\sigma\,,
\end{equation*} 
is continuous in the vectors $\hat{V}$ and $\hat{u}$. Such continuity can be easily verify by  using polar coordinates with zenith $\hat{u}$.  Moreover, the integrand  \eqref{Pov0}being smaller or equal than unity,  with respect to the measure $\left(b(\hat u\cdot\sigma)  d\sigma\right)$,   is strictly decreasing for $r>1$ up to a set of measure zero (namely at $\sigma=\{\pm\hat{V}\}$).  
As a consequence $\mu_{r_{1}}(\hat{V},\hat{u})>\mu_{r_{2}}(\hat{V},\hat{u})$ for any $r_{1}<r_{2}$.  Thus, the
 $\sup_{(V,\hat u)}\left(  \mu_{r_{1}}(\hat{V},\hat{u})-\mu_{r_2}(\hat{V},\hat{u})\right)> 0$, when   taken over any unitary vector pairs $(\hat{V}, \hat{u} )\in \mathbb{S}^{d-1} \times \mathbb{S}^{d-1} $, which, by continuity,  implies $\mu_{r_{1}}-\mu_{r_{2}}>0$.  Finally, monotonic convergence gives that $\mu_r\searrow0$, with  $\mu_r<\|b \|_{L^1(\mathbb S^{d-1})}$, as soon as $r>1$.\\

In the case that the scattering kernel $b$ is integrable in $L^{p}(\mathbb{S}^{d-1})$ with $p>1$,  estimate \eqref{povzner0.1} follows by H\"{o}lder's inequality
\begin{align}\label{Povp1}
\mu_r\leq \int_{ \mathbb{S}^{d-1} }\left| \tfrac{1 \pm \hat V\cdot\sigma}2 \right|^{r} &b(\hat u\cdot\sigma) d\sigma \leq \|b\|_{L^{p}(\mathbb{S}^{d-1})}\bigg( \int_{ \mathbb{S}^{d-1} } \left| \tfrac{1 \pm \hat V\cdot\sigma}2 \right|^{p'r}  d\sigma \bigg)^{1/p'} \leq  2  \|b\|_{p}  \left(\frac{|\mathbb{S}^{d-2}|}{1+rp'}\right)^{1/p'}\\
=   2^{(2p-1)/p}  &\|b\|_{L^p(\mathbb{S}^{d-1})}|\mathbb{S}^{d-2}|^{(p-1)/p}   \bigg(\frac{p-1}{p-1+pr}\bigg)^{(p-1)/p}\!\!\!= \mathcal{O}\left((1+\frac{p}{p-1}r)^{-\frac{p-1}p}\right),  \nonumber
\end{align}
and
\begin{align}\label{Povp2}
\mu_r\leq  4  \|b\|_{L^{\infty}(\mathbb{S}^{d-1})}| \mathbb{S}^{d-2}| ((1+r)^{-1}), \qquad \text{for } p\equiv \infty. \qquad \qquad  
\end{align}
\end{proof}

Recall  that, since the unit sphere is a bounded set in $\mathbb R^d$, then   the      $L^p$-integrability of the transition probability function $b(\hat u\cdot \sigma)$ on the unit sphere $\mathbb S^{d-1}$, with any  $p\in[1,\infty] $, is control from below from the $L^1$ integrability norm, that is 
$ \|b\|_{L^1(\mathbb{S}^{d-1})} \le  \|b\|_{L^p(\mathbb{S}^{d-1})}  |\mathbb{S}^{d-2}|^{(p-1)/p} $.

\subsection{Moment inequalities for binary interactions}\label{sec:mom}

The sharp form of Povzner Lemma in Theoren~\ref{BGP-JSP04} has a profound impact in the theory of existence and uniqueness of the Boltzmann equation in the space homogeneous case because it allows to propagate and generate statistical moments in the case of hard potentials.

\medskip

\begin{lem}[\bf{Angular Averaging for the Binary Weight Function }]\label{PovznerII}
The Boltzmann weight function \eqref{weightuv2} associated to the weak form \eqref{bina-weak2} for  $\varphi(v)=\langle v\rangle ^{sk}$  for $1<k\in\mathbb{N}$ and $s\in(0,1]$, now denoted by $G_{sk}(v_*,v)$, satisfies the  following inequality for any $sk\geq1$,
\begin{align}\label{Pov2}
\begin{split}
G_{sk}(v_*,v) &= \int_{S^{d-1}}\Big( \langle v' \rangle^{2sk}   +   \langle v'_*\rangle^{2sk}  - \langle v\rangle^{2sk} - \langle v_*\rangle^{2sk}\Big) B(|u|, \hat u\cdot\sigma)\, d\sigma\\
&\le  \Phi(|u|)\Big( \mu_{sk}\Big( \langle v \rangle^{2} + \langle v_{*}\rangle^{2} \Big)^{sk}  -\|b\|_{L^1(\mathbb{S}^{d-1})}\left( \langle v\rangle^{2sk} + \langle v_*\rangle^{2s k}\right) \Big)\\[3pt]
&\hspace{-1cm}\le  \Phi(|u|)\Big( 2^{J_{2sk}}\mu_{sk}\Big(\langle v \rangle^{2}\langle v_{*}\rangle^{2(sk-1)} + \langle v\rangle^{2(sk-1)}\langle v_{*}\rangle^{2} \Big) - (\|b\|_{L^1(\mathbb{S}^{d-1})}- \mu_{sk})\Big( \langle v\rangle^{2 sk} +  \langle v_*\rangle^{2sk}\Big) \Big)\,. 
\end{split}
\end{align}
\end{lem}
\begin{proof}   
It follows by a trivial application, first, of Theorem \ref{BGP-JSP04} and, then, of Lemma \ref{BGP04-bino-sum} with $x=\langle v\rangle^{2}$ and $y=\langle v_{*}\rangle^{2}$ and $J_{2sk}= sk-1$.
\end{proof}


\subsection{Moment  estimates - Version $\mathrm I$}\label{L^1-MM}
 
In order to obtain moment's estimates  associated to the solution of the Boltzmann equation that will also serve to calculate the summability of moments as a function of the potential rate $\gamma$ from \eqref{pot-Phi_1}, 
 defined the $sk^{th}$-moment of  the probability density $f(v,t)$ by  $m_{sk}[f](t) = \int f(v,t) \langle v\rangle^{2sk}  dv,$  as defined in \eqref{k-moment}, for any  $sk\geq0$, with $0<s\le1.$
    
\begin{lem}[\bf{Lebesgue moment's estimates}]\label{mom-lemma-I}
Each  $sk^{th}$ moment of the Boltzmann equation associated to the Cauchy problem Theorem~\ref{CauchyProblem}  
satisfies the following Ordinary Differential Inequality (ODI) 

\begin{align}\label{momentineq}
\frac {d}{dt} m_{sk}[f](t) &\leq 2^{sk-1+\frac\gamma 2}\mu_{sk}  C_\Phi  \Big(m_{1+\gamma/2}[f](t) m_{sk-1}[f](t)  + m_{1}[f](t) m_{sk-1+\gamma/2}[f](t) )\Big)\\
&\ \ \ \  + \|b\|_{L^1(\mathbb{S}^{d-1})}\,m_{sk}[f](t) m_{\gamma/2}[f](t)  - 2^{-\frac{\gamma} 2}c_{\Phi} (\|b\|_{L^1(\mathbb{S}^{d-1})}-\mu_{sk}) m_{0}[f](t) m_{sk+\gamma/2}[f](t),  \nonumber\\[6pt]
& \ \ \ \   \ \text{for any }  0<s\le 1, \ 0<\gamma \le 2,\  \text{and} \  sk> s\kc\ge 1, \   \nonumber
\end{align}

with the {$sk^{th}$} contracting factor  $\mu_{sk} < \|b\|_{L^1(\mathbb{S}^{d-1})}$, and the lower bound constant  $2^{-\frac{\gamma}2}c_{\Phi }$ from the intermolecular potential function $\Phi(u)$.
\end{lem}

\begin{proof}
A direct application of Theorem \ref{BGP-JSP04} and Lemma \ref{lem:binom} leading to  Lemma \ref{PovznerII}, enable estimates for any  $sk^{th}$-moments of the collisional form  $Q(f,f)(v,t)$, for any real valued $sk\geq1$,   as defined in \eqref{k-moment},  that is
\begin{align}\label{L^1-MM-1}
\begin{split}
\frac {d}{dt} m_{k}[f](t) = \int_{\real^d}& Q(f,f)(v,t) \, \langle v\rangle^{2sk}  dv   =  \int_{ \real^{2d} } f(v) f(v_*) \, G_{sk}(v_{*},v)\, dv_* dv  \\
&= \int_{ \real^{2d} } f(v) f(v_*) \, \left(G^+_{sk}(v_{*},v) - G^-_{sk}(v_{*},v) \right) \, dv_* dv  \\
&\quad\le  2^{sk  -1 +\frac\gamma 2}C_{\Phi}\mu_{sk}   \int_{ \real^{2d} } f(v) f(v_*)\left(\langle v\rangle^{2}\,\langle v_{*}\rangle^{2(sk-1)} +\langle v\rangle^{2(sk-1)}\langle v_*\rangle^{2}\right ) \Phi(|u|)  dv_* dv    \\
&\qquad  -  ( \|b\|_{L^1(\mathbb{S}^{d-1})} -\mu_{sk})\int_{\real^{d}}f(v)\langle v\rangle^{2sk}\Big(\int_{\real^{d}}f(v_{*})\Phi(|v-v_*|) dv_{*}\Big)dv\,.
\end{split}
\end{align}
Thus, the first binary integral term,  or positive part of  \eqref{L^1-MM-1}  is integrated by Fubinni  with respect to $v$ and $v_*$, which combined with  the upper of $\Phi_\gamma$ from \eqref{pot-Phi_2}  yields the following  upper bound  given by the bilinear forms  of  shifted  moments
  \begin{align}\label{momentinequp}
\begin{split}
 \int_{\real^{2d} }f(v,t) f(v_*,t) G^+_{sk}(v_{*},v)dv_* dv  \!  \leq \! 2^{sk -1+\frac\gamma 2} C_\Phi \mu_{sk}  \Big(\!&m_{1+\gamma/2}[f]m_{sk-1}[f]  \!+\! m_{1}[f] m_{sk-1+\gamma/2}[f] \!\Big)(t).\end{split}
\end{align}

The   negative part of  \eqref{L^1-MM-1},  $\int_{\real^{2d} }f(v,t) f(v_*,t) G^-_{k}(v_{*},v)dv_* dv$  ,   needs a lower lower bound for the collision frequency factor  that uses the pointwise   bound estimates  for $\Phi_\gamma$
 from \eqref{pot-Phi_2}, in terms of the Lebesgue weights,  
\begin{align}\label{momentineqlow}
\int_{\real^{d}}f(v_{*})\Phi_\gamma(|u|)dv_{*} &=   m_{0}[f],\qquad  \text{ for}\  \gamma = 0 \ \text{(Maxwell type interactions case)}\,, \nonumber\\
\text{and}\qquad  \qquad \qquad\qquad&\ \\
\int_{\real^{d}}f(v_{*})\Phi_\gamma(|u|)dv_{*}&\geq  c_{\Phi}(2^{-\frac\gamma 2}  m_{0}[f]\langle v\rangle^{\gamma} - m_{\gamma/2}[f])\,, \qquad \text{ for} \ 0 < \gamma \le 2 \ \text{(Hard potentials case)}\,, \nonumber
\end{align} 

Combining both estimates for the positive and negative parts of the collisional integrals estimates \eqref{momentinequp} and 
\eqref{momentineqlow} yields an upper estimate for the right  hand side the $k^{th}$-moment ODI derived from the corresponding moments inequality \eqref{L^1-MM-1}, and so   \eqref{momentineq} holds. 
\end{proof}

\begin{remark}\label{remark-moment-low} 
Note that the bound from below of $\Phi_\gamma$ introduces a positive contribution, and so coarsening the negative contribution of the lower bound.   However, this lower bound constants does not need any information about the solutions moment's order. 
\end{remark}
\begin{remark}\label{remark-moment-0} 
Note that the third term positive term in the righthand side of  \eqref{momentineq} disappears for the case  Maxwell interactions when $\gamma=0$.  
\end{remark}

\

It is important to point out  that  the above  Ordinary Differential  inequality    \eqref{momentineq} enables global in time propagation and generation of moments estimates, that are sufficient to prove  the existence and uniqueness theory associated to the initial value problem in the proof of Theorem~\ref{CauchyProblem} in a suitable subspace of the positive cone of $L^1_{2k}(\real^d)$ spaces.  A minor and yet profound impact of the lower bound contribution from \eqref{momentineqlow}, is that requires to revisit the classical $k^{th}$-moment inequalities calculated  with weights of the form $|v|^{2k}$ rather the Lebesgue weight $\langle v\rangle^{2k}$.   
Because of this ``wrinkle" in the lower bound, we need  to appeal to a  
the following classical interpolation results for weighted moments associated to  probability densities.

\begin{lem}[\bf Moment interpolation inequalities]\label{Lemma:mom-Holder}
The moment $m_{j} :=m_{j}[f]  :=\int_{\real^d} f(v,t) \, \psi_{ j} (v)\, dv $, for either $\psi_{ j}(v)=  |v|^{j}$  or  $ \langle v \rangle^{j}$,  satisfies
\begin{equation}\label{mom-Holder}
m_{j} \leq m_{{j}_1}^\tau m_{{j}_2}^{1-\tau},
\end{equation}
where the positive constants ${j}, {j}_1, {j}_2, \tau$ satisfy $0<{j}_1< {j}< {j}_2,$ $0<\tau<1$, and ${j}=\tau{j}_1+(1-\tau){j}_2$.
\end{lem}
\begin{proof}
The proof of this statement is straightforward since $\psi_{j} (v)$ is convex for $j\ge1$.  Indeed, H\"{o}lder's  and Jensen's inequality imply
\begin{align*}
m_{j}[f] \leq \int_{\mathbb{R}^d}\psi(v)^{j} f(v) dv \ &\leq\ 
\int_{\mathbb{R}^d} \psi(v)^{ {j}_1 \tau+  {j}_2 (1-\tau) } f(v) dv \\
&\left(\int_{\mathbb{R}^d}\psi(v)^{{j}_1} f(v) dv\right)^\tau\, \left(\int_{\mathbb{R}^d}\psi(v)^{{j}_2} f(v) dv \right)^{1-\tau}= m_{{j}_1}^\tau[f]\, m_{{j}_2}^{1-\tau}[f]\,.
\end{align*}
\end{proof}

\subsection{Moment estimates for the collision operator in $L^1_{\ell}(\R^d)$}\label{Coll-mom-estimates}

This task focus in obtaining   Lebesgue moments estimates  for the collisional integral $Q_\gamma(f,f)](v) $ associated to the Boltzmann flow, by estimating  its $\ell^{th}$-order Lebesgue weighted  integrability in probability space $v\in\R^d$ and $\ell$ a real valued number.  That is,  to estimate $m_k[Q(f,f)]= \int_{\mathbb R^d}Q(f,f)(v) \langle v\rangle^{2k} dv$  independently of time.
This is a delicate point as the goal is not only to find bound depending on $q$ and the problem data, but also give an rather acquire description of the dependance of this bound with respect to the coerciveness factor $\mu_r$ obtained by the Angular Averaging Lemma~\ref{BGP-JSP04}. 
Previous results along these pages from Subsections~\ref{sec:mom} and ~\ref{L^1-MM}, namely Lemma~\ref{PovznerII} on Angular Averaging for the Binary Weight Function, Lemma~\ref{mom-lemma-I} on  Lebesgue moment's estimates, Lemma~\ref{Lemma:mom-Holder} on  Moment interpolation inequalities and moments lower bounds \eqref{momentineqlow}, are the starting points of the following Theorem~\ref{mom-coll-op} and its proof  developing the refined  bounds.
 
 \medskip
 
\begin{thm}[{\bf Lebesgue Moments of the collisional integral for hard potentials}] \label{mom-coll-op}

  Let $f(v)\in L^1_{\ell}(\R^d)$ and $Q_\gamma(f,f)](v) $ to the Boltzmann collision operator  as posed in \eqref{collision3} and \eqref{collision3.2}, with the   with the collision kernel or transition probability  $B(|u|, \hat u\cdot\sigma)= \Phi(|u|) b(\hat u\cdot\sigma)$,  with $\Phi(u)$ satisfying conditions \eqref{pot-Phi_1}, \eqref{pot-Phi_2.0} and \eqref{pot-Phi_2} for $\gamma\in (0,2]$ and  the angular transition part $b$  integrable in $L^p(\mathbb{S}^{d-1}), \ p\in[1,\infty)$.

 Set $\ell=2sk$, for  $s\in (0,1]$ and integer $k$, to define the parameter 
   \begin{equation}\label{ck}
0<c:=c_\gamma(sk):=\frac{\gamma/2}{sk-1}
 \end{equation}
  then,   for any reference function $f_0(v)\in L^1_{\ell}(\R^d)$, and any other $f(v)$ satisfying $m_i[f]=m_{i\frac{\ell}2}[f] = m_{i\frac{\ell}2}[f_0]$, $i=0,2$,  the $sk^{th}$-Lebesgue moment of the Boltzmann Collision operator is majorized by
    \begin{align}\label{Qmoment-0}
m_{sk}[Q_\gamma(f,f)]  &= \| Q_\gamma(f,f)\|_{L^1_{2sk}(\R^d)}= \int_{\mathbb R^d}Q_\gamma(f,f)(v,t) \langle v\rangle^{2sk} dv \\
&\le \mathcal{L}(m_{sk}[f]):= 
\B_{  sk}\  -\  \frac{A_{{  s}\kc}}2\,  {m^{-c }_{1}[f_0]}  m^{1+c}_{  sk}[f], 
\ \ \text{for any } \ \ k\ge \kc>1, \text{and}\    sk\ge s\kc >1 . \nonumber
  \end{align}

      The constants $A_{s\kc}$ and  $B_k$ are determined solely from the  Cauchy problem data associated  for the Boltzmann flow stated in Theorem~\ref{CauchyProblem} and the order of $k$ of the Lebesgue polynomial moment. They are 
\begin{align}\label{moment-factors}
& A_{{  s}\kc}=   c_{\Phi } 2^{-\frac{\gamma}2} \,(\|b \|_{L^1(\mathbb S^{d-1})} -\mu_{{  s}\kc}) \,{m_0[f_0]} \quad \text{the coercive constant, }    \nonumber \\
\text{ and }   \qquad\qquad&\ \\
& \B_{  sk}:=
 C_{\Phi} \left(   {2^{sk+\frac\gamma 2}}   \mu_{{  sk}}  + 2\|b\|_{L^1(\mathbb{S}^{d-1})} \right)   \mathcal{C}_{  sk} :=\beta_{  sk} \, \mathcal{C}_{  sk}   \quad \text{the  upper bound constant,} \nonumber
\end{align}  
both depending on   the constant $k^{th}$-contracting factor  $\mu_{{  sk}}\in (0,1)$  calculated in the Angular Averaging Lemma~\ref{BGP-JSP04},  and the data  associated to the Cauchy Problem posed in Theorem~\ref{CauchyProblem}. 

In particular,  there exists a rate $\kb=\kb(\|b \|_{L^1(\mathbb S^{d-1})} ,m_0[f_0])$, such that for any  $k\ge  \kmax:=\max\{\kb,\kc\} \ge \kmin:=\min\{\kb,\kc\}>1, $
the quotient 
 \begin{align}\label{delta-k}
\delta_{  sk}\ &:= \ \frac{ A_{{  s}\kc} }{\beta_{  sk}  } \le \frac{ \|b\|_{L^1(\mathbb{S}^{d-1})}  \,m_0[f_0] }{ 2^{sk +\gamma } \mu_{ sk}  + 2^{1+\frac\gamma2}\|b\|_{L^1(\mathbb{S}^{d-1})}} <1. 
\end{align}

Hence, the constant $\mathcal{C}_{  sk}$  estimated by the moments interpolation parameters is estimated by above by the constant $\mathcal{C}^{\text{up}}_{  sk}$, related as follows, 
 
%
 \begin{align}\label{mathcal-C-kc}
\mathcal{C}_{sk} 
:=  \ttm_{1,k}[f_0] \, R_{\C}(\delta^{-1}_{sk})  =  \ttm_{1,k}[f_0] \, R_{\C}(\beta_{sk}/ A_\kc)
\end{align} 
with
\begin{align}\label{RateC}
 R_{\C}(\delta^{-1}_{sk})& :=   R_{\C}(\beta_{sk}/ A_\kc) :=
 \left[ \left(   \frac{ \beta_{  sk}  } { A_{{  s}\kc}} \right)^{e_\alpha}   \mathbb{I}_{k\geq \kgamma} +
 \left(  \frac{ \beta_{  sk} } { A_{{  s}\kc}} \right)^{e_\theta}   \mathbb{I}_{k<\kgamma}\right], \quad
\nonumber \\
\text{for}  \ \ e_\alpha=\frac1{c_\gamma(sk)}, \quad  &\text{and}\quad\  e_\theta=\left({sk}-1\right) \left(1+c_\gamma(sk) \right)^2 +c_\gamma(sk),
 \quad\text{for }  \  s\kgamma:= \frac{1+(\gamma/2)^2}{1-\gamma/2} >1,  \end{align}
for  the factor  $\kgamma$ depending only on the potential rate $\gamma$ and  on fixed moment order $sk\ge s\max\{\kb,\kc\}=s\kmax >1$, with $\kb$ depending on $m_0[f_0]$ and $b(\hat u\cdot\sigma)$.  

The factor $\ttm_{1,k}[f_0] $ depends order of moment $k$ and  on data, namely,   the invariant moments $m_0[f_0]$ and $m_1[f_0]$ and $\gamma$ , 
\begin{align}\label{invariant k factor}
 0<\ttm_{1,k}[f_0] :=  \max\left\{   m_1^{2c +1}[f_0] \, ;\,    m^{{  sk}+\frac{\gamma}2}_{1}[f_0]   m^{(1+\frac{\gamma}2)\left( 1+c\right)}_{0}[f_0]  \right\} 
 \end{align}

\medskip

 Finally, combining  \eqref{moment-factors}, \eqref{delta-k} and \eqref{mathcal-C-kc},  set  $\B_{sk}:= \beta_{sk} \, \mathcal{C}_{sk}$,  enables the definition of the unique equilibrium state associated to the operator $ \mathcal{L}(m_{sk}[f])(t)$ from \eqref{Qmoment-0} defined as upper bound of the $k^\th$-Lebesgue moment of the collisional integral as to be the unique root of 
     \begin{align}\label{equi root}
     \mathcal{L}(\mathbb{E }_{k,\gamma}[f]) &= 0, 
    \end{align}
    with $\mathbb{E }_{k,\gamma}[f]$ is the unique equilibrium state, fully characterized by    
   \begin{align}\label{RateE}
   \mathbb{E}_{k,\gamma} &:=  m_1^{\frac{c}{1+c}}[f_0]\left(\frac{\B_{sk}}{A_\kc}\right)^{\frac1{1+c}} = m_1^{\frac{c}{1+c}}[f_0]\left(\frac{\beta_{sk} \C_{sk}}{A_\kc}\right)^{\frac1{1+c}} \nonumber\\
  &\, =  \left( m_1\,\ttm_{1,k}\right)^\frac{sk-1}{sk-1+\gamma/2}\![f_0] \, := \ R_{\mathbb E}(\delta^{-1}_{sk}) \ := \,R_{\mathbb E}(\beta_{sk}/ A_\kc)  
\\ 
\text{with} \quad R_{\mathbb E}(\delta^{-1}_{sk}) &:= \,R_{\mathbb E}(\beta_{sk}/ A_\kc) :=\left[ \left( \frac{ \beta_{  sk}}{ A_{{  s}\kc} }\right)^{\e_\alpha} \mathbb{I}_{ k\ge \kgamma } \!+\!
  \left(\frac{ \beta_{sk}}{ A_{{  s}\kc} }\right)^{sk+\frac\gamma2} \mathbb{I}_{ k< \kgamma } \right],  \nonumber
 \end{align}
after following the analog notation from \eqref{RateC}. 

The rate factors  $R_\C(\delta^{-1}_{sk})$ and  $R_\mathbb E(\delta^{-1}_{sk})$ defined in   \eqref{RateC} and \eqref{RateE}, respectively,  depend on mass and energy conserved quantities $m_0[f]$ and  $m_1[f]$, the integrable angular transition $b(\hat u\cdot\sigma)$,  potential rate $\gamma$, and Lebesgue moment order $sk$. 
\end{thm}

\smallskip

\begin{remark}
The statement of  Theorem~\ref{mom-coll-op}   replaces the definitions of   $A_{s\kc}$  and  $B_k$  in \eqref{moment-factors},  to  $A_{s\kc}=   c_{lb}\,(\|b \|_{L^1(\mathbb S^{d-1})} -\mu_{{  s}\kc}) \,{m_0[f_0]}$ and 
$\B_{  sk}:= C_{\Phi}  {2^{sk+\frac\gamma 2}}   \mu_{  sk} \,  \mathcal{C}_{  sk} :=\beta_{  sk} \, \mathcal{C}_{sk},$ respectively, and so
the new $\beta_{sk}:=C_{\Phi}  {2^{sk+\frac\gamma 2}}   \mu_{  sk}.$
\end{remark}

\smallskip

\begin{proof}[Proof  of Theorem~\ref{mom-coll-op}]

First note that  if $f(v)\in L^1_{2k}(\R^d)$, for $k>1$,  $m_1[f]=m_1[f_0]$, then $f(v)\in L^1_{2sk}(\R^d)$ as $s\in(0,1]$. Equivalently, for $f$ non-negative,  $f(v)\in L^1_{2sk}(\R^d) = m_{sk}[f] $, and so $m_{sk}[f]$ is finite. In addition we assume that $f $ is positive, non-singular measure, so $0<m_0[f]  <m_1[f]$.

The  $sk^{th}$-Lebesgue  moment of the  collisional integral was already estimated 
in the right hand side the corresponding moment recursion Ordinary Differential Inequalities (\ref{L^1-MM-1},\ref{momentineq}) and by below in \eqref{momentineqlow} for   any fixed  $k$,  such that  $0<s\le 1< sk <\infty$, after making use of the Angular Averaging
Lemma~\ref{BGP-JSP04}, where the contractive constant $\mu_{sk}$ was introduced.

The proof consists in  estimating just  the moment  of the collisional integral appearing in the right hand side  of the time dependent equations \eqref{momentineq}. These bounds  are obtained  by means of 
interpolation  estimates  from \eqref{mom-Holder}, for any  $\gamma\in(0,2]$ and ${\kc}>1$,  that control all its positive terms  bounded by above by the higher order moment $m_{k+\gamma/2}[f](t)$, 
 while controlling the negative contributions from below by coercive estimates proportional	to $m_{k}^{1+ \gamma/(2(sk-1)}[f]$, both  multiplied by  constants that  only  on the moment order ${\kc}>1$ and its corresponding ${\kc}^{th}$-contracting factor $\mu_{\kc}$, as well as on the mass and energy  $m_0[f]=m_0[f_0]$,  $m_1[f]=m_1[f_0]$,  where $f(v)$ is fixed.

Hence, these collision operator moment estimates are obtained from setting  the non-linear factor 
\begin{align}\label{little c}
0< c\ =\ c_\gamma(sk)\ :=\ \frac{\gamma}{2(sk-1)},  \qquad \text{with}\ \gamma \in (0,2]  \ \text{and}\ sk\geq\s\kc>1 \ge s>0,
\end{align}
that controls the maximum super linear growth to be determined by a lower bound of the moments of the loss operator.  
Thus,  while a Jensen's   inequality is enough to control part of it,  the use of interpolations enables a good control for all the terms, since  while mower moment of $f$ are associated to the conservation properties of the Boltzmann flow,  all other moments above $sk>1$ will have estimates that are easily obtained from interpolation arguments. 

Therefore,  starting from the following interpolation, using the short $c$ notation for $c_\gamma(sk)$, 
 \begin{align}\label{mom-interpol}
m_{1+\gamma/2}[f]  &\leq m^{\frac{1}{1+c}}_{1}[f_0]\, m^{\frac{c}{1+c}}_{sk+\gamma/2}[f]\, ,\qquad \ \ \ \ \ 
\text{and, } \\
m_{sk-1}[f] & \le m^{\frac{s+\gamma/2}{sk+\gamma/2}}_{0}[f_0]\, m^{\frac{sk-1}{sk+ \gamma/2}}_{sk+\gamma/2}[f] = m^{\frac{s+\gamma/2}{sk+\gamma/2}}_{0}[f_0]\, m^{\frac{1}{sk/(sk-1)+ c}}_{sk+\gamma/2}[f] \,. \nonumber 
\end{align}
With the first term in the right side of the inequality \eqref{momentineq} is controlled as
\begin{align}\label{first-mom-holder}
m_{s+\gamma/2}[f]m_{sk-1} [f]&\leq C_{1,k}[f_0]\,m^{\theta_{k}}_{sk+\gamma/2}[f]\,, \nonumber\\
\text{with}  \qquad  \qquad   \qquad  \qquad  \qquad  &\\
C_{1,k}[f_0]= C_{1,sk}[f_0]\!:=\!  m^{\frac{1}{1+c}}_{1}\![f_0] m^{\frac{s+\gamma/2}{sk+\gamma/2}}_{0}\![f_0],  \quad  &\text{and} \quad \theta_{k}\!:=\!\frac{c}{1\!+\!c}\! +\! \frac{1}{\frac{sk}{(sk-1)}\!+\! c} \!=\! 
\frac{\gamma/2}{\gamma/2 \!+\!sk\!-\!1}\! +\! \frac{sk\!-\!1}{sk\!+\!\gamma/2}\!   <\!1.\nonumber
\end{align}


Furthermore, $m_{\gamma/2}[f]\leq m_{1}[f] = m_{1}(f_0)<\infty$ and, interpolating again,
\begin{align}\label{second-mom-holder}
&m_1[f_0] m_{sk-1+\gamma/2}[f]\leq m_1[f_0] \, m_{sk}[f] \leq C_{2,sk} [f_0]  \,m^{\alpha_{k}}_{sk+\gamma/2}[f]\,, \nonumber \\
&  \text{with}    
   \ \ { C_{2,sk}[f_0] =C_{2,sk}[f_0] :=  m^{\frac{c}{1+c}}_{1}[f_0]  \ \ \text{and}\ \      \alpha_{k}:=\frac{1}{1+c}<1} \,.
\end{align}

Furthermore,  both $\{m_{\gamma/2}[f]\,;\, \m_s[f]\}\leq m_{1}[f] = m_{1}(f_0)<\infty$,  interpolating again 
\begin{align}\label{inter msk}
m_{sk}&< m_1^{\frac{c}{1+c}}\, m_{sk +\gamma/2}^{\frac{1}{1+c}}
\end{align}
then,
\begin{align}\label{second-mom-holder}
&m_1[f_0] m_{sk-1+\gamma/2}[f]\leq m_1[f_0] \, m_{sk}[f] 
\leq C_{2,k} [f_0]  \,m^{\alpha_{k}}_{sk+\gamma/2}[f]\,, \nonumber \\
&  \text{with}    
   \ \ { C_{2,sk}[f_0] =C_{2,sk}[f_0] :=  m^{\frac{2c+1}{1+c}}_{1}[f_0]  \ \ \text{and}\ \      \alpha_{k}:=\frac{1}{1+c}<1} \,.
\end{align}

Therefore, from \eqref{first-mom-holder}, \eqref{second-mom-holder}, and the right hand side  moment inequality \eqref{momentineq},  the following estimate hold
\begin{align}\label{hardpotmom-1}
m_{sk}[Q(f,f)]  &\leq \left(C_{\Phi}2^{(sk-1)(1+c_\gamma(sk))} \mu_{sk} +\|b\|_{L^1(\mathbb{S}^{d-1})}\right)\left[ C_{1,k}[f_0]\,m^{\theta_{k}}_{sk+\gamma/2}[f] +  C_{2,k}[f_0]\,m^{\alpha_{k}}_{sk+\gamma/2}[f]\right]   \nonumber\\
&\qquad - c_{\Phi}2^{-\gamma/2}( \|b\|_{L^1(\mathbb{S}^{d-1})}   -\mu_{s\kc})\,m_0(f_0)\,m_{sk+\gamma/2}[f], 
\end{align}
valid for any $\gamma\in(0,2]$, $ 0<s\le 1$, and $k\geq\kc>1.$

Note that  these $sk^{th}$-moment estimates provide a {\em coercive factor} independent of $k$, defined by   
\begin{align}\label{Akc}
 A_{s\kc} :=
  c_{\Phi}2^{-\gamma/2}( \|b\|_{L^1(\mathbb{S}^{d-1})}   -\mu_{s\kc})\,m_0(f_0), \qquad \text{for }\  sk>s\kc> 1\ge s>0.
\end{align}

\medskip 

Next,  the following   key point enables the control of the right hand side of inequality \eqref{hardpotmom-1} by proving that, after the use of a further interpolation, it is possible to obtain an upper  bound for the positive contribution  only depending on data and $k$.

For such task,  take  the couple of  strictly less that unity exponents $\theta_{k}$ and $\alpha_{k}$ from \eqref{first-mom-holder} and \eqref{second-mom-holder}, respectively,   as well as constants  $C_{1,k}[f_0]$ and $C_{2,k}[f_0]$  defined in    \eqref{first-mom-holder} and  \eqref{second-mom-holder}, respectively,  and invoke the $\delta$-Young's inequality, namely, $|a\,b|\leq \delta^{-p/q} p^{-1} |a|^p + \delta q^{-1} |b|^q$, for $\delta>0$ to be chosen sufficiently small,  and $p^{-1}+q^{-1} =1$, making  the following choice, for $c=c_\gamma(sk)$ and $\theta_k$ from \eqref{little c} and \eqref{first-mom-holder}, respectively,
\begin{align*}
&a=C_{1,k}[f_0], \qquad \qquad\qquad  \ \ b=  m^{\theta_{k}}_{k+\gamma/2}[f]\,,  
\nonumber\\ 
\text{with} \ \ &p_{1}= \frac{1}{1-\theta_{k}}     > 1,
\  \text{and its convex conjugate}\ \  q_{1}=\frac{1}{\theta_{k}} >1,
\end{align*}
to obtain, again using the short $c$ notation for $c_\gamma(sk)$, 
\begin{align}\label{hardpotmom-22}
&C_{1,k}[f_0]\, m^{\theta_{k}}_{k}[f] \leq  \delta^{-\frac{ \theta_k}{1-\theta_k} } C^{\frac{1}{1-\theta_{k}}}_{1,k}[f_0] +
\delta \, m_{k+\gamma/2}[f],
 \nonumber\\
\text{for}\  \ &\frac{1}{1-\theta_{k}} =   (1+ c)(sk+ (sk-1) c) =  \left(1+ c \right) \left(sk+{\gamma}/2 \right), \\
\text{and} \ \ 
&\frac{ \theta_k}{1-\theta_k}=\left(sk+ (sk-1) c\right) c+ (1+ c)(sk-1) +  c=  (sk-1)(1+ c)^2 + c , \nonumber
\end{align}
 which controls the first term  of  in \eqref{hardpotmom-1}.
A similar control for the second term  of the positive contribution in  \eqref{hardpotmom-1}, can be obtained choosing $c$ and $\alpha_k$ from \eqref{little c} and \eqref{first-mom-holder}, respectively, to yield
\begin{align*}
&a=  C_{2,k}[f_0], \qquad \qquad\qquad  \ \ b=  m^{\alpha_{k}}_{k+\gamma/2}[f],  \nonumber\\ 
\text{with} \ \ & p_{2}= \frac{1}{1-\alpha_{k}}=\frac{1+ c}{ c} >1,\  \text{and convex conjugate}\ \  q_{2}=\frac{1}{\alpha_{k}}=1+c >1,  
\end{align*}
to obtain  $\frac{\alpha_{k}}{1-\alpha_{k}} = \frac{1}{ c} = \frac{ sk-1}{\gamma/2} $  ,

 Hence, combining with  \eqref{second-mom-holder},
\begin{align}\label{hardpotmom-3}
C_{2,k}[f_0]\, m^{\alpha_{k}}_{{  sk}+\gamma/2}[f]  &\leq  \delta^{-\frac{\alpha_k}{ 1-\alpha_k}}
 C^{\frac{1}{1-\alpha_{k}}}_{2,k}[f_0] + \delta \, m_{{  sk}+\gamma/2}[f] \\
 &= \delta^{-\frac1c}
 m_1^{\frac{c+1}{c}} [f_0] + \delta \, m_{{  sk}+\gamma/2}[f]. \nonumber
\end{align}

These previous calculations from  moment interpolations enable  the formulation os a new upper bound for  $sk^{th}$-moments Inequalities \eqref{hardpotmom-1},     by   estimating  upper bounds of the maximum of estimates  obtained in  \eqref{hardpotmom-22}, and \eqref{hardpotmom-3}, respectively.

For the sake of simplicity of the calculations that follow, we set 

 \begin{align}\label{betak}
 \beta_{  sk}:= 2{C}_{\Phi}\left(  2^{(sk-1)(1+ c)} \mu_{  sk} +\|b\|_{L^1(\mathbb{S}^{d-1})} \right) = {C}_{\Phi}\left(  2^{(sk +\frac\gamma 2)} \mu_{  sk} + 2\|b\|_{L^1(\mathbb{S}^{d-1})} \right) .
\end{align}

The corresponding $k^{th}$-moments bounds of the collisional  from \eqref{hardpotmom-1} are majorized the the following Ordinary Differential Inequality
 \begin{align}\label{hardpotmom-3.1}
m_k[Q_\gamma(f,f)] &\leq \beta_{  sk} C_{sk} \  -\  \Big( A_{{  s}\kc} - \frac{\beta_{  sk}}2{\delta}\Big)\,m_{{  sk}+\gamma/2}[f]\, . 
\end{align}

\medskip

The next step consists in choosing a suitable $\delta=\delta_{  sk}$,  and associated  constant  ${C}_{  sk}$ to be defined below,  that will preserve the coerciveness to secure a strict absorption effect from the higher order moment $m_{{  sk}+\gamma/2}[f]$.  It is very  important to notice that these choices only depend on the data, that are the  potential rate $\gamma\in (0,2]$, the moment order  for any fixed fixed $sk\ge s\kc >1$, and  the initial mass $m_0[f_0]$ and energy $m_1[f_0]$, both time invariants by the Boltzmann flow.

Indeed, this is shown in the following Proposition.

\begin{problem2}\label{prop delta k} There exist a large enough moment order $\kb=\kb(m_0[f_0], b(\hat u\cdot\sigma))$,  such that  time independent factor  $\delta_{sk}$,  defined  for  any $sk\geq {s\kc}>1$,  satisfies
\begin{align}\label{hardpotmom-4}
0 < \delta := \delta_{  sk}\   = \frac{  A_{{  s}\kc} }{ \beta_{sk}} \ <\ 1.
\end{align}
\end{problem2}

\begin{proof} This property is easy to prove since  this choice is very natural, after evaluating the above  quotient using \eqref{Akc} and \eqref{betak}, short labels for $A_{{  s}\kc}$ and $\beta_{  sk}$, respectively. Indeed,  the denominator in the above defenition of $\delta$, can be shown to grow as large as needed, but means of the angular averaging Lemma~\ref{BGP-JSP04}. Indeed,   from  \eqref{mu_r},   \eqref{Povp1} and  \eqref{Povp2},   
the contracting factor $\mu_{  sk}< \mu_1 =\|b\|_{L^1(\mathbb{S}^{d-1})}$, so on one hand 
 \begin{align}\label{k0-1}
0< A_{{  s}\kc} =   c_{\Phi}2^{-\frac\gamma2}( \|b\|_{L^1(\mathbb{S}^{d-1})}   -\mu_{{  s}\kc})\,m_0[f_0] <   c_{\Phi}2^{-\frac\gamma2}\|b\|_{L^1(\mathbb{S}^{d-1})}  \,m_0[f_0].
\end{align}
On the other hand, from the definition \eqref{mu_r} for $\mu_{sk}$ after multiplying both sides by $2^{sk+\frac{\gamma}2}$, yield the following estimate
\begin{align}\label{k0-2}
\lim_{sk\to \infty} 2^{sk+\frac{\gamma}2} \mu_{sk} \ge\lim_{sk\to \infty} \,  2^{\frac{\gamma}2} \int_{ \mathbb{S}^{d-1} }
  \bigg( \left| 1 - \hat u\cdot\sigma \right|^{sk} +   \left| 1 + \hat u\cdot\sigma  \right|^{sk} 
  \bigg)  b(\hat u\cdot\sigma)\, d\sigma = \infty \, 
    \end{align}

In particular, since  data parameters  from \eqref{pot-Phi_1} assert  $c_\Phi \leq C_\Phi$,  there exists a  
\begin{align}\label{bf kb}
\kb = \kb(m_0[f_0], b(\hat u\cdot\sigma)),
\end{align}
 large enough such that, for all  ${sk} \ge s\max\{\kb,\kc\}=:s\kmax>1$,   
 \begin{align}\label{choose kb}
 0 <  \delta_{  sk}\   = \frac{  A_{{  s}\kc} }{ \beta_{sk}} \le\frac{  c_{\Phi} \|b\|_{L^1(\mathbb{S}^{d-1})}  \,m_0[f_0] }{{C}_{\Phi} \left( 2^{sk +\gamma } \mu_{{  sk}} + 2^{\frac\gamma2+1}\|b\|_{L^1(\mathbb{S}^{d-1})}\right) } < 
 \frac{ \|b\|_{L^1(\mathbb{S}^{d-1})}  \,m_0[f_0] }{ 2^{sk +\gamma  } \mu_{ sk} }< 1.
  \end{align}
\end{proof}  

\medskip

Finally, invoking the interpolation from \eqref{inter msk}, 
yields the following lower bound  for  $m_{{  sk}+\gamma/2}[f]$, \begin{align}
\label{interpol-2}
   \  m^{-c}_{1}[f_0]\, m_{  sk}^{1+c}[f]  
 \leq  m_{{  sk}+\gamma/2}[f]\, , 
\end{align}
%

 %
%
Which gives a bound from above to the righthand side in\eqref{hardpotmom-3.1},   with constants  depending the moment order $k$, on the potential rate $\gamma$, the first two moments of the initial data $f_0$ and the integrability properties of the angular function $b(\hat u\cdot\sigma)$. In particular, this estimate   yields a super-linear  negative term for the  Lebesgue polynomial moment $m_k[f]$, namely 
 \begin{align}
\label{hardpotmom-3.11}
m_k[Q(f,f)] &< 
 \beta_{  sk}  C_{sk}  -\  \frac{A_{{  s}\kc}}2 \,{m^{-c}_{1}[f_0]}\,m^{1+c}_{  sk}[f] \ =:\  \mathcal{L}_k(m_k[f])(t). 
\end{align}

Hence,  the  estimates take the simple form 
{where} the $k$-independent  coercive factor $A_{{  s}\kc}$, defined on \eqref{Akc}, and  the $sk$-dependent positive factor $ \beta_{  sk}   \mathcal{C}_{  sk}$  takes the product  the factor $\beta_{sk}$ from \eqref{betak}, with $k\ge \kb$, and 
 $ \mathcal C_{  sk} $  is characterized and estimated by above 
\begin{align}\label{hardpotmom-3.2}
&C_{sk}= \max\Big\{ \!       \delta^{- \frac{\theta_{k}}{1-\theta_{k}} }\, {C^{\frac{1}{1-\theta_{k}}}_{1,{k}}[f_0]}\, ,\,  \delta_{sk}^{ -\frac{\alpha_{k}}{1-\alpha_{k}} }\,{C^{\frac{1}{1-\alpha_{{k}}}}_{2,{k}}[f_0]}  \Big\}\nonumber\\[4pt]
&\le 
\max\left\{ \!{C^{\frac{1}{1-\theta_{k}}}_{1,{k}}[f_0]};  {C^{\frac{1}{1-\alpha_{k}}}_{2,{k}}[f_0]} \right\}
\max\left\{  \delta^{- \frac{\theta_{k}}{1-\theta_{k}} } ;   \delta_{sk}^{- \frac{\alpha_{k}}{1-\alpha_{k}} }  \right\}  \nonumber\\[4pt]
&  = \ttm_{1,k}[f_0]  \max\Big\{ \delta_{sk}^{- [(k-1)(1+c)^2 +c]} \, ;\,  \delta_{sk}^{-\frac1c}   \Big\}, 
\qquad\  \text{for any }   sk> \max\{\kb,\kc\}=\s\kmax >1, 
\end{align}
with the factor  $\ttm_{1,k}[f_0] $ being the one defined at \eqref{invariant k factor},  in the statement of this theorem. 

To this end, the right hand  of \eqref{hardpotmom-3.2}, can be further estimated by calculating  $e_{\text{max}}$ to be the exponent controlling   both, $e_\alpha=e_\alpha(\gamma,sk)$  and $e_\theta=e_\theta(\gamma,sk)$, for  fixed $\gamma\in(0,2]$ and $c=c_\gamma({  sk})=\frac{\gamma/2}{sk-1}$,  obtain by setting
\begin{align}\label{e max}
     e_{\text{max}}:=  e_{\text{max},\gamma}  &= \begin{cases}
  e_{\alpha}:=e_{\alpha}(\gamma, sk) = {\frac{\alpha_k}{1-\alpha_k}} =  \frac{1}{c}, &\ \ \ \quad \text{for }sk \ge s\kgamma>1, \\[3pt]
 e_{\theta}:=    e_{\theta}(\gamma,sk) = {\frac{\theta_k}{1-\theta_k}} ={\left(sk-1\right)\left(1+c\right)^2 +c}, &\ \ \ \quad \text{for } 1<sk < s\kgamma
    \end{cases} \nonumber\\[6pt]
    &=: \left[ e_{\alpha}    \mathbb{I}_{k\geq \kgamma} + e_{\theta}    \mathbb{I}_{k< \kgamma}\right],
    \end{align}
with a fixed $\kgamma$ parameter characterized by how  the order $k$ compared to $\kgamma$,    for any $0<\gamma \le 2$. That means, the parameter $\kgamma$ is exactly calculated  from  setting ${\bf c}=c_\gamma(s\kgamma)=\frac{\gamma/2}{s\kgamma-1}$  and finding the value  for both exponents $ e_{\theta}(k,\gamma)$ and $ e_{\alpha}(k,\gamma)$ to coincide. That is,    $s\kgamma$  only depends on  the hard potential rate exponent $\gamma\in (0,2]$, since
\begin{align}\label{compare delta expo}
e_{\theta} &:=  ({s\kgamma}-1) (1+{\bf c})^2 +{\bf c}  =   \frac1{\bf c} =: e_{\alpha} \ \ 
\iff \ \  
  \frac{\gamma}2 =  \frac{(s\kgamma-1) -\frac{\gamma}2}{(s\kgamma-1) + \frac{\gamma}2} \nonumber \\[4pt]
&\qquad\qquad \qquad \ \iff\  \ \ s\kgamma =  \frac{1+(\gamma/2)^2}{1-\gamma/2} >1.
\end{align}

As a consequence,   taking $\delta_{sk}^{-1} :=   \beta_{sk}\,A^{-1}_{s\kc}  \ >\ 1$ from \eqref{hardpotmom-4},   with $\beta_{sk}$ from \eqref{betak},  
the second factor  in \eqref{hardpotmom-3.2} defines a coercive rate quotient	
 \begin{align}\label{coer-rate-0}
 R_{\mathcal C}(\delta_k^{-1}) &:=  \max\Big\{ \delta_{sk}^{- [(k-1)(1+c)^2 +c]} \, ;\,  \delta_{sk}^{-\frac1c}   \Big\}\\
& \equiv  \left[   \delta_{sk} ^{-e_{\alpha}}   \mathbb{I}_{k\geq \kgamma} +
   \delta_{sk}^{-e_{\theta}}   \mathbb{I}_{k<\kgamma}\right] 
\equiv \left[ \left(   \frac{ \beta_{  sk}  } { A_{{  s}\kc}} \right)^{e_{\alpha}}   \mathbb{I}_{k\geq \kgamma} +
 \left(  \frac{ \beta_{  sk}  } { A_{{  s}\kc}} \right)^{e_{\theta}}   \mathbb{I}_{k<\kgamma}\right]  
  \equiv  R_\C (   \beta_{  sk} / A_{{  s}\kc} ) \nonumber.
 \end{align}

It follows that  the right hand side from 
\eqref{hardpotmom-3.2}  is  majorized by the  upper estimate, now written in terms of the coercive rate quotient	
 defined in \eqref{coer-rate-0}
\begin{align}\label{Bkc}
 \B_{sk}&:=\beta_{sk} \mathcal {C}_{  sk} :=  \ttm_{1,{  sk}}[f_0] \, \beta_{  sk}
 R_{\mathcal C} (   \beta_{  sk} / A_{{  s}\kc} )\\
 & \equiv \ttm_{1,{  sk}}[f_0] \, \beta_{  sk}
 \left[ \left(   \frac{ \beta_{  sk}  } { A_{{  s}\kc}} \right)^{\frac1{c}}   \mathbb{I}_{k\geq \kgamma} +
 \left(  \frac{ \beta_{  sk}  } { A_{{  s}\kc}} \right)^{\left({sk}-1\right) \left(1+c\right)^2 +c}   \mathbb{I}_{k<\kgamma}\right],\end{align}
for any ${  sk}\ge {  s}\max\{\kb,\kc\}=s\kmax$, after gathering the exponents from \eqref{e max} depending on $\kgamma$  from \eqref{compare delta expo}.   
\\ 

The final upper bound for the $sk^{th}$-moment of the collision operator   from \eqref{hardpotmom-3.11} is controlled 
 by the following   operator $\mathcal{L}_{sk} (m_{ks}[f])$, acting on the $sk^{th}$-Lebesgue moment of $f\in L^1_{2sk}(\R^d)$, 
 \begin{align}\label{equi-1}
m_{sk}[Q(f,f)] < \mathcal{L}_{sk}(m_k[f])=   \beta_{  sk}  \C_{sk}  -\  \frac{A_{{  s}\kc}}2 \,{m^{-c}_{1}[f_0]}\,m^{1+c}_{  sk}[f], 
\end{align}
whose unique root  
 $\mathbb{E}_{k,\gamma}$ is referred as to the equilibrium state, that is 
\begin{align}\label{equi-1.1}
\mathbb{E}_{c_\gamma(sk)} &:= m_1^{\frac{c}{1+c}}[f_0]\left(\frac{\B_{sk}}{A_\kc}\right)^{\frac1{1+c}}.
   \end{align}

Written in terms of  \eqref{Bkc}, $\mathbb{E}_{k,\gamma}$   can be fully characterized by the data and $m_0[f]$ and $m_1[f]$  moments by 
    \begin{align}\label{equi-2}
\mathbb{E}_{c_\gamma(sk)} &:= m_1^{\frac{c}{1+c}}[f_0]\left(\frac{\beta_{  sk} \C_{sk}}{A_\kc}\right)^{\frac1{1+c}}
\equiv m_1^{\frac{c}{1+c}}[f_0]  R^{\frac1{1+c_\gamma(sk)}}_\C(     \beta_{  sk} / A_{{  s}\kc} )\nonumber \\
&=\left( m_1\,\ttm_{1,k}\right)^\frac{c}{c+1}\!\![f_0] \, 
 \left[ \left( \frac{ \beta_{  sk}}{ A_{{  s}\kc} }\right)^{\frac{1+e_\alpha}{1+c} } \mathbb{I}_{ k\ge \kgamma } \!+\!
  \left(\frac{ \beta_k}{ A_{{  s}\kc} }\right)^\frac{1+e_\theta}{1+c } \mathbb{I}_{ k< \kgamma } \right]\nonumber\\
  &\, = \left( m_1\,\ttm_{1,k}\right)^\frac{sk-1}{sk-1+\gamma/2}\![f_0] 
 \left[ \left( \frac{ \beta_{  sk}}{ A_{{  s}\kc} }\right)^{\frac{2(sk-1)}\gamma } \mathbb{I}_{ k\ge \kgamma } \!+\!
  \left(\frac{ \beta_k}{ A_{{  s}\kc} }\right)^{sk+\frac\gamma2} \mathbb{I}_{ k< \kgamma } \right]\nonumber\\
&=:  \left( m_1\,\ttm_{1,k}\right)^\frac{sk-1}{sk-1+\gamma/2}\![f_0] \, R_\mathbb E(\delta^{-1}_{sk}) \,
 \quad \text{for any} 
  \ \ sk\ge s\kmax,
   \end{align}
after defining  $R_{\mathbb E} (\beta_{sk}/A_\kc):=  R^{\frac1{1+c_\gamma(sk)}}_\C(     \beta_{  sk} / A_{{  s}\kc} )\equiv 
 R_\mathbb{E}(\delta^{-1}_{sk})$. These new  exponents  obtained by taking $e_{\alpha}(sk)$ and $e_{\theta}(sk)$   defined in \eqref{e max}, 
 devided by $(1+ c),$ respectively, for $c=c_\gamma(sk)=\gamma/{(2(sk-1))}$,  yields the simpler exponent associated to the the characterization 
 of  $\mathbb{E}_{c_\gamma(sk)}$ as follows
  \begin{align}\label{equi-2.2}
\frac{1+e_\alpha}{1+c}
= \frac1c= \frac{2(sk-1)}\gamma,\qquad \text{and}
\qquad \frac{1+e_\theta}{1+c}= sk +\frac \gamma 2.
 \end{align}

Thus, the proof on Theorem~\ref{mom-coll-op} is now complete.
\end{proof}
\ \\
    
    \begin{remark}\label{exp-max}
    Note that, for example,  in the classical case for {  $s=1$},   $e_{\text{max}}=e_{\theta}$  for $\gamma =2$ and    $e_{\text{max}}=e_{\alpha}$ for values of $\gamma \approx 0^+$. In addition,  in the case of hard spheres corresponding to  $\gamma=1$, the exponent may change depending on the $k^{th}$-moment order, as  $e_{\text{max}}=e_{\alpha}$ for  if $k\ge \kgamma= 5/2 $, but   $e_{\text{max}}=e_{\theta}$ if $k<\kgamma=5/2 $.

In general, it follows that  for any $sk>s\max\{\kb,\kc\}>1$,   the constant  $\mathcal {C}_k$, as characterized in \eqref{hardpotmom-3.2}, can be recasted by 
\begin{align}\label{new Ck}
 {C}_{  sk}   &\le \mathcal {C}_{  sk}:=\ttm_{1,{  sk}}[f_0] 
 \left[ \left(   \frac{ \beta_{  sk}  } { A_{{  s}\kc}} \right)^{\frac1{c}}   \mathbb{I}_{k\geq \kgamma} +
 \left(  \frac{ \beta_{  sk} } { A_{{  s}\kc}} \right)^{\left({sk}-1\right) \left(1+c\right)^2 +c}   \mathbb{I}_{k<\kgamma}\right]
\end{align}
\end{remark}

\bigskip

 \begin{remark}\label{reduced mathcalC} In the case where  the initial mass $m_0[f_0]=1 <m_1[f_0]$, after rescaling the Boltzmann solution by its initial mass, since $\delta=\delta_{  sk}<1$, then  for  $sk>  \max\{\kgamma,\kb,\kc\}>1$   by condition \eqref{choose kb},  the constant $\mathcal C_{  sk} $   is  
\begin{align}\label{hardpotmom-3.3}
 \mathcal {C}_{  sk} & :=\left( \frac{ \beta_{  sk} }{ A_{s\kc} }\right)^{\frac{1}{c}} 
  m_1^{ sk+\frac{\gamma}2 + (1+\frac{\gamma}2)\left( 1+c_\gamma(sk)\right)  } [f_0],  
  \quad  \text{for any}\ \  sk\ge {  s}\max\{\kgamma, \kb,\kc\}>1\,.
\end{align} 

\end{remark}

\bigskip 

\begin{remark}\label{kc=eps}
This remark is fundamental to proof the existence of solutions to any moment either propagated or generated  for as long the initial data is a non-negative $f_0\in L^1_{2\kc}$, with  $\epsilon:=2(\kc -1) $, for $\epsilon\in (0,\gamma)$. In particular,  setting the factor  $s=1$, in the moments estimates, 
the  $sk^{th}$ - moments of the collision operator  remain valid when taking, for example   $\mu_{k}< \mu_{\epsilon+1}< \mu_1=
\|b\|_{ L^1(\S^{d-1})}=1$ and $m_0[f_0](0)=1$.  
With this choice, which suitable for the scalar binary elastic interacting Collisional form, whether both the initial mass $m_0[f_0]$  and the angular transition function $b(\hat u\cdot \sigma)\in  L^1(\S^{d-1})$ is renormalized to unit spherical average, naturally yields the  order $\kb$ and $\kgamma,$ defined in \eqref{bf kb} and \eqref{compare delta expo}, respectively, are larger than $\kc=(2+\epsilon)/2.$ 

This observation will be brought to readers' attention in the proof of Lemma~\ref{Lemma_ODE-extension2}, that extends the concept of solutions to any initial data $f_0\in  L^1_{2+\epsilon}. $
\end{remark}

 \bigskip
 
\begin{remark}\label{existence thm remark}
 The estimates obtained in  Theorem~\ref{mom-coll-op}  for the moments of the collision operator are sufficient  to prove the existence and uniqueness  of the Boltzmann flow solutions in Section~\ref{existence-uniqueness}, as posed in  Theorem~\ref{CauchyProblem} on Section~\ref{hardpotsection},  by means of Ordinary Differential Equations dynamics for time dependent flows in  Banach spaces $L^1_k(\R^d)$. 

More specifically,  Theorem~\ref{CauchyProblem} has two parts.  The first one  is to show existence and uniqueness  in $L^1_{\ell}(\R^d)$ for $\ell>2 +2\gamma$.
The second part    is shown in Lemma~\ref{Lemma_ODE-extension2} proving  existence  for lower order initial Lebesgue weights, that is in $L^1_{\ell}(\R^d)$ for $2+ \eps \le \ell <2 +2\gamma$.
\end{remark}

\bigskip

\subsection{Moment Ordinary Differential Inequalities for Boltzmann flow solutions in $L^1_{\ell}(\R^d)\times\R^+$}\label{moments-ODE's}

The following results derives   a priori estimates for the Lebesgue polynomial  moment  time dynamics to solutions  of the Boltzmann flow and rather accurate bounds as functions of the Cauchy Problem data.  These estimates  are, not only  invoked  in the second	part of the of the proof of Theorem~\ref{CauchyProblem} as mentioned in the last Remark~\ref{existence thm remark}, but also are critical for the moment summability  properties that yield  propagation and generation  exponential moments  globally in time with exponential rates fully characterized as a function of the coercive	constant $A_{s\kc}$ from \eqref{moment-factors} and the Cauchy problem data as well.   In addition these  global moment estimates and bounds  enable  global in time propagation estimates for solutions in  $L^p_\ell(\R^d)\times R^+$, for $1\le \ell\le \infty$.
 
\medskip

\begin{thm}[{\bf Lebesgue Polynomial Moments  \textit{A priori} Estimates for hard potentials}] 
\label{propagation-generation}

  Let $f(v,t)\geq0$ be a solution of the Cauchy problem associated to the initial value problem for the Boltzmann flow  in Theorem~\ref{CauchyProblem}, i.e.  with collision operator   $Q_\gamma(f,f)](v, t) $ as defined in  Theorem~\ref{mom-coll-op}. 

Let $0<c:=c_\gamma(sk):=\frac{\gamma/2}{sk-1}$ be the positive parameter defined in Theorem~\ref{mom-coll-op}, identity \eqref{ck}.  Then,  a super-solution to \eqref{momentineq}  is constructed  by solving the following Ordinary Differential Inequality 
  \begin{equation}\label{ODI-0}
 \frac {d}{dt} m_{  sk}[f](t)   \le \B_{  sk}\  -\  \frac{A_{{  s}\kc}}2\,  {m^{-c }_{1}[f_0]}  m^{1+c}_{  sk}[f](t)\,,\qquad sk\ge s\kmax >1,
  \end{equation}
  whose right hand side is an upper bound to the  $sk^{th}$-Lebesgue moment of the Boltzmann collision operator from  inequality \eqref{Qmoment-0} in the statement of Theorem~\ref{mom-coll-op}.
    The $sk$-dependent constants $A_{s\kc}$ and  $B_{sk}$ are determined in Theorem~\ref{mom-coll-op}, equations \eqref{moment-factors} to \eqref{RateE}, which depend
    solely from the  Cauchy problem data associated  for the Boltzmann flow stated in Theorem~\ref{CauchyProblem} and the $sk^{th}$-order  Lebesgue polynomial moment under consideration. 

 Hence, the following estimates and  bounds  hold for any fixed $sk$ value $s\in(0,1]$, and integer $k>1$. The first one is characterized by the following propagation of the initial data, after setting  $c=c_\gamma(sk)$.
 
{\bf Moment's propagation estimates:}  Let  $f_0\in L^1_{2sk}(\R^d)$ non-negative, or equivalent 
$m_{2sk}[f_0] $ is finite  for any $sk\ge s\kmax$. Then  
 \begin{align}\label{mom-prop}
m_{  sk}&[f](t) \leq \max\ \left\{ m_{  sk}[f_0] \ ; \mathbb{E}_{c_\gamma(sk)}\right\}=m_{  sk}[f_0] \, \mathbb{I}_{m_{  sk}[f_0] < \mathbb{E}_{c_\gamma(sk)}} + \mathbb{E}_{c_\gamma(sk)}\,\, \mathbb{I}_{m_{  sk}[f_0] \ge \mathbb{E}_{c_\gamma(sk)}}  \nonumber \\
 &= m_{  sk}[f_0] \, \mathbb{I}_{m_{  sk}[f_0] < \mathbb{E}_{c_\gamma(sk)}}  \nonumber \!\\
 &\qquad+\!  \left( m_1\,\ttm_{1,k}\right)^\frac{c}{c+1}[f_0] 
  \left[ \left(\! \frac{ \beta_{  sk}}{ A_{{  s}\kc} }\!\right)^{\frac{2(sk-1)}{\gamma} }\! \mathbb{I}_{ k\ge \kgamma }
  \!+\! \left(\!\frac{ \beta_{sk}}{ A_{{  s}\kc} }\!\right)^{sk+\frac\gamma2}  \mathbb{I}_{ k <  \kgamma } \right] \!  \mathbb{I}_{m_{  sk}[f_0] \ge  \mathbb{E}_{c_\gamma(sk)}},
      \end{align}
  or, equivalently, for the rate factors from \eqref{RateE} 
      \begin{align}
  m_{  sk}&[f](t) \leq m_{  sk}[f_0] \, \mathbb{I}_{m_{  sk}[f_0] < \mathbb{E}_{c_\gamma(sk)}} \!+\!  \left( m_1\,\ttm_{1,k}\right)^\frac{c}{c+1}[f_0] 
  \, R_\mathbb{E}(\delta^{-1}_{sk})\,  \mathbb{I}_{m_{  sk}[f_0] \ge  \mathbb{E}_{c_\gamma(sk)}}, \nonumber
    \end{align}
uniformly in $ t>0,$\ for any $ sk> s\kmax.$ 

Global moments generation estimates  are characterized by  the generation of any polynomial $sk^{th}$-order from just an initial data $f_0\in L^1_{2s\kmax}(\R^d)$,   with $s\kmax >1$, are as follows. \\

{\bf Moment's generation estimates:} Let  $f_0\in L^1_{2s\kmax}(\R^d)$  non-negative, or equivalent 
$m_{2s\kmax}[f_0] $ is finite, that all moments  $m_{sk}[f](t)$ are finite for $sk>s\kmax>1$, with $\min_{t\to 0} m_{sk}[f](t) =\infty$ satisfying the following  global estimate for positive times, that is
\begin{align}\label{mom-gen}
  m_{  sk}&[f](t) \le   \mathbb{E}_{c_\gamma(sk)} + m_1[f_0]\bigg( \frac1{c  A_{{  s}\kc}}\bigg)^{\frac{1}{c}} \, t^{-\frac{1}{c} } \\
  & = \left( m_1\ttm_{1,{  sk}}\right)^{\!\frac{sk-1}{sk-1+\frac\gamma2}}\![f_0] 
  \left[ \left( \frac{ \beta_k}{ A_{s\kc} }\right)^{\frac{2(sk-1)}{\gamma} } \!\!\mathbb{I}_{ k\ge \kgamma }\!+\!
  \left(\frac{ \beta_{  sk}}{ A_{s\kc} }\right)^{sk+\frac\gamma2} \! \!\mathbb{I}_{ k< \kgamma } \right]  \! +\! m_1[f_0]\bigg( \frac{2(sk-1)}{\gamma\,  A_{{  s}\kc}}\bigg)^{\!\!\frac{2(sk-1)}{\gamma}} \!\!\! t^{-\frac{2(sk-1)}{\gamma} }\!, \nonumber
  \end{align}
 or, equivalently, for the rate factors from \eqref{RateE} 
      \begin{align*}
  m_{  sk}[f](t) \leq    \left( m_1\,\ttm_{1,k}\right)^\frac{c}{c+1}[f_0] 
  \, R_\mathbb{E}(\delta^{-1}_{sk})\,  \mathbb{I}_{m_{  sk}[f_0] \ge  \mathbb{E}_{c_\gamma(sk)}} +\! m_1[f_0]\bigg( \frac{2(sk-1)}{\gamma\,  A_{{  s}\kc}}\bigg)^{\!\!\frac{2(sk-1)}{\gamma}} \!\!\! t^{-\frac{2(sk-1)}{\gamma} }, \nonumber
    \end{align*}

 uniformly in  time $t>0,$ for any $ sk> \max\{ s\kb,s\kc \}=s\kmax.$
 
 For both statements,  constants  $\beta_{sk},\  \mathcal{C}_{sk} $ and the equilibrium state $\mathbb{E}_{c_\gamma(sk)} $,  are obtained in Theorem~\ref{mom-coll-op}, are  characterized in \eqref{moment-factors},  \eqref{mathcal-C-kc} and  \eqref{RateE}, respectively.

\end{thm}

\medskip

\begin{proof}  
Starting from  the moments estimates and bounds for the Boltzmann collisional integral \eqref{collision3}, \eqref{collision3.2}  in the Banach space $L^1_q(\R^d)$, with $q=2sk$,  developed in Theorem~\ref{mom-coll-op}, 
the proof of the current theorem  consists in  estimating  the Ordinary Differential Inequalities  for   $sk^{th}$-Lebesgue moments  associated to the solutions of the Boltzmann flow, as developed  in Lemma~\ref{mom-lemma-I}, whose right hand side of  inequality \eqref{momentineq} was estimated in much detailed by analyzing the $sk^{th}$-Lebesgue moment bounds for  the collisional integral. 

%

More  specifically, starting from  right hand side of the  moments inequality in \eqref{momentineq},  and invoking the estimates from Theorem~\ref{mom-coll-op}, it follows that any solution for the Boltzmann flow  with initial data $f_0(v)\in L^1_{2+\eps}$ satisfies, for any $\eps\ge\max\{s\kb,s\kc\}-1>0$ and  $\gamma\in(0,2]$,  
\begin{align}\label{sharp coercive}
\frac {d}{dt} m_{  sk}[f](t) \ &<  \B_{  sk} -  \frac{A_{{  s}\kc}}2\,m^{-c}_{1}[f_0]\,m_{  sk}^{1+c}[f](t)\,,\quad {  sk}\ge \max\{s\kb,s\kc\}>1, 
\end{align}
{where} the $k$-independent  coercive factor $A_{s\kc}$ was defined on \eqref{Akc}, and  the $k$-dependent positive factor $ \B_{sk} =\beta_{sk}\,   \mathcal {C}_{  sk}$  defined in \eqref{Bkc} and $\beta_{sk}$ defined in \eqref{betak}.
\\ 

Hence, the proof of Theorem~\ref{propagation-generation}  needs to be completed by showing the  that, for each $sk$ fixed  satisfying ${  sk}\ge \max\{s\kb,s\kc\}>1$, its corresponding  super linear Ordinary Differential Inequality \eqref{sharp coercive}
  admits an upper barrier, that is a super solution with a finite bound independent of time, for each $sk$ fixed, by either propagating or generating  the  initial data, supplemented by a comparison principle to ODI's, under the assumption that $m_{sk}[f](t)$ is finite for any  positive time $t$.

Therefore, the goal is to construct  to upper bound form of the right hand side  the initial value problem associated to the Ordinary  Differential Inequality \eqref{sharp coercive}, after setting $c=c_\gamma(sk)= \gamma/(2(sk-1))$ once again, 
\begin{align}\label{ODE1}
\begin{cases}&y'(t) <  \B_{sk}  -  A_{s\kc}\, m_1^{-c}[f_0]y^{1 + c}(t)\,, \quad\text{for}\ c(k)=\frac{\gamma/2}{sk-1}, \quad t>0\,, \\
&y(0)=y_0;
\end{cases}
\end{align}
where $y(t) = m_{sk}[f](t)$ for $t>0$, and  the corresponding value of initial data is $y_0$ finite that may be any $j$-moment   $m_{j}[f](0)$ finite with $1<j\leq {  sk} $.

In particular, for the case where $s=1$ and $j=k$, the propagation of moment results hold if the initial data is $y(0)=y_0= m_k[f_0](0)$.  The generation of moments result  holds if   just  $j=\max\{s\kb,s\kc\}>1$,   $y(0)=y_0= m_j [f_0](0)$, i.e. only mass, kinetic energy and the $j^{th}$-moments are bounded, while estimating $y(t)=m_{sk} [f](t)$ for any $sk>j$.
Both results follow from comparison arguments for Ordinary Differential Inequalities whose right hand side are continuous in $t$ and Lipschitz in $y=y(t).$

 By classical  ODEs comparison arguments (or the maximum principle in ODEs), it is easy to verify that any  $y_*(t)$, the solution of the upper ODE problem to \eqref{ODE1} with the initial data $y_*(0):=y_{*,0} \ge y_0$, that is 
\begin{align}\label{upperODE}
\begin{cases}&y_*'(t) = \B_{  sk}  -  A_{{  s}\kc}\,  m_1^{-c} y_*^{1 + c}(t)\,,\qquad t>0\,,  \qquad  \\
&y_*(0)=y_{*,0}, 
\end{cases}\end{align}
controls the solution from above,  i.e. $y(t)<y_*'(t)$ unifomly in time $t\ge0$.

 It is relevant to notice identify the equilibrium state  $\mathbb{E}_{c_\gamma(sk)}$ associated to  this upper ODE  \eqref{upperODE} is exactly the unique root calculated for the moment operator $\mathcal L_{sk}(y)$ introduced in \eqref{RateE} and \eqref{equi root}, respectively,  where can be easily verify  from the representation \eqref{equi-1.1}, and the fact 
 $\lim_{sk\to 1} c_{\gamma}(sk)=\infty$, that 
 \begin{equation}\label{equi-solution}
 \mathbb{E}_{c}:= m_1^{\frac{c}{1+c}}[f_0]\left(\frac{\B_{sk}}{A_\kc}\right)^{\frac1{1+c}}
\longrightarrow_{{  sk}\to 1}\  m_1 [f_0]\, ,
\end{equation}
 is the unique $y_*^{\text{eq}}=\mathbb{E}_{c}$  time independent equilibrium solution  of   $y'_*=0$. 
 
 In addition, since the ODE \eqref{upperODE}  is autonomous,  clearly satisfies that 
 the solution  $y_*(t)$ monotonically increases in time $t$  to the equilibrium state $y_*^{\text{eq}}$, if the initial data $y_0 <  y_*^{\text{eq}} $; while monotonically decreases to $y_*^{\text{eq}}$,  provided $y_0 >  y_0^{\text{eq}}$.
  
 Therefore,  the propagation of the $k^{th}$-moment  property follows easily from  such    comparison arguments between the solution of ODI \eqref{ODE1}, and the one for the ODE  \eqref{upperODE},    yields the upper control   to the solution of the initial value problem \eqref{ODE1}, with the following propagation of $y(t)=m_k[f](t)$ associated to the initial data  $y(0)=y_0=m_k[f_0](0)$,   that is 
 \begin{equation*}
 y(t)\le \max\Big\{  y(0)\, ; \, \mathbb{E}_{c} \Big\}\, , \nonumber\\
 \end{equation*}
 or, equivalently,  for the initial data $y_0=m_{sk}[f_0]$  the global in time bound for each $sk^{th}$-moment  
 \begin{align}\label{moments-bounds}
m_{sk}[f](t) &\leq \max\Big\{ m_{{  sk}}[f_0] \, ;\,  m_1^{\frac{c}{1+c}}[f_0]\big(\tfrac{\B_{  sk}}{A_{{  s}\kc}}\big)^{\frac{1}{1+c}}\Big\}\,,  \quad\text{for}  \quad {  sk} > 1, \quad t>0, 
\end{align}
with the coercive constant factor $A_{{  s}\kc}$ and the upper bound constant $\B_{  sk}$ defined on \eqref{Bkc}.  
In particular,   this estimate   \eqref{moments-bounds}  complete the moments propagation property stated  \eqref{mom-prop}  of  Theorem~\ref{propagation-generation}.

\medskip

The part of Theorem~\ref{propagation-generation}'s proof concerning the   ${  sk}^{th}$-moments  generation associated to the Boltzmann flow, 
for any $sk$-order with only  initial finite mass and energy and  a moment of order $1+\eps$, with $\eps\ge s\kmax-1.$ This statement   means  that there is an instantaneous regularization, in the sense that it is possible  to find an upper bound to solutions $m_{  sk}[f](t)$,   \eqref{upperODE}, or equivalently,  \eqref{ODE1} for $t>0.$\\

The following comparison arguments for generation of moment argument confronts  the fact that the  expected moment estimate at the initial time is actually defined by the unbounded $m_{sk}[f](0)=+\infty$, for $sk>s\max$.  Yet, the following lemma remains true, as we describe the comparison argument.

\begin{lem}{\bf ODI's Comparison Lemma for Moments Generation } \label{ODI Comparison} 
Let $A, B$ and $c$ be positive constants and consider a function $y(t)$ which is absolute continuous in $t\in(0,\infty)$ and satisfies 
\begin{equation*}
y'(t) \leq  B  -  Ay^{1 + c}(t)\,, \qquad \text{for} \ \ c>0, \quad t>0\,.
\end{equation*}
Then,
\begin{equation*}
y(t)\leq \bar{y}(t):=E\bigg(1+ \frac{K}{t^{\beta}}\bigg)\,,\qquad t>0\,,
\end{equation*}
for the choice
\begin{equation*}
E=\mathbb E_c\bigg(\frac{B}{A}\bigg)^{\frac{1}{1+c}}\,,\qquad \beta=\frac{1}{c}\,,\qquad K=\bigg(\frac{1}{c\,A}\bigg)^{\frac1c}\bigg(\frac{A}{B}\bigg)^{\frac{1}{1+c}}\,.
\end{equation*}
\end{lem}
\begin{proof} 
Note that in one hand
\begin{equation*}
\bar y'(t) = -\beta E K t^{-(\beta+1)}\,.
\end{equation*}
On the other hand,
\begin{equation*}
B-A\,\bar y^{1+c}(t) = B - AE^{1+c}\bigg(1+\frac{K}{t^\beta}\bigg)^{1+c}\leq B - AE^{1+c}\bigg(1+\frac{K^{1+c}}{t^{\beta(1+c)}}\bigg)
\end{equation*}
where we invoked the binomial inequality  $( 1+ \frac{K}{t^\beta} )^{1+c} \geq  1+ (\frac{K}{t^\beta})^{1+c}$.  With the choice of $E,\,\beta$ and $K$ as suggested in the statement of the lemma it follows then 
\begin{equation}\label{upbd}
B-A\,\bar y^{1+c}(t) \leq \bar y'(t)\,,\qquad t>0\,.
\end{equation}
Next, note that for any $\epsilon>0$, the translation $\bar{y}_\epsilon(t):=\bar{y}(t-\epsilon)$ satisfies \eqref{upbd} for $t>\epsilon$ due to the time invariance of the inequality.  Moreover, since $y(t)$ is absolute continuous in $t\in[\epsilon,\infty)$, there exists $\eta_*>0$ such that $\bar{y}_{\epsilon}(\epsilon+\eta)\geq y(\epsilon+\eta)$ for any $\eta\in(0,\eta_*]$.  Consequently, we conclude the following setting for any $\epsilon>0$ and $\eta\in(0,\eta_*]$:
\begin{align}\label{ODE02}
\begin{cases}
&y'(t) \leq B  -  A\,y^{1 + c}(t)\,,\qquad  t>\epsilon+\eta\,\\
&\bar y'_\epsilon(t) \ge B  -  A\, \bar y_\epsilon^{1 + c}(t)\,,\qquad  t>\epsilon+\eta\,,\\
&+\infty >\bar y_\epsilon(\epsilon+\eta) \geq y(\epsilon+\eta)\,. 
\end{cases}
\end{align}
Consequently, a standard comparison in ode implies that $\bar{y}_\epsilon(t)\geq y(t)$ in $t>\epsilon+\eta$ valid for any $\epsilon>0$ and $\eta\in(0,\eta_*]$.  In this particular case this fact can be readily proved denoting
\begin{equation*}
\Phi(t) = \frac{y^{1 + c}(t) - \bar y_\epsilon^{1 + c}(t)}{y(t) - \bar y_\epsilon(t)}\,.
\end{equation*}
Since $ |x^{1+c} - y^{1+c} |\leq (1+c)\max\{x^c , y^c \}|x-y|$, it follows that $\Phi(t)$ is bounded by
\begin{equation*}
\big| \Phi(t) | \leq (1+c)\max\big\{ y^c(t)\,,\,\bar y^c_\epsilon(t) \big\}<+\infty\,.
\end{equation*}
Thus, for $w(t)=y(t) - \bar y_\epsilon(t)$ it follows from \eqref{ODE02} that
\begin{equation*}
w'(t) + A\,\Phi(t)\,w(t)\leq0\,,\qquad w(\epsilon+\eta)\leq0\,,\qquad t>\epsilon+\eta\,,
\end{equation*}
which implies that
\begin{equation*}
w(t)\leq w(\epsilon+\eta)e^{-A\int^{t}_{\epsilon+\eta}\Phi(s){\rm d}s}\leq 0\,.
\end{equation*}
Now, sending first $\eta$ to zero and then $\epsilon$ to zero it follows that $\bar{y}(t)\geq y(t)$ for any $t>0$ which is the statement of the lemma.
\end{proof}
Recall that the moments of $f(v,t)$ satisfy the inequality \eqref{ODE1}, specifically, setting $c=c_\gamma(sk)=\frac{\gamma}{ 2(sk-1)}$
\begin{equation*}
m_{sk}'(t) \leq  \B_{sk}  -  A_{s\kc}\, m_1^{-c}[f_0]\,m_{sk}^{1 + c}(t)\,, \quad  \text{for any}\ \  sk\geq s\kmax >1, \quad  t>0\,.
\end{equation*}
  Thus, applying  Lemma~\ref{ODI Comparison} with
\begin{equation*}
B=\B_{sk}\,,\quad A = A_{s\kc}\, m_1^{-c}[f_0]\,,\quad 
\end{equation*}
it follows that
\begin{align}\label{generation-bound}
   m_{  sk}[f](t) \ &\le \ \bigg(\frac{\B_{  sk}}{A_{{  s}\kc}}\bigg)^{\frac{c}{ 1+c}} 
   m_1^{\frac 1{1+c}} [f_0] + 
    m_1[f_0]\bigg( \frac{1}{c\, A_{{  s}\kc} \, t}\bigg)^{\frac{1}{c }} \nonumber
    \\
   & \le \  \mathbb E_{c_\gamma(sk)} + 
    m_1[f_0]\bigg( \frac{1}{c\, A_{{  s}\kc} \, t }\bigg)^{\frac{1}{c }} \,.
\end{align}

That is, $m_{  sk}[f](t)$ is finite for any $t>0$  and for all ${  s}k\ge { s}\max\{\kb , \kc\} >1$ fixed.  In addition, it becomes uniformly bounded in time for say $ t\in[1,\infty)$ satisfying 
\begin{equation} \label{generation-bound-limit}
\lim_{t \to \infty}  \left( m_k[f](t) -  \mathbb E_{c}  \right) \le m_1[f_0]\,  \lim_{t \to \infty}  \left(  \frac{1}{ c \, A_{{  s}\kc}\, t}\right)^{\frac{1}{c}} = 0\, , 
\end{equation}
since the fixed  term $A_{s\kc}>0$ for any choice of  $ {\kc}>1$,  which is  the coercive  factor  defined in \eqref{moment-factors}.

\medskip
The proof of Theorem~\ref{propagation-generation} is completed after  the classification  of the  {\bf Moments bounds estimates in terms of the coercive constant $A_{{  s}\kc}$. }
After combining the definitions  of the coercive constant factor $A_{s\kc}$,  and the upper bound constant $\B_{  sk}$ in \eqref{Bkc},  all depending on ${  sk}\ge \max\{s\kb,s\kc\}>1$, a more explicit expression for the upper bound generation  emerges when specifying the equilibrium state $\mathbb E_{c_\gamma(sk)}$ as a function of $\beta_{sk}$ and $\C_{sk}$ and the initial moments associated to $m_0[f_0]$ and  $m_1[f_0]$,  as characterized in \eqref{equi-1}.
Therefore,   the moment bounds estimates, focusing on the dependance of the coercive constant $A_{{  s}\kc}$ are recasted, after invoking the moment propagation estimates  \eqref{mom-prop},  as stated in Theorem~\ref{propagation-generation},

{\bf Moment's propagation estimates:} 
 \begin{align}\label{moments-bounds-2}
m_{  sk}[f](t) & \!\leq  \!\max \!\left\{ m_{  sk}[f_0] \, ;  
  \,  \left( m_1\,\ttm_{1,k}\right)^\frac{ks-1}{sk-1+\frac\gamma2}[f_0] \left[ \left( \frac{ \beta_{  sk}}{ A_{{  s}\kc} }\right)^{\frac{2(ks-1)}{\gamma} } \mathbb{I}_{ k\ge \kgamma } \!+\!
  \left(\frac{ \beta_k}{ A_{ s \kc} }\right)^{sk+\frac\gamma2  }  \mathbb{I}_{ k< \kgamma } \right] \right\}, 
  \end{align}
uniformly in $ t>0,$\ for any ${  s}k> {  s}\kmax $.  Analogously, the moments generation estimates \eqref{mom-gen}  are characterized by {\bf Moment's generation estimates:}
\begin{align}\label{generation-bound-2}
  m_{sk}[f](t) &\le  
  \left( m_1\,\ttm_{1, sk}\right)^{\!\frac{ks-1}{sk-1+\frac\gamma2}}[f_0] 
  \left[ \left( \frac{ \beta_{sk}}{ A_{s\kc} }\right)^{\frac{2(sk-1)}{\gamma}} \!\! \mathbb{I}_{ k\ge \kgamma }\!+\!
  \left( \frac{ \beta_{sk}}{ A_{s\kc} }\right)^{sk+\gamma/2  } 
\!\!  \mathbb{I}_{ k< \kgamma } \right] \nonumber\\
   &\qquad \qquad \qquad \qquad\qquad
  \qquad + m_1[f_0]\bigg( \frac {2(sk-1)}{\gamma\, A_{s\kc}\, t}\bigg)^{\frac{2(sk-1)}{\gamma}} \, . 
\end{align}
In addition, if $sk> \max\{\kb,\kc \}>1$, then $\beta_{sk} A^{-1}_{\kc} >1.$ 
 
 Thus, using the notation \eqref{RateE} as defined earlier, one can write 
 \begin{align}
\left( \frac{ \beta_{sk}}{ A_{s\kc} }\right)^{\frac{2(sk-1)}{\gamma}} \!\! \mathbb{I}_{ k\ge \kgamma }\!+\!
  \left( \frac{ \beta_{sk}}{ A_{s\kc} }\right)^{sk+\gamma/2  } =  R_\mathbb{E}(\delta^{-1}_{sk})= R_\mathbb{E}(\beta_{sk}/A_\kc).
  \end{align}
The proof of Theorem~\ref{propagation-generation} is now completed.
\end{proof}

\medskip

\begin{remark}\label{coercive-remark} {\bf Coerciveness.} {\em The control of the upper bound any ${sk}^{th}$-moment that solves the differential inequality \eqref{Akc}, degenerates as the term    $A_{{ s}\kc} := c_{\Phi}2^{-\frac{\gamma}2} (\|b\|_{L_1(\mathbb{S}^{d-1})}-\mu_{\kc})  {m_0[f_0]}$, defined in \eqref{mu_r}, would vanish.  The derivation of this term 
 depends on the choice of  of the positive term \, $0<\delta_{ sk}<1$ made in \eqref{delta-k}, for any  $sk>s\max\{\kb,\kc \}>1 $. This choice  depends on the decaying property of  the contractive constant $\mu_{s}k$  for any ${ sk}\ge {{s}\kc} >1$ calculated in the Angular Averaging Lemma~\ref{BGP-JSP04},  the pointwise factors  for the potential  from \eqref{pot-Phi_1} and  \eqref{pot-Phi_2}  depending  on the data.
In addition,  that the constant $A_{{s}\kc}$ is calculated without using that $f$ solves the Boltzmann equation, it only depends
on the  spherical integrability (or angular averaging) of the transition probability $B(\Phi(|u|), \hat u\cdot\sigma) \in L^1(\mathbb{S}^{d-1}, d\sigma)$, and   lower constants $c_\Phi$ and $C_\Phi$ \eqref{pot-Phi_1} and \eqref{pot-Phi_2}, associated to the the potential rate $\Phi(v).$ 

That means   $A_{{ s}\kc}$ can be viewed as the  {coercive constant} related to the Cauchy problem  to Boltzmann flow for elastic interactions with transition probabilities corresponding to  hard potentials   with an integrable angular transition probability, posed in the Banach space  with the natural $k$-Lebesgue weight norm.  

Hence, the a priori ${sk}^{th}$-moment estimates obtained in Theorem~\ref{propagation-generation} induced an  upper uniform bound in time for each ${ sk}$ fixed with a  coercivity property for any ${sk} \ge { s}\max\{\kb,{\kc}\} >1$.  One can  view  these a priori estimates   as  analog to significance of the Sobolev-Poincare and coercive constants  associated to elliptic and parabolic equations in divergence form.   The corresponding existence and uniqueness theorem in these Banach spaces follow in the next section.
}

\end{remark}

\bigskip

We conclude this section by stating the corresponding $L^1_k(\mathbb{R}) $ or $K^{th}$-moment estimates for the Boltzmann flow with Maxwell type collisional forms corresponding to potential rates $\gamma=0$, corresponding to transition probability with  potential  rate $\gamma=0$ and integrable angular factor $(\hat u\cdot \sigma)$ 

\subsection{\em Case for Maxwell type of interactions corresponding to  $\gamma=0$. } This is the well known case   where the estimate \eqref{momentineq}  
is obtained by a moment recursion formula for classical cases where the transition probability of collision kernel  simple depend on an integrable $B(\hat u\cdot\sigma) \in L^1(\mathbb{S}^{d-1}, d\sigma)$. 

However, the interpolation estimate techniques developed  in 
Theorem~\ref{mom-coll-op} and Theorem~\ref{propagation-generation} extend to the case  $c_{\Phi} b(\hat u\cdot\sigma) \le B(\Phi(|u|), \hat u\cdot\sigma) \le C_{\Phi} b(\hat u\cdot\sigma)$, with $b(\hat u\cdot\sigma)\in  L^1(\mathbb{S}^{d-1}, d\sigma)$ and $c_{\Phi} $ and $C_{\Phi} $ the bounds from \eqref{pot-Phi_1} and \eqref{pot-Phi_2}, respectively.

Indeed, 
similar estimates by above for  the gain operator by above using interpolation,  while an exact negative identity for the moments of the loss, which clearly maintains the linearity property of the moments estimates,  to obtain  the  corresponding ODI
\begin{align}\label{moments odi Maxwell}
\frac {d}{dt} m_k[f](t) & =  m_k[Q_0[f,f](t) \leq C_{\Phi}2^{k}\mu_km_{1}[f_0]m_{k-1}[f] - c_{\Phi} (\|b\|_{L^1(\mathbb{S}^{d-1})}-\mu_k)m_{0}[f_0]m_{k}[f] \nonumber \\[3pt]
&\leq K_{1,k}[f_0]m^{\frac{k-2}{k-1}}_{k}[f] - (\|b\|_{L^1(\mathbb{S}^{d-1})}-\mu_k)m_{0}[f_0]m_{k}[f]\\[3pt]
&\leq K_{2,k}[f_0] - \tfrac{c_{\Phi}}{2}(\|b\|_{L^1(\mathbb{S}^{d-1})}-\mu_k)m_{0}[f_0]m_{k}[f] = 
 K_{2,k}[f_0] - \frac{c_{\Phi}}2 A_\kc m_{0}[f_0]m_{k}[f],\nonumber\\[3pt]
 &=: \mathcal L^M_k(m_{k}[f]), \qquad 
 \quad k\geq \kc\ge1\,, \nonumber
\end{align}
 where $\mathcal L^M_k(y)$ is a linear operator  in $y=m_{k}[f]$, with 
  constants   $A_\kc$  the coercive one for $\kc>1$,  and   $ K_{i,k}[f_0]$, for $I=1,2$, depend on  $C_{\Phi}, \mu_k$ the interpolation factors. This case can only  fulfilled the moment propagation  property of a  given initial data, as the negative term in the $k^{th}$-moment ODI   makes a linear upper ODE  with an absorption term $m_k[f](t)$ and with a bounded in time right hand side.  Hence,  this inequality leads to estimate \eqref{mom-prop} and the is no mechanism that can generate a moment or order $k$ from just an initial data  with bounded mass and energy.  More precisely, for any $f_0(v) \in L^1_{2k}(\R^d)$ the $k^{th}$ Lebesgue moment	 associated to the Maxwell type of interactions  Ordinary Differential Inequality \eqref{moments odi Maxwell} is global in time controlled by 
 \begin{align}\label{moments odi Maxwell2}
 m_k[f](t) &\le  \mathbb E_k   + m_k[f_0] \exp (-\frac{c_{\Phi}}2 A_\kc m_{0}[f_0] \, t), \\[4pt]
 \text{with} \qquad \mathbb E_k &:=\frac{2K_{2,k}}{ c_{\Phi} A_\kc m_{0}[f_0] }\ \text{being the unique root of  linear operator}  \ \mathcal L^M_k(m_{k}[f])\, \nonumber
 \end{align}
 making  $ \mathbb E_k $ the unique equilibrium state of an associated  upper ODE $y'(t)= \mathcal L^M_k(y(t))$.
 
 In particular the existence theorem to be developed in the next sections applies to this case as well, although the classical theory of ODEs, for comparisons to linear theory, since the collision operator can be shown to be a Lipschitz form, and the existence and uniqueness of  $f(v,t)$ solving the Cauchy problem if obtained by a  classical Picard type iteration constructs a fulfilling the global moments bounds \eqref{moments odi Maxwell2}. Much work has been written for the theory of conservative and non-conservative models for Maxwell type of interactions, and a sample of them can be found at \cite{Bo76, B88, BG-06, BoCeGa2-06,BG-17-maxw-bounds, Bobylev-book20} and references therein.
 
 \medskip

\bigskip

\section{Proof of  existence and uniqueness Theorem \ref{CauchyProblem}}\label{existence-uniqueness}

All  estimates  for $k^{th}$-Lebesgue polynomial moments for $1<k\in \R$,   obtained in Theorem~\ref{mom-coll-op},  are enough to  prove the  existence and uniqueness of solutions to the initial value problem \eqref{collision3}, for the space homogeneous problem of Theorem~\ref{CauchyProblem} for Maxwell and hard potentials for $\gamma\in[0,2]$ and integrable angular scattering transition probability rate.
This result is based on the following result on Banach spaces, Theorem~\ref{Theorem_ODE}, which was initially proposed in the unpublished note \cite{bressan} in order  to solve the Boltzmann equation for hard spheres in three dimensions, that is $\gamma=1$, and uniform scattering kernel.  While the ideas to prove this theorem are somehow classical, see for example \cite{Martin}, we revisit  this approach here showing     the use of the a-priori estimates on the $k^{th}$-Lebesgue polynomial moments  and their  impact to verify that  the   sub-tangent condition is satisfied,  for any real valued $1<k\in \ \R$. The proof of the  sub-tangent   is crucial for the selection of the solvability invariant set. This point not covered in the proposed proof  in  \cite{bressan}.   The following argument fully completes the argument sub-tangent condition being a sufficient one.  It also extends the argument for the aforementioned more general models for of collision cross sections for the homogeneous Boltzmann equation for elastic binary interactions, and other broader problems such collisional multi-species gas mixtures, as we well as polyatomic gas models.

Following the notation  used earlier for the Lebesgue norm  $\|g\|_{L^1_\ell(\R^d)}$ for $\ell\in \R$,  and the  $k^{th}$-moment  $m_k=\int_{\R^d} g(v) \langle v\rangle^{2k} dv$, as defined in\eqref{k-norm} and \eqref{k-moment}, respectively,  notice that the norm and moments coincide for a positive $g(v)$ and $2k=\ell$.  
 
Hence, focusing the  proof of initial value problem stated Theorem~\ref{CauchyProblem}, initial data is given by moments of positive probability   density  functions taking the identical notation   $m_0:=m[f_0] =\|f_0\|_{L^1(\mathbb R^{d})}$ and  $m_1:=m_1[f_0] =\|f_0\|_{L^1_{2}(\mathbb R^{d})}$ which are  initial mass and energy invariants of the Boltzmann flow, respectively.  In addition, since $k^{th}$-moments associated to to time dependent probability densities $f(v,t)$,   the short notation $m_k(t) := m_k[f](t)$ will be used, unless it is relevant  for the the Boltzmann solution flow.

%
%

\begin{thm}\label{Theorem_ODE}
Let $\mathcal{E}:=(\mathcal{E},\|\cdot\|)$ be a Banach space, $\Omega$  be a bounded, convex and closed subset of $\mathcal{E}$, and $Q:\Omega\rightarrow \mathcal{E}$ be an operator  satisfying the following properties:
\begin{itemize}
\item[{\bf i)}]  H\"{o}lder continuity condition
\begin{equation}\label{Holder_C}
\big\|Q[f] - Q[g]\big\| \leq C\|f - g\|^{\beta},\quad \beta\in(0,1), \quad \forall\,f,g\in\Omega\,.
\end{equation}
\item[{\bf ii)}]  Sub-tangent condition
\begin{equation}\label{subtangent}
\liminf_{h\rightarrow0^+}h^{-1}\text{Dist}\big(f+h\,Q[f],\,\Omega\big)=0,\qquad \forall\,f\in\Omega\,,
\end{equation}
where $\text{Dist}(v,\Omega) :=  \inf_{\omega\in\Omega} \|v-w\|$.
\item[{\bf iii)}]   One-sided Lipschitz condition 
\begin{equation}\label{one-sided-lip}
\big[ Q[f] - Q[g], f - g \big]_{-}\  \leq C\ \|f - g\|,\qquad \forall\,f,g\in\Omega\,,
\end{equation}
where $\big[ \varphi,\phi \big]_{-}\ : = \  \lim_{h\rightarrow 0^{-}}h^{-1}\big(\| \phi + h\varphi \| - \| \phi \| \big)$.
\end{itemize}
Then, the abstract \text{ode} problem 
\begin{equation*}
\partial_t f = Q[f] \quad\text{on}\quad [0,\infty)\times \mathcal{E},\qquad f(0)=f_0 \in \Omega,
\end{equation*}
has a unique solution in $\mathcal{C}([0,\infty),\Omega)\cap\mathcal{C}^1((0,\infty),\mathcal{E})$.
\end{thm}

\begin{remark}
In the Banach space $L^{1}_{1}$, the one-sided Lipschitz bracket reduces to
\begin{align*}
&\big[\varphi, \psi\big]_{-}\  = \ \lim_{h\rightarrow 0^{-}}h^{-1}\big(\| \varphi + h\psi \|_{L^1_{2}} - \| \varphi \|_{L^1_{2}} \big)\\
&\hspace{-.5cm}= \lim_{h\rightarrow 0^{-} } \int h^{-1}\Big(\big| \varphi + h \psi \big | - | \varphi | \Big) \ \langle v \rangle^{2} \ dv  =     \  \int \psi \,\text{sign}(\varphi)  \ \langle v \rangle^{2} \ dv \,.
\end{align*}
The last step follows by Lebesgue's dominated convergence because
\begin{equation*}
\big| h^{-1} \big|\,\Big|\big| \varphi + h \psi \big | - | \varphi | \Big| \leq |\psi|\,,
\end{equation*}
and,
\begin{equation*}
h^{-1}\Big(\big| \varphi + h \psi \big | - | \varphi | \Big) = \frac{ 2\,\varphi\,\psi + h\,\psi^{2} }{\big| \varphi + h \psi \big | + | \varphi | } \rightarrow \psi \,\text{sign}(\varphi)\quad\text{as}\quad h\rightarrow0\,.
\end{equation*}
We point out that  the one-sided Lipschitz condition \eqref{one-sided-lip} is sufficient to obtain uniqueness for the Boltzmann problem as noticed by \cite{diblassio}.
\end{remark}

The proof of theorem~\ref{Theorem_ODE} is shown in the Appendix,  in the exact version as done by A. Bressan in the unpublished work of \cite{bressan}.  For completeness, it is included in his manuscript.

\medskip

\begin{proof}{\em The Cauchy problem as posed  in  Theorem~\ref{CauchyProblem}.}
 without loss of generality  taking just  $s=1$. This proof consists in two stages. The first stage  shows that  conditions {\bf i), ii)} and {\bf iii)} from Theorem~\ref{Theorem_ODE} are fulfilled  for the collision operator $Q:=Q_\gamma(f,f)(v)$ as defined in the statement of  Theorem~\ref{mom-coll-op} with all properties proved in that theorem, after setting Banach space $\E := L^1_2(\R^d)$
and the subset $\Omega \subset L^1_2(\R^d)$ defined as follows. 
Set $k=1+\gamma$ be moment $m_{1+\gamma}[f]=\int_{\R^d}f(v)\langle v\rangle^{2(1+\gamma)}dv$ and consider the constant $\mathfrak{h}_{1+\gamma} := \mathbb{E}_{1+\gamma} + \B_{1+\gamma}$, where 
 $\mathbb{E}_{1+\gamma}$ is the only root of the operator $\mathcal L_{k}(y):=\B_{k} -  \frac{A_{\kc}}2 m^{-c_\gamma(k)}_1[f_0] y^{1+c_\gamma(k)}$ as defined in \eqref{Qmoment-0}.

Hence, define the   bounded convex subset $\Omega\subset L^{1}_{2}({\R}^{d}) $ by 
%
\begin{equation}\label{SetOmega}
\begin{split}
\Omega_{\gamma} := &\Big\{ \ g \in L^{1}_{2} (\real^d) \ : \ g \geq 0,
\int g(v)\,dv \!=\!m_0,  \int g(v) |v|^2 dv \!=\! m_{1}-m_0>0,  \\
&\qquad\qquad  \int  g(v) \langle v \rangle^{2(1+\gamma)} dv =m_{1+\gamma} \!\le\!   K \!<\!\infty \ \Big\},
\end{split}
\end{equation}
for any finite constant $K >  \mathfrak{h}_{1+\gamma}.$ 
In fact,   conditions {\bf i)} and {\bf iii)} are satisfied with any initial moment to have at least $1+\gamma$ moments, that is  $f_0\in L^{1}_{2+2\gamma}(\mathbb{R}^{d})$,     while the Sub Tangent condition  {\bf ii)} proof  enables the existence  as well characterizes the  bounded, convex, closed subset $\Omega \in \mathcal E,$  by that any element $g\in \Omega$ may have   moment $m_{1+\gamma}[g]\ge  \mathfrak{h}_{1+\gamma}.$

  The second stage consists in showing the solution set can be enlarge  to solution set posing just initial data in  $L^{1}_{2^+}(\mathbb{R}^{d})$ by means of relaxing the initial data to the lower moment condition $f_0\in L^{1}_{2+\epsilon}(\mathbb{R}^{d})$, for  some $\epsilon>0$.  It covers two delicate points. The first one if to construct a sequence of approximants  in the set $\Omega$  to the initial data $f_0\in L^{1}_{2+\epsilon}(\mathbb{R}^{d})$ by density and, the second one consists	in  showing that the sequence of solutions is uniformly bounded strongly in  $L^{1}_{2+\epsilon}$ and so has a limit that generates the solution for such  $f_0\in L^{1}_{2+\epsilon}(\mathbb{R}^{d}).$

\begin{remark}\label{moment mathfrakh} It is important to notice that the constant $\mathfrak{h}_{1+\gamma}$   is  fully  determined only by the transition probability, or collision kernel term $B_\gamma(|u|,\hat{u}\cdot \sigma)$, and the initial conserved quantities $m_0$ and $m_{1}$, and so it  
only depends on the data associated to the Cauchy problem  to solved in Theorem~\ref{CauchyProblem}. 
\end{remark}

Without loss of generality, the binary collisional Boltzmann operator   $Q_{\gamma}(\cdot,\cdot)$ if considered to  have $B_\gamma(|u|,\hat u\cdot \sigma):=|u|^{\gamma}\,b(\hat{u}\cdot\sigma)$, with $\gamma\in[0,2]$, and   the angular transition measure  $b(\hat u\cdot \sigma)$  is in $L^1({\mathbb{S}^{d-1}}).$
The subindex $\gamma\in[0,2]$ on $Q_{\gamma}(\cdot,\cdot)$ denotes the potential rate.  Thus, $Q_{0}(\cdot,\cdot)$ represents the Maxwell Type of  interactions. 

\medskip

This following Lemma shows the  H\"older property {\bf i)} hold for any   $g\in L^1_{\ell}(\mathbb{R}^d)$, with $\ell\ge2(1+\gamma)$. It does not need  the conservation property stating that  any pair of  functions must have   the same   $L^1_{\ell}(\mathbb{R}^d)$ for $\ell=1,2$, since $L^1_{\ell}(\mathbb{R}^d) \subset L^1_{2(1+\gamma)}(\mathbb{R}^d)$

\begin{lem}[{\bf H\"older estimate  for elements  $L^1_{\ell}(\mathbb{R}^d)$ for any $\ell\ge2(1+\gamma)$  -- Property i)}]\label{Holder-estimate}

Let $f$ and $g$ functions lying in $L^1_{\ell}(\mathbb{R}^d)$ for any $\ell\ge2(1+\gamma)$, with their  $L^1_{2(1+\gamma)}$-norms are bounded. 
Denote by $h=f-g$ and $H=f+g$, noticing that $H$ is nonnegative whenever $f$ and $g$ are nonnegative.
Then, 
\begin{equation*}
\int_{\real^d}  \big|Q_{\gamma}(f,f) - Q_{\gamma}(g,g)\big| \langle v\rangle^{2} \ dv \leq  \|b\|_{ L^{1}( \mathbb{S}^{d-1} ) }\,C_{\Omega}\,\|f-g\|^{1/2}_{L^{1}_{2}}\,.
\end{equation*}
with constant given by $C_{\Omega}:=6\, K^{3/2}$.
\end{lem}
\begin{proof} 
Observe that for $h=f-g$ and $H=f+g$,  the  following identity holds
\begin{equation*}
Q_{\gamma}(f,f) - Q_{\gamma}(g,g) \ =\ \tfrac{1}{2}Q_{\gamma}(H, h) + \tfrac{1}{2}Q_\gamma(h,H) \,. 
\end{equation*}
In addition, note that
\begin{equation*}
|u| = |u'| \leq \min\big\{\langle v \rangle \langle v_{*} \rangle  ,  \langle v' \rangle \langle v'_{*} \rangle  \big\}\,, \quad \text{and} \quad \langle v \rangle \leq \langle v' \rangle\langle v'_{*} \rangle\,, 
\end{equation*}
where the latter uses local conservation of energy $|v'|^2 + |v'_*|^2 = |v|^2 + |v_*|^2$.

Therefore,
\begin{equation*}
\big|Q_{\gamma}(h,H)\big| \langle v \rangle^{2} \leq Q^{+}_0(|h|\langle v \rangle^{2+\gamma}, H \langle v \rangle^{2+\gamma}) + Q^{-}_0(|h|\langle v \rangle^{2+\gamma}, H \langle v \rangle^{2+\gamma})\,.
\end{equation*}

Moreover, multiplying  $Q^{\pm}_0(|\psi|,|\varphi|)(v)$  by  any $\psi, \varphi \in L^{1}(\mathbb{R}^{d})$,  in $v\in\mathbb{R}^{d}$,   the following identity holds%
\begin{equation*}
\int_{\mathbb{R}^{d}}Q^{\pm}_0(|\psi|,|\varphi|)(v)dv = \|b\|_{\mathbb{S}^{d-1}}\| \psi \|_{L^{1}(\mathbb{R}^{d})} \| \varphi \|_{L^{1}(\mathbb{R}^{d})}\,,
\end{equation*}
that allow us to conclude that $Q_{\gamma} (h,H)$ is control by the $(1+\gamma/2)$-moments of $h$ and $H$, respectively, proportionally to the integral of the angular transition function $b(\hat u\cdot \sigma)$, that is 
\begin{equation*}
\int_{ \mathbb{R}^{d} } \big| Q_{\gamma} (h,H) \big| \langle v \rangle^{2} \ dv \leq 2\ \|b\|_{ L^{1}( \mathbb{S}^{d-1} ) } \  \|H\|_{L^{1}_{2+\gamma}} \ \|h\|_{L^{1}_{2+\gamma}}\,.
\end{equation*}
Further, using Lebesgue's interpolation, which uses the $L^1_{2(1+\gamma)}$-norm  of  $H=f+g$,   yields 
{\begin{equation}\label{holder 1}
\|h\|_{L^1_{2 + \gamma}} \ \le \|h\|^{1/2}_{L^1_{2(1+\gamma)}} \|h\|^{1/2}_{L^1_{2}} \leq \|H\|^{1/2}_{L^1_{2(1+\gamma)}} \|h\|^{1/2}_{L^1_{2}}\,,
\end{equation}}
where the second inequality follows by observing that for nonnegative functions $|h| \leq H $.  Therefore,
\begin{equation}\label{holder 2}
\int_{ \mathbb{R}^{d} } \big| Q_{\gamma} (h,H) \big| \langle v \rangle^{2} \ dv \leq 2\ \|b\|_{ L^{1}( \mathbb{S}^{d-1} ) } \  \|H\|^{3/2}_{L^{1}_{2(1+\gamma)}} \ \|h\|^{1/2}_{L^{1}_{2}}\leq 6\,K^{3/2} \|b\|_{ L^{1}( \mathbb{S}^{d-1} ) } \|h\|^{1/2}_{L^1_{2}}\,.
\end{equation}
The latter estimate holds since $\|H\|_{L^1_{2+2\gamma}}\le 2m_{1+\gamma}[f]\leq  2K$ 
for any $f,g\in\Omega$. 

\end{proof}
\begin{lem}[{\bf Sub-tangent condition -- Property ii)}]\label{prop_subtan}
Fix $f\in\Omega$.  Then, for any $\epsilon>0$ there exists $\eta_{*}:=\eta_{*}(f,\epsilon)>0$ such that the ball centered at $f+\eta\,Q_{\gamma}(f,f)$ with radius $\eta\,\epsilon>0$ intersects $\Omega$.  That is,
\begin{equation}\label{non-empty-set}
B(f+\eta\,Q_{\gamma}(f,f),\eta\,\epsilon)\cap\Omega,\;  \text{is non-empty for any}\; 0<\eta\leq \eta_{*}\ .
\end{equation}
In particular,
\begin{equation}\label{invariant region limit}
\lim_{\eta\to 0+} \eta^{-1}\text{Dist}\big(f+\eta\,Q_{\gamma}(f,f),\Omega\big)=0\,.
\end{equation}
\end{lem}

\begin{proof}
The proof consists, for a given $f(v)\in \Omega$, we construct an associated function, denoted by $w_{R,f}(v)$ and first show that it is possible to find a small parameter $\eta(\epsilon, f)$,  such that $w_{R,f}(v)\in \Omega$. 
The second part simple shows that, for the given arbitrary $\epsilon>0$,  also $ w_{R,f}(v)\in B(f+\eta\,Q_{\gamma}(f,f),\eta\,\epsilon)$. 
 
\smallskip

\noindent {\em Part 1:} Start by recalling $m_0:=m_0[f](0)$ and  $m_1:=m_1[f](0)$, the conserved in time quantities mass and mass plus energy, respectively for any distribution density  $f(v,t)\in \Omega$.  

Next,  consider the angular transition $b(\hat u\cdot \sigma)$  integrable 
and define  the following cut off function depending on the fixed $f\in \Omega$  with a couple of positive real valued parameters $R$ and $\eta$, to be chosen later, given by  
\begin{align}\label{wRf}
w_{R,f}(v) &:=  f(v) + \eta\,Q(f_{R},f_{R})(v) \\
&\ = f(v) + \eta\Big(Q^{+}_{\gamma}(f_{R},f_{R})(v) - f_{R}(v) \|b\|_{L^1(\S^{d-1})}\big(f_{R}\ast \Phi \big)(v)\Big)\,, 
 \nonumber
\end{align}
with  $ f_{R}(v):= \text{1}_{\{|v|\leq R\}}\,f(v)\geq 0$,  and    the potential function $\Phi(v)$ satisfies  \eqref{pot-Phi_1} and \eqref{pot-Phi_2}.  
 
 The  arguments that follow prove that  $w_{R,f}(v)$ is an element in  the set \eqref{non-empty-set} for $R>0$,  chosen sufficiently large,  and    $\eta>0$  in \eqref{invariant region limit}  chosen sufficiently small,  both parameters only depending on $\epsilon$ and $f$.  
 
Indeed,  observing that $\gamma\in (0,2]$,   and invoking  \eqref{pot-Phi_2}, the collision frequency  is bounded above by, 
\begin{align}\label{Cm0}
\big(f_{R}\ast \Phi \big)(v)\leq 2^{\frac{\gamma}2} C_{\Phi}  m_0[f](0) (1+ |v|^{\gamma})  = C(m_0) (1+ |v|^{\gamma})\, , 
\end{align}
then,  choosing $0< \eta \leq 1/\big(\|b\|_{L^1(\S^{d-1})}C(m_{0})\big(1 + R^{\gamma}\big)\big)$, the function $w_{R,f}(v)$ can be estimated from below by
\begin{align}\label{wRf lower bound}
w_{R,f}(v)&\geq f(v) + \eta\,\Big(Q^{+}_{\gamma}(f_{R},f_{R})(v) - \text{1}_{\{|v|\leq R\}} \, f(v) \|b\|_{L^1(\S^{d-1})} \big(f_{R}\ast \Phi \big)(v,t)    \Big)\\
&\geq f(v,t)\Big(1 - \eta \|b\|_{L^1(\S^{d-1})}C(m_{0}) \big(1+R^{\gamma}\big)\Big)\geq0\, , \qquad \text{for} \  C(m_{0}):= 2^{\frac{\gamma}2} C_{\Phi}  m_0[f] .     \nonumber
\end{align}  

In addition, since the definition of $w_{R,f}(v,t)$ in \eqref{wRf}  preserves the binary structure of the collisional form
$Q(f_{R},f_{R})(v)$, then by conservation of mass and energy, both such Lebesgue moments of   $w_{R,f}(v)$ coincide with the ones for $f(v)$, that is 
\begin{equation}\label{moments-wRf}
\int_{\real^{d}}w_{R,f}(v,t) \left(\!\begin{matrix} 1 \\  |v|^2 \end{matrix}\!\right)  dv \!=\! \int_{\real^{d}}f(v) \left(\!\begin{matrix} 1 \\  |v|^2 \end{matrix}\!\right) dv \!+\!\int_{\real^{d}}Q(w_{R,f}, w_{R,f})(v) \left(\!\begin{matrix} 1 \\  |v|^2 \end{matrix}\!\right) dv  \!=\!
\left(\!\begin{matrix} m_0 \\  m_1-m_0 \end{matrix}\!\right), \  \forall \ R>0\,, 
\end{equation}
or, equivalently,
\begin{equation}\label{conserve-wRf}
\int_{\real^{d}}w_{R,f}(v)  - f(v)\, \left ( 1 +  |v|^2\right)  dv = 0.
\quad \forall \ R>0\,.
\end{equation}

Thus, in order to show $w_{R,f}(v,t)$ is an element in  the set \eqref{non-empty-set} for $R>0$, recall first the \textit{a priori} estimates developed for the $k^{th}$-Lebesgue moments of the collisional integral $Q(f,f)(v,t)$  in Theorem~ \ref{propagation-generation}, after setting $s=1$,  namely, for the nonlinear operator  acting on $k$-Lebesgue moments, 
\begin{align}\label{Lk operator}
\int_{\mathbb{R}^d} \!Q(f,f) \langle v\rangle^{2k} dv \!\leq \! \B_{k} \!- \!A_{\kc}m_1^{-c}[f_0(0)]m_{k}[f]^{ 1 + c }\! =:\!  \mathcal{L}_{k}\big(m_{k} [f] \big), \  \text{for} \ c\!=\!\frac{\gamma}{2(k-1)}, \ k \!\ge \!{\kc}\! >\!1\,,
\end{align}
maps $\mathcal{L}_{k}(y):[0,\infty)\rightarrow \mathbb{R}$, 
with strictly positive  constants $\B_{k}$ and  $A_{\kc}$ (the coercive factor) defined in \eqref{moment-factors}  depending on $f$ only through the initial moments mass and energy initial moments $m_0$ and $m_1$, respectively,  the fundamental constant  $\mu_k $ from the angular averaging  Lemma~\ref{BGP-JSP04},  and the parameter $\gamma$, and the bounds of $\Phi$ from \eqref{pot-Phi_1} depending on the data, for any $k>1$.  
In addition, the functional   $\mathcal{L}_k(y)$ is  non-linear form  function acting on $y$,  that is strictly decreasing since it has a super linear negative term due to the fact that $c(k, \gamma)$ is strictly positive for any $\gamma\in(0,2]$ and $k>1.$
As a consequence,   $\mathcal{L}_k(y)$ has {\em at most one positive root} characterized by $\mathcal{L}_{k}(\mathfrak{h}_k^{*})=0 $, depending on  $c= \gamma/(2(k-1))$.   

It is worthy to note that the root $\mathfrak{h}_k^{*}$ coincides  with the equilibrium solution associated to the flow $y_t(t)=\mathcal{L}y(t)$, denoted by  $\mathbb{E}_k$ defined in \eqref{RateE},  to the upper ODE problem \eqref{upperODE},  with a positive initial $k$-Lebesgue moment that controls the a priori estimates for the $k$-Lebesgue moments ODE, so the root  $\mathfrak{h}_k^{*}$  is identify with the notation introduced in \eqref{RateE}, namely
 %
%
\begin{equation}\label{Lk 0}
\mathfrak{h}_k^{*}:= \mathbb{E}_k =
m_1^{\frac {c}{1+c}}\Big(
\tfrac{\B_{k}}
{ A_{\kc}}\Big)^{\frac1{1+c}}\,,\qquad  k \ge {\kc} >1\,,
\end{equation}   
at which $\mathcal{L}_{k}(y)$ changes from positive to negative,  if $y<\mathbb{E}_k$ or $y>\mathbb{E}_k$ respectively.  

In addition, looking at  the explicit formula   \eqref{Lk operator},   
the maximum value of the non-linear operator $\mathcal{L}_{k}(y) $ is achieved at $\mathcal{L}_{k}(y)\mid_{y=0} :=\B_k $, since $\partial_y \mathcal{L}_{k}(y) <0$ if $c> 0$, or, equivalently, if $\gamma\ge 0$, and $k>1$.

 Therefore,  choosing  $\kc$, such that $k=1+\frac\gamma 2\ge \kc$, not only the exponent factor becomes $c(1+\gamma)=1/2$, but also the sum of  the equilibrium solution $ \mathbb{E}_{1+\gamma} $ plus  the maximum value of the nonlinear operator   $\max_{y\in \Omega}{\mathcal{L}}_{1+\gamma}(y):=  \B_{1+\gamma}$,  provides a characterization to find a sufficient bound  on the moment $m_{1+\gamma}[g] <   {\mathfrak{h}}_{1+\gamma} \le  K,$ that secures the set $\Omega$ to be the invariant region, that is    
 \begin{align}\label{Lk-frac-h}
  {\mathfrak{h}}_{1+\gamma}    &:=  \mathbb{E}_{1+\gamma} + B_{1+\gamma}
= m_1^{\frac13}\Big(
\tfrac{\B_{1+\gamma}}
{ A_{\kc}}\Big)^{\frac 23} + \B_{1+\gamma}= m_1^{\frac13}\Big(\tfrac{\beta_{1+\gamma}}
{ A_{\kc}}\,  \mathcal C_{1+\gamma} \Big)^{\frac 23}+ \beta_{1+\gamma}\mathcal C_{1+\gamma}\\
&= \! m_1^{\frac13} \left[ \left(\! \frac{ \beta_{1+\gamma}}{ A_{\kc} }\!\right)^2 \mathbb{I}_{ k\ge \kgamma }
  \!+\! \left(\frac{ \beta_{1+\gamma}}{ A_{{  s}\kc} }\!\right)^{1+\frac\gamma3}  \mathbb{I}_{ k <  \kgamma } \right] 
 + \beta_{1+\gamma}\left[ \left(\! \frac{ \beta_{1+\gamma}}{ A_{\kc} }\!\right)^2 \mathbb{I}_{ k\ge \kgamma }
  \!+\! \left(\frac{\beta_{1+\gamma} }{ A_{\kc} }\!\right)^{\frac12+\frac94\gamma }  \mathbb{I}_{ k <  \kgamma }\right] \le K,\nonumber \\[4pt]
  \text {for} \quad&\beta_{1+\gamma}=2C_\Phi (2^{\frac32\gamma}\mu_\kc +1).  \nonumber
\end{align}

Now, it becomes  clear that for  any  given $f(\cdot,t)\in\Omega$ with $t$ fixed, the constructed it $w_{R,f}$ is an element in the intersection set  defined in \eqref{non-empty-set}, for a choice of $\eta(\eps)$ that makes the limit \eqref{invariant region limit} to hold, so  showing that $\Omega$ is an invariant region for the Boltzmann flow under consideration.

 First, observe there are only two options  for the associated moment $y:=m_{1+\gamma}[f]$, namely either $y$ is smaller than  $ \mathbb{E}_{1+\gamma}$,   or larger than it. 

Thus for the first case, when $y =m_{1+\gamma}[f] \leq \mathbb{E}_{1+\gamma},$
   and for any $\eta<1$  from   \eqref {Lk-frac-h}, it clearly follows that
\begin{align}\label{wRf-case1}
m_{1+\gamma}[w_{R,f}] &= \int_{ \mathbb{R}^{d} } w_{R,f}(v) \langle v\rangle^{2(1+\gamma)}  dv = 
\int_{\mathbb{R}^{d}}\big( f(v) + \eta \, Q(f_{R},f_{R})(v)\big) \langle v\rangle^{2(1+\gamma)}  \text{d}v  \nonumber\\
&\leq  \mathbb{E}_{1+\gamma} + \eta\  
 \max_{y\in\Omega}{\mathcal{L}}_{1+\gamma}   =    \mathbb{E}_{1+\gamma}+ \eta\  \B_{1+\gamma}   \leq \mathfrak{h}_{1+\gamma} \le K,
\end{align} 

And for the one when $y:=m_{1+\gamma}[f] > 
\mathbb E_{1+\gamma}$,  make the choice  of letting the cut-off parameter $R:=R(f)$  from \eqref{wRf} to be sufficiently large  for  $m_{1+\gamma}[f_{R}]\geq\mathbb E_{1+\gamma}$, implying
\begin{equation*}\label{wRf-case2-0}
\int_{\mathbb{R}^d} Q(f_{R},f_{R}) \langle v\rangle^{2(1+\gamma)} dv \leq \mathcal{L}_{1+\gamma}\big(m_{1+\gamma}[f_{R}]\big)\leq 0\,, 
\end{equation*}
and, hence, 
\begin{equation}\label{wRf-case2}
m_{1+\gamma}[w_{R,f}] = \int_{\mathbb{R}^{d}} \big(f(v)+\eta \, Q(f_{R},f_{R})\big) \langle v \rangle^{2(1+\gamma)} dv \leq \int_{\mathbb{R}^{d}} f(v) \langle v\rangle^{2(1+\gamma)} dv \leq \mathfrak{h}_{1+\gamma}  \le K. 
\end{equation}
Thus, {\em Part 1} is complete, since   for any $f\in\Omega$, we have proven in  \eqref{wRf lower bound}  that  $w_{R,f} \ge  c(\eta, R) \, f(v) \ge 0$,  from  \label{moments-wRf} that the  mass and energy of $w_{R,f}(v)$ coincide with those of $f(v,t)$ and,   from \eqref{wRf-case1} and \eqref{wRf-case1},   that $m_{1+\gamma}[w_{R,f} ] \leq \mathfrak{h}_{1+\gamma} $, and then  the constructed $w_{R,f}$ belongs to  the invariant set $\Omega$.

\smallskip

\noindent {\em Part 2:} In order to show that    the constructed $ w_{R,f}(v)$ in {\em Part 1} belongs to $B(f+\eta\,Q_{\gamma}(f,f),\eta\,\epsilon)$ for the given arbitrary $\epsilon>0$, one needs to calculate the normed distance 
\begin{equation*}
\big\|w_{R,f} - f - \eta \, Q_{\gamma}(f,f) \big\|_{L^{1}_{2}(\real^{d})} = \eta\, \big\|Q(f_{R},f_{R}) - Q_{\gamma}(f,f) \big\|_{L^{1}_{2}(\real^{d})}\leq  6\,\eta\,  \mathfrak{h}_{1 + \gamma}^{3/2} \|b\|_{ L^{1}( \mathbb{S}) }\|f_{R} - f\|^{1/2}_{L^{1}_{2}(\real^{d})}\leq \eta \, \epsilon\,, 
\end{equation*}
which follows from taking the definition of 	$ w_{R,f} - f = Q(f_{R},f_{R})$ from \eqref{wRf} and invoking  the H\"{o}lder property proven in the previous  Lemma \ref{Holder-estimate}, estimate \eqref{holder 2},  when taking  $R\geq R(\epsilon, f)>0$, sufficiently large such that $\|f_{R} - f\|^{1/2}_{L^{1}_{2}(\real^{d})}\leq \epsilon$.  

To end, adjust the choice of  $w_{R,f} \in B(f+\eta\,Q[f],\eta\,\epsilon)$  for $\bar R =\max\{R(f),R(\epsilon,f)\}$  and  a possible smaller   $0< \eta_*(\epsilon,f) \leq 1/\big(C(m_{0})\big(1 + \bar R^{\lambda}\big)\big)$ to obtain,  not only $w_{\bar R,f}\in \Omega$, but also  for an arbitrary $\epsilon$, it follows that 
\begin{equation}\label{limit 0+}
\eta^{-1}\text{dist}\big( f + \eta\, Q_{\gamma}(f,f) , \Omega\big) \leq \epsilon\,, \qquad  \forall \ 0<\eta\leq  \eta_{*}\,.
\end{equation}

Therefore  taking  the limit as $\eta\to 0^+$ of \eqref{limit 0+} vanishes, proving that    property \eqref{invariant region limit} holds, and so  the proof of Lemma \ref{prop_subtan} is now complete. \\
\end{proof}

The last condition to check that is satisfied is {\bf iii)},  the One-Sided Lipschitz condition,  which is a sufficient to show that if   the initial $f_0\in L^{1}_{2+2\gamma}(\mathbb{R}^{d})$, solutions of  Theorem~\ref{CauchyProblem} are unique . The original proof was performaned by G. Di Blassio in \cite{diblassio} for the case $\gamma =1$, were polynomial moments were defined with weights $|v|^{2\gamma}$.

\begin{lem}[{\bf one-sided Lipschitz condition -- Property iii)}]\label{one-Lips-prop3}
For any function $f$ and $g$ lying in $\Omega$, it holds that
\begin{equation}\label{osLip}
\big[f-g,Q_{\gamma}(f,f) - Q_{\gamma}(g,g)\big]_{-}\leq \|b\|_{L^{1}(\mathbb{S}^{d-1})}\,C_{\Omega}\,\|f-g\|_{L^{1}_{2}(\mathbb{R}^{d})}\,,
\end{equation}
with $C_{\Omega}:= 8K$.
\end{lem}
\begin{proof} Recall that in $L^{1}_{2}(\real^{d})$, the Lipschitz bracket simply reads
\begin{equation}\label{LipBracket}
\big[\varphi, \psi\big]_{-} =    \int_{\real^{d}} \psi \,\text{sign}(\varphi)  \,\langle v\rangle^2 dv \,.
\end{equation}
Then, using the weak formulation (with the shorthand $\int=\int_{\real^d}\int_{\real^d}\int_{\mathbb{S}^{d-1}}$ and $h=f-g$)
\begin{align}\label{osLip1}
\big[f-g&,Q_{\gamma}(f,f) - Q_{\gamma}(g,g)\big]_{-} =\int_{\real^{d}}\big(Q_{\gamma}(f,f) - Q_{\gamma}(g,g)\big)\,\text{sign}(f-g)\,\langle v \rangle^{2} dv\\
&= \int \big(f\,f_{*} - g \, g_{*}\big)\Delta(v,v_{*})|u|^{\gamma}b(\hat{u}\cdot\sigma) \text{d}\sigma \text{d}v_*\text{d}v =\int \big( h\,f_{*} + g\,h_{*} \big)\Delta(v,v_{*})|u|^{\gamma}b(\hat{u}\cdot\sigma) \text{d}\sigma \text{d}v_*\text{d}v\,.\nonumber
\end{align}
Here, $\Delta(v,v_{*}) = \langle v' \rangle^{2}\,\text{sign}(h') + \langle v'_{*} \rangle^{2}\,\text{sign}(h'_{*}) - \langle v \rangle^{2}\,\text{sign}(h) - \langle v_{*} \rangle^{2}\,\text{sign}(h_{*})$.  Observe that
\begin{align}\label{osLip2}
h\, \Delta(v,v_{*})&\leq |h|\,\Big(\langle v' \rangle^{2} + \langle v'_{*} \rangle^{2} + \langle v_{*} \rangle^{2} - \langle v \rangle^{2}\Big) = 2\,|h|\,\langle v_{*} \rangle^{2}\,,\quad\text{ and }\nonumber\\
h_{*}\, \Delta(v,v_{*})&\leq |h_{*}|\,\Big(\langle v' \rangle^{2} + \langle v'_{*} \rangle^{2} + \langle v \rangle^{2} - \langle v_{*} \rangle^{2}\Big) = 2\,|h_{*}|\,\langle v \rangle^{2}\,.
\end{align}
Combining,  both \eqref{osLip1} and   \eqref{osLip2} estimates, local conservation of energy.  Therefore, (with the shorthand $H=f+g$)  
\begin{align}\label{one-sided}
 \big[f-g,&Q_{\gamma}(f,f) - Q_{\gamma}(g,g)\big]_{-} \leq 2\,\|b\|_{L^{1}(\mathbb{S}^{d-1})}\int_{\real^{d}}\int_{\real^{d}}\big(|h|\,f_{*}\,\langle v_{*}\rangle^{2} + |h_{*}|\,g\,\langle v \rangle^{2} \big)\,|u|^{\gamma} \, dv_{*}\,dv\\
&= 4\,\|b\|_{L^{1}(\mathbb{S}^{d-1})}\int_{\real^{d}}\int_{\real^{d}}|h|\,H_{*}\,\langle v_{*}\rangle^{2}\,|u|^{\gamma} \, dv_{*}\,dv\leq 4\,\|b\|_{L^{1}(\mathbb{S}^{d-1})}\|H\|_{L^{1}_{2+2\gamma}}\|h\|_{L^{1}_{\gamma}}\,.\nonumber
\end{align}
Since $\|H\|_{L^{1}_{2+2\gamma}}= m_{1+\gamma}[f] + m_{1+\gamma}[g]  \,  \le 2K $,  estimate  \eqref{osLip} follows.\\
\end{proof}

  %

\bigskip

Hence, Lemmata \ref{Holder-estimate}, \ref{prop_subtan}, and \ref{one-Lips-prop3} ensures that conditions $\bf{i), ii), iii)}$ of Theorem~\ref{Theorem_ODE},  and so  the following existence and uniqueness result,  with 
$f_0\in L^1_{2+2\gamma}({\mathbb S^{d-1})}$ to the initial value problem for the space homogeneous Boltzmann equation under consideration holds, 
\begin{equation*}
f_t = Q(f,f)\,, \qquad \text{with initial data}\qquad  f(v,0) = f_0(v) \in \Omega\,,
\end{equation*}
has a unique solution $0\leq f(v,t)\in \mathcal{C} \big([0,\infty); \Omega\big)\cap  \mathcal{C}^1\big((0,\infty);L^{1}_{2}\big(\mathbb{R}^{d}\big)\big)$.

\medskip

Hence, the second stage of proof of Cauchy Problem, Theorem~\ref{CauchyProblem}, are fulfilled and, as a consequence, the homogeneous Boltzmann problem, and the new    moment condition on the initial data can be relaxed to having an initial data that belongs to $L^1_{2^+}$ for $\gamma\in(0,2]$.   In fact, under Remark~\ref{kc=eps} conditions, set $\epsilon =  2(\kc-1) \in (0, \gamma)$.

The following Lemma justifies this statement.

\begin{lem}\label{Lemma_ODE-extension2}
 Let $\gamma\in(0,2]$ and $0\leq f_0\in L^{1}_{2+\epsilon}(\mathbb{R}^{d})$, with  $\epsilon\in(0,\gamma)$ fixed. 
Then, 
\begin{itemize} 
\item[{\it i)}]\  There exists a non negative sequence of initial data $\{f^{j}_{0}\}\subset L^{1}_{2+\gamma}(\mathbb{R}^{d})\cap \Omega$ strongly converging to $f_0$, that is 
\begin{equation} \label{fj-convergence}
f^{j}_{0} \longrightarrow f_0\,,\quad\text{strongly in} \; L^{1}_{2+\epsilon}(\mathbb{R}^{d})\,.
\end{equation} 

\item[{\it ii)}]\ There is a sequence of solutions of the Boltzmann Cauchy problem $f^j(v,t) \in L^{1}_{2+\gamma}(\mathbb{R}^{d})$,  with initial data $f^{j}_{0}\in L^{1}_{2+\gamma}(\mathbb{R}^{d})$,  converging  to   $f\in \mathcal{C}\big([0,\infty); L^{1}_{2}(\mathbb{R}^{d}) \big)\ \cap\ \mathcal{C}^{1}\big((0,\infty); L^{1}_{\ell}(\mathbb{R}^{d}) \big),$ with $2\le \ell<2(1+\gamma), $ that solves  the Cauchy problem with initial data $f_0\in L^{1}_{2+\epsilon}(\mathbb{R}^{d}).$ 
\end{itemize} 

\end{lem}

\begin{proof}
Let $0\leq f_0\in L^{1}_{2+\epsilon}(\mathbb{R}^{d})$, by density, there exits an approximating sequence   $\{f^{j}_0\} \subset \Omega\subset L^1_2(\R^d)$, with $K$ arbitrarely large,   such that  $\lim_{j\to\infty} f^{j}_0 =f_0$  in $L^{1}_{2+\epsilon}(\mathbb{R}^{d})$ satisfying property {\it i)}, such that  $\|f^{j}_0\|_{L^{1}_{2+\epsilon}(\mathbb{R}^{d}) }< C_{\epsilon}[f_0], $ uniformly in $j$.
Such approximations can be obtained, for instance, by taking $f^j_0 (v):=  f_0(v) \, \mathbb I_{v<j}\in L^1_\ell$ for all $\ell\le 2(1+\gamma)$, which implies $\{f^{j}_0\} \subset \Omega^j_\gamma$ and  clearly $\|f^{j}_0 - f_0\|_{L^{1}_{2+\epsilon}(\mathbb{R}^{d})}$  converges  strongly to $0$.

In particular,  for each $f^{j}_0\in \Omega_\gamma^j\subset L^1_2(\R^d)$,  the first stage of the Cauchy problem  Theorem~\ref{CauchyProblem} secures the existence and uniqueness of solutions, invoking   the tools from Lemmas \ref{Holder-estimate}, \ref{prop_subtan}, and \ref{one-Lips-prop3} which are  sufficient  conditions $\bf{i), ii), iii)}$ for Theorem~\ref{Theorem_ODE} to hold for each $j$ fixed, that is  there exists and $f^j(v,t)\in \C(0,\infty, \Omega^j_\gamma) \ \cap\ \mathcal{C}^{1}\big((0,\infty); L^{1}_{2}(\mathbb{R}^{d}) \big).$
 
Next, it is important to verify the uniformity of the $\{ f^j(\cdot,t)\}$  sequence   in the $ L^{1}_{2+\epsilon}(\mathbb{R}^{d}) \big)$ topology. 
To this end,  start from invoking the multi linearity structure of the collisional form  and   the  one-sided Liptchitz condition \eqref{one-sided} from  Lemma \ref{one-Lips-prop3},  which clearly remains valid  independently of the pair $f^{j_1}$ and $f^{j_2}$ are in different $\Omega_\gamma^{j_1}$ and  $\Omega_\gamma^{j_2}$, respectively, to obtain
the following estimate  for the difference on any pair $f^j$ and $f^l$ of the approximating sequence $\{ f^j(\cdot,t)\}$, 
\begin{align}\label{e1Cauchy}
\frac{d}{dt}\|f^{j}(t) - f^{l}(t)\|_{L^1_2}  &= \frac{d}{dt}\int_{\R^d} (f^j(t) - f^l(t)) \,\text{sign}(f^j(t) - f^l(t)) \langle v \rangle^2 dv  \nonumber\\
%
%
&= \int_{\R^d} \left( Q_\gamma(f^j+f^l,f^j-f^l )  \,\text{sign}(f^j - f^l) \right)(v,t) \langle v \rangle^2 dv
 \\
&\leq 4\|b\|_{L^{1}(\mathbb{S}^{d-1})}
\|f^{j}(t) \!+\! f^{l}(t)\|_{L^{1}_{2+2\gamma}}\|f^{j}(t) \!-\! f^{l}(t)\|_{L^{1}_{\gamma}} =: \mathcal A(t) \,\|f^{j}(t) \!-\! f^{l}(t)\|_{L^{1}_{\gamma}}, \nonumber
\end{align}
for   $\A(t):= 4\|b\|_{L^{1}(\mathbb{S}^{d-1})}
\|f^{j}(t) \!+\! f^{l}(t)\|_{L^{1}_{2+2\gamma}}$. 

Hence,  since $\|f^{j}(t) - f^{l}(t)\|_{L^1_2}= \|f^{j}(0) - f^{l}(0)\|_{L^1_2}\exp(-\int_0^t\A(\tau) d\tau)$,  then, the  uniformity of the  sequence $\{ f^j(\cdot,t)\}\in L^{1}_{2+\epsilon}(\mathbb{R}^{d}) $   follows from the uniform integrability of  $\A(t), $ independently of sequence's index $j$.

%
%
Therefore,   invoking the propagation of  moments and global in time  estimates \eqref{mom-prop} from Theorem~\ref{propagation-generation}   for each $j$-term in the sequence,  it ensures a boundedness of 
$\{ f^{j}(t)\}$  in the $L^{1}_{2+\epsilon}(\R^d)$  topology,  with  for all $m_0[f^j_0]=m_0[f_0]$ and  $m_1[f^j_0]=m_1[f_0]$, not only, for $c_\gamma(\kc)=\frac{\gamma}{\epsilon}$, 
\begin{align}\label{M1eps}
\sup_{t\geq0}\|f^{j}(t)\|_{L^{1}_{2+\epsilon}}
 \le \sup_{t\geq0} m_{1+\frac\epsilon 2}[f^{j}](t)
&= \max\{  m_{1+\frac\epsilon 2}[f^{j}](0),  \mathbb E^j_{c_\gamma(1+\epsilon)}  \}\nonumber\\
&=: M_{1+\frac\epsilon 2 }(m_0[f_0], m_1[f_0],  b(\hat u\cdot\sigma), \gamma, \epsilon), 
\end{align}
but also,  since the convergence $f^{j}_0\to f_0$ in the $L^{1}_{2+\epsilon}(\mathbb{R}^{d})$ also holds in $L^{1}_{2}(\mathbb{R}^{d})$, that means there is a constant $C_{\epsilon}[f^{j}_0]$ that only depends on $m_0[f_0]$ and 
$m_1[f_0]$,  converging to the mass and energy of $f_0$,   yielding the uniform bound estimate
\begin{equation}\label{uniform j}
\sup_{t\geq0}\|f^{j}(t)\|_{L^{1}_{2+\epsilon}} \leq \max\big\{M_{1+ \epsilon }[f_0],\, C_{\epsilon}[f_0], \gamma, \epsilon\big\}=: C_0(f_0, b(\hat u\cdot\sigma),\gamma,\epsilon), \ \ \text{ independent of} \ j.  
\end{equation}

Furthermore, for each fixed $j$,  invoking the generation  of  moments and its  priori global in time estimates Theorem \ref{propagation-generation},  each solution $f^j(t)\in \C(0,\infty, \Omega^j_\gamma) \ \cap\ \mathcal{C}^{1}\big((0,\infty); L^{1}_{2}(\mathbb{R}^{d}) \big)$ constructed from the initial data  $f^j_0$ es estimated by  invoking  generation global bounds \eqref{mom-gen},   with $k=1+\frac\gamma 2+\frac\epsilon 2$,  so $c^{-1}_\gamma(k)= \frac{\gamma +\epsilon}{\gamma} $ to obtain
\begin{equation}\label{gen j}
\|f^{j}(t)\|_{L^{1}_{(2+\gamma+\epsilon)}}\leq \tilde{C}_0\Big(1 + t^{- \frac{\gamma+\epsilon}{\gamma}}\Big)\,,\quad \text{for any}\ \  \epsilon\in (0,\gamma),
\end{equation}
with constant $ \tilde{C}_0 =\tilde{C}_0(f_0,b(\hat u\cdot\sigma), \gamma,\epsilon)$ independent of $j$, after combining   by the aforementioned estimates \eqref{e1Cauchy} with \eqref{M1eps} and \eqref{uniform j}.  
Therefore, by interpolation,
\begin{align}\label{interpol-free}
\|f^{j}(t)\|_{L^{1}_{2+\gamma}(\mathbb{R}^{d})} \ &\leq \ \|f^{j}(t)\|^{\frac{\gamma-\epsilon}{\gamma}}_{L^{1}_{(2+\gamma+\epsilon)}(\mathbb{R}^{d})}\|f^{j}(t)\|^{\frac{\epsilon}{\gamma}}_{L^{1}_{2+\epsilon}(\mathbb{R}^{d})}\\
&\le\  \tilde{C}_0^{\frac{\gamma-\epsilon}{\gamma}} \left(1 + t^{- \frac{\gamma+\epsilon}{\gamma}}\right)^{\frac{\gamma-\epsilon}{\gamma}}\, \|f^{j}(t)\|^{\frac{\epsilon}{\gamma}}_{L^{1}_{2+\epsilon}(\mathbb{R}^{d})} 
\leq \  \tilde C_1 \, \mathcal{O}\!\left(1 + t^{-\theta}\right), 
\quad \text{with }\theta = \frac{\gamma^{2}-\epsilon^{2}}{\gamma^{2}}  \in(0,1)\, , \nonumber 
\end{align}
and $\tilde C_1 =\tilde{C}_0^{\frac{\gamma-\epsilon}{\gamma}} C_0^{\frac{\epsilon}{\gamma}}$,   uniformly in $j$.

This last  estimate  is crucial: not only shows the control the term in \eqref{e1Cauchy}, but also allows the control of the 
$ \|f^{j}(t)\|_{L^{1}_{2+\gamma}(\mathbb{R}^{d})}$ by  $ \|f^{j}(t)\|^{\frac{\epsilon}{\gamma}}_{L^{1}_{2+\epsilon}(\mathbb{R}^{d})}$  proportional to  $\kappa(t) := \tilde{C}_0^{\frac{\gamma-\epsilon}{\gamma}}\mathcal{O} \left(1 + t^{- \frac{\gamma+\epsilon}{\gamma}}\right)^{\frac{\gamma-\epsilon}{\gamma}}$  in any time interval $t\in(0, T), \ T>0$, uniformly in the $j$-index.

Hence,  the first listed fact implies the control of $\A(t)$ in \eqref{e1Cauchy}, to obtain 
\begin{align}\label{mathcalA(t)-uni j}
\A(t)&:= 4\|b\|_{L^{1}(\mathbb{S}^{d-1})}
\|f^{j}(t) \!+\! f^{l}(t)\|_{L^{1}_{2+\gamma+\epsilon}} 
\le \tilde C_1 \mathcal{O}\!\left(1 \!+\! t^{-\theta}\right) \,  \|f^{j}(t)\|^{\frac{\epsilon}{\gamma}}_{L^{1}_{2+\epsilon}(\mathbb{R}^{d})}=
\tilde{C}_1\mathcal{O}\!\left(1 \!+\! t^{-{\frac{\gamma^{2}-\epsilon^{2}}{\gamma^{2}}}}\!\right)\!, 
\end{align}
   where the constants $\tilde C_0$ and $\tilde C_1$ depend only on $f_0, b(\hat u\cdot\sigma), \gamma,\epsilon$, independently of the sequence's index $j$.
   
Therefore, inserting  this last estimate in \eqref{e1Cauchy} 
leads to, after performing the time integration
\begin{equation}\label{uni-j-bound}
\|f^{j}(t) - f^{l}(t)\|_{L^{1}_{2}} \leq \|f^{j}_0 - f^{l}_0\|_{L^{1}_{2}} \exp\left(\tilde C_1\mathcal{O}\!\left(t +\tfrac{\gamma^{2}}{\epsilon^{2}}t^{1-\theta}\right) \right)\,,\quad \text{with}\  \theta \in(0,1), \ \  0< t\le T.
\end{equation}

As a consequence, the sequence  $\{f^{j}(t)\}$ is Cauchy, and so it  converges strongly to a unique $0\le f(v,t) \in L^{\infty}\big([0,T);L^{1}_{2}(\mathbb{R}^{d})\big)$,  for any  arbitrary $T>0$.  And thus, the second estimate from \eqref{interpol-free}, ensures that the limiting function $f(v,t)$ satisfies for any terminal $T>0,$
\begin{equation}\label{f-bound 0 }
\|f(t)\|_{L^{1}_{2+\gamma}(\mathbb{R}^{d})}\ \le\  \tilde{C}_1\mathcal{O} \!\left(1\! +\! t^{- \frac{\gamma+\epsilon}{\gamma}}\right)^{\frac{\gamma-\epsilon}{\gamma}} \|f(t)\|^{\frac{\epsilon}{\gamma}}_{L^{1}_{2+\epsilon}(\mathbb{R}^{d})} \quad \text{for any}\  \ \epsilon\in(0,\gamma), \ \ \ 0<t\le T.
 \end{equation}

In addition,  the convergence $f^j(v,t) \to f(v,t)$  in  $L^{\infty}\big([0,T);L^{1}_{2}(\mathbb{R}^{d})\big)$ holds. Invoking the interpolation estimates developed the H\"older's estimates     Lemma~\ref{Holder-estimate},  \eqref{holder 1} and \eqref{holder 2} in  $L^1_{\ell}(\mathbb{R}^d), \ell\ge2+\gamma$, 
 the  collision  operator  $Q_\gamma(f^j,f^j)$ is continuous operator in  from  ${L^1_2(\R^d)}$ into itself,  shown as follows.
%

 Indeed, recalling the previous argument from \eqref{e1Cauchy}, set
 \begin{equation}\label{oslim 1}
\frac{d}{dt}\left( f^j-f^l\right) (v,t) = \left(Q_\gamma(f^j,f^j)-Q_\gamma(f^l,f^l)\right)(v,t) 
=Q_\gamma(H,h)(v,t),  
\end{equation}
for $H=f^j+f^l$ and  $h=f^j-f^l.$
 Its $L^1_2(\R^d)$ norm is
   \begin{equation}\label{oslim 2}
 \int_{\mathbb{R}^{d}} Q^-_\gamma(H,h)(v,t) \sign(h)  \langle v\rangle^2 \text{d}v +  \frac{d}{dt} \| h(t) \|_{L^{1}_{2}}(t) = 
  \int_{\mathbb{R}^{d}} Q^+_\gamma(H,h)(v,t) \sign( h) \langle v\rangle^2 \text{d}v,
  \end{equation}
 noticing that both the gain and loss collisional forms associated  $\int_{\mathbb{R}^{d}} Q^\pm_\gamma(H,h)(v,t) \sign(h) \, 
\langle v\rangle^2 \text{d}v $ are similar due to the local conservation property $\langle v'_*\rangle^2 +\langle v'_*\rangle^2=\langle v_*\rangle^2 +\langle v_*\rangle^2$.  This difference of gain and loss operators are estimated by
 \begin{align}\label{oslim 4}
 \int_{\mathbb{R}^{d}} Q^\pm_\gamma(H,h)(v,t) \sign( h) & \langle v\rangle^2 \text{d}v  
 \le \int_{\mathbb{R}^{2d}\times\mathbb{S}^{d-1}} H_*\,h |u|^{\gamma}\left( \sign(h'_*) \langle v'_*\rangle^2 + \sign(h') \langle v'\rangle^2 \right) b(\hat{u}\cdot\sigma) \text{d}\sigma \text{d}v_*\text{d}v \nonumber\\
&\le C_\Phi 2^{\frac\gamma 2 +1}\|b\|_{L^{1}(\mathbb{S}^{d-1})} \int_{\mathbb{R}^{2d}} \left(\langle v\rangle^\gamma+ \langle v_*\rangle^\gamma\right) \left( \langle v_*\rangle^2  + \langle v\rangle^2 \right)
|h|\,  |H_*| \text{d}v_*\text{d}v \nonumber \\
&\le
C_\Phi 2^{\frac\gamma 2+2}\, \|b\|_{ L^{1}( \mathbb{S}^{d-1} ) }  \|H\|^{\frac{\gamma-\epsilon}{\gamma}}_{L^{1}_{2+\gamma+\epsilon}} \, \|h\|^{\frac{\epsilon}{\gamma}}_{L^{1}_{2}}\\
&\le C_\Phi 2^{\frac\gamma 2+4}\, \|b\|_{ L^{1}( \mathbb{S}^{d-1} ) } \,  \tilde{C}_1 \mathcal{O}\!\left(1\! + \!t^{-\theta}\right)^{\frac{\gamma-\epsilon}{\gamma}} \, \|h\|^{\frac{\epsilon}{\gamma}}_{L^{1}_{2}} =:
{\bf \it K}_0   \|h\|^{\frac{\epsilon}{\gamma}}_{L^1_2},\nonumber
 \end{align}
 %
%
for  any $t\in(0,T)$ for $T$ arbitrary,   ${\bf \it K}_0 \ge C_\Phi 2^{\frac\gamma 2+4} \|b\|_{ L^{1}( \mathbb{S}^{d-1} ) }   \tilde{C}_1 \mathcal{O}\!\left(1 \!+\! t^{-\theta}\right)^{\frac{\gamma-\epsilon}{\gamma}}$,  is uniform a time $t$  and index $j$ depending only on  $f_0, b(\hat u\cdot\sigma), \gamma,\epsilon$,  for any $\gamma\in (0,2]$,  $\epsilon\in(0,\gamma)$.

Therefore, \  $(f-f^j)(t)=\lim_{l\to\infty}(f^l-f^j)(t)$ in $L^1_2(\R^d)$, and consequently, 
  inserting estimates \eqref{oslim 4} into or   \eqref{oslim 1},  yields 
    \begin{equation}\label{oslim 5}
 - {\bf \it K}_0\,  \|f^j-f\|^{\frac{\epsilon}{\gamma}}_{L^1_2}(t)\ \le\   \|\frac{d}{dt} (f^j-f) - Q_\gamma(f^j+f, f^j-f)\|_{L^{1}_{2}}(t) \ <\  {\bf \it K}_0\,    \|f^j-f\|^{\frac{\epsilon}{\gamma}}_{L^1_2}(t),
  \end{equation}
uniformly in the sequence index $j$ and for all $t\in (0,T)$, for  arbitrary large $T>0$.

\medskip

This last estimate proves   $f(v, t)  \in L^1_{2}$ is  the unique solution of the Cauchy Boltzmann  problem with initial data $f_0\in L^1_{2+\epsilon}$ after the following  regularity in time is shown.  

In order to accomplish this task,   replace $f^l$ by the limiting $f$ into the 
  identity \eqref{oslim 1} to   set
   $\bar{\A}(t):= 4\|b\|_{L^{1}(\mathbb{S}^{d-1})}
\| f^{j}+f \|_{L^{1}_{2+\gamma}}(t)$, is the $j$-uniform upper bound  from  \eqref{mathcalA(t)-uni j}, yield  estimates  for $\bar{\A}(t) \le \tilde C_1(f_0, b(\hat u\cdot\sigma),\gamma,\epsilon)\mathcal{O}\!\left(1 \!+\! t^{-\theta}\right)$ 
with $\theta = \left({\gamma^{2}-\epsilon^{2}}\right)/{\gamma^{2}}\in(0,1)$ and $\epsilon\in (0,\gamma).$

\begin{equation}\label{oslim 2}
\|f- f^{j}\|_{L^{1}_{2}}(t) \leq \|f_0 - f^{j}_0\|_{L^{1}_{2}} \exp\left(\tilde C_1\left(t +\tfrac{\gamma^{2}}{\epsilon^{2}}t^{1-\theta}\right) \right)\,,\quad \text{with}\  \theta \in(0,1), \ \  t\geq0.
\end{equation}


To this end,  using estimates developed in  \eqref{oslim 4} for the continuity of the $L^{1}_{2}(\R^d)$-norm  of the collision operator, the differentiability in time follows by observing that, for any $t\ge 0.$
Then, since $f(v,t)$ is non-negative, both collisonal forms, gain  $Q_\gamma^{+}(f,f)$ and loss $Q_\gamma^{-}(f,f)$, are positive operators implying that  $\|Q_\gamma^{\pm}(f,f)\|_{L^{1}_{2}(\R^d)}(t) = \int_{\mathbb{R}^{d}}\!Q_\gamma^{\pm}(f,f)(vat)\langle v \rangle^{2}\text{d}v $, therefore
  \begin{equation*}
-\tilde{C}_1\mathcal{O}\left(1+\frac{1}{t^{\theta}}\right) \leq -\|Q_\gamma^{-}(f,f) \|_{L^{1}_{2}(\R^d)} \leq \frac{d}{dt} \| f(t) \|_{L^{1}_{2}} \leq \|Q_\gamma^{+}(f,f) \|_{L^{1}_{2}(\R^d)}\leq\tilde{C}_1\mathcal{O}\left(1+\frac{1}{t^{\theta}}\right).
\end{equation*}

Consequently, time integration in the interval $(s,t)$ leads to
\begin{equation}\label{oslim 7}
\Big| \| f(t) \|_{L^{1}_{2}} - \| f(s) \|_{L^{1}_{2}} \Big| \leq \frac{\tilde C_1}{1-\theta}\mathcal{O}\left(\left|t - s\right|^{1-\theta}\right)\,.
\end{equation}

Hence,  the limiting sequence $\lim_{j\to\infty} f^j $ in $ L^1_2(\R^d)$  makes the limiting $f(v,t)$ the unique solution of the Boltzmann equation in $\C^1((0,\infty);  L^1_2(\R^d))$ with initial data $f_0\in L^{1}_{2+\epsilon}(\R^d)$ for any $\epsilon\in (0,\gamma).$

 This strong convergence of the sequence   implies that the conservation laws hold as well.

Finally, using the generation of moments, estimates extends to  
\begin{equation*}
\Big|\frac{d}{dt}f(t)\Big| = \Big| Q_\gamma(f,f)(t)\Big| \in L^{1}_{\ell}(\mathbb{R}^{d})\, \ \text{for}\ \  \ell> 2, \  t>0.
\end{equation*}\\


The proof of Lemma~\ref{Lemma_ODE-extension2} is complete.
\end{proof}

Consequently, the proof of Theorem~\ref{CauchyProblem} is also complete.  \\
\end{proof}

\noindent
\begin{remark}\label{no-entropy need}
The proof does not need initial bounded entropy. However, if initially so, the following estimate holds
\begin{equation*}
\int_{\real^d}f(v,t)\ln\big( f(v,t) \big) dv \leq \int_{\real^d}f(v,0)\ln\big( f(v,0) \big) dv\,,\qquad t\geq0\,.
\end{equation*}
\end{remark}
\bigskip

\smallskip
\bigskip
\bigskip


  \section{Propagation and generation of exponential moments.}\label{exponential-tails}

One may view the study of  exponential moments associated to a unique solutions in $f(v,t)\in L^1_k(\R^d)$ for all $k\ge 2^+$, as   constructed in sections\ref{hardpotsection} and~\ref{existence-uniqueness}, 
 is the  property that states such solution remain in $L^1(\R^d)$ when exponentially weighted.  In fact, from probability methods  it is well establish that the sumability of polynomial expectations (referred as  moments in this manuscript)  is related to an exponential expectation with a rate related to a few initial moments. Such observation is closely related to the use of the Fourier transform  that transfer the topology of probability measure to the topology of  bounded functions with the $L^\infty(\R^d)$ norm in Fourier space. In the case of classical Boltzmann binary, elastic collisional flows, such proof is subtle, as it depends on the collision kernel, both potential and angular transition parts is in need to be verified.

In the case of  the Boltzmann equation for binary elastic particle interactions, with  hard sphere potentials, in three dimensions with a constant angular transitions,  the first were  developed in \cite{boby97}, who  showed that it was possible to prove that the summability of moments  was true allowing to conclude that the probability density  would  be  . These results were extended seven years later  in \cite{BGP04} for the hard sphere case for binary inelastic collisions with bounded angular transitions, and in \cite{GPV09}, to show, in the classical case of binary elastic interactions for variable hard potentials, the solutions are exponential weighted in $L^1(\R^d)$ if initially does as well, with a new slower rate need to be strictly positive and finite.
  
While much has been said about the summability of moments in the last twenty five years, the work in this manuscript clarified missing points in many of the previous approaches to this problem. 
 That means the extension to summability properties can now be extended to the case of hard potential rate $\gamma\in (0,2]$ with a transition angular functions just satisfying integrability, as  ran

 This moment summability properties arise from the  relation between the rate of an exponential form to its Taylor expansion in relations to and polynomial moments of the associated probability density that solves the initial valued problem to the Boltzmann flow. Hence we start by  formally writing
\begin{align}\label{exp-moment}
E_s[f](t, z(t)) &:=   \int_{\real^d} f(v,t) \, e^{z(t) \left\langle v \right\rangle^{2s}} dv
= \int_{\real^d} f(t, v) \sum_{k=0}^\infty  \langle v\rangle ^{2sk}\frac{z^k(t)}{k!} 
dv=\sum_{k=0}^\infty \frac{z^k(t)}{k!} \, m_{s k}[f](t), \quad\text{for} \ t>0,  
\end{align}
which we refer as to exponential moments with order $2s$, with $s\in (0,1]$,  $s=1$ corresponding to Gaussians,  and rate $z(t)    > 0$ where $z(t)$ may be constant.
There are two fundamental cases to study. The first one will be the propagation of exponential moments, for which the rate $ z(t) =z$ is  a constant function in time. The second one is the generation of exponential moments, for which the rate $ z(t) =zt$ for $0<t\le1$ and  $ z(t) =z$ is  a constant function for   $1\le t$.

Such moments   summability propagates whenever the initial data satisfies the same  summability property, and generates when starting with $2^+$-moments, and invoking the moments generation results from the previous two sections, the summability of all generated moments will be finite as well. 

 by showing that there is an associated geometric convergent series of moments, with a positive and finite convergence  radius resulting in integrability of the $L^1$- exponentially weighted norm, uniformly in time, whose  is the rate  of decay  depends of the data associated to the initial flow as much as from the moments bounds in the space $L^1_k(\mathbb R^n)$ studied in subsection~\ref{L^1-MM-1}.

In fact, these properties follows from  the constructed solutions  in  Section~\ref{existence-uniqueness}, which  propagate  or generate exponential moments depending on the integrability properties of the exponential moment of the initial data can be shown to be summable. 
  Propagation of initial data in this context  means given $f_0\in \Omega$, an order factor $0<s    \leq 1$ and a rate $0<z_0$ such that   $E_s[f]( 0,z_0) $ is finite, then the solution of the Cauchy problem for $f(t, v)$ posed in Theorem~\ref{CauchyProblem},    satisfies  $E_s[f](t,z) $ is finite for all time $t\ge0$    and $0<z\leq z_0.$ 
However, generation of data means the following stronger property:  given just  $f_0\in \Omega$,  (not necessarily   $E_s[f](0,z_0)$ finite for any order $2s$, $ 0<s\leq 1$, and $0<z_0:=z(0)$) the solution of the corresponding Cauchy Problem,Theorem~\ref{CauchyProblem},  satisfies  that $E_s[f](z(t),t) $ is finite for all time $t>0$, for some order factor $2s$ and rate $z(t)$ to be found depending on the data. 

The following Theorem proves the accuracy of these two statements. Their proofs consists in developing ordinary differential inequalities for  the quantities $E_s[f](z(t),t) $ valid for $t>0$, whose initial data is referred as  to $E_s[f](z(t),t)\mid_{t=0}$.  

The proof that follows shows  the generated exponential moment, associated to such  initial polynomial $k_{in}^{th}$-moment initial, has exponential order $\gamma   \in (0,2]$ with a rate $z$ depending on the evolution variable $t$, which may be view as the "instantaneous generation time" for a transition from a  $k_{in}^{th}$-polynomial moment to an exponential tail of order $\gamma$,  as well as depending on such initial data, as much as on the coercive factor, and,  naturally, on the potential rate $\gamma$.  This result, surely relies   on the  developed generation of moments estimates presented in   Theorem~\ref{propagation-generation},   \eqref{ODI-0} and 	 \eqref{moment-factors}.

The techniques developed in this manuscript, are an improvement with respect to the existing ones, in the sense that developed a unified framework for solutions of Boltzmann flows for binary elastic interactions, for hard potential and integrable  differential crossections  conforming the transition probabilities in such interacting particle flows.
 The specifics developed here  can be applied to any   solution $f(v,t) \in \Omega$ of the Boltzmann equation constructed by Theorem~\ref{CauchyProblem}, were inspired  in  the works on \cite{boby97, BGP04, GPV09, AG-JMPA08, mouhot06} for elastic and inelastic theories,  by framing the summability of moments into shifted power series (or Mittag-Leffler form), as well as on the work of \cite{AlonsoCGM} introducing the technique of study the summability of partial sums for hard potentials with integrable angular part of the transition probability,  and by non-integrable cross sections  in \cite{TAGP, PC-T18} for Maxwell type of interactions with or without the angular integrability requirement.
 
 These new techniques developed in this revision, not only are extended  to Lebesgue norms in Banach Spaces, but also revised several misses on some of the cited works, but improve the characterization of calculating rather explicit exponential rates by the coercive factors,  as  just  functions of the data. These techniques have been recently implemented  \cite{IG-P-C-mixtures} for a system of Boltzmann equation for disparate masses, and  most recently applied  to polyatomic gas models in  in  \cite{IG-P-C-poly-2020}.
 
 \smallskip
 
Starting from the   solution $f(v,t) \in \Omega$ of the Boltzmann equation constructed by Theorem~\ref{CauchyProblem},

 In theorem \ref{propagation-generation} we have shown that solutions of  the Boltzmann equation, with  finite initial mass and energy, satisfy propagation and generation property of  polynomial moments of any order or degree.  In this subsection we prove that, in fact,  exponential tails, up to the order of the potential $\gamma$, are generated and, even more, Gaussian tails are propagated \cite{BGP04, GPV09, AG-JMPA08}.  %

\begin{thm}[\bf{Exponential moments propagation}]\label{thm:exp-moment-propag}
Let $f\geq0$ be a solution to the Cauchy problem associated to the  Boltzmann flow, as posed in Theorem~\ref{CauchyProblem}, with a potential rate $\gamma\in(0,2]$ and integrable angular transition function $b(\hat u\cdot \sigma)\in L^1(\mathbb{S}^{d-1})$.  Assume, moreover, that the initial data satisfies for some $2s\in (0, 2]$
\begin{equation}\label{expotail-prop-data} 
E_s[f](z(t),t) \mid_{t=0} = E_s[f_0](0,{z_0}) = \int_{\real^d} f_0(v)  \, e^{z_0 \left\langle v\right\rangle^{2s}} dv =: M_P  < \infty, 
\end{equation}
Then,  for the coercive factor  satisfying the relation $A_{s\kc}\equiv c_{lb}^{-1} 2^{\gamma/2}\A_\kc$, for $c_{lb}$ from \eqref{c-lb},  there is an exponential  rate constant 
\begin{align}\label{exp-prop-rate} 
z(t):=z <\min \left\{
 z_0,\,  1, \,   
 \left(c_{lb}2^{\gamma/2} A_{s\kc}  (M_P-m_0[f_0]) \,
\frac{ \mathcal R(c_{lb}2^{\gamma/2} A_{s\kc}/\beta_{sk_0+\gamma/ 2})}{\kappa^P_{sk_0} }\, \right)^{2/\gamma} \right\},
\end{align} 
with a rate factor $ \mathcal R(\A_{s\kc}/\beta_{sk_0+\gamma/2})= \mathcal R(\delta_{sk_0+\gamma/2}) $  defined in \eqref{coercive rate max inverse},   satisfying  $\mathcal R(c_{lb}2^{\gamma/2} A_{s\kc}/\beta_{sk_0+\gamma/2})/\kappa^P_{sk_0} < 1$, is actually the maximum of  $  \mathcal R(\A_{s\kc}/\beta_{sk_0+\gamma/2})=\max\{  R_\C(\A_{s\kc}/\beta_{sk_0+\gamma/2}\; ; \;   R_{\mathbb E}(\A_{s\kc}/\beta_{sk_0+\gamma/2} \}$, with  $R_{\mathbb E}$ was defined in \eqref{RateE},  
as characterized by the choice of exponential rate $z$ from \eqref{choose z1} and  inequality \eqref{choose z1-2}, for any $0<\gamma\le 2$. 

 The coercive constant  ${A}_{s\kc} =  c_{\Phi} 2^{-\frac\gamma2 }( \|b\|_{L^1(\mathbb{S}^{d-1})}-\mu_{s\kc}) $,  the contractive factor $\mu_{\kc}$,   the integrability order of the angular transition function $b(\hat u\cdot\sigma)\in L^1(S^{d-1})$, the potential rate $\gamma$, also depending on the lower bound constant  $c_{lb}$ from \eqref{c-lb}, and positive and finite initial mass and positive energy $m_1[f_0]:=\|f_0\|^1_1$,  the initial exponential moment $M_P$ with order $s\in (0,1]$, rate $z_0$,  as well as the positive  constant $\kappa^P_{sk_0 } $ is  defined in \eqref{kappa-g},  depends only on the Cauchy problem data whose moment order $sk_0$, selected in \eqref{choose ko}, is explicitly characterized in  \eqref{choose ko-22} if the angular transition function $b\in L^p(\S^{d-1})$, for  $1\le p\le\infty$.  
 
 In addition, this constant $z$ defines the global in time estimate for the propagation of the following exponential moment of order $2s$ with rate $z$, 
%
 the infinite sum
\begin{align}\label{exp-moment-3} 
\sum_{k=0}^\infty \frac{z^k(t)}{k!} \, m_{s k}[f](t) &<
\sum_{k=\infty}^\infty \frac{z^k(t)}{k!} \, m_{s k}[f](t) =  
\int_{\real^d} f(v,t) \, e^{z(t) \left\langle v \right\rangle^{2s}} dv   =E_s[f](t,z) \nonumber\\
 & \leq  E_s[f_0](0,z_0) + m_0[f_0] =  M_P+ (e-1)m_0[f_0]   
\quad\text{for} \ t>0. 
\end{align}
\end{thm}
 \ \\ 

%
%

In addition, the following generation of exponential moments is proven in the following theorem.

\begin{thm}[\bf{Exponential moments generation}]\label{thm:exp-moment-gen}
Let $f\geq0$ be a solution to the homogeneous Boltzmann equation with potential $\gamma\in(0,2]$ and integrable angular transition $b$ and with finite initial mass and energy denoted by $m_1(t)=m_1(0)=\|f_0\|_1^1$.

Assume that the prescribed initial data satisfies the conditions to be in the solution set $\Omega$ defined in \eqref{SetOmega}, that is, for some fixed $k_{in}$-Lebesgue moment associated to the initial data $f_0(v)$, 
\begin{equation}\label{initial data exp gen}
m_{k_{in}}[f_0]    =\int_{\real^d} f_0(v) \, \left\langle v \right\rangle^{2k_{in}} \, dv =: M_G  < \infty, \qquad k_{in}\ge 1\,.
\end{equation}
Then,  for the coercive factor  satisfying the relation $A_{s\kc}\equiv c_{lb}^{-1} 2^{-\gamma/2}\A_\kc$, for $c_{lb}$ from \eqref{c-lb}, 
\begin{align} \label{exp-gene-rate}
0<z<  c_{lb}2^{\gamma/2-1}\min\left\{ A_{s\kc}  \,; \,  A_{s\kc} \frac{(M_G-m_0[f_0]) \, {\mathcal R}\left(\delta_{sk_0+\frac\gamma2}
\right) \, e^{-{c_{lb} 2^{\gamma/2}} A_{s\kc}}}{2 \, \kappa^G_{k_0} \left(  {\A}_{s\kc} + 1  \right)  
\left(1+\mathcal R(\delta_{sk_0+ \frac\gamma 2}) 
  {\bf \C}_{\text{\bf max}}^G \right) 
  }   \right\} ,
\end{align}
satisfying 
depending only on the coercive constant  ${A}_{s\kc} =2^{-\frac{\gamma}{2}} c_{\Phi}  ( \|b\|_{L^1(\mathbb{S}^{d-1})}-\mu_{s\kc})$,   the rate factor ${\mathcal R}(\delta_{sk_0+\frac\gamma2}):={\mathcal R}( c_{lb}2^{\gamma/2}A_{s\kc}/\beta_{sk_0+\gamma/2} )$ the  same rate in the propagation rate estimate \eqref{exp-prop-rate}defined in \eqref{coercive rate max}, with the moment order $sk_0$ selected in \eqref{choose ko-22-g} depending only on the Cauchy problem data.
    The positive constants $ {\bf\C}^G_{k_0}, \kappa^G_{sk_0}$ and $ {\bf \C}_{\text{\bf max}}^G$ are  
defined in  \eqref{CGk0}, \eqref{kappa-g}and   \eqref{CG max}, respectively
depends only on the Cauchy problem data.   

Hence,
\begin{equation*}
M^*_G := \sup_{t\geq0}\int_{\real^d} f(t)(v) \, \left\langle v \right\rangle^{2k_{in}} \, dv\,,
\end{equation*}
the exponential moment is generated and then bounded uniformly in time by the initial polynomial moment with estimate
\begin{equation}\label{exp generation p}
E_s[f](t, z \min\left\{t,1\right\})
\leq M^\star_G \,,\qquad  
\text{for all}\quad t \geq 0, \quad  \text{and} \quad  s\in\left(0,\frac\gamma2\right].
\end{equation}
In particular, when initial moment corresponds to the kinetic energy $k_{in}=1$ one has $M^{*}_G=M_G$ due to conservation of mass and energy.
\end{thm}

\begin{remark}\label{prop-gen-expo-rates}
One should expect that   for same model parameters and same initial data the  exponential rate $z$ are smaller for the  generated exponential high energy tails than for the propagated ones. This effect may be interpreted as the generating $L^1$-exponential bounds from above (or coarser)  than the propagating ones. 
\end{remark}

\medskip

The method of proof for these two theorems presented in this manuscript is an improvement from the one  inspired in the ideas \cite{AlonsoCGM} by the authors with collaborators by estimating the convergence of partial sums associated to \eqref{exp-moment}, that followed from the original work in \cite{boby97}  and \cite{GPV09} using estimates to estimate the radius of convergence of the geometric series on moments \eqref{exp-moment}.  

The strategy consists in generating Ordinary Differential Inequalities for partial sum sequences associated to the moment series defined in \eqref{exp-moment}, denoted by $ E^n_s(t,z)$,    which admits  global in time solutions,  uniform  in the parameter $n$,  and whose  limit  in $n$  yields the  generation and propagation of  exponential moment associated to the solution of the Boltzmann flow, characterizing their corresponding exponential order $2s$ and rates $z$, both   depend on the Cauchy problem data, on the global estimates  upper bound estimates obtained for the  $sk^{th}$-moments  Theorem~\ref{propagation-generation}, as well as  a moments lower bound  given by  the general property  proven in Lemma~\ref{lblemma} for showing that convolutions of  
probability distribution function, with bounded  mass, energy, a arbitrary over the energy moment,  with the potential rate function $\Phi_\gamma(|u|)$ satisfying (\ref{pot-Phi_1}-\ref{pot-Phi_2})  are controlled from below by a constant proportional to  the Lebesgue bracket $\langle v\rangle^\gamma$.

\begin{defn} For any function a function $f(v,t)\in \Omega$,  for $t$  fixed and  set $\Omega$ from \eqref{SetOmega} the solution set for the Cauchy problem shown in Theorem~\ref{CauchyProblem},  the partial $n^{th}$-sum the exponential moment  $ E_s[f](z(t), t)$ defined in \eqref{exp-moment}  of order $2s$  with $s\in(0,1)$ and rate $z>0$, is given by the shortest notation, 
\begin{equation}\label{expotail5}
E_{s}^n(t,z):= E^n _s[f](z(t),t)  :=   \sum_{k=0}^n m_{sk}[f](t) \frac{z^k(t)}{k!}\, , \qquad \text{ for all} \ \, n \in \mathbb N\,.
\end{equation}
\end{defn}

In order to derive an ODI, in time $t$,  for the partial exponential sums  $ E^n_{z,s}[f](t)$ while avoiding the use of interpolation techniques that we used in the derivation  of  moments estimates  and bounds from \eqref{mom-prop} and \eqref{mom-gen} for establishing a-priori estimates and the subsequent existence and uniqueness of the Boltzmann flow solution to the Cauchy problem from Theorem~\ref{CauchyProblem}.

The first estimate consists in a sharper form of binary sums  obtained from Lemma~\ref{BGP04-bino-sum} 
inequalities  \eqref{binsum} and \eqref{binsum-2}, which enable the control of partial $n^{th}$-sums  associated to the $k^{th}$-moments collisional form  $Q(f,f)(v,t)$  above,  and  the second one  is to invoke  Lemma~\ref{lblemma} to obtain a 
lower bound to the contribution  from the negative part from the moment of the binary collision operator by controlling  averaged
the potential part of the transition probability with respect to the solution density.

Thus, the  first  new improvement needs upper estimates for the $sk$-moment of the binary collisional form   $Q^+(f,f)(v,t)$ that arises from  the Corollary of the Angular Averaging Lemma, by  modifying Lemma~\ref{PovznerII} on the angular average of moments estimates  obtained in  previous sections showing global in time  propagation and generation of polynomial moments, as stated in Theorem~\ref{propagation-generation}.

Thus, the  first  new improvement needs upper estimates for the $sk$-moment of the binary collisional form   $Q^+(f,f)(v,t)$ that arises from the following second version of the Angular Averaging Lemma $\mathrm{II}$, by  invoking the Lower Bound Lemma~\ref{lblemma}, and
modifying Lemma~\ref{PovznerIII} on the angular average of moments estimates  obtained in  previous sections showing global in time  propagation and generation of polynomial moments, as stated in Theorem~\ref{propagation-generation}.   

These new techniques contain two novel components. 
The first one is to invoke to use the Lower Bound Lemma~\ref{lblemma} giving control from below the collision frequency associated to the potential function $\Phi(|u|)$ introduced in \eqref{pot-Phi_1}, instead of simply using the pointwise estimates \eqref{pot-Phi_2}. 
That means the $sk^{th}$-moment of the collisional integral estimates, as calculated in Theorem~\ref{mom-coll-op}, bounded by the $\mathcal L(m_{sk})[f]$ form defined in \eqref{Qmoment-0} is now written as
\begin{align}\label{Qmoment-0-expo}
m_{sk}[Q_\gamma(f,f)]  \le \mathcal{L}_{c_{lb}}(m_{sk}[f]):= 
\B_{  sk} \!-\!  \frac{\A_{{  s}\kc}}2 {m^{-c_\gamma(sk+\frac\gamma2) }_{1}[f_0]}  m^{1+c_\gamma(sk+\frac\gamma2)}_{ sk}[f], 
\  \text{for any } \ sk\ge s\kc>1, 
  \end{align}
with the  new  coercive constant $\A_{s\kc}:=c_{lb} 2^{\frac\gamma2} A_\kc$.

The second novel component consists in developing, a quantitative lower bound for the exponential moments  rates and orders, much as developed in  \cite{AlonsoCGM}, \cite{boby97}  and \cite{GPV09},  and more recently developed in  \cite{Alonso-Gamba-Tran}, \cite{IG-P-C-mixtures},  \cite{IG-P-C-poly-2020};
and  under development  in
\cite{DLC-IG-PC-mixturesL1} and \cite{Ampatzoglou_G_P_T-23}.
The following form of Angular Averaging Lemma an upper forum for the collisional  weight form associated to the weak formulation when testing the collision operator  by the  Lebesgue weight function $\langle v\rangle^{sk}$, resulting in a  the weighted angular average of the transition probability form that follows.

\begin{lem}[\bf{Angular Averaging Lemma $\mathrm{II}$}]\label{PovznerIII}
The Boltzmann weight function \eqref{weightuv2} associated to the weak form \eqref{bina-weak2} for  $\varphi(v)=\langle v\rangle ^{sk}$  for $sk\geq1$,  $s\in(0,1]$ and $k\in\mathbb{N}$,  now denoted by $G_{sk}(v_*,v)$, is estimated by 
\begin{align}\label{Pov2}
\begin{split}
G_{sk}(v_*,v) &= \int_{S^{d-1}}\Big( \langle v'\rangle^{2sk}   +   \langle v'_*\rangle^{2sk}  - \langle v \rangle^{2sk} -  \langle v_*\rangle^{2sk}\Big) B(|u|, \hat u\cdot\sigma)\, d\sigma\\
&\le \Phi(|u|)\Big( \mu_{sk}\big( \langle v \rangle^{2s} + \langle v_{*}\rangle^{2s} \big)^{k}  - \langle v\rangle^{2sk} -  \langle v_*\rangle^{2sk} \Big)\\
&\hspace{-1cm}\leq \Phi(|u|)\bigg( \mu_{sk}\sum^{k_j}_{j=1} {k \choose j} \langle v \rangle^{2s(k-j)}\langle v_{*}\rangle^{2sj}  - (\|b\|_{L^i(\mathbb{S}^{d-1}}- \mu_{s\kc})\big( \langle v\rangle^{2 sk} +  \langle v_*\rangle^{2 sk}\big) \bigg)\,, 
\end{split}
\end{align}
with $k_j = \lfloor{k + 1}\rfloor$, the contractive parameter $\mu_{sk}$  from (\ref{povzner0}, \ref{povzner0.1})  derived in the proof of the first  Angular Averaging Lemma~\ref{BGP-JSP04}.
\end{lem}
\begin{proof}   
Applying the first  Angular Averaging Lemma \ref{BGP-JSP04} with $j=sk$, and invoking  the elementary  estimate
\begin{equation*}
\big(\langle v \rangle^{2}+\langle v_{*}\rangle^{2}\big)^{sk} \leq \big(\langle v\rangle^{2s} + \langle v_{*}\rangle^{2s}\big)^{k}\,,\quad s\in(0,1]\,,
\end{equation*}
and using the exact combinatorial   binomial expansion \eqref{binsum}, applied with $x=\langle v \rangle^{2s}$ and $y=\langle v_{*}\rangle^{2s}$,  yields estimate \eqref{Pov2} for the weight function $G_{sk}(v_*,v)$ associated to the weak form of the Boltzmann flow. A detailed proof of this statement can also be found in \cite{BGP04}, Lemma~2. 
\end{proof}

\nn In addition, a straight forward consequence of Lemma \ref{PovznerII}, the upper and lower inequalities  from \eqref{pot-Phi_2} 
yield  a sharper estimates to the collision operator's moments,  than the ones obtained in \eqref{momentineq},  namely
\begin{align}\label{binom-1}
\tfrac{\partial}{\partial t}m_{sk}&=
\int_{\real^{d}}Q(f,f)(t, v)\langle v\rangle^{2sk}dv \nonumber \\
&\leq  2^{\frac{3\gamma}2 -1}   C_\Phi \mu_{sk}\sum^{k_j}_{j=1} {k \choose j}m_{sk-j}[f]m_{j+\gamma/2}[f]
- ( \|b\|_{L^1(\mathbb{S}^{d-1})} - \mu_{\kc}) \int_{\real^{d}}f(v)\langle v \rangle^{2sk} (f\ast\Phi)(v) dv \\
&:= 2^{\frac{3\gamma}2-1}   C_\Phi\mu_{sk}\,S_{sk} - \mathcal{A}_{s\kc} \, m_{sk+\gamma/2}[f]\,. \nonumber
\end{align}
 corresponding to the $sk^{th}$-moments Ordinary Differential  Inequality, where now $S_{sk} $ is a bilinear sum expressing the discrete convolutional structure 
 \begin{align}\label{binom-1-2}
S_{sk} [f]:= \sum^{k_j}_{j=1} {k \choose j}m_{s(k-j)}[f]m_{sj+\gamma/2}[f]
\end{align}
and   the  coercive constant  that is now is given by 
\begin{align}\label{new coercive}
\mathcal{A}_{s\kc}=c_{lb}c_{\Phi}( \|b\|_{L^1(\mathbb{S}^{d-1})}-\mu_{s\kc})m_0[f_0],  \qquad {with} \ \ b(\hat u\cdot \sigma)\in L^1(\mathbb{S}^{d-1}),
\end{align}
with the lower bound $c_{lb}$ from  \eqref{c-lb}  explicitly calculated in the  Lower Bound Lemma~\ref{lblemma}.

\begin{remark}\label{new-coercive-2} 
This new  ODI  resulting from the $sk^{th}$-moment estimates, using the operator $\mathcal L_{c_{lb}}$ from \eqref{Qmoment-0-expo}, is sharper  than the one  calculated for  Theorem~\ref{Coll-mom-estimates} in the proof of the both propagation and generation estimates. It is important to note that   the coercive constant $\mathcal{A}_{\kc}$, just defined in \eqref{new coercive},  is different from the  calculated constant  $A_{s\kc}:=c_{\Phi } 2^{-\frac{\gamma}2} \,( \|b\|_{L^1(\mathbb{S}^{d-1})}-\mu_{s\kc}) \,{m_0[f_0]}= 
\left(c_{lb} 2^{\frac{\gamma}2}\right)^{-1}  \mathcal{A}_{s\kc} $,  coercive constant \eqref{Akc}  (as defined  in  \eqref{moment-factors}  in the statement of Theorem~\ref{mom-coll-op}).  The main difference  is that    $A_{s\kc}$ depends only on their   first conserved moments of the initial data with respect to $k$, while  $\mathcal{A}_{\kc}$ also depends, inversely proportional to the upper bound $B$ of a $\beta^{th}$ classical moment introduced in the Lower Bound Lemma~\ref{lblemma}, with $\beta>1$, expressed explicitly in the the conformation of the the $c_{lb}$ constant.  This observation is relevant as the  exponential rate estimates depend on the $sk^{th}$-moment bounds, inversely with respect to $\mathcal{A}_{\kc}$. 
%
\end{remark}

\medskip

Therefore, it is necessary to briefly revisit the impact of the new coercive form $\A_{s\kc}$ in the development of precise global moments bound obtain by either propagation and generation of the initial data. 
As in the case for Theorem~\ref{mom-coll-op} for the collision moments estimates and Theorem~\ref{propagation-generation} recall the form   $c=c_\gamma(sk)=\gamma/(2(sk-1)) <1$ for $sk>1$, and  define the moment orders $\kmin$ and $\kmax$ by 
\begin{align}\label{kmaxmin}
0<  1< \kmin:= \min\{\klb,\kc \} \le \max\{\klb,\kc \}=:\kmax\,,
   \end{align}
for $\kc$ fixed from Theorem~\ref{propagation-generation}, but  $\klb$ replaces  $\kb$ in the defined in \eqref{bf kb} and Lemma~\ref{Lemma_ODE-extension2} that  now sets,    for any  $sk$ satisfying $2\le s\kmin\le s\kc\le s\kmax < sk$,  needs to be replace by
  \begin{align}\label{Bsk}
\ \ &\delta_{sk}:=\frac{\A_{s\kc}}{\beta_{sk}}=\frac{ c_{lb}2^\frac\gamma2 A_{s\kc}}{\beta_{sk}} <1; \quad  \nonumber\\
 \ \  \ &\beta_{sk}  =  {C}_{\Phi}(2^{(sk+\frac\gamma2)}\mu_{s\kc} +2\|b\|_{L^1(\mathbb{S}^{d-1})})\le\|b\|_{L^1(\mathbb{S}^{d-1})} \left({C}_{\Phi}2^{(sk+\frac\gamma2)} +2\right) ,  \\
\text{and}\qquad   &\B_{sk} :=  \beta_{sk}\, \C^{\text{up}}_{sk}(\delta),\ \  \ 
\ \ \, \mathcal{C}^{\text{up}}_{sk} 
= \ttm_{1,sk}[f_0]  
\left[  \delta_{sk}^{-e_\alpha}  \mathbb{I}_{ k\ge \kgamma } +  \delta_{sk}^{-e_\theta}   \mathbb{I}_{ k<\kgamma }  \right] , 
\nonumber\end{align}
 for the time invariant constant $\ttm_{1,sk}[f_0]  $  defined in \eqref{invariant k factor}.
 
 Note that $\klb$ and $\kb$ are not ordered, since the factor $ c_{lb}2^\frac\gamma2$ may or may not be bigger than unity.

Therefore under these conditions   the existence  and uniqueness properties  shown in Theorem \eqref{Theorem_ODE} hold, as well as propagation and generations estimates hold, meaning that 
for  any initial data $M_P =m_{k_{\text{ini}}}[f_i](0)$, coercive constant   $\A_{s\kc} = c_{lb} 2^{\frac\gamma2} A_\kc$,   following from \eqref{Qmoment-0-expo}, 
   replacing   on \eqref{Akc} in  the  proof of Theorem~\ref{propagation-generation}, 
 for any $s\in (0,1]$ yields the same results whose quantitative estimates replace $ A_\kc$ by $ \A_\kc$.


The exponents  
 $e_\alpha(sk)= {c^{-1}_\gamma(sk)} $ and $  e_\theta(sk)=\left(sk-1\right)\left(1+c_\gamma(sk) \right)^2 +c_\gamma(sk)  $, as defined in  \eqref{e max},  
varying with respect to $\kgamma $, as in  \eqref{compare delta expo}, satisfying, both  
\begin{align}\label{expk-sk2}
 e_\alpha(sk)= {c^{-1}_\gamma(sk)} = \frac{(sk-1)}{\gamma/2}>1  \ \ \text{and }\ \   e_\theta(sk)=\frac1{sk-1} \left( (sk-1+\gamma/2)^2+ {\gamma/2}\right)  \ge 1,  
 \end{align}
 are well defined for any $sk\ge 2$ and, as described in  \eqref{equi-2.2}, these exponents satisfy 
 \begin{align*}
\frac{1+e_\alpha}{1+c_\gamma(sk)} =e_\alpha=\frac{2(sk-1)}\gamma,  \qquad\text{and}\qquad \frac{1+e_\theta}{1+c_\gamma(sk)}=sk+\frac\gamma2. 
\end{align*}

 \medskip
 
Next, the focus is to find from   these last three inequalities a  characterization for  the coerciveness factor as functions of the quotient  $\beta_{sk}/\A_{s\kc}$.

Thus,  it becomes  important  to  study  the upper bounds of the a moment $2\le sk_0$, for $k_0>\kmax$ and $s\in(0,1]$ in order to study the summability of moments and   the rate of convergence of exponential moments.

Starting from the choice of $\kb$ in connections with the definition of  making the quotient $\delta=\A_{s\kc}/\beta_{sk}<1$,  the first goal is to estimate both  factors containing $\delta$ from the constant $\mathcal{C}^{\text{up}}_{sk} $
from \eqref{Bsk} and from the definition of equilibrium constant $\mathbb{E}_{c_\gamma(sk)}$ from \eqref{equi-2}. 

To this end, it convenient to invoke the  rate constants depending on $sk$ and $\gamma$  as defined in \eqref{RateC} and \eqref{RateE}, respectively,
  \begin{align}\label{coercive rate  RC & RE}
   R_{\C}(\delta_{sk}) &:=  R_{\C}(\A_{s\kc}/\beta_{sk}) :=  \delta_{sk}^{e_\alpha}  \mathbb{I}_{ k\ge \kgamma } +  \delta_{sk}^{e_\theta}   \mathbb{I}_{ k<\kgamma },\\   
   R_{\mathbb E}(\delta_{sk}) &:=   R_{\mathbb E}(\A_{s\kc}/\beta_{sk}) :=   \delta_{sk}^{e_\alpha}  \mathbb{I}_{ k\ge \kgamma } +  \delta_{sk}^{sk+\frac\gamma 2}   \mathbb{I}_{ k<\kgamma }, \nonumber
     \end{align}
   with the exponents, recall from 
     \begin{align}\label{coercive rate  R(/A)}
    R_{\C}^{-1}(\delta_{sk}):= \delta_{sk}^{-e_\alpha}  \mathbb{I}_{ k\ge \kgamma } +  \delta_{sk}^{-e_\theta}   \mathbb{I}_{ k<\kgamma }\qquad 
    R^{-1}_{\mathbb E}(\delta_{sk}):= \delta_{sk}^{-e_\alpha}  \mathbb{I}_{ k\ge \kgamma } +  \delta_{sk}^{-(sk+\frac\gamma 2)}   \mathbb{I}_{ k<\kgamma }.
      \end{align}
  
Their algebraic inverse follows the relation by inverting the quotient $\delta =(\A_{s\kc}/\beta_{sk})$,  namely   both
 \begin{align}\label{coercive rate  R inverse}
 R^{-1}_{\C}(\delta_{sk}):= R_{\C}(\delta_{sk}^{-1})\qquad\text{and}\qquad 
R^{-1}_{\mathbb E}(\delta_{sk}):= R_{\mathbb E}(\delta_{sk}^{-1}).
\end{align}

\smallskip

The first observation relates  exponents for different of $sk$ when compared to $\kgamma$ defined in \eqref{compare delta expo}. Starting from  noting that exponents $e_\alpha(sk)$ and $e_\theta(sk)$ are both bigger than $1$ for $sk\ge2$, as shown in \eqref{expk-sk2}, for any $k\ge \kmax=\max\{\kb,\kc\}$, with $\kb$ was defined in \eqref{choose kb} to make 
   \begin{align}\label{coercive rate 1}
 &1< \left(\frac{ 2^{sk-2 +\gamma } \mu_{ sk}+\|b\|_{L^1(\mathbb{S}^{d-1})}}{\|b\|_{L^1(\mathbb{S}^{d-1})}m_0[f_0] }\right)^{-e_\alpha(sk)}\mathbb{I}_{k\geq \kgamma} 
 + \left(\frac{ 2^{sk-2 +\gamma } \mu_{ sk}+\|b\|_{L^1(\mathbb{S}^{d-1})}}{\|b\|_{L^1(\mathbb{S}^{d-1})}m_0[f_0] }\right)^{-e_\theta(sk)}\mathbb{I}_{k< \kgamma} \nonumber\\
&\leq\delta_{sk}^{-e_\alpha(sk)}  \mathbb{I}_{k\geq \kgamma} + \delta_{sk}^{-e_\theta(sk)}   \mathbb{I}_{k<\kgamma} 
\le  R_{\C}(\delta_{sk}^{-1}) =:
  R_{\C}(\beta_{sk}/\A_{s\kc}),
\end{align}
 or, equivalently,  using the inverse relations from\eqref{coercive rate  R inverse}, implies	
$ R_{\C}(\A_{s\kc}/\beta_{sk}) =  R^{-1}_{\C}(\beta_{sk}/\A_{s\kc})< 1.$

 Analog as shown in \eqref{coercive rate R inverse},  the  estimate for the equilibrium  factor \eqref{equi-2} and identity   \eqref{equi-2.2},  is  
  \begin{align}\label{coercive rate  inverse 2}
   R_{\mathbb E}(\delta_{sk}) :=
   \delta_{sk}^{e_\alpha}  \mathbb{I}_{ k\ge \kgamma }  +   \delta_{sk}^{sk+\frac\gamma 2}   \mathbb{I}_{ k<\kgamma }
 =   \left( \frac{ \A_{{  s}\kc}  }{\beta_{  sk}} \right)^{\frac{2(sk-1)}\gamma }  \mathbb{I}_{ k\ge \kgamma }+ 
  \left(\frac{ \A_{{  s}\kc}  }{\beta_{  sk}}\right)^{sk+\frac\gamma2}\! \!   \mathbb{I}_{ k< \kgamma }<1. 
\end{align}


In addition,   since 
$\frac{\|b\|_{L^1(\mathbb{S}^{d-1})}m_0[f_0] } { 2^{sk-2 +\gamma } \mu_{ sk} +\|b\|_{L^1(\mathbb{S}^{d-1})}} <1, $ for any $sk>s\kmax \ge s\kmin\ge 2$ and $2\le\gamma$, then  
 the minimum from  rates calculations  \eqref{coercive rate 1} and \eqref{coercive rate inverse 2} is defined by and satisfies
\begin{align}\label{coercive rate min}
  \mathcal{R}(\delta_{sk}) := \mathcal{R}(\A_{s\kc}/\beta_{sk}) := \min\left\{{R_{\C}(\A_{s\kc}/\beta_{sk}) }  \, ;\,{   R_{\mathbb E}(\A_{s\kc}/\beta_{sk})}  \right\}  
\equiv   R_{\C}(\delta_{sk}) \leq 1, 
\end{align}
or, equivalently, its inverse 
is estimated from above by 
\begin{align}\label{coercive rate min-inverse}
1\le  \mathcal{R}^{-1}(\delta_{sk})
%
\!:= \!\max_{2\le sk} \left\{R_{\C}(\delta^{-1}_{sk})   ;  R_{\mathbb E}(\delta^{-1}_{sk})   \right\}\! =\! \max_{2\le sk}  \left[  \delta^{- \bar e_\alpha(sk)}  \mathbb{I}_{k\geq \kgamma} \!+\!
  \delta^{-\bar e_\theta(sk) }   \mathbb{I}_{k<\kgamma} \right] < \infty ,  
  \end{align}
  where the  following minimal and maximal  form of the exponent $ e_\alpha(sk)$ and $ e_\theta(sk)$, respectively,  are 
  \begin{align}\label{expo maxs}
  \underline e_\alpha(sk) &:= \frac\gamma2 \le \min_{2\ge sk}\left\{ e_\alpha(sk)  \,;\,\frac{1+ e_\alpha(sk)}{1+c}  \right \}\equiv 
  e_\alpha(sk) = \frac2\gamma (sk-1),    \nonumber
 \\
\quad \text{but} 
\quad  \underline e_\theta(sk) \!&:=\!\min_{2\le sk}\left\{ e_\theta(sk)  ;\frac{1+ e_\theta(sk)}{1+c}  \right \} \!=\!
\min_{2\le sk}\left\{(sk-\!1) \left(1\!+\!\frac{\gamma/2}{sk-1}\right)^2 \!+\!\frac{\gamma/2}{sk-1}  ; sk\!+\!\frac\gamma2 \right\}\!\\
&<\!
\end{align}
 
%
%
\bigskip

%

It is interesting to observe that ${\underline e}_\alpha=e_\alpha=\bar{e}_\alpha$. However,  the same can not be assess for the $e_\theta$ exponents, and in fact the form of ${\underline e}_\theta, e_\theta$ and $\bar{e}_\theta$ are different since, each  the maximum and minimum values of  $e_\theta$  change with respect  the value of $\gamma\in(0,2].$

Therefore,  for orders $\max{\{2<s\kmin< sk\le  sk_0\}}$ these  two minimal and maximal exponents  follow from the notation introduced above, \eqref{coercive rate min-inverse}, namely
\begin{align}\label{exp k_0-1}
 \max_{\{2<s\kmin< sk\le  sk_0\}} \frac{1+e_\alpha}{(1+c_\gamma(sk))} 
 &:= \frac{2(sk_0-1)}\gamma, 
 \quad \text{and} \quad
  \max_{\{2<s\kmin< sk\le  sk_0\}}\!\frac{1+e_\theta}{(1+c_\gamma(sk))}  &:=   sk_0 +\frac \gamma 2,
\end{align}
and their maximum exponents values \eqref{coercive rate min-inverse},  for any $k_0$ with $sk_0>2 $, are denoted by
\begin{align}\label{exp k_0}
\bar e_\alpha(sk_0)  &:=\max_{\{2\le sk_0\}}e_\alpha(sk) = \frac2\gamma (sk_0-1),
  \qquad \qquad \text{and}\\[4pt]
\bar e_\theta(sk_0) &:= \max_{\{2\le sk_0\}} \bar e_\theta(sk)  = \left\{(sk_0-\!1) \left(1\!+\!\frac{\gamma/2}{sk_0-1}\right)^2 \!+\!\frac{\gamma/2}{sk_0-1}  ; sk_0\!+\!\frac\gamma2 \right\}\!
 \nonumber
\end{align}


 In particular upper bound for both $ R(\beta_{sk}/\A_{s\kc})$ and $ R_{\mathbb E}(\beta_{sk}/\A_{s\kc}) $ follows, since the reader can easily check from estimates \eqref{coercive rate 1} and \eqref{coercive rate  R(/A)}, respectively, as well as from \eqref{coercive rate min-inverse}, yield
 
\begin{align}\label{coercive rate max}
  \mathcal R(\delta^{-1}_{sk_0}) :=    \mathcal R(\beta_{sk_0}/\A_{s\kc}) :=     
  %
  \left(\frac{\beta_{sk_0}} {\A_{s\kc}}\right)^{  \frac{2(sk_0-1)}{\gamma} }  \mathbb{I}_{k\geq \kgamma} +
  \left(\frac{\beta_{sk_0}} {\A_{s\kc}}\right)^{\bar e_\theta(sk_0)  }   \mathbb{I}_{k<\kgamma}   > 1, 
  %
\end{align}
and, $ \mathcal R(\delta^{-1}_{sk_0}) = \mathcal R^{-1}(\delta_{sk_0}) $ after invoking the relation  \eqref{coercive rate  R inverse} for inverse rates, it follows that 
\begin{align}\label{coercive rate max inverse}
  \mathcal R(\delta_{sk_0}) :=   \mathcal R(\A_{s\kc}/\beta_{sk_0}) &:=  
  %
  \left(\frac{\A_{s\kc}}{\beta_{sk_0}} \right)^{  \frac{2(sk_0-1)}{\gamma} }  \mathbb{I}_{k\geq \kgamma} +
  \left(\frac {\A_{s\kc}}{\beta_{sk_0}} \right)^{\bar e_\theta(sk_0) }   \mathbb{I}_{k<\kgamma}  < 1.
\end{align}

In this notation we can reformulate the global bound for moments' propagation \eqref{moments-bounds-2} 
 \begin{align}\label{moments-prop-sum}
m_{  sk}[f](t) \leq  \max \!\left\{ m_{  sk}[f_0] \, ;  \mathbb{E}_{c_\gamma{sk}}\right\} 
  = \max \!\left\{ m_{  sk}[f_0] \, ;  
  \,  \left( m_1\,\ttm_{1,k}\right)^\frac{ks-1}{sk-1+\frac\gamma2}[f_0] R_{\mathbb E}(\delta^{-1}_{sk}) \right\}, 
  \end{align}
  and the corresponding global estimate for moments' generation \eqref{generation-bound-2}
 \begin{align}\label{moments-gen-sum}
m_{  sk}[f](t) &\leq   \mathbb{E}_{c_\gamma(sk)} + m_1[f_0]\left( \frac 1{c_\gamma(sk)\, \A_{s\kc}}\right)^{\frac1{c_\gamma(sk)}} t^{-\frac1{c_\gamma(sk)} }\nonumber\\
&= \max 
  \,  \left( m_1\,\ttm_{1,k}\right)^\frac{ks-1}{sk-1+\frac\gamma2}[f_0] \, R_{\mathbb E}(\delta^{-1}_{sk}) + 
   m_1[f_0]\bigg( \frac {2(sk-1)}{\gamma\, \A_{s\kc}\, t}\bigg)^{\frac{2(sk-1)}{\gamma}},
  \end{align}
 hold for any $sk\ge s\kmax>2$.

Hence,  two fundamental estimates are need  to obtain and upper ODI for the partial sums of moments of the Cauchy problem solution,  set on Theorem~\ref{CauchyProblem}, with estimates from \eqref{binom-1}. One of them,   estimates  each partial sums coming from  positive contributions of  moments of the collision operator containing the term  $S_{sp}$. The other one  estimates  from below each partial sums of negative  contributions from   the collision operator moments. 

Such upper bound can be   obtained  by  a weighted discrete double mixing convolution that controls partial sums  of  $S_{sp}$ by   the binary interaction of partial sums expressed by a bilinear  product of the $n^{th}$-partial sum of the exponential expansion $E^{n} _s (t,z)$ from \eqref{expotail5}, and  of the  $n^{th}$-partial sum of the shifted $sp+\gamma$-moment  $I^{ell,n} _{s,\gamma}(t,z)$, defined by 
\begin{equation}\label{expotail6}
I^{n} _{s,\gamma}[f](z(t),t) := \sum_{k=0}^n m_{sk+\gamma/2}[f](t) \frac{z^{k}(t)}{k!}\, ,
\end{equation}
 arises in the calculation of  upper bounds for $n^{th}$-partial sums   $k^{th}$-moments associated to the collision operator  $Q(f,f)(v,t)$, as shown in the next Lemma.

%
%

  \medskip  
    
%
 
 From now on, and without loss of generality, the dependance of these partial sums on the solution $f$ is omitted in the notation of moments partial sums, i.e. $E^{n} _s (t,z):=E^{n} _s [f](t,z)$ and $I^{n} _{s,\gamma}(t,z):= I^{n} _{s,\gamma}[f](t,z)$, as all estimates that follow are uniform in $f(v,t)$, for $t>0$, as it will be shown that only depends of the a finite number of initial moments associated to the initial data that solves the 
 Cauchy problem, Theorem~\ref{CauchyProblem},  by means of the moment  propagation and generation estimates obtained in 
 Theorem~\ref{thm:exp-moment-gen}, as it will be indicated in the sequel

\begin{lem}[\bf{Discrete convolution for partial sums of $k^{th}$-moments of the binary collisional form}] \label{lem:convolution}
For any $\gamma, s \in(0,1]$ and $\tilde{k}<n$, both natural numbers in $\mathbb{N}$ with $s\tilde{k}  \geq 1$, and $S_{sk}(t)$ defined in  \eqref{binom-1-2}, the discrete convolutional form is estimated by 
\begin{equation}\label{expotail7}
\sum_{k=\tilde{k}} ^n \frac{z^k}{k!} S_{sk}[f](t) \le E^n _s (t,z)\, I^n _{s,\gamma}(t,z)\, ,
\end{equation}
where the product terms in the righthand side of this inequality have been defined  in \eqref{expotail5} and \eqref{expotail6}, respectively. 
\end{lem}
\begin{proof}
By a simple combinatorial counting argument on the left hand side of \eqref{expotail7}, it follows that
\begin{align*}
\sum_{k=\tilde{k}} ^n \frac{z^k}{k!} S_{sk}[f](t) = \sum_{k=\tilde{k}}^n \frac{z^k}{k!} &\sum_{j=1}^{k_j} {k \choose j}  m_{sj + \gamma/2}[f]\,m_{s(k - j)}[f]  
= \sum_{k=\tilde{k}}^{n} \sum_{j=1}^{k_j} m_{sj + \gamma/2}[f] \frac{z^{j}}{j!}  \, m_{s(k - j)}[f] \frac{z^{k-j}}{(k-j)!}\\
&= \sum^{n-1}_{j=1} m_{sj + \gamma/2}[f] \frac{z^{j}}{j!}  \sum_{k=\max\{\tilde{k}, j\}}^n m_{s(k - j)}[f]
\frac{z^{k-j}}{(k-j)!}  \leq I^n _{s,\gamma} (t,z)\, E^n _s(t,z), 
\end{align*}
which proves Lemma~\ref{lem:convolution}. %
\end{proof}
%
%
%

\nn 
\medskip

\begin{proof}[Proof of Theorem~\ref{thm:exp-moment-propag}] 
 This proof focus  on the moments summability properties from  the global in time moments propagation estimates \eqref{mom-prop}, which in fact it is enough to show it  for any $t\in (0,1]$.
 
\medskip


\subsection{Estimates for summability of moments} 
In oder to calculate global in time estimates for the propagation of  summability of  $sk$-moments for any rate $0\le sk $.
This task is performed by controlling  the partial sums $E^{n}_s(z, t)= \sum_{k=0}^{n}m_{sk}[f](t)$, for initial data  $f_0\in L^1_{2sk}$, with $sk\ge\kmax\ge\kmin>2$,   uniformly in time $t$ and partial summability index $n$.

The way to obtain these estimates entices to invoke  the estimates for  of  Ordinary Differential Inequalities \eqref{ODI-0} to obtain  an    Ordinary Differential Inequality for the $n$-sum  $E^{n}_s(z, t)$.  

Here reader needs to observe the following a delicate point regarding the following two needed properties.  The first one is the propagation global $sk$-moment estimates needs to be  performed under the conditions from  of  Theorem~\ref{mom-coll-op}
and Theorem~\ref{propagation-generation}, for any $sk\ge s\kmax$. 

The second one consists  on the observation that  estimates on an ODI  for partial sum of for $E^{n}_s(z, t)$,  needs also global upper bounds  for  $\frac{d}{dt} m_{sk}[f](t)=:m'_{sk}[f](t)$ for any   $0\le sk\le s(\kmax+1)$. 

To this end, start by setting the ODI for $E^{n}_s(z, t)$ adding up to order $0<k\le n$ the modified $sk^{th}$-moments ODIs    obtained now  from the $sk^{th}$-moments of the Boltzmann Equation from Angular Averaging Lemma II (Lemma~\ref{PovznerIII}), and \eqref{binom-1-2}, to get
  \begin{align}\label{mskprime estimates}
  m_{sk}'[f](t) &\le  \left[2^{\frac{3\gamma}2 -1}   C_\Phi \mu_{s\kc}\sum^{k+1}_{j=0} {k \choose j}m_{s(k-j)}[f](t)\,m_{sj+\gamma/2}[f](t)
  \right]  \mathbb{I}_{ k <\kmax+1} + \B_{sk}\,\mathbb{I}_{ k \ge\kmax+1} \\
&< \|b\|_{L^1(\mathbb {S}^{d-1})}  2^{\frac{3\gamma}2 -1}   \kmax C_\Phi  \,  2^{\frac k2 +1} \, m_{1+\frac{\gamma}2}[f_0]\, m_{s\kmax}[f_0]
\mathbb{I}_{ k <\kmax+1}
  +\ttm_{1,sk}[f_0] \, \beta_{sk} \,R_{\C}(\delta^{-1}_{sk} ) \, \mathbb{I}_{ k \ge\kmax+1}.  \nonumber
  \end{align}

Hence, for the first property,  upper bounds for moment's estimates for $m'_{sk}[f](t)$, 
one can obtain  upper bounds for $sk^{th}$-moment derivatives, for any $1<sk\le sk_0$,  such that $2< s\kmax\le sk_0$.

However,  for any initial $f_0\in L^1_{2sk}$, with $sk\ge s\kmax$,   clearly  $m_{sk}[f_0](t)\le m_{s\kmax}[f_0](t)$, for any $0\le sk<\kmax$, and so an upper bound estimate for  $m'_{sk} $ is  need for any $0\le sk\le s\kmax$ to be calculated as follows.

Therefore, after gathering all these estimates related to the coefficients and  moments estimates associated to Lebesgue (or polynomial) moments, conditions are sufficient to  study  moments summability properties  in order to find the tails associated to polynomial moments, denoted by  $m_{sk}:=m_{sk}[f](t)$ in the sequel.

Starting by  fixing $s\in(0,1]$, summing side-by-side over $k = 0,\cdots n$  the sharp form estimates   \eqref{binom-1} of the  $sk^{th}$-moment of the collision operator,  the following estimate to the time variation of  partial sum $E^{n}_s(z, t)$, for a choice of   $1<\kmax<k_0<n$ to be specified later,   yields
\begin{align}\label{expotail-9}
\frac{d}{dt} E_{s}^{n} (t)   &= \frac{d}{dt}\sum_{k=0}^n  m_{s k} \frac{z^k}{k!} 
= \sum_{k=0}^{k_0} \frac{z^{ k}}{k!}  m'_{s k} +  \sum_{k=k_0+1}^n \frac{z^{ k}}{k!} m'_{s k},
\nonumber\\
& \ \ \leq    \sum_{k=0}^{k_0}  m'_{s k} \frac{z^k}{k!} +2^{\frac{3\gamma}2-1}C_{\Phi} \sum_{k=k_0+1}^n \frac{z^k}{k!} \mu_{s k} S_{s, k} - \mathcal{A}_{s\kc} \sum_{k=k_0+1}^{n} m_{s k + \gamma/2} \,\frac{z^k}{k!} 
\\ 
&=    \sum_{k=0}^{k_0}  m'_{s k} \frac{z^k}{k!} +2^{\frac{3\gamma}2-1}C_{\Phi} \sum_{k=k_0+1}^n \frac{z^k}{k!} \mu_{s k} S_{s, k} 
-  \mathcal{A}_{s\kc} \left( I^{n} _{s,\gamma}(t,z) - \sum_{k=0}^{k_0} m_{s k + \gamma/2} \,\frac{z^k}{k!} \right)   \nonumber, \\
&=   \sum_{k=0}^{k_0}\left(  m'_{s k}  + \mathcal{A}_{s\kc} \sum_{k=0}^{k_0} m_{s k + \gamma/2} \right) \frac{z^k}{k!} +2^{\frac{3\gamma}2-1}C_{\Phi} \sum_{k=k_0+1}^n \frac{z^k}{k!} \mu_{s k} S_{s, k} 
-  \mathcal{A}_{s\kc}  I^{n} _{s,\gamma}(t,z)  \,\frac{z^k}{k!}    \nonumber, 
\end{align}
where   now { $\mathcal{A}_{s\kc}= c_{lb} 2^\frac{\gamma}2 A_{s\kc}$} is the renormalized coercive constant, defined in \eqref{new coercive},  as the moment estimates bounds \eqref{moments-prop-sum} are also written in terms of  $\mathcal{A}_{s\kc}$.

%

Next we focus on  estimates of each terms in the partial sum \eqref{expotail-9},  for the case of propagation of initial data given by $E_{s}^{n} (0)= \sum_{k=0}^{n}m_{sk}[f](0)$, with all $m_{sk}[f](0)$ finite, that can be estimated by means of  
combining inequalities \eqref{Bsk}, \eqref{moments-prop-sum} and \eqref{mskprime estimates}. They are 
 \begin{align}\label{mskprime estimates 2}
 m_{sk}'&[f](t) < \|b\|_{L^1(\mathbb {S}^{d-1})}  2^{\frac{\frac k2+3\gamma}2 }  \kmax C_\Phi  \, m_{1+\frac{\gamma}2}[f_0]\, m_{s\kmax+\frac{\gamma}2}[f_0] \mathbb{I}_{ k <\kmax+\frac\gamma 2}
  +\ttm_{1,sk}[f_0] \, \beta_{sk} \,R_{\C}(\delta^{-1}_{sk}) \, \mathbb{I}_{ k \ge\kmax+\frac\gamma 2}\nonumber\\
  &<   \|b\|_{L^1(\mathbb {S}^{d-1})} 2^{\frac{\frac \kmax2+3\gamma}2 }   \kmax C_\Phi  
    \!\left[ m_{1+\frac{\gamma}2}[f_0]\, m_{s\kmax+\frac\gamma 2} [f_0] 
   \mathbb{I}_{m_{s\kmax+\frac{\gamma}2}[f_0] \ge  \mathbb E_{s\kmax+\frac\gamma 2}} \right. \nonumber\\
    &\left.\qquad\qquad + \mathbb E_{c(1+\frac{\gamma}2)}\mathbb E_{c(s\kmax+\frac\gamma 2)}  \mathbb{I}_{m_{s\kmax+\frac{\gamma}2}[f_0] <  \mathbb E_{s\kmax+1}}    
\right]  \mathbb{I}_{ k < \kmax+1}+\ttm_{1,sk}[f_0] \, \beta_{sk} \,R_{\C}(\delta^{-1}_{sk}) \, \mathbb{I}_{ k \ge\kmax+1}  \nonumber\\
&<   \|b\|_{L^1(\mathbb {S}^{d-1})} \left[  2^{\frac{\frac{\kmax}2 +3\gamma}2 }   \kmax C_\Phi  
    \left( m_{1+\frac{\gamma}2}[f_0]\, m_{s\kmax+\frac\gamma 2} [f_0]  +
     R_{\mathbb E}(\delta^{-1}_{s+\frac{\gamma}2})\,R_{\mathbb E}(\delta^{-1}_{s\kmax+\frac\gamma 2})
     \right) \right] \mathbb{I}_{ k <\kmax+1} \nonumber\\
&\quad\quad\quad\quad\quad\quad +\ttm_{1,sk}[f_0] \, \beta_{sk} \,R_{\C}(\delta^{-1}_{sk} ) \, \mathbb{I}_{ k \ge\kmax+1} \nonumber\\
 &=:  \| b\|_{L^1(\mathbb S^{d\!-\!1})} \kmax C_\Phi \left(2^{sk_0+\frac\gamma2}\! +\! 2\!+\!2^{\frac{\frac \kmax2+3\gamma}2 }\right)  
 {\bf C}^P_k(f_0, \kmax, \gamma) \,R_{\C}(\delta^{-1}_{sk} ), 
      \end{align}
with the constant
      \begin{align}\label{constant CPk}
  {\bf C}^P_k(f_0, \kmax, \gamma)\!=\! \max\left\{\kmax C_\Phi    m_{1+\frac{\gamma}2}[f_0]\, m_{s\kmax+\frac\gamma 2} [f_0]  \!+\! 
      R_{\mathbb E}(\delta^{-1}_{s+\frac{\gamma}2})\,R_{\mathbb E}(\delta^{-1}_{s\kmax+\frac\gamma 2})\, ;\,
\ttm_{1,sk}[f_0] \!\right\}.
      \end{align}\\

%
%

Consequently, taking any $k_0>  \kmax + 2$,  and setting $\delta_{sk_0 +\frac\gamma2}=  \A_{s\kc}/ \beta_{sk_0+ \frac\gamma 2}$,  the  first term in that last RHS of \eqref{expotail-9} is estimated by 
\begin{align}\label{expotail-91}
\sum_{k=0}^{k_0}  \left( m'_{s k}(t) + \mathcal{A}_{s\kc} m_{s k + \frac\gamma 2} \right)  \frac{z^k}{k!} 
&<\left[ 2^{\frac{k_0  + 3\gamma}2} {\bf C}^P_{k_0}(f_0, \kmax, \gamma)   +  \A_{s\kc}  
\right]  \mathcal{R}(\delta^{-1}_{sk_0 +\frac\gamma 2} )  \sum_{k=0}^{k_0} \frac{z^k}{k!} \nonumber\\
& <  \kappa^P_{sk_0}\,  \mathcal{R}(\delta^{-1}_{sk_0 +\frac\gamma 2} )   \,  e^z\  \leq \   \kappa^P_{sk_0}\,  \mathcal{R}(\delta^{-1}_{sk_0 +\frac\gamma 2} ), 
\end{align}
for the rate  factor  $ \mathcal{R}(\delta^{-1}_{sk_0 +\frac\gamma 2} )=\mathcal R(\beta_{sk_0+ \frac\gamma 2}/\A_{s\kc}) $ as defined in \eqref{coercive rate max},  
with the propagation factor 
\begin{align}\label{kappa-2}
\kappa^P_{sk_0}= \| b\|_{L^1(\mathbb S^{d\!-\!1})} \kmax C_\Phi \left(2^{sk_0+\frac\gamma2}\! +\! 2\!+\!2^{\frac{\frac \kmax2+3\gamma}2 }\right) \,  
 {\bf C}^P_k(f_0, \kmax, \gamma)  +\A_{s\kc}, \
 \end{align}
with the constant rate     $\mathcal{R}(\beta_{sk_0+ \frac\gamma 2}/\A_{s\kc})$ 
 is calculated  from \eqref{coercive rate max}, provided that $k_0>\kmax+1$,  with  ${\bf C}^P_{k_0}(f_0, \kmax, \gamma)$  depending on $\kmax, \gamma$ and  on the initial $m_{sk_0+\frac\gamma2}[f_0]$ moments, and exponential rate $z$ chosen to be  
\begin{align}\label{zlne}
 z\leq \ln e =1 .
\end{align}

The second term from the right hand side estimate  \eqref{expotail-9} is controlled   by the inequality   from the  Discrete Convolution Lemma~\ref{expotail7} and the fact that $\mu_{sk}$ decreases in $sk$,  to obtain 
\begin{equation*}\label{Cpo2}
2^{\frac{k_0  + 3\gamma}2}C_{\Phi}\sum_{k=k_0+1}^n \frac{z^k}{k!} \mu_{s k_0}S_{s, k} \leq 2^{\frac{k_0  + 3\gamma}2}C_{\Phi} \mu_{s k}\,E^{n}_{s}(t,z)\,I^{n}_{s,\gamma}(t,z)\,,\qquad   sk_0\geq s\kmax>1.
\end{equation*}


Therefore,   gathering estimates \eqref{expotail-91} and  \eqref{Cpo2} 
 the right hand side 
of \eqref{expotail-9} is estimated by
\begin{align}\label{expotail-10}
\frac{d}{dt} E_{s}^{n} (t,z)   \le  \kappa^P_{sk_0}\, \mathcal{R}(\delta^{-1}_{sk_0 +\frac\gamma 2} )    - \Big( \A_{s\kc} - 2^{\frac{3\gamma}2-1}C_{\Phi}\mu_{s k_0}\,E^{n}_{s}(t,z)\Big) I^{n} _{s,\gamma}(t,z)\,,\qquad sk_0\geq s\kmax>1 \,. 
\end{align}

Next, step is motivated by the search of an exponential rate $z$ small enough that enables an upper bound for $E^{n}_{s}(t,z)$. To this end,  fix  $z \in (0,z_0]$, where $z_0$ is the exponential rate on the initial data  set in \eqref{expotail-prop-data}, 
 and introduce  the set of  times 
 \begin{equation}\label{timen}
t_n:=\sup\big\{ t\geq0\,:\, E^{n}_{s}(\tau,z) \leq M_P + (e-1) m_0[f_0], \; \forall\, \tau\in[0,t] \big\}\,,\quad \text{and}\ \ \ n\in  \mathbb{N}.
\end{equation}
where $M_P$ is the initial moment with exponential tail rate $z_0$ and the factor $(e-1)$ proportional to the initial mass $m_0[f_0]$ is the smallest one for any the initial exponential moment with rate $z_0$ and order $s$, as it will become a transparent choice  in the developing estimates.  

Clearly, the  value of $0<z_0$  need to be chosen sufficiently small   after the choice of rate \eqref{zlne}, making  $z\le\min\{ z_0, \ln e\}\le1$,  for each $n$ fixed. Hence  the partial moment sum $E^{n}_{s}(t,z)$ is uniformly bounded in $t$ and $n$, by showing  $E^{n}_{s}(t,z) \leq M_P + (e-1) m_0[f_0]$ for an unbounded $t \in \mathbb{R}^+$.  
%

Indeed,   first note that  when $0 < z < z_0$,
\begin{align}\label{expotail-11}
E^{ n}_{s}(0,z) =  \sum_{k=0}^n \frac{z^{ k}}{k!} m_{s k}(0) \leq  E_{s}({z_0},0)= M_P,  
\quad\text{uniformly in} \ n\in \mathbb{N}, \ \text{for any} \ z\le\min\{ z_0, \ln e\}. 
\end{align}

Now, since each $m_sk[f]$(t) moment  is continuous in time $t$, then  the partial sum $E_{s}^n(t,z)$ is also continuous in $t$. Hence,  since  $E_{s}^n(t, z)\leq  M_P + (e-1)m_0[f_0]$, for $t\in [0,t_n]$, then the right hand side of inequality \eqref{expotail-10} is dominated by 
\begin{align}\label{expotail-12}
\frac{d}{dt} E_{s}^{ n} (t)   &\le    \kappa^P_{sk_0}  \mathcal{R}(\delta^{-1}_{sk_0 +\frac\gamma 2} ) 
\!- \! \left(   \A_{s\kc} \!
-  \!   2^{\frac{3\gamma}2-1}C_{\Phi}\mu_{s k_0} (M_P \! + \! (e \!- \!1) m_0[f_0])\right)I^{ n} _{s,\gamma}(t,z) , 
\end{align}
for any $z\le\min\{ z_0, \ln e\}$, with rate  $\mathcal{R}(\delta^{-1}_{sk_0 +\frac\gamma 2} )$  from \eqref{coercive rate max} and $   sk_0 > s \kmax> 2$.\\

Therefore,  since  $0<\mu_{sk} \searrow 0$ as $k\to 0$,  then it is possible to choose  a sufficiently large $k_0$,   cutting-off in the initial  partial sums in \eqref{expotail-9}, such that 
  \begin{align}\label{choose ko}
  &\A_{s\kc}-  2^{\frac{3\gamma}2-1}C_{\Phi} \mu_{s k_{0}}(M_P+(e-1)m_0[f_0]) \ge \frac{  \A_{s\kc}}2, \ \ \text{or, equivalently,} \nonumber\\
  &\qquad  \qquad\qquad\qquad \mu_{s k_{0}} \ge  \frac{\A_{s\kc}}{2^{\frac{3\gamma}2} C_{\Phi} (M_P+ (e-1)m_0[f_0])}. 
\end{align}

\begin{remark}\label{k_0} Note that if  the angular transition $b(\hat u\cdot\sigma) \in L^p(S^{d-1})$, \ for $1<p\le\infty$, then by \eqref{Povp1},  the $k_0$ moment order can be chosen by taking  
  \begin{align}\label{choose ko-22}
    sk_0=  \mathcal{O}\left(\frac{2^{\frac{3\gamma}2} C_{\Phi} (M_P+ (e-1)m_0[f_0])}{\A_{s\kc}}\right)^q,  \qquad \text{for} \  q= \frac{p}{p-1} \ \ \text{and} \ \ \ sk_0>s\kmax.
\end{align}
\end{remark}

\ \\

 These estimates immediately lead to Ordinary Differential Inequality 
\begin{align}\label{expotail-12}
\frac{d}{dt} E_{s}^n (t, z)  \le\mathcal{R}(\delta^{-1}_{sk_0 +\frac\gamma 2} ) \, \kappa^P_{sk_0}   -  \frac{\mathcal{A}_{s\kc}}2  \, I^n _{s,\gamma}(t,z)\,,\qquad \ sk_0>s\kmax. 
\end{align}

In addition, a lower bound  for the shifted moments' partial sum follows form a modified estimate  \cite{AlonsoGThar} is
\begin{align}\label{lower bound Isgamma}
I^{n} _{s,\gamma} (t,z) &\geq \sum^{n}_{k=0}\frac{z^k}{k!}\int_{\{\langle v \rangle \geq z^{-1/2}\}} f(v,t)\langle v \rangle^{2(sk + \gamma/2)}\text{d}v \nonumber\\
& \geq z^{-\gamma/2}\bigg( E^{n}_{s}(t,z) - \sum^{n}_{k=0}\frac{z^k}{k!}\int_{\{\langle v \rangle \leq z^{-1/2}\}} f(v,t)\langle v \rangle^{2sk}\text{d}v \bigg)\nonumber\\
&\geq z^{-\gamma/2}\bigg( E^{n}_{s}(t,z) - m_0[f_0]\sum^{n}_{k=0}\frac{z^{k(1-s)}}{k!} \bigg) > z^{-\gamma/2} \Big( E^{n} _{s} (t,z) - e^{z^{1-s}}\,m_0[f_0] \Big) \\
& 
\geq z^{-\frac{\gamma}2} \left(E_{s}^{n}(t,z) - e  m_0[f_0] \right)\, . \nonumber
\end{align}
 And, since  the rate is already chosen to be  $0<z\le\min\{ z_0, \ln e\},$  then  $ 1< e^z< e^{z^{1-s}} \le e$  in the range of    $0<s\le 1$ values, as the initial data may be an arbitrary rate $z_0>0.$\\
 
 
The ODI  solution estimate  for partial the sums $E_{s}^{n}(t,z)$ is finalized after defining the state variable
\begin{equation}\label{XODI}
X^{n}_{k_0}(t,z):= E_{s}^{n}(t,z) -  e\,m_0[f_0],  
\end{equation}
whose initial condition is $X^{n}_{k_0}(0,z_0):= E_{s}^n(0,z_0) -  e\,m_0[f_0] = M_P -  e\,m_0[f_0].$ Therefore,  writing an new ODI  for $X^n(t,k)$ after invoking  the lower estimate for   $I^{n} _{s,\gamma} (t,z)$ from  \eqref{lower bound Isgamma}, applied  into \eqref{expotail-12} yields  
\begin{equation}\label{expotail-13}
\frac{d}{dt}X^{n}_{k_0}(t,z) <-\tfrac{\mathcal{A}_{s\kc}}{2z^{\gamma/2}} \, X^{n}_{k_{0}}(t,z) +  \kappa^P_{sk_0}  \mathcal{R}(\delta^{-1}_{sk_0 +\frac\gamma 2} ) \qquad t\in(0,t_{n})\,,
\end{equation}
whose time integration of this differential inequality leads to
\begin{equation}\label{expotail-13}
X^{n}_{k_0}(t,z) < X^{n}_{k_0}(0,z)\,e^{-\tfrac{\mathcal{A}_{s\kc}\,t}{2z^{\gamma/2}}} +  \frac{2z^{\gamma/2} \kappa^P_{sk_0}  \mathcal{R}(\delta^{-1}_{sk_0 +\frac\gamma 2} )}{\A_{s\kc}} \Big(1 - e^{-\tfrac{\mathcal{A}_{s\kc}\,t}{2z^{\gamma/2}}}\Big),\qquad t\in(0,t_{n})\,.
\end{equation}
Thus, the choice of  rate $z>0$ small enough  can be made setting 
 \begin{equation*}
 z^{\gamma/2} \, \frac{2\mathcal{R}(\delta^{-1}_{sk_0 +\frac\gamma 2} ) \, \kappa^P_{sk_0} }{\mathcal{A}_{s\kc}} =   M_P - m_0[f_0]  >0,\qquad t\in(0,t_{n})\,,
\end{equation*}
since  $M_P - m_0[f_0] = \sum_{k=1}^n \frac{z_0^{ k}}{k!} m_{s k}(0)>0$,  or equivalently, identifying a lower bound rate  $z_1$ by
 \begin{equation*}
  z\leq z_1 := \, \left(    \frac{ 
   {\A}_{s\kc} \mathcal R(\A_{s\kc}/\beta_{sk_0+ \frac\gamma 2})  (M_P-m_0[f_0])} { \kappa^P_{sk_0} }     \right)^{2/\gamma}
,\qquad t\in(0,t_{n})\,, 
\end{equation*}
with the rate  factor    $\mathcal R(\A_{s\kc}/\beta_{sk_0+ \frac\gamma 2}):= \mathcal R(\delta_{sk_0 +\frac\gamma 2} ) =\mathcal R^{-1}(\delta^{-1}_{sk_0 +\frac\gamma 2} ) <1$ as shown in  \eqref{coercive rate max inverse}, with  $\beta_{sk_0+\frac\gamma2}$ calculated from \eqref{betak}. 

As a consequence the parameter  rate  $z_1$  is uniform in $n\in{\mathbb N}$ and time $t$.

It is also worth to notice that  the propagation constant $\kappa^P_{sk_0}$,  calculated  \eqref{kappa-2} is $2\kappa^P_{sk_0}=2^{O({\frac{k_0  + 3\gamma}2} +1)},$ 
as a consequence a quantitative estimate for the $z_1$ rate can be obtained as follows
\begin{align}\label{choose z1}
 z_1 &:=  2 
 \left(\A_{s\kc}   (M_P-m_0[f_0])\, 
 \frac{\mathcal R(\delta_{sk_0 +\frac\gamma 2} )}{\kappa^P_{sk_0}}\right)^{\frac{2}{\gamma}},\qquad t\in(0,t_{n})\,, 
  \end{align}
  with 
 \begin{align}\label{choose z1-2}
\left(\frac{ \mathcal R(\delta_{sk_0 +\frac\gamma 2} )}{\kappa^P_{sk_0} }\right)^{\frac{2}{\gamma}}
\! =\!  \mathcal{O} \left({2^{-\frac{k_0 +3\gamma + 2}\gamma} }\right)
  \left[  \!  \left(\delta_{sk_0 +\frac\gamma 2} \right)^{ e_\alpha(sk_0) }   \! \!  \mathbb{I}_{ k\ge\kgamma }   
\!  +\!   \left(\delta_{sk_0 +\frac\gamma 2}\right)^{\bar e_\theta(sk_0) }   \! \!  \mathbb{I}_{ k < \kgamma }  \right]^{\frac{2}{\gamma}} \! <\! 1, 
\end{align}
for $e_\alpha(sk_0) =c^{-1}_\gamma(sk_0)$ and $\bar e_\theta(sk_0)$ defined in \eqref{exp k_0}.
Clearly $z_1$ degenerates the choice moment rate  as $sk_0$ grows, so the best possible rate should be taking with the smallest $sk_0$ satisfying condition \eqref{choose ko}, or  \eqref{choose ko-22}, depending on the integrabitlity order of the angular functions $b(\hat u\cdot\sigma)$.

In particular, $X^{n}_{k_0}(0,z) < E_{s}^{n}(0,z_0) - m_0[f_0]$, which implies together with \eqref{expotail-13} and \eqref{choose z1} that 
 \begin{equation*}
E_{s}^{n}(t,z)  -e\,m_0[f_0]   <   M_P -  m_0[f_0] \,,  \quad  \text{for any} \ \ t\in(0,t_{n})\,.
\end{equation*}
with the upper bound uniform in $n\in{\mathbb N}$,or equivalently, 
 \begin{align}\label{expotail-15}
 E_{s}^{n}(t,z)    <   M_P   + (e-1) m_0[f_0]    \qquad           \text{for all}\   t\in(0,t_{n})\ \ \text{and}   \  z\le \min\{z_0, \ln 2, z_1 \} <1, 
\end{align}
  for any $k> k_0$, with $k_0$ large enough to satisfy condition \eqref{choose ko}, and a uniform in time exponential rate $z$ the minimum between $z_0$ is the initial exponential rate, $z_1$ from condition \eqref{choose z1}.

 Now, the continuity of $E_{s}^{n}(z,t)$ with respect to time $t$ ensures that a strict inequality holds, at least, in a slightly larger time interval $[0, t_n + \varepsilon)$, $\varepsilon >0$. Yet, this fact  contradicts the maximality of $t_n$ unless $t_n = + \infty$.  Therefore,  
$E_{s}^{n}(z, t) < M_P + (e-1)\,m_0[f_0]$ for all $t\geq 0$ and $n\in \mathbb{N}$. 

Thus, letting $n\rightarrow \infty$ we conclude
\begin{align*}
E_s^{n}[f](z, t) &= \lim\limits_{n\rightarrow \infty}E_{s}^{n}[f](z, t)  \leq (e-1) m_0[f_0] + M_P  \\
 &\qquad< (e-1)m_0[f_0] + \int_{\mathbb{R}^{d}} f_0(v) \, e^{z_0 \left\langle v \right\rangle^{2s}} dv = (e-1)m_0[f_0] + 
 E_s[f](z_0, 0)
, \quad \forall t\geq 0,
\end{align*}
i.e. the solution $f(v,t)$ to   Boltzmann equation with finite initial exponential moment of order $2s$ and rate $z_0$ propagates the exponential moments of order $2s$, but at a lower constant  rate $z$ than the initial  rate $z_0$ ,  with $z$ defined by 
\begin{equation}\label{expotail-16}
z(t):= z <\min \left\{ z_0,\,  \ln e, \, z_1 \right\}=\left\{ z_0,\,  1, \, z_1 \right\}, 
\end{equation}
for $z_1$ defined by \eqref{choose z1}, and well characterized by  \eqref{choose z1-2}.

\end{proof}

\begin{proof}[Proof of Theorem~\ref{thm:exp-moment-gen}] ({\bf Generation of Exponential moments})
This proof, has similarities with respect to the one of  previous Theorem~\ref{thm:exp-moment-propag}, mostly in the sense that we seek  an ODI for   the exponential moment, denoted for simplicity of notation, by $E_{s}[f](t) := E_{z,s}[f](t)  $ where the exponential order factor $2s$ and rate $z(t)$ need to be found and fully characterized, by   
the initial data    $f_{in}\in L^1_{k_{in}}(\mathbb{R}^{d})$ for $k_{in}\ge 2\epsilon +1$, by the potential rate $\gamma\in(0,2]$ and and by properties on the binary collision operator as described for  Cauchy Problem stated in Theorem~\ref{CauchyProblem}. That means  the initial data is defined takes the form  $\int_{\mathbb{R}^{d}} f_0(v) \, \left\langle v \right\rangle^{k_{in}}  \mathrm{d}v =: M_G $ finite,        with $k_{in}\ge\kc> 1$,  as defined in \eqref{exp generation p}.  

We require $z(t)$ to be continuous, increasing  for $t \in (0,1]$ and with $z(0)=0$ since the exponential weight is not available at $t=0$.

As  the proof for propagation of exponential tails was initiated, in Theorem~\ref{thm:exp-moment-propag},   recall  the first  the $sk^{th}$-moments generation global bounds derived in Theorem~\ref{propagation-generation}, \eqref{generation-bound},   which can be large for short times. 

Set constant $C_{\gamma k_0} > 0$ defined by the upper bound of first $sk_0$ finite moments,  and invoke the $sk$-parameters  introduced in the polynomial moments estimates, to obtain, for any $0<s\le1<s\kc<sk.$
Then,  starting  from  the solution of the Boltzmann flow for  $f_{in}(v) $, constructed in Theorem~\ref{CauchyProblem}, with the solution set $\Omega$ as defined in \eqref{SetOmega},
our task consists on studying the time evolution of generated moment partial power series in $z(t)$ up to order $n$, denoted  $E^{n}_{s}(z(t),t)$,  with the goal  to find uniform in $n$  rates $z(t)>0$ dependent  by the given  initial data. 

As performed in the proof of the propagation of exponential tail estimates, the corresponding moments ODI satisfied by the solution $f$ of the Cauchy problem posed in Theorem~\ref{CauchyProblem} written as in \eqref{binom-1-2} becomes again the starting point and the  representation of $I^n _{s,\gamma}[f]$ defined explicitly  in \eqref{expotail6}, is now invoked evaluated at $(z(t),t)$ to obtain  
\begin{equation}\label{expotail6}
I^n _{s,\gamma}[f](z(t), t) := \sum_{k=0}^n m_{sk+\frac\gamma2}[f](t) \frac{z^k(t)}{k!}\,.
\end{equation}

Thus, starting first for values of  $t\in(0,1]$ and invoking  the product rule and the fact that contractive constant $ \mu_k$  decreasing in $k$,    $ I^n _{s,\gamma}(z(t),t)$ given now by  the $n^{th}$-partial sum of the shifted to the $sk+\gamma/2$-moment  defined in \eqref{expotail6},  after some elementary algebra it follows that for $t\in(0,1]$.
Hence,  taking the rate of change of the partial moment's sums 
$ E^{n}_{s}(z(t),t)$, with respect to time $t$, and invoking the $n^{th}$-partial $sk^{th}$-Lesbegue moments sums and the more accurate binary expansion for  moments estimates  from the collisional form  \eqref{PovznerIII}, the for a suitable $k_0 > \kmax$ to be chosen later,  the following  estimate holds
\begin{align}\label{ODI tail gen}
\frac{d}{dt} E^{n}_{s}(z(t),t) &=  \sum_{k=0}^{k_{0}}  m'_{sk}[f] \frac{z^k(t)}{k!}  + \sum_{k=k_0+1}^n  m'_{sk}[f] \frac{z^k(t)}{k!} + 
z'(t)\sum_{k=1}^{n}  m_{sk}[f] \frac{z^{k-1}(t)}{(k-1)!}\nonumber\\
&\leq \sum_{k=0}^{k_{0}}  m'_{sk}[f] \frac{z^k(t)}{k!}  + \sum_{k=k_0+1}^n \frac{z^k(t)}{k!}  \Big(
2^{\frac{3\gamma}2-1}C_{\Phi} \mu_{sk} S_{sk}[f] -\mathcal {A}_{s\kc} m_{sk+\frac\gamma2}[f] \Big) \nonumber\\
&\hspace{4.3cm} + z'(t) \sum_{k=0}^{n-1} m_{sk+\frac\gamma2}[f] \, \frac{z^k(t)}{(k)!}\\
&\le\sum_{k=0}^{k_{0}}  \left(m'_{sk}[f] + \mathcal {A}_{s\kc}   m_{sk+\frac\gamma2}[f] \right) \frac{z^k(t)}{k!}  + 2^{\frac{3\gamma}2-1}C_{\Phi}\mu_{sk}\sum_{k=k_0+1}^n \frac{z^k(t)}{k!}  S_{sk}[f] \nonumber\\
&\hspace{1.5cm}  - \mathcal {A}_{s\kc}   I^n _{s,\gamma}(z(t),t) +z'(t) I^{n-1}_{s,\gamma}(z(t),t), \nonumber
\end{align}
much like the split made in the Propagation of partial some case \eqref{expotail-9}, and so  the right hand side from \eqref{ODI tail gen} will be  estimated for times adapting the estimates performed  in \eqref{expotail-9}, taking into account that the unknown exponetial rate $z=z(t)$ depends on time $t\in (0,1]$.

Hence,  invoking the global  moments generation estimates from \eqref{generation-bound-2}, for  $c_{\gamma}(sk) = \frac{\gamma}{2(sk-1)}$,  for any $sk\ge s\kmax>2$, 
the first step starting from  consists in  performing the $m'_{sk}[f]$ estimates similar to those on  \eqref{mskprime estimates}, for any   $0<sk \le s\kmax$ but now  applying the generation upper bounds  \eqref{moments-gen-sum} when needed.
   In addition,  if  the initial data $M_G =m_{s\kmax}[f_1](0)$, and $sk_0\ge  s\kmax$, then
    \begin{align}\label{mskprime estimates gen 2}
  m_{sk}'[f](t)&<\! \ttm_{1,sk}[f_0] \, \beta_{sk} \,R_{\C}(\delta^{-1}_{sk}) \mathbb{I}_{ k \ge\kmax+1}\ \nonumber\\
  &\qquad\quad+ \|b\|_{L^1(\mathbb {S}^{d-1})} \ 2^{\frac{\frac \kmax 2+3\gamma}2 }  \kmax C_\Phi   m_{1+\frac{\gamma}2}[f](t) m_{s\kmax+\frac{\gamma}2}[f](t) \mathbb{I}_{ k <\kmax+1}
\nonumber \\
  &<\ttm_{1,sk}[f_0] \, \beta_{sk} \,R_{\C}(\delta^{-1}_{sk}) \, \mathbb{I}_{ k \ge\kmax+1}+  \|b\|_{L^1(\mathbb {S}^{d-1})} 2^{\frac{\frac{\kmax}2 +3\gamma}2 } \!  \kmax C_\Phi  
   \left[\left(R_{\mathbb E}(\delta^{-1}_{s+\frac\gamma 2 }) + \frac{m_1[f_0]}{\gamma\, A_{s\kc}}  \nonumber \right)\right.\\
   &\left.\quad\quad
    \left(R_{\mathbb E}(\delta^{-1}_{s\kmax+\frac\gamma 2 }) +  m_1[f_0]\max_{k<\kmax+1} \left(  \!\frac {1}{c_\gamma(sk\!+\!\frac\gamma2) \A_{s\kc}}\right)^{\frac{1}{c_\gamma(sk\!+\!\frac\gamma2)}}
  \right)  t^{ -\frac{1}{c_\gamma(sk\!+\!\frac\gamma2)}} \right] \, \mathbb{I}_{ k <\kmax+1} \nonumber\\
  &  <   \kappa^G_{k_0}
    \left( \mathcal R (\delta^{-1}_{sk}) +  
{\C}^G_{\bar k} \, t^{-\frac{2sk\!+\!\gamma-2}{\gamma}}\right),   
     \end{align}
     with the time constant factors
 \begin{align}\label{kappa-g}
\kappa^G_{k_0} :=   \| b\|_{L^1(\mathbb S^{d\!-\!1})} \left(  C_\Phi \left(2^{sk_0+\frac\gamma2}\! +\! 2\!\right)+\!2^{\frac{\frac \kmax2+3\gamma}2 }\right)   \,  \max\left\{\left(m_1 ;\ttm_{1,sk_0+\frac\gamma2};\left( m_1\,\ttm_{1, sk_0}\right)^{\!\frac{sk_0-1}{s\kmin+\frac\gamma2}}\right) \right\}[f_0] 
\end{align}
and 
      $ {\C}^G_{\kbar}$ denoting, for any moment's order $1<k\le \kbar$ such that  $ s\kbar\ge sk> 1$ and $0<s\le 1$, 
          \begin{align}\label{CGk}
     {\bf \C}^G_{\kbar}:=
 \max_{k<\kbar+1}\!\!\left(  \!\frac {c^{-1}_\gamma(sk\!+\!\frac\gamma2)}{ \A_{s\kc}}\right)^{c^{-1}_\gamma(sk\!+\!\frac\gamma2)}.
 \end{align}       %

Since  the positive  coercive constant $\A_\kc=c_{lb}2^{\frac\gamma2}{ A_\kc}$, may or may not be bigger that unity, and the power associated to the quotient $c^{-1}_\gamma(sk+\gamma/2)/ \A_{s\kc}$ needs to be maximize between $1<sk<s \kbar$, then
 the corresponding constant upper $ {\bf \C}^G_{k_0}$ bound associated to  \eqref{CGk}   is simply  stated to be 
 \begin{align}\label{CGk0}
{\bf \C}^G_{k_0} := \max_{k\le k_0+1} \left(  \!\frac {sk-1+\gamma/2}{ \A_{s\kc} \, \gamma/2}\right)^{\!\!\frac{sk-1+\gamma/2}{\gamma/2}}.
  \end{align}

Therefore, for any $t\in(0, 1]$ and  moment order  $k_0\ge 0$,  the  first term in that last RHS of \eqref{moments-gen-sum}  estimates the partial moments sum majoring with the  constant $ {\C}^G_{k}$ from \eqref{CGk} by $ {\C}^G_{sk_0}$, and natural upper bounds for $m'_{sk+\frac\gamma 2}$ from \eqref{mskprime estimates gen 2}
as follows
\begin{align}\label{ODI tail gen 1} \nonumber
\sum_{k=0}^{k_{0}}  &\left( \mathcal{A}_{s\kc}  m_{s k+\frac\gamma2}[f] +m'_{sk}[f] \right)  \frac{z(t)^k}{k!} < \sum_{k=0}^{k_{0}}  \left[{\A}_{s\kc}  m_{s k+\frac\gamma2}[f]     + \kappa^G_{k_0}
    \left( \mathcal R (\delta^{-1}_{sk}) +  
{\C}^G_{\bar k} \, t^{-\frac{2sk\!+\!\gamma-2}{\gamma}}\right) \right] \frac{z(t)^k}{k!}\nonumber\\
&\ \  <  {\A}_{s\kc}\left(  \left( m_1\,\ttm_{1, sk_0}\right)^{\!\frac{sk_0-1}{s\kmin+\frac\gamma2}} \mathcal R(\delta^{-1}_{sk_0+ \frac\gamma 2} )  \max_{0<t\le 1} e^{z(t)} +
  m_1[f_0]\,{\C}^G_{sk_0}\sum_{k=0}^{k_0}t^{-{\frac{sk-1+\gamma/2}{\gamma/2}}}  \frac{z(t)^k}{k!}\!\right)   \nonumber \\
& \qquad \quad+   \kappa^G_{k_0} \mathcal R (\delta^{-1}_{sk_0}) \max_{0<t\le 1} e^{z(t)}  +  \kappa^G_{k_0} {\C_\kmax^G} 
   \, \sum_{k=0}^{\kmax+1} t^{-\frac{2sk+\gamma-2}{\gamma}}    \frac{z(t)^k}{k!}\\
 &\quad < 
\kappa^G_{k_0} \left(  {\A}_{s\kc} + 1  \right)  \max_{0<t\le 1} e^{z(t)} \mathcal R(\delta^{-1}_{sk_0+ \frac\gamma 2}) +  \left[ {\A}_{s\kc}\, m_1[f_0]\, 
  {\bf\C}_{sk_0+\frac\gamma2}^G+   \kappa^G_{k_0} \,
  \C_\kmax^G \right]  
  \sum_{k=0}^{k_{0}}   \, t^{-\frac{2sk+\gamma-2}{\gamma}}    \frac{z(t)^k}{k!}\nonumber . 
  \end{align}
with the constant rate factor     $\mathcal R(\delta^{-1}_{sk_0+ \frac\gamma 2}):=\mathcal{R}(\beta_{sk_0+ \frac\gamma 2}/\A_{s\kc})$ 
 is from  \eqref{coercive rate max}, provided that $sk_0>1$.

%
 It follows that estimate \eqref{ODI tail gen 1} can be expressed term of these constant factors  \eqref{kappa-g} by
%
 \begin{align}\label{ODI tail gen 1} 
\sum_{k=0}^{k_{0}}  &\left( \mathcal{A}_{s\kc}  m_{s k+\frac\gamma2}[f] +m'_{sk}[f] \right)  \frac{z(t)^k}{k!} 
<\kappa^G_{k_0} \left(  {\A}_{s\kc} + 1  \right)  
\left(\max_{0<t\le 1} e^{z(t)}\,\mathcal R(\delta^{-1}_{sk_0+ \frac\gamma 2}) 
  +   {\bf \C}_{\text{\bf max}}^G   
  \sum_{k=0}^{k_{0}}   \, {t}^{-\frac{2sk_0+\gamma-2}{\gamma}} \frac{z(t)^k}{k!}\right),
  \end{align}
for
 \begin{align}\label{CG max} 
 {\bf \C}_{\text{\bf max}}^G   =  \max\left\{  m_1[f_0]\,{\bf\C}_{sk_0+\frac\gamma2}^G  \,
  \, ; \,  \C_\kmax^G  \right\}
  \end{align}

The second term from \eqref{ODI tail gen} is estimated by the Discrete Convolution Lemma~\ref{lem:convolution} set up for  $\tilde{k}=k_0$, inequality \eqref{expotail7}, 
now for any $s\le\gamma/2$, 
yields
\begin{equation}\label{ODI tail gen 2}
 2^{\frac{3\gamma}2-1}C_{\Phi}\mu_{sk_0}\sum_{k=k_0+1}^n \frac{z^k(t)}{k!}  S_{sk}[f]   \le  2^{\frac{3\gamma}2-1}C_{\Phi} \mu_{sk}\  E^n _{s} (z(t),t)\, I^n _{s,\gamma}(z(t),t),
\quad \text{for any} \    0< t \le 1. \end{equation}

\smallskip

Then, the last term from the right hand side of \eqref{ODI tail gen} is gathered by the fact that $I^{n-1}_{s,\gamma}(z(t),t)\leq I^{n}_{s,\gamma}(z(t),t)$ together with  \eqref{ODI tail gen 1} and \eqref{ODI tail gen 2} to obtain  upper bounds 
 to  the ODI for  $ E^{n}_{s}(t,z\,t)$ as follows
\begin{align}\label{ODI tail gen5}
\frac{d}{dt} E^{n}_{s}(t,z(t) ) &\le  \kappa^G_{k_0} \left(  {\A}_{s\kc} + 1  \right) 
\left(\max_{0<t\le 1} e^{z(t)} \mathcal R(\delta^{-1}_{sk_0+ \frac\gamma 2}) 
  +   {\bf \C}_{\text{\bf max}}^G   
  \sum_{k=0}^{k_{0}}   \, {t}^{-\frac{2sk_0+\gamma-2}{\gamma}} \frac{z(t)^k}{k!}\right)
 \nonumber\\
&\qquad\quad -  I^n _{s,\gamma}(z(t),t)\left( \A_{s\kc} - z'(t) - 2^{\frac{3\gamma}2-1}C_{\Phi} \mu_{sk_0}\  E^n _{s} (z(t),t) \right)
     \quad \text{for any} \ \ 0<t\le 1, 
\end{align}


Further,  the negative term proportional to the shifted moments partial sums  $I^n _{s,\gamma} (z(t),t)$,  is estimated  from below by
\begin{align}\label{K3-1}
 {I}^n_{s; \gamma}(z(t),t) &:=  \sum_{k=0}^{n} \,  \frac{z^k(t)} {k!} m_{sk+\gamma/2}(t) 
 =\frac{1}{z(t)} \left( \sum_{k=0}^{n}(k+1) \,  \frac{z^{k+1}(t)} {(k+1)!} m_{sk+\gamma/2}(t) \right)   \\
 &\ge \frac{1}{z(t)} \left( \sum_{k=1}^{n+1} \,  \frac{z^k(t)} {k!} m_{s k}(t) \right) >     \frac{1}{z (t)}  \left(E^{n} _{s} (z(t),t) - m_0[f_0]\right)\geq0, \nonumber
\end{align}
which implies, under conditions $0<t\le 1$ and  $0< s\le\gamma/2$
\begin{align}\label{ODI tail gen 6}
\frac{d}{dt} E^{n}_{s}(t,z(t) )  &< 
\kappa^G_{k_0} \left(  {\A}_{s\kc} + 1  \right) 
\left(\max_{0<t\le 1} e^{z(t)}\,\mathcal R(\delta^{-1}_{sk_0+ \frac\gamma 2}) 
  +   {\bf \C}_{\text{\bf max}}^G   
  \sum_{k=0}^{k_{0}}   \, {t}^{-\frac{2sk_0+\gamma-2}{\gamma}} \frac{z(t)^k}{k!}\right) \nonumber\\
&
\qquad\quad  -\frac{(E^{n} _{s} (z(t),t) - m_0[f_0])}{z(t)}  \left(\A_{s\kc} - z'(t) - 2^{\frac{3\gamma}2-1}C_{\Phi} \mu_{sk_0}\  E^n _{s} (z(t),t) \right). 
%
\end{align}

Since the latter term in the right side needs to be an absorption term, we require the condition on the exponential rate   $z(t)$ that maximizes is the largest possible one satisfying  the condition
\begin{align}\label{choose ko-0}
z'(t) \leq \frac{\A_{s\kc}}2   \ \ \ \ \text{for any}\   0 < t <1. 
\end{align}
In particular, after performing a time integration recalling that $z(0)=0$, the natural choice is $0\leq z(t)  := zt \leq \frac{\A_{s\kc}}2 t$, and therefore the $z$ factor  must satisfy
 \begin{align}\label{choose z-ge}
z \leq \frac{\A_{s\kc}}2
\end{align}
 
Hence,
\begin{align*}
\frac{d}{dt} E^{n}_{s}(t,z(t) )  &<\kappa^G_{k_0} \left(  {\A}_{s\kc} + 1  \right) 
\left(\max_{0<t\le 1} e^{z(t)}\,\mathcal R(\delta^{-1}_{sk_0+ \frac\gamma 2}) 
  +   {\bf \C}_{\text{\bf max}}^G   
  \sum_{k=0}^{k_{0}}   \, {t}^{-\frac{2sk_0+\gamma-2}{\gamma}} \frac{z(t)^k}{k!}\right)
 \nonumber  \\
&\qquad\qquad - \frac{(E^{n} _{s} (z(t),t) - m_0[f_0])}{z(t)} \left( \frac{\A_{s\kc}}2 - 2^{\frac{3\gamma}2-1}C_{\Phi} \mu_{sk_0}\  E^n _{s} (z(t),t)
\right).
\end{align*}

Now,  invoking \eqref{choose z-ge} for the choice $z(t)=\frac{\A_\kc}2 t$, then  for any $t\in (0,1]$, it follows 
\begin{align}\label{sum-mom-control}
  \sum_{k=0}^{k_{0}}  & \, {t}^{-\frac{1}{c_\gamma(sk)}} \frac{z(t)^k}{k!}= \sum_{k=0}^{k_{0}} t^{-\frac{2sk -\gamma + 2}{\gamma }}  t^k  \frac 1{k!}\left( \frac {\A_\kc}{2}\right)^k  = \sum_{k=0}^{k_{0}} t^{ {k} (1-\frac{2s}\gamma) +(1-\frac2\gamma)}   \frac{1}{k!} \left( \frac {\A_\kc}{2}\right)^k  <e^{\frac{\A_{s\kc}}2 } , 
\end{align}
for any $s\in (0,\gamma/2]$ and $ \gamma\in (0,2]$.

\medskip
\noindent
To this end, a reduced ODI for the partial sums of $sk^{th}$-moments arises 
\begin{align}\label{ODI tail gen 6}
\frac{d}{dt} E^{n}_{s}(t,zt )  &<  \kappa^G_{k_0} \left(  {\A}_{s\kc} + 1  \right)  e^{\frac{\A_{s\kc}}2 }
\left(\mathcal R(\delta^{-1}_{sk_0+ \frac\gamma 2}) 
  +   {\bf \C}_{\text{\bf max}}^G \right)\\
 & \qquad- \frac{(E^{n} _{s} (t,zt) - m_0[f_0])}{zt} \left( \frac{\A_{s\kc}}2 - 2^{\frac{3\gamma}2-1}C_{\Phi} \mu_{sk_0}\  E^n _{s} (z(t),t)
\right) , \ \ \text{for any }\ 0<t< 1.\nonumber
\end{align}

\medskip

Next, recall that 
\begin{equation}\label{ODI tail gen7}
m_0[f_0]  <  \int_{\mathbb{R}^d} f_0(v)\langle v\rangle^{2k_{in}} dv =: M_G, 
\quad\text{for} \ \ n\in \mathbb{N}\,.
\end{equation}

Therefore, since $m_0[f_0](0)=E^{n}_{s}(0 , z(0)) <M_G$, proceeding in an analog way that we did for the exponential propagation, we set the following sequence of times $t_n$, motivated by the search of the exponential rate $z$ small enough that enables an upper bound for $E^{n}_{s}(t,zt)$, given by 
\begin{equation}\label{timen}
t_n:=\sup\big\{ t\in(0,1]\,:\, E^{n}_{s}(\tau , z(\tau)) \leq M_G, \; \forall\, \tau\in[0,t] \big\}\,,\quad \text{for}\ \ \ n\in  \mathbb{N},
\end{equation}
where  $t_n\le 1$ is the first touching time where $E^{n}_{s}(t , z(t)) = M_G$, with 
$M_G$ is the initial moment with $k_{in}^{th}$-order Lebesgue weighted tail.   

In addition  invoking  moment generation estimates  up to $k=n$, for any time $0<t<1$, it holds that, 
\begin{equation*}
t^{k} \, m_{sk} (t) \le t^k\, \kappa^G_{k_0}  \mathcal R(\delta^{-1}_{sk_0+ \frac\gamma 2})  +{\A}_{s\kc}^{-\frac{2sk-\gamma+2}\gamma} \,  t^{k(1-\frac2\gamma s)+\frac2\gamma}\,,\quad \text{for}\quad 1\leq k \leq n\,,
\end{equation*}
so that, if $s\in(0,\frac\gamma2]$ and $\gamma\in(0,2]$, the limit as $t\searrow0$ of this product vanishes.  Hence,
\begin{align}\label{timen 2}
E^{n}_{s}(t , z(t))&:= \sum_{k=0}^{n} \,  \frac{(zt)^k} {k!} m_{sk}(t)= m_0[f_0]+\sum_{k=1}^{n} \,  \frac{(zt)^k} {k!} m_{sk}(t) \le m_0[f_0]\nonumber \\
& \qquad\qquad +  \left(e^{\frac{\A_{s\kc}}2} -1\right) 
\left(\kappa^G_{k_0}  \mathcal R(\delta^{-1}_{sk_0+ \frac\gamma 2})  +{\A}_{s\kc}^{-\frac{2sk-\gamma+2}\gamma} \right)
 \,   t^{\frac 2\gamma}\,.
\end{align}
Consequently, since $m_0<M_G$, it holds that 
\begin{equation*}
E^{n}_{s}(t , z(t)) \le M_G \quad \text{in the interval} \quad 0 < t\leq \left(\frac{{\A}_{s\kc}^{\frac{2sk-\gamma+2}\gamma}\left(M_{G}-m_0[f_0]\right)}{ \left(e^{\frac{\A_{s\kc}}2} -1\right)\, \left(\kappa^G_{k_0}  \mathcal R(\delta^{-1}_{sk_0+ \frac\gamma 2})  {\A}_{s\kc}^{\frac{2sk-\gamma+2}\gamma}+1 \right) 
  }
\right)^{\frac\gamma{2}}\ =:\ T_{k_0}\,.
\end{equation*}
This guarantees that the differential estimate \eqref{ODI tail gen 6} satisfied by $E^{n}_{s}(t , z(t))$ holds in the non-empty interval $t\in (0,t_n]$ where
\begin{equation*}
1\ge t_{n}\geq T_{k_0} \,, \quad\text{for all}\quad  n\in\mathbb{N}\,, 
\end{equation*}
or equivalently, the set \eqref{timen} is non-empty.\\

In addition, due to the fact  that the contractive factor $0<\mu_{sk_0} \searrow 0$, with  a sufficiently large $k_0:=k_0(\kc,M_G)$, renders the lower bound estimate
\begin{align*}
&\frac{\mathcal {A}_{s\kc}}{2} - 2^{\frac{3\gamma}2-1}C_{\Phi} \mu_{sk_0}\  E^n _{s} (t,zt) \geq \frac12 \mathcal {A}_{s\kc} - 2^{\frac{3\gamma}2-1}\,C_{\Phi}\, \mu_{sk_0}\,  M_G   \ge \frac{ \mathcal {A}_{s\kc}}4, 
\end{align*}
to hold for the precise, uniform in $n\in \mathbb N$, choice of 
    \begin{align}\label{choose ko-1}
  &\qquad  \qquad\qquad\qquad\qquad \mu_{s (k_{0}-1)} \leq  \frac{\mathcal {A}_{s\kc} }{2^{\frac{3\gamma+2}2} C_{\Phi}  M_G}, \qquad\text{for any}\quad  0<s \le \frac\gamma2 \quad \text{and} \quad   0< t\le t_n\le 1.
\end{align}
In particular,   {\bf Remark~\ref{k_0}} stated  that for any  the angular transition $b(\hat u\cdot\sigma) \in L^p(S^{d-1})$, \ for $1<p\le\infty$,   invoking \eqref{Povp1} or  \eqref{Povp2},  the $k_0^{th}$-moment becomes explicit 
\begin{align}\label{choose ko-22-g}
s k_0>  \mathcal{O}\left(\frac{2^{\frac{3\gamma+4}2} C_{\Phi}  M_G}{ \mathcal {A}_{s\kc}}\right)^{\frac{p}{p-1}},  \qquad \text{for} \ k_0>\kc \ \text{and}\  0<s\le 1.
\end{align}
Therefore, directly from estimate \eqref{ODI tail gen 6}, substituting $z(t)=zt$ it follows that
\begin{align} \label{ODI tail gen7-1} 
\frac{d}{dt}E^{n}_{s}(t,zt)  &<   \kappa^G_{k_0} \left(  {\A}_{s\kc} + 1  \right)  e^{\frac{\A_{s\kc}}2 }
\left(\mathcal R(\delta^{-1}_{sk_0+ \frac\gamma 2}) 
  +   {\bf \C}_{\text{\bf max}}^G \right)  -  
\frac {\mathcal {A}_{\kc}}{4zt} \big(E^{n} _{s} (t,zt) - m_0[f_0]  \big)\\
&\text{valid for}\quad 0 < t\leq t_n\leq1 \quad\text{and} \quad s\in \left(0, \frac\gamma2\right]\, .\nonumber
 \end{align}
where the coefficients of this new ODI are uniform in $n$.

\smallskip
\noindent
Hence, in order to find bound to solutions of the initial value problem to inequality \eqref{ODI tail gen7-1},  it is convenient  to set  $X^{n}_{k_0}(t,zt):= E_{s}^n(t,zt) -  m_0[f_0]\geq0$, to obtain  the  contractive ODI 
\begin{align} \label{ODI tail gen7-11} 
&\frac{d}{dt}X^n(t,zt)  <  \kappa^G_{k_0} \left(  {\A}_{s\kc} + 1  \right)  e^{\frac{\A_{s\kc}}2 }
\left(\mathcal R(\delta^{-1}_{sk_0+ \frac\gamma 2}) 
  +   {\bf \C}_{\text{\bf max}}^G \right)
 -  \frac {\mathcal {A}_{s\kc}}{4zt}\,X^{n} (t,zt)\,,\qquad 0<t\leq t_n\leq1\,, \nonumber\\
 &X^n(0,0):= \lim_{t\to 0^+}  E_{s}^n(t,zt) -  m_0[f_0]=0   \qquad \text{its inital data}.
\end{align}

A function satisfying the differential inequality \eqref{ODI tail gen7-11} is controlled by direct integration of the inequality $\bar X'+c\, t^{-1} \bar X < K$ after using the multiplicative factor $t^c$ so that $(t^c \bar X)' \leq t^c\,K$, for $t>0$.  We arrive to
\begin{equation*}
\bar X < \bar X(0) + t(1+c)^{-1} B\,,\quad t>0\,.
\end{equation*}
When this super solution is applied with  $c:= \frac{\mathcal {A}_{\kc}}{2z}$ and $K:=  \kappa^G_{k_0} \left(  {\A}_{s\kc} + 1  \right)  e^{\frac{\A_{s\kc}}2 }
\left(\mathcal R(\delta^{-1}_{sk_0+ \frac\gamma 2}) 
  +   {\bf \C}_{\text{\bf max}}^G \right)$, and initial $\bar X(0)=0$ the target function $X^n(t,zt)$ is then controlled by
\begin{align} \label{upper ODE gen1} 
X^n(t,zt) &<   \frac{4zt} {\mathcal {A}_{s\kc} +4z}  \kappa^G_{k_0} \left(  {\A}_{s\kc} + 1  \right)  e^{\frac{\A_{s\kc}}2 }
\left(\mathcal R(\delta^{-1}_{sk_0+ \frac\gamma 2}) 
  +   {\bf \C}_{\text{\bf max}}^G \right) \\
 &\ \ \ <    \frac{4zt} {\mathcal {A}_{\kc}}  \kappa^G_{k_0} \left(  {\A}_{s\kc} + 1  \right)  e^{\frac{\A_{s\kc}}2 }
\left(\mathcal R(\delta^{-1}_{sk_0+ \frac\gamma 2}) 
  +   {\bf \C}_{\text{\bf max}}^G \right)   \qquad \text{for any}\  0<t\leq t_n\le 1. \nonumber
\end{align} 
Therefore choosing $zt>0$ such that
\begin{align} \label{upper ODE gen1} 
 &\frac{4zt} {\mathcal {A}_{s\kc}}  \kappa^G_{k_0} \left(  {\A}_{s\kc} + 1  \right)  e^{\frac{\A_{s\kc}}2 }
\left(\mathcal R(\delta^{-1}_{sk_0+ \frac\gamma 2}) 
  +   {\bf \C}_{\text{\bf max}}^G \right) <   M_G-  m_0[f_0],   \nonumber\\
\text{or equivalently, }&\ \ \ \ \ \ \ \\
&{0<zt} <  \frac{(M_G-m_0[f_0]) \mathcal {A}_{s\kc}   e^{-\frac{\mathcal{A}_{s\kc}}2}  }{4 \,   \kappa^G_{k_0} \left(  {\A}_{s\kc} + 1  \right)  \left(\mathcal R(\delta^{-1}_{sk_0+ \frac\gamma 2}) 
  +   {\bf \C}_{\text{\bf max}}^G \right) },  \qquad \text{for any}\  0<t\leq t_n\le 1,\nonumber
\end{align} 
it follows  that, for any $0<t\leq t_n\le 1$, 
\begin{align*}
E_{s}^n(t,zt) -  m_0[f_0] &<  \frac{4zt} {\mathcal {A}_{s\kc}}  \kappa^G_{k_0} \left(  {\A}_{s\kc} + 1  \right)  e^{\frac{\A_{s\kc}}2 }
\left(\mathcal R(\delta^{-1}_{sk_0+ \frac\gamma 2}) 
  +   {\bf \C}_{\text{\bf max}}^G \right)  < M_G -m_0[f_0], 
\end{align*}
implying 
\begin{align}\label{ODI tail gen7-2} 
E_{s}^n(t,zt)  < M_G,  \qquad \text{for any}\  0<t\leq t_n\le 1. 
\end{align}
Since $t\leq1$ we conclude that any constant parameter $z$, uniformly in $n\in\mathbb N$,  satisfying
 \begin{align}\label{choose z gen-final}
 0<z<  \min\left\{ \frac{\A_{s\kc}}2\  ;\  \mathcal {A}_{s\kc}\frac{ \mathcal {R}(\delta_{sk_0+\frac\gamma2}) (M_G-m_0[f_0])  \,   e^{-\frac{\mathcal{A}_{s\kc}}2}  }{4   \kappa^G_{k_0} \left(  {\A}_{s\kc} + 1  \right)  
\left(1+ \mathcal R(\delta_{sk_0+ \frac\gamma 2})   {\bf \C}_{\text{\bf max}}^G \right) 
  }  \right\} ,
\end{align}
is the admissible rate,  such that $E^{n}_{s}(t,zt)<M_G$ for $0<t\leq t_n\le 1$,  for any  $s\in (0,\gamma/2]$ and the constant $  \kappa^g_{sk_0}$ defined in \eqref{kappa-g}. 

Since the rate factor $0<{\mathcal R}(\A_{s\kc}/\beta_{sk_0+\frac\gamma2})=\mathcal{R}(\delta_{sk_0+\frac\gamma2})<1$, the exponential  rate $z$ will not degenerate for any $0<\gamma\le 2$. In addition, recalling  $\mathcal{A}_{\kc} := c_{lb} 2^{\frac{\gamma}2} A_{s\kc}$, the form of the rate of decay for the exponential tail,  depending on the coercive factor $A_{s\kc}$, by replacing on  \eqref{choose z gen-final}, to obtain the rate \eqref{exp-gene-rate} in the statement of generation of exponential moments.

 Note that the inequality  \eqref{ODI tail gen7-2} for the choice of such rate $z(t)=zt$ is strict, that means there is no touching time in $(0,t_n]$, so it must be the case that $t_n=1$ due to time continuity of the moments evolution.  

\smallskip
\noindent
Sending $n\rightarrow\infty$ one concludes 
\begin{align} \label{ODI tail gen10} 
E_{s}(t,zt)   <   M_G,   \quad \ \text{for all}\quad 0\leq t\le 1,\quad s\in\left(0, \frac\gamma2\right]\,.
\end{align}

Note that the dependence of $z$ is in terms of the conserved quantities, that is, the mass and energy of $f$ and on $M^\star_G$, defined as $M^\star_G:=\sup_{t\geq0} m_{k_{in}}[f(t)]$, replacing $M_G$.  All such quantities are independently on time.    Therefore, by the time-invariance of the equation, one redoes the argument in the intervals $t\in [\kmax/2,1+\kmax/2]$, with $\kmax=1,2,\cdots$, to conclude that
\begin{align} \label{ODI tail gen10} 
E_{s}(t,zt)  \le  M^*_G,   \quad \ \text{for all}\quad t\geq0,\quad s\in\left(0, \frac\gamma2\right]\,.
\end{align}

\bigskip

The proof Theorem~\ref{thm:exp-moment-gen} is now completed.

\end{proof}

\section{Convolution inequalities and gain of integrability for the collision operator}\label{convo-ine+gain int}
After the groundbreaking work of Arkeryd \cite{Arkeryd-Linfty} in 1982 on propagation of $L^\infty$ estimates for bounded angular cross-sections, in late eighties and nineties Gustaffson \cite{Gustaff}, Wennberg \cite{WennbergLp}, Lions \cite{LionsIII}, among others, studied $L^r(\R^d)$ estimates, for $1\le r <\infty$, for the collision operator.  Lions introduced the concept of {\em gain of integrabilty}, by meaning that it is possible to obtain smoothing properties of the collision operator under certain regularity conditions of the potential kernel.  His work culminated in estimates of the form
\begin{equation*}
\int_{\mathbb{R}^{d}} Q^+(f, g)(v,t)f(v,t) dv  \leq  C(\|g\|_{L^1})  \| f\|_{L^2}^{2-\delta},\quad \text{for some} \quad \delta\in(0,1)\,.
\end{equation*}
Yet,  both constants $C(\|g\|_{L^1})$ and $\delta\in(0,1)$ where shown to exist but hard to calculate.  These results provided a venue to obtain energy estimates suitable for the study of the evolution of $L^r$-norms for solutions to the Boltzmann equation later in \cite{TV, MV04}.  Such initial results where limited to potentials satisfying certain regularity conditions in both kinetic and angular variables, including angular cut-off near $\hat u\cdot\sigma=\pm1$.

\smallskip
\nn The  new approach presented here comprises recent versions of  functional inequalities that will enable to address the study of the propagation for the $L^{r}$-norms for the Boltzmann equation, not only  under general integrable angular kernels for the first time, but also  following estimates based in previous work \cite{AC,ACG-cmp10,AG-krm11},  it is possible to obtain explicit constants important to describe  the dependence of the key parameters at play such as the angular transition function.

\smallskip
{
\nn Motivated by this goal, we revisit the Young's inequality's proof  with techniques  introduced by the authors and Carneiro in \cite{AC, ACG-cmp10},  based on radial symmetrization techniques applied to the weak form of the gain collision Boltzmann operator, not by invoking   the classical representation (\ref{bina-weak2}, \ref{weightuv2}),  but rather written in terms of   center of mass and relative velocity coordinates as described in \eqref{eq:9.3-Vu} taking the form $v'_* -v= -u^{+}:=-\tfrac{1}{2}\big(u+|u|\sigma\big)$ and  $v'-v = u^{-}:=\tfrac{1}{2}(-u+|u|\sigma)$, as follows.

\subsection{Young's type inequality for the gain collision operator} In the late eighties, it was show in \cite{Gustaff} that the gain part of the collision operator satisfies Young's convolution inequalities.  Indeed, as the reader may have noticed, the gain collision operator is a special type of nonlinear convolution between the entries and the collision kernel.  The reader can find in  \cite{AC, ACG-cmp10} the following version of such convolution inequalities.  These estimates will be important for the  propagation of $L^r$-norms, for the range $1< r \leq \infty$, to be shown as an application to the Boltzmann equation in the next section.
\begin{thm}[\bf{Young's type inequality}]\label{Young}
Let $1 \leq p,\,q,\,r \leq \infty$ with $1/p + 1/q = 1 + 1/r$.  Assume that $$B(|u|,\hat{u}\cdot\omega)=b(\hat{u}\cdot\omega)\in L^1(\mathbb{S}^{d-1})\,.$$
Then, the bilinear operator $Q^{+}$ satisfies
\begin{equation}\label{YoungIn}
\left\| Q^{+}(f,g)\right\|_{L^{r}(\R^d)} \leq C_{b}(p,q) \,
 \|f\|_{L^p(\R^d)} \, \|g\|_{L^q(\R^d)}.
\end{equation}
The constant $C_{b}(p,q)$ is given by
\begin{align}\label{YoungInC1-0}
\begin{split}
C_{b}(p,q)=&\left(\int_{\mathbb{S}^{d-1}}\left(\tfrac{1 + e_{1}\cdot\sigma}{2}\right)^{-\tfrac{d}{2r'}}b(e_{1}\cdot\sigma)d\sigma\right)^{\tfrac{r'}{p'}}\left(\int_{\mathbb{S}^{d-1}}\left(\tfrac{1 - e_{1}\cdot\sigma}{2}\right)^{-\tfrac{d}{2r'}}b(e_{1}\cdot\sigma)d\sigma\right)^{\tfrac{r'}{q'}}.
\end{split}
\end{align}
In the case $(p,q,r) = (1,1,1)$ the constant $C_{b}(1,1)$ is understood as
\begin{equation}\label{YoungInC1}
 C_{b}(1,1) = \int_{\mathbb{S}^{d-1}}b(e_{1}\cdot\sigma)d\sigma = \|b\|_{L^{1}(\mathbb{S}^{d-1})}.
\end{equation}
In the cases that $p=1$ or $q=1$ one interprets $(\cdot)^{r'/p'}=(\cdot)^{0} = 1 $ and $(\cdot)^{r'/q'}=(\cdot)^{0} = 1$ respectively.
\end{thm}
\medskip

 It is of interest to recall that Young's inequality   framework's proof  that developed by introducing  
 a weighted $\mathcal P$-operator as a function of the  the angular integration, written in center of mass and relative velocity coordinates  \eqref{eq:9.3-Vu} is defined by
\begin{align}\label{P-operator1}
\mathcal{P}({\it\Psi}, {\it\Phi})(u) &:= \int_{\mathbb{S}^{d-1}} {\it\Psi}(u^-(u,\sigma) {\it\Phi}(u^+(u,\sigma)) b(\hat{u} \cdot \sigma) d\sigma,\  \qquad u^{\pm}:=\tfrac{1}{2}\big(u\pm|u|\sigma\big)\,  
\end{align}

Indeed, the operator $\mathcal{P}$ in \eqref{P-operator1} used in the weak formulation plays the central role of angular averaging occurring in $Q^{+}(f,f)(v)$ as will be explained below, after Theorem \ref{HPop}.
From this representation, it is clear that the integrability properties of $\mathcal{P}$ are central to the understanding of the integrability properties of $Q^{+}$ itself.
 More precisely, using radial symmetrization one can readily prove the following result, 
\begin{thm}\label{HPop}
Let $1 \leq p,q,r \leq \infty$ with $1/p+1/q=1/r$. Then, 
\begin{align}\label{HPop1}
&\|\mathcal{P}_b({\it\Psi}, {\it\Phi})\|_{L^r(\mathbb{R}^{d})}\leq \tilde{C}_{b}(p,q) \|{\it\Psi}\|_{L^p(\mathbb{R}^{d})}\|\it\Phi\|_{L^q(\mathbb{R}^{d})}, \nonumber\\
\text{where}\qquad \qquad \qquad &\ \\
&\tilde{C}_{b}(p,q)=\int_{\mathbb{S}^{d-1}} \left(\tfrac{1 - e_{1}\cdot\sigma}{2}\right)^{-\frac{d}{2p'}} \left(\tfrac{1 + e_{1}\cdot\sigma}{2}\right)^{-\frac{d}{2q'}}\, b\big( e_{1}\cdot\sigma \big)d\sigma. \nonumber
\end{align}
In addition, if either functions $\it\Psi$ or $\it\Phi$ are unity, but not simultaneously, then the statement \eqref{HPop1} is modified by setting $q=\infty,$ $r=p$, since  $1=\phi\in L^\infty(\mathbb{R}^{d}\times\mathbb{S}^{d-1})$ to estimate, for \ $\mathbb{S}_{\mp}^{d-1}:= \{e_{1}\cdot\sigma \lessgtr 0 \}, $ \\
\begin{align}\label{HPop2}
\|\mathcal{P}^-_b({\it\Psi},1)\|_{L^p(\mathbb{R}^{d})}:=\|\mathcal{P}_b({\it\Psi},1)\|_{L^p(\mathbb{R}^{d})}
&\leq \tilde{C}_{b}(p,\infty) \|\it\Psi\|_{L^p(\mathbb{R}^{d})},\nonumber\\[4pt]
 \ \text{with}\ 
 \tilde{C}_{b}(p, \infty)&=\int_{\mathbb{S}_-^{d-1}} \left(\tfrac{1 -e_{1}\cdot\sigma}{2}\right)^{-\frac{d}{2p'}} b\big( e_{1}\cdot\sigma \big)d\sigma, \nonumber\\[4pt]
\text{or setting}\  p=\infty, r=q,\ \text{for} \ 1=\psi \in L^\infty(\mathbb{R}^{d}\times\mathbb{S}^{d-1}), &\ \text{to estimate}  \qquad   \qquad   \qquad &\ \\[4pt]
\|\mathcal{P}^+_b(1,{\it\Phi})\|_{L^q(\mathbb{R}^{d})}:= \|\mathcal{P}_b(1,{\it\Phi})\|_{L^q(\mathbb{R}^{d})}
&\leq \tilde{C}_{b}(\infty,q) \|\it\Phi\|_{L^q(\mathbb{R}^{d})}, \nonumber\\[4pt]
  \text{with} \ \  \tilde{C}_{b}(\infty,q)&=\int_{\mathbb{S}_+^{d-1}} \left(\tfrac{1 + e_{1}\cdot\sigma}{2}\right)^{-\frac{d}{2q'}} b\big( e_{1}\cdot\sigma \big)d\sigma. \nonumber
\end{align}
\end{thm}

All details on how to proof Therorem~\ref{HPop} can be found at \cite[Theorem 1]{AC},  as well as at \cite{ACG-cmp10} for further applications to Young's inequalities for gain operators.
Theorem \ref{HPop} is the base to prove Theorem \ref{Young}.  It also explains the peculiar expression of the constants in such theorem.  Interestingly, the constants given in Theorem \ref{HPop} are known to be sharp in the sense of functional inequalities.

However, it is important to notice that  the definition of  $\mathcal{P}_b(\psi,\phi)(u)$ in \eqref{HPop1} depends on test functions $\psi$ and $\phi$ acting on pairs $u^\pm(u,\sigma)= \frac12(|u|\sigma\pm u)$, respectively. In addition, 
$|u^\pm| = |u| \,\frac12(1\pm e_{1}\cdot\sigma)^{1/2}$. This last identity  means  that $ |u^+|=|u| \,\frac12(1 + e_{1}\cdot\sigma)^{1/2}>0 $, i.e. non-vanishing, whenever $\sigma\in \mathbb{S}_+^{d-1}:= \{e_{1}\cdot\sigma >0\}$ and, analogously,  $ |u^-|=|u| \,\frac12(1 -  e_{1}\cdot\sigma)^{1/2}>0 $ whenever $\sigma\in \mathbb{S}_-^{d-1}:= \{e_{1}\cdot\sigma <0\}.$

Thus, while the reader can observe that  the definition of the constant $C_{b}(p,q)$  in \eqref{HPop1} may contain two singularities at the points $e_{1}\cdot\sigma=\pm1$, (i.e in the north and south poles of the scattering sphere where $e_1=\hat u$ fixed), 
it becomes transparent  in the next paragraph that the  weak formulation associated to the  bilinear collisional form  $Q^+(f,g)(v)$, written in center of mass and relative velocity or scattering direction coordinates \eqref{eq:9.3-Vu},  takes a form in terms of the simpler  estimates  for  $\mathcal{P}$ from \eqref{P-operator1} described in \eqref{HPop2}, by associating the definition of the domain of spherical integration of the test  functions to whether it is evaluated on $u^+(u,\sigma)$ with $\sigma\in \mathbb{S}_+^{d-1}$; or $u^-(u,\sigma)$ with $\sigma\in \mathbb{S}_-^{d-1}$.

Indeed, the characterization of the operator $\mathcal{P}({\it\Psi}, {\it\Phi})(u):= \int_{\mathbb{S}^{d-1}} {\it\Psi}(u^+(u,\sigma) {\it\Phi}(u^-(u,\sigma)) b(\hat{u} \cdot \sigma) d\sigma,$ 
in the weak formulation of the  $Q^+(f,g)(v)$ is conveniently  written in scattering direction coordinates  \eqref{eq:9.3-Vu}, namely $v'=v+u^-(\sigma)$ and $v'_*=v-u^+(\sigma)$. 

Thus, recalling  that $\tau$ is the translation operator $\tau_v\psi(x) := \psi(x-v)$ and $\mathcal{R}$ the reflexion operator $\mathcal{R}\psi(x):=\psi(-x)$, then   
for any given test function $\psi$ set  ${\it\Psi}(u^-(u,\sigma)):=  \psi(v+u^-(u,\sigma))= (\tau_{-v} \psi)(u^-(u,\sigma))$,  and 
${\it\Phi}(u^+(u,\sigma)) :=  \psi(v-u^+(u, \sigma))= (\tau_{-v}\mathcal{R} \psi)(u^+(u,\sigma))$
and applied the estimates from Theorem~\ref{HPop} 
\begin{align}\label{P-operator}
\mathcal{P}({\it\Psi}, {\it\Phi})(u)&:= \int_{\mathbb{S}^{d-1}} {\it\Psi}(u^-(u,\sigma))\,{\it\Phi}(u^+(u,\sigma)) b(\hat{u} \cdot \sigma) d\sigma,
\end{align}

Thus,  splitting the spherical integration domain into  $\sigma\in\mathbb{S}^{d-1} = \{\hat u\cdot\sigma<0 \}\cup\{\hat u\cdot\sigma\ge0  \} =: \mathbb{S}_-^{d-1} \cup\mathbb{S}_+^{d-1}$ 
 is given by, the weak form of $Q^+(f,g)(v)$ becomes 
\begin{align}\label{P-operator2}
\int_{\mathbb{R}^d} \!Q^{+}&(g,f)(v)\psi(v) dv = \int_{\mathbb{R}^{d}} f(v)\int_{\mathbb{R}^{d}}\Phi(|u|) g(v-u)\left(\int_{\mathbb{S}_-^{d-1}}\psi(v+u^-(\sigma)) {b(\hat u \cdot\sigma)} d\sigma  \right)\!(u) du dv \nonumber\\
&+ \int_{\mathbb{R}^{d}} g(v_*)\int_{\mathbb{R}^{d}}\Phi(|u|) f(v_*+u)\left(\int_{\mathbb{S}_+^{d-1}}\psi(v-u^+(\sigma)) {b(\hat u \cdot\sigma)} d\sigma  \right)\! (u)du dv_*\nonumber\\
&=:\!\int_{\mathbb{R}^{d}}\! f(v)\int_{\mathbb{R}^{d}}\Phi(|u|) ,\tau_{-v}\mathcal{R}g(u)  \left[ \mathcal{P}^-_{b}\left((\tau_{-v}\psi)(u^-), 1\right)\!+\!   \mathcal{P}^+_{b}(1,(\tau_{-v}\mathcal{R}\psi)(u^+) \right]\!(u) du dv\nonumber\\
&=:\int_{\mathbb{R}^{d}}\! f(v)\int_{\mathbb{R}^{d}}\Phi(|u|) \, \tau_{-v}\mathcal{R}g(u)\, \left[ \mathcal{P}^-_{b}\left({\it\Psi}(u^-), 1\right)\! + \!  \mathcal{P}^+_{b}(1,{\it\Phi})(u^+) \right]\!(u) du dv, 
\end{align}
since  $\it\Psi(u^-) := (\tau_{-v}\psi)(u^-)$ and $\it\Phi(u^+) := (\tau_{-v}\mathcal{R}\psi)(u^+)$.

\medskip

The potential part of the transition scattering rates
$ \Phi(|u|)$ is the potential function defined on \eqref{pot-Phi_1} under the  assumption \eqref{pot-Phi_2.0}, where for the time being, we illustrate the technique for  the case of $\gamma=0$, which implies	 $0<c_{\Phi}\le  \Phi(|u|)\le  C_{\Phi}<\infty$ uniformly in $|u|$.  Section~8.2  will extended  the following estimate  to general potential, satisfying this conditions, for any $\gamma\in (0,2]$.

Therefore,  an estimate of \eqref{P-operator2} needs to evaluate  the first term noting that  the test function $\psi(v+u^-(\sigma))$ acts on the first component of the $\mathcal P$ operator  \eqref{P-operator1}, while the second term now the test function, not only is evaluated at a different velocity vector, $\psi(v-u^+(\sigma))$, but also acts on the second component of the $\mathcal P$.

This observation implies  the constant  $\tilde{C}_{b}(p,q)$ can be calculated  finite using the estimates \eqref{HPop2} from Theorem~\ref{HPop}.
 Therefore, since  the first term of  identity \eqref{P-operator2}  has test function $\psi(v-u^-(u,\sigma))$ acting on the first component of the $\mathcal P$  from \eqref{P-operator1} while  the test function  in the second component $\phi$ of  the $\mathcal P$ operator is taken to be  unity,   
hence, Theorem~\ref{HPop}  can be applied  with estimate 
\eqref{HPop2} by choosing  $r = p = s$ and $q=\infty$,  
to obtain, after invoking   H\"{o}lder's inequality, finite inequality 
\begin{equation*}
\int_{\mathbb{R}^{d}}(\tau_{-v} \mathcal{R}g)(u)\mathcal{P}^-_{b}({\it\Psi}, 1)(u) du 
\leq \tilde{C}_{b}(s,\infty)\|\tau_{-v}g\|_{L^{s}}\| \tau_{-v}\psi \|_{L^{s'}}=\tilde{C}_{b}(s,\infty)\|g\|_{L^{s}}\| \psi \|_{L^{s'}}\,, 
\end{equation*}
with 
\begin{equation*}
\tilde{C}_{b}(s,\infty)=\int_{\mathbb{S}_-^{d-1}} \left(\tfrac{1 - e_{1}\cdot\sigma}{2}\right)^{-\frac{d}{2s'}} b( e_{1}\cdot\sigma)d\sigma \leq 2^{\frac{d}{2s'}} \|b\|_{L^{1}(\mathbb{S}_-^{d-1})}\,.
\end{equation*}

Analogously,  the second term in  \eqref{P-operator2} has the first component of  $\mathcal P$ from  \eqref{P-operator1} in $L^\infty(\mathbb{R}^{d}\times \mathbb{S}_+^{d-1})$,  while the test function $\psi(v-u^+(\sigma))$ is acting on the second component of the $\mathcal P$ on  the  integrating domain $ \mathbb{S}_-^{d-1}$. Therefore, Theorem~\ref{HPop} is applied  with estimate \eqref{HPop2}, with the choice  $q=r=s$ and $p=\infty$,   yielding	 
\begin{align*}
\int_{\mathbb{R}^{d}}(\tau_{-v}g) (u)\mathcal{P}^+_{b}(1,{\it\Phi})(u)  du \leq \tilde{C}_{b}(\infty,s)\| \tau_{-v}g \|_{L^{s}}\| \tau_{-v}\mathcal{R}\psi \|_{L^{s'}}=\tilde{C}_{b}(\infty,s)\| g \|_{L^{s}}\| \psi \|_{L^{s'}(\mathbb{S}_+^{d-1})}\,,
\end{align*}
where
\begin{equation*}
\tilde{C}_{b}(\infty,s)=\int_{\mathbb{S}_+^{d-1}} \left(\tfrac{1 + e_{1}\cdot\sigma}{2}\right)^{-\frac{d}{2s'}} b( e_{1}\cdot\sigma)d\sigma\leq 2^{\frac{d}{2s'}} \|b\|_{L^{1}(\mathbb{S}_+^{d-1})}\,.
\end{equation*}
Combining these two last estimates, it follows that
\begin{align}\label{Qfg Lp estimates}
\int_{\mathbb{R}^d}Q^{+}_{b}(fig)(v) \psi(v) dv &\leq \tilde{C}_{b}(s,\infty)\| g\|_{L^1} \|f\|_{L^{s}}\| \psi \|_{L^{s'}} + \tilde{C}_{b}(\infty,s)\|f\|_{L^1} \| g \|_{L^{s}}\| \psi \|_{L^{s'}}\\
&\leq 2^{\frac{d}{2s}}\| \psi \|_{L^{s'}} \Big(\|b\|_{L^{1}(\mathbb{S}_-^{d-1})}\| g\|_{L^1} \|f\|_{L^{s}} + \|b\|_{L^{1}(\mathbb{S}_+^{d-1})} \|f\|_{L^1} \| g \|_{L^{s}} \Big)\,, \nonumber
\end{align}
or, by duality,
\begin{equation}\label{Qfg  Lp estimates dual}
\|Q^{+}_{b}(f,g) \|_{L^{s}(\mathbb{R}^d)} \leq 2^{\frac{d}{2s}}\Big(\|b\|_{L^{1}(\mathbb{S}_-^{d-1})}\| g\|_{L^1} \|f\|_{L^{s}} +  \|b\|_{L^{1}(\mathbb{S}_+^{d-1})}\|f\|_{L^1} \| g \|_{L^{s}} \Big)\,.
\end{equation}
%
%
In the quadratic case $f=g$  estimate \eqref{Qfg  Lp estimates dual} reduces to
\begin{equation}\label{Lp Q+ estimate}
\|Q^{+}_{b}(f,f) \|_{L^{s}(\mathbb{R}^d)} \leq 2^{\frac{d}{2s}}\|b\|_{L^{1}(\mathbb{S}^{d-1})}\| f\|_{L^1} \|f\|_{L^{s}}\,,\qquad s\in[1,\infty]\,.
\end{equation}

\smallskip

The expression \eqref{P-operator2} does not necessarily holds for gas mixtures with disparate masses or inelastic interacting systems since the definition of vectors $u^\pm:=u^\pm(u, \sigma)$ do not enjoy the symmetry $u^{\pm}(u,-\sigma)=u^{\mp}(u, \sigma)$ in such cases.  Such vectors will change depending on the masses or restitution coefficients respectively.  This, of course, makes the analysis a bit more subtle.
}

\subsection{The Carleman representation, mixing convolution, and  gain of integrability}
Carleman in his work \cite{Ca57} represents the gain part of the collisional integral with the integration in the sphere replaced by an integration along special planes.  In a sense, he introduced a trick to project the sphere in $\mathbb{R}^{d-1}$ in the context of the collision operator.  A similar trick has been implemented in the context of radiative transfer in the non-cutoff case in \cite{AS} and has the benefit of flatten the geometry of the operator at the expense of adding weights in the integration.  This representation can be quite handy in the study of fine properties of the collision operator.  A modern derivation of such representation is given in \cite{AC, ACG-cmp10,AG-krm11, GPV09}, which revised the Carleman representation in higher dimension ( as is given by\footnote{This formula correct a misprint in \cite{AG-krm11}, equation (9) and (10).}. )

\smallskip
 Starting by setting  the variables  $z:=-u^+=\frac12\left(|u|\sigma+u\right)$ and $x:=u^- = \frac12\left(|u|\sigma-u\right)$. This choices makes  
 to  obtain $-z-u= u^+-\frac12 u = u^- $, which  implies 
$|u|=|z+x|$.

In addition, since $|z|^2= \frac12|u|^2(1+ \hat u\cdot \sigma)$, then  $\cos \theta ={\hat u\cdot \hat{u'}}= 1- 2\tfrac{|z|^2}{|u|^2}= 1- 2\tfrac{|z|^2}{|z+x|^2}, $  and the $d$-dimensional spherical integration element $d\sigma =|u|^{2-\di}d(|u| \sigma +u) = |z+x|^{2-d}dx$, 
yields the transition probability rate element  in the collisional integral  is written by
\begin{equation}\label{transition rate}
 B(u,\hat u\cdot\sigma) \frac1{|u|^{\di-2}} d\sigma=  B\left( -(z+x), 1-2\frac{|z|^2}{|z+x|^2} \right)  \frac1{ |z+x|^{d-2}}dx.
\end{equation}

Thus, letting $g$ and $h$ be a pair suitable functions for which the following integral is well defined, evaluated $g(v'_*)=g(v-u^+)=g(v-x)$ and $h(v')=h(v+u^+)=h(v+z)$,  the Carleman's representation takes the form, 
 \begin{equation}\label{carle2}
 \begin{split}
 Q^+(g,h) (v) =
\ {2^{\di-1}}\!\int_{x\in\R^\di}
\frac{g(v+z)}{|x|}\int_{z\cdot x=0} h(v+x)\,
\!|z+x|^{2-\di}\,B\Big( -(z+x), 1-2\tfrac{|z|^2}{|z+x|^2} \Big)\,d\pi_z\, dx, 
\end{split}
\end{equation}
where the integral element  $d\pi_z $ is the Lebesgue measure on the hyperplane $\{z :   \hat{x}\cdot z = 0\}.$

Using the change of variables $x\to x-v$ yields
\begin{align}\label{Car}
\begin{split}
&Q^{+}(g,h)(v)=2^{d-1}\int_{\mathbb{R}^{d}}\frac{g(x)}{\left|v-x\right|}\int_{(v-x)\cdot z=0}\frac{\tau_{-x}h(z+(v-x))}{\left|z+(v-x)\right|^{d-2}}\,\tilde{a}\left(v,z,x\right) d\pi_z\,dx\\[2pt]
&\text{where} \qquad \tilde{a}\left(v,z,x\right) := B\left(|z+(v-x)|,1-2\tfrac{|z|^2}{|z+(v-x)|^2}\right).
\end{split}
\end{align}

As noticed in \cite{AC}, applying this representation to the linear form $Q^{+}(\delta_0,h)$ one obtains the $d$-dimensional weighted Radon transform 
\begin{equation}\label{CarCT}
Q^{+}(\delta_0,h)(v)=\frac{2^{d-1}}{|v|}\int_{v\cdot z=0}\frac{h(z+v)}{|z+v|^{d-2}}\;B\left(|z+v|,1-\frac{2|z|^{2}}{|z+v|^{2}}\right)d\pi_z.
\end{equation}
Combining \eqref{Car} and \eqref{CarCT} we obtain another representation for the collision operator, the double mixing convolution 
\begin{equation}\label{dmconvo}
Q^{+}(g,h)(v)=\int_{\mathbb{R}^d}g(x)\,\tau_{x}Q^{+}(\delta_0,\tau_{-x}h)(v)\,dx.
\end{equation}
Using this representation, the $L^2$-norm of $Q^{+}$ can be estimated in terms of the $L^{2}$-norm of $Q^{+}(\delta_0,\cdot)$ after invoking Minkowski's integral inequality
\begin{align}\label{mi}
\begin{split}
\big\|Q^{+}(g,h)\big\|_{L^{2}(\real^{d})}&=\left(\int_{\mathbb{R}^d}\left(\int_{\mathbb{R}^d}g(x)\tau_{x}Q^{+}(\delta_0,\tau_{-x}h)(v)dx\right)^2dv\right)^{1/2}\\
&\leq\int_{\mathbb{R}^d}\left(\int_{\mathbb{R}^d}\left(\tau_{x}Q^{+}(\delta_0,\tau_{-x}h)(v)\right)^2dv\right)^{1/2}g(x)dx\,.
\end{split}
\end{align}
Just as a minor remark, if we define the operator $\mathcal{P}_o(\cdot) = \mathcal{P}(\cdot,1)$, then $Q^{+}(\delta_0,\cdot)=\mathcal{P}^{\ast}_o(\cdot)$ the adjoint of $\mathcal{P}_o$.

\smallskip
\nn The following proposition is a slight improvement (in terms of moments) of its analog in \cite{AG-krm11}.  It shows that adding conditions on $b$, such as boundedness, one can improve the Young's inequality of Theorem \ref{Young}.  Other example of this fact can be found in \cite{LionsIII,WennbergRT,BD} that deal with Sobolev regularisation.  
\begin{prop}[\bf{Gain of integrability for angular transition $b\in L^\infty(\mathbb{S}^{d-1}) $} ]\label{lem-cancel}
Let $h \in L^{1}(\real^{d}) \cap L^{\frac{2d}{2d-1}}(\real^{d})$, with $d\geq3$.  Assume cutoff hard sphere potential $\Phi(u)=|u|\textbf{1}_{|u|\geq\varepsilon}$, with $\varepsilon\in(0,1]$, and with scattering kernel $b\in L^\infty(\mathbb{S}^{d-1}) $.  Then,
\begin{equation}\label{sp}
\big\|Q^{+}_{\Phi}(\delta_0,h)(v)\big\|_{L^{2}(\real^{d})}\leq \frac{C_{d}}{\varepsilon^{\frac{d-3}{2}}}\|b\|_{\infty}\|h\|_{L^{\frac{2d}{2d-1}}(\real^{d})}\,,
\end{equation}
for some constant $C_{d}$ depending only on the dimension.
\end{prop}
\begin{proof}
Using Radon's representation \eqref{CarCT} it follows that (for the current potential $B(u,s) = |u|\textbf{1}_{|u|\geq\varepsilon}\,b(s)$)
\begin{align*}
Q^{+}_{1}(\delta_0,h)(v)&=\frac{2^{d-1}}{|v|}\int_{z\cdot v=0}\frac{h(z+v)}{|z+v|^{d-2}}\;B\left(|z+v|,1-\frac{2|z|^{2}}{|z+v|^{2}}\right)d\pi_z\\
&=\frac{2^{d-1}}{|v|}\int_{z\cdot v=0}\frac{h(z+v)\textbf{1}_{|z+v|\geq\varepsilon}}{|z+v|^{d-3}}\;b\left(1-\frac{2|z|^{2}}{|z+v|^{2}}\right)d\pi_z\\
&\leq 2^{d-1}\left\|b\right\|_{\infty}\frac{1}{|v|}\int_{z\cdot v=0}\tilde{h}(z+v)d\pi_z\,,
\end{align*}
where $d\pi_z$ is the $\mathbb{R}^{d-1}$ Lebesgue's measure and $0\leq\tilde{h}(\cdot) = h(\cdot)\textbf{1}_{|\cdot|\geq\varepsilon}/|\cdot|^{d-3}$.  Using polar coordinates $v=r\sigma$, it is possible to estimate its $L^{2}$-norm in the following way
\begin{align*}
&(2^{d-1}\left\| b \right\|_{\infty})^{-2}\int_{\mathbb{R}^{d}}\big| Q^{+}(\delta_0,h)(v) \big|^{2}dv\\
&\hspace{-.5cm}\leq\int_{ \mathbb{S}^{d-1}} \int_{\mathbb{R}}\int_{z_1\cdot\sigma=0}\int_{z_2\cdot\sigma=0}\tilde{h}(z_1+r\sigma)\tilde{h}(z_2+r\sigma)d\pi_{z_1}d\pi_{z_2}r^{d-3}drd\sigma=:I.
\end{align*}
Perform the change of variables, for fixed $\sigma$, $x:=z_1+r\sigma$.  Note that $r=|x\cdot\sigma|$, therefore,
\begin{align*}
I=\int_{\mathbb{S}^{d-1}}\int_{\mathbb{R}^{d}}&\int_{z_2\cdot\sigma=0}\tilde{h}(x) \tilde{h}(z_2+(x\cdot\sigma)\sigma)d\pi_{z_2}|x\cdot\sigma|^{d-3}dxd\sigma\\
&\leq\int_{\mathbb{S}^{d-1}}\int_{\mathbb{R}^{d}}\int_{z_2\cdot\sigma=0}h(x) \tilde{h}(z_2+(x\cdot\sigma)\sigma)d\pi_{z_2}|\hat{x}\cdot\sigma|^{d-3}\textbf{1}_{|x|\geq\varepsilon}dxd\sigma\,.
\end{align*}
Writing
\begin{equation*}
z_2+(x\cdot\sigma)\sigma=x+(z_2+(x\cdot\sigma)\sigma-x):=x+z_3,
\end{equation*}
it easy to check that  $z_3\in\{z:z\cdot\sigma=0\}$.  Hence,
\begin{equation*}
I\leq\int_{\mathbb{S}^{d-1}}\int_{\mathbb{R}^{d}}\int_{z_3\cdot\sigma=0}h(x) \ \tilde{h}(x+z_3)d\pi_{z_3}|\hat{x}\cdot\sigma|^{d-3}\textbf{1}_{|x|\geq\varepsilon} dx d\sigma.
\end{equation*}
Next, invoking the identity
\begin{equation*}
\ir\delta_0(z\cdot y)\,\ph(z)\,dz = |y|^{-1} \int_{z\cdot y=0} \ph(z)\,d\pi_z \quad\iff\quad  
\ir\delta_0(z\cdot \hat{y})\,\ph(z)\,dz =  \int_{z\cdot \hat{y}=0} \ph(z)\,d\pi_z 
\end{equation*}
 valid for any smooth $\ph$ and $d\pi_z $ is the Lebesgue measure on the hyperplane $\{z :   \hat{y}\cdot z = 0\}$,  transforms  the  integration  in the hyperplane  
 $\{z_3\cdot \sigma=0\}$ into an integration in $\real^\di$, to get
\begin{align*}
\int_{z_3\cdot\sigma=0}\tilde{h}(x+z_3)d\pi_{z_3}&=\int_{\mathbb{R}^{d}}\delta(z\cdot\sigma)\tilde{h}(x+z)dz\\
&=\int_{\mathbb{R}^{d}}\delta(\hat{z}\cdot\sigma)\frac{\tilde{h}(x+z)}{|z|}dz.
\end{align*}

Hence,
\begin{equation*}
I\leq\int_{\mathbb{R}^d}\int_{\mathbb{R}^d}\frac{ h(x) \tilde{h}(x+z)}{|z|}\bigg(\int_{\mathbb{S}^{d-1}}\delta(\hat{z}\cdot\sigma)|\hat{x}\cdot\sigma|^{d-3}d\sigma\bigg)\textbf{1}_{|x|\geq\varepsilon} dx dz\,.
\end{equation*}
Observing that using polar coordinates with zenith $\hat{z}$ it holds that
\begin{equation*}
\int_{\mathbb{S}^{d-1}}\delta(\hat{z}\cdot\sigma)(\hat{x}\cdot\sigma)^{d-3}d\sigma\leq \int_{\mathbb{S}^{d-1}}\delta(\hat{z}\cdot\sigma)d\sigma=|\mathbb{S}^{d-2}|\,.
\end{equation*}
Moreover, in the region $\{|x|\geq\varepsilon\}$ it follows that $\tilde{h}\leq \varepsilon^{-(d-3)}h$.  Thus, one concludes that
\begin{equation*}
I\leq \varepsilon^{-(d-3)}|\mathbb{S}^{d-2}|\int_{\mathbb{R}^d}\int_{\mathbb{R}^d}\frac{ h(x) h(z)}{|z-x|}dzdx\,.
\end{equation*}
Invoking Hardy-Littlewood-Sobolev (HLS) inequality, we have for any $r\in(1,\infty)$,
\begin{align*}
I&\leq \varepsilon^{-(d-3)}|\mathbb{S}^{d-2}|\,\|h\|_{L^{r'}(\real^{d})}\|h\ast|x|^{-1}\|_{L^{r}(\real^{d})} \leq \varepsilon^{-(d-3)}C_{d,r,p}\,\|h\|_{L^{r'}(\real^{d})}\|h\|_{L^{p}(\real^{d})},
\end{align*}
where $1/p+1/d=1+1/r$ and $C_{d,r,p}$ is proportional to the HLS constant.  Taking the optimal choice $p=r'$, that is $p=\frac{2d}{2d-1}$,  yields $I\leq \frac{C_{d}}{\varepsilon^{d-3}}\|h\|^{2}_{L^{\frac{2d}{2d-1}}(\real^{d})}$, which is the desired conclusion.
\end{proof}
%
%

\nn In this way, estimate \eqref{mi} and Proposition \ref{lem-cancel} readily imply that for $b\in L^\infty(S^{d-1}) $,
\begin{equation}\label{gofi}
\big\| Q^{+}_{\Phi}(g,h) \big\|_{ L^{2}(\real^{d}) }\leq \tfrac{ C_{d} }{ \varepsilon^{\frac{d-3}{2}} }\|b\|_{\infty} \|g\|_{ L^{1}(\real^{d}) }\|h\|_{L^{\frac{2d}{2d-1}}(\real^{d})}\,\qquad \Phi(\cdot) = |\cdot|\textbf{1}_{|\cdot|\geq\varepsilon}\,,\;\; \varepsilon\in(0,1]\,.
\end{equation}
Estimate \eqref{gofi} is understood as a gain of integrability for the collision operator since the $L^{2}$-norm of $Q^{+}(g,\cdot)$ is estimated by the lower $L^{\frac{2d}{2d-1}}$-norm when considering a fixed integrable first entry $g$.  One can use Riesz--Thorin interpolation theorem to extend this estimate to the $L^{p}$-spaces when combined with the convolution estimates of Theorem \ref{Young}, see for instance \cite{MV04,AG-krm11}.  

\begin{thm}[{\it Gain of integrability}]\label{gainmaxwell}
Take functions $g, f\in (L^{1}\cap L^{p})(\real^{d})$ with $p\in(1,\infty)$.  Let $\tilde\Phi(u)=\textbf{1}_{\{|u|\geq\varepsilon\}}$, with $\varepsilon\in(0,1]$, and scattering kernel $b\in L^\infty(\mathbb{S}^{d-1})$.  Write $b=b^+ + b^-$ where
\begin{equation*}
b^+(s) = b(s)\text{1}_{\{ s\geq0 \} }\qquad \text{and} \qquad b^{-}(s)=b(s)\text{1}_{ \{ s<0 \}}\,,
\end{equation*}
are the forward and backward scattering components in $\mathbb{S}^{d-1}_{\pm}$.  Then,
\begin{align}
\big\|Q^{+}_{\tilde\Phi,b^+}(g,h)(v)\big\|_{L^{r}(\real^{d})}&\leq \frac{C}{\varepsilon^{\theta\frac{d-1}{2}}}\|b^+\|_{\infty}\|g\|_{ L^{1}(\real^{d}) }\|h\|_{L^{p}(\real^{d})}\,,\label{sp1}\\
\big\|Q^{+}_{\tilde\Phi,b^-}(g,h)(v)\big\|_{L^{r}(\real^{d})}&\leq \frac{C} {\varepsilon^{\theta\frac{d-1}{2}}}\|b^-\|_{\infty}\|h\|_{ L^{1}(\real^{d}) }\|g\|_{L^{p}(\real^{d})}\,,\qquad \theta=\tfrac{2}{\max\{r,r'\}}\,,\label{sp2}
\end{align}
for some constant $C$ depending only on the dimension $d\geq3$.  The integral exponents relate as
\begin{equation*}
r>p=
\left\{
\begin{array}{ccl}
\frac{r}{r(1-1/d)+1/d} & \text{if} & r \in (1,2] ,\\[0.2cm]
\frac{r}{2-1/d} & \text{if} & r \in [2,\infty) \,.
\end{array}\right.
\end{equation*}
As a consequence,
\begin{align}\label{sp3}
\begin{split}
\big\|Q^{+}_{\tilde\Phi,b}(g,h)(v)\big\|_{L^{r}(\real^{d})} &\leq \frac{C\,\|b\|_{\infty}}{\varepsilon^{\theta\frac{d-1}{2}}}\Big(\|g\|_{ L^{1}(\real^{d}) }\|h\|_{L^{p}(\real^{d})} + \|h\|_{ L^{1}(\real^{d}) }\|g\|_{L^{p}(\real^{d})}\Big)\,.
\end{split}
\end{align}

\end{thm}
\begin{proof}
When $r\in(1,2]$ one uses the end point (Theorem \ref{Young} for Maxwell interactions with $(p,q,r)=(1,1,1)$)
\begin{equation*}
\big\| Q^{+}_{\tilde\Phi,b}(g,h) \big\|_{ L^{1}(\real^{d}) } \leq \|b\|_{L^{1}(\mathbb{S}^{d-1})}\|g\|_{ L^{1}(\real^{d}) }\|h\|_{L^{1}(\real^{d})}\,,\qquad \tilde\Phi = \textbf{1}_{|\cdot|\geq\varepsilon}\,,\;\; \varepsilon\in(0,1]\,.
\end{equation*}
Interpolating with estimate \eqref{gofi}, for any $r\in(1,2]$, one deduces that (note that $\tilde\Phi \leq \varepsilon^{-1}\Phi$)
\begin{equation*}
\big\| Q^{+}_{\tilde\Phi,b}(g,h) \big\|_{ L^{r}(\real^{d}) } \leq \tfrac{ C }{ \varepsilon^{\theta\frac{d-1}{2}} }\|b\|^{\theta}_{L^{\infty}(\mathbb{S}^{d-1})}\|b\|^{1-\theta}_{L^{1}(\mathbb{S}^{d-1})}\|g\|_{ L^{1}(\real^{d}) }\|h\|_{L^{p}(\real^{d})} \,,\quad p=\frac{r}{r(1-1/d)+1/d}<r\,.
\end{equation*}
For the case $r\in[2,\infty)$ we use the end point (Theorem \ref{Young} for Maxwell interactions with $(p,q,r)=(\infty,1,\infty)$)
\begin{equation*}
\big\| Q^{+}_{\tilde\Phi,b^+}(g,h) \big\|_{ L^{\infty}(\real^{d}) } \leq C\|b^+\|_{L^{1}(\mathbb{S}^{d-1})}\|g\|_{ L^{1}(\real^{d}) }\|h\|_{L^{\infty}(\real^{d})}\,.
\end{equation*}
Note that $b^{+}$ is necessary to keep the constant finite.  Interpolation with \eqref{gofi} proves that
\begin{equation*}
\big\| Q^{+}_{\tilde\Phi,b^+}(g,h) \big\|_{ L^{r}(\real^{d}) } \leq \tfrac{ C }{ \varepsilon^{\theta\frac{d-1}{2}} }\|b^+\|^{\theta}_{L^{\infty}(\mathbb{S}^{d-1})}\|b^+\|^{1-\theta}_{L^{1}(\mathbb{S}^{d-1})}\|g\|_{ L^{1}(\real^{d}) }\|h\|_{L^{p}(\real^{d})} \,,\quad p=\frac{r}{2-1/d}<r\,.
\end{equation*}
Here the interpolation parameter $\theta\in(0,1]$ is given by
\begin{equation*}
\theta=
\left\{
\begin{array}{ccl}
\frac{2}{r'} & \text{if} & r \in (1,2] ,\\[0.2cm]
\frac{2}{r} & \text{if} & r \in [2,\infty) \,.
\end{array}\right.
\end{equation*}
This proves inequality \eqref{sp1} since $\|b^+\|_{L^{1}(\mathbb{S}^{d-1})}\leq |\mathbb{S}^{d-1}|\,\|b^+\|_{L^{\infty}(\mathbb{S}^{d-1})}$.  As for inequality \eqref{sp2} one performs the change $\sigma\rightarrow -\sigma$ to notice that
\begin{equation*}
Q^{+}_{\tilde\Phi,b^-}(g,h)(v) = Q^{+}_{\tilde\Phi,\mathcal{R}b^-}(h,g)(v)
\end{equation*}
where $(\mathcal{R}b^-)(s)=b^-(-s)$ is supported in $\mathbb{S}^{d-1}_+ $.  Then one can apply inequality \eqref{sp1} and use that $\|\mathcal{R}b^-\|_\infty=\|b^{-}\|_\infty$.  For the latter inequality \eqref{sp3} simply note that
\begin{align*}
\big\|Q^{+}_{\tilde\Phi,b}(g,h)(v)\big\|_{L^{r}(\real^{d})} &= \big\|Q^{+}_{\tilde\Phi,b^+}(g,h)(v)+Q^{+}_{\tilde\Phi,b^-}(g,h)(v)\big\|_{L^{r}(\real^{d})}\\
&\leq \big\|Q^{+}_{\tilde\Phi,b^+}(g,h)(v)\big\|_{L^{r}(\real^{d})}+\big\|Q^{+}_{\tilde\Phi,b^-}(g,h)(v)\big\|_{L^{r}(\real^{d})}\,,
\end{align*}
and apply inequalities \eqref{sp1} and \eqref{sp2} accordingly.
\end{proof}
%
%
%
\section{An energy method: $L^r$-propagation theory}\label{Lr-propagation}

Propagation of the $L^{r}$-norms for the Boltzmann equation has been of great interest in the community for quite some time, see for instance \cite{WennbergLp, TV, MV04}.  In \cite{MV04} a complete study was  presented addressing not only the $L^r$-propagation of solutions to the space homogeneous Boltzmann equation, for $r\in(1,\infty)$, but also the propagation of Sobolev norms and convergence toward equilibrium.

In this section we use the work done in the previous section, based on \cite{ACG-cmp10} and \cite{AG-krm11}, to address the $L^{r}$-propagation theory filling along the way some of the issues that remained open in \cite{MV04}.  We free the scattering kernel $b$ of any higher integrability assumption, the only requirement is $b\in L^{1}(\mathbb{S}^{d-1})$, and deal with the case $r=\infty$ as well.  In addition, we present a quite natural argument to find a lower bound for the negative collision operator to clearly address this issue with a minimal requirement on solutions.  For a brief comparison between cutoff and non-cutoff $L^r$-propagation theories we refer the reader to \cite{alonso-BD}.

\medskip
Throughout the section we use the $L^{p}_{\ell}$ spaces which, we recall, are defined as:

\begin{equation}\label{Lpk-norm}
L^{p}_{\ell}(\real^d)\  = \  \Big\{ \ g \; \text{measurable} \ : \  \|g\|_{L^p_\ell} \  = \| \langle \cdot \rangle^{\ell} g \|_{L^{p}} \ < \ \infty\Big\}\,, \qquad p\geq1,\quad \ell\in\mathbb{R}\,.
\end{equation}

\subsection{Estimating the Dirichlet form of the collision operator for hard potentials }
This subsection focuses on  estimates of  following from Theorem \ref{Young} and Theorem \ref{gainmaxwell} in order to find a suitable estimate for the collision Dirichlet form
\begin{equation}\label{t123D}
\int_{\mathbb{R}^d} Q^+(f,f)(v)\,f(v)^{r-1} \langle v \rangle^{r\,\ell} dv\,,\qquad \ell \geq0,\quad r\in(1,\infty)\,, 
\end{equation}
with a transition probability $B(|u|, \hat u \cdot\sigma)= \Phi(|u|) b(\hat u \cdot\sigma)$ with the potential function $\Phi$ satisfying conditions from \eqref{pot-Phi_1} and $b\in L^{1}(\mathbb{S}^{d-1})$. Such estimates are central when implementing an energy method for the Boltzmann equation.  As stated earlier, these ideas started in \cite{TV, WennbergLp} and were successfully concluded in \cite{MV04}.  They are inspired in older ideas used for the analysis of the linearized Boltzmann equation where the gain operator is broken down in a big compact (smoothing) part that satisfies a gain of integrability property and a small rough part.  Then, one uses the loss part of the collision operator to control such small rough part.  In the following argument, we use Theorem \ref{gainmaxwell} which reduces considerably the technicalities of the original arguments given in \cite{MV04}. 

The main technicality of this new  argument is to handle the general case $b\in L^{1}(\mathbb{S}^{d-1})$.  To this end, one has to make a proper decomposition of the integration domains by writting  the collisional kernel as the sum of nonnegative functions as follows
\begin{equation}\label{CK-split}
|\cdot|^{\gamma} = \Phi(\cdot) = \Phi_{l}(\cdot) + \Phi_{s}(\cdot)\,,\qquad b(\cdot) = b^{\infty}(\cdot) + b^{1}(\cdot) \in L^{1}(\mathbb{S}^{d-1})\,.
\end{equation}
That is,  the component $b^{\infty}(\mathbb{S}^{d-1}) \in L^{\infty}$ is a good approximation for $b\in L^{1}(\mathbb{S}^{d-1})$ and $\Phi_{l}(u)$ relate to large relative velocities of the potential by choosing  sufficiently large   positive parameters $R$ and $\tilde{a}$ to make 
\begin{align}\label{R-tilde-a}
\Phi_{l}(u)=\Phi(u)\,\text{1}_{\{|u| \geq 1/R \}} 
\qquad \text{and} \qquad
\| b^{1} \|_{L^{1}(\mathbb{S}^{d-1})} \leq \tilde{a}^{-1}\|b\|_{L^{1}(\mathbb{S}^{d-1})}\, ,
\end{align}
  in order to estimate
 the gain part of collision operator with a collision cross section $|u|^{\gamma} \, b(\hat u\cdot\sigma)$ as a sum of  terms
\begin{equation}\label{t123}
Q^{+}_{\Phi,b}(f,f) = Q^{+}_{\Phi_{l},b^{\infty}}(f,f) + Q^{+}_{\Phi,b^{1}}(f,f)+Q^{+}_{\Phi_{s},b^{\infty}}(f,f)\,.
\end{equation}   
Before estimating each term in the right of \eqref{t123},  recall center of mass and relative velocities coordinates introduced in \eqref{eq:9.3-Vu} to be used in the form
\begin{equation*}
u^{+}=\frac{u+|u|\sigma}2=v-v'_*\qquad \text{and}\qquad  u^{-}=\frac{u-|u|\sigma}2=v'-v\,,
\end{equation*}
as defined in \eqref{eq:9.3-Vu}.   Then the  following inequalities hold
\begin{align}\label{u-abs-estimates}
&{\bf i)} \  \text{If}\ \  \widehat{u}\cdot\sigma\geq0, \ \text{then}  \ |u| = |u^{+}|\sqrt{\tfrac{2}{1+\widehat{u}\cdot\sigma}} =|v-v'_*| \sqrt{\tfrac{2}{1+\widehat{u}\cdot\sigma}} \leq \langle v'_{*} \rangle \langle v \rangle \sqrt{\tfrac{2}{1+\widehat{u}\cdot\sigma}} \leq \sqrt{2}\,\langle v'_{*} \rangle\, \langle v \rangle\,,\nonumber\\[3pt]
&\text{and}\  \\[2pt]
&{\bf ii)}\  \text{If}\ \  \widehat{u}\cdot\sigma\leq0, \ \text{then}  \ |u| = |u^{-}|\sqrt{\tfrac{2}{1-\widehat{u}\cdot\sigma}} =|v-v'| \sqrt{\tfrac{2}{1-\widehat{u}\cdot\sigma}} \leq \langle v' \rangle \langle v \rangle \sqrt{\tfrac{2}{1-\widehat{u}\cdot\sigma}} \leq \sqrt{2}\,\langle v' \rangle\, \langle v \rangle\,,\nonumber\
\end{align} 
Thus, from \eqref{u-abs-estimates} it follows for any $r\in(1,\infty)$ and $\frac1r+\frac{1}{r'}=1$
\begin{align}\label{ugamma-estimate}
\begin{split}
|u|^{\gamma} &= |u'|^{\frac{\gamma}{r}}|u|^{\frac{\gamma}{r'}} \leq 2^{\frac{\gamma}{2r'}}\,\langle v'_{*} \rangle^{\gamma}\langle v' \rangle^{\frac{\gamma}{r}}\, \langle v \rangle^{\frac{\gamma}{r'}}\quad \text{ for the case {\bf i)}}\,,\\
|u|^{\gamma} &= |u'|^{\frac{\gamma}{r}}|u|^{\frac{\gamma}{r'}} \leq 2^{\frac{\gamma}{2r'}}\,\langle v' \rangle^{\gamma}\langle v'_* \rangle^{\frac{\gamma}{r}}\, \langle v \rangle^{\frac{\gamma}{r'}}\quad \text{ for the case {\bf ii)}}\,.
\end{split}
\end{align}
Further, write $b^\infty = b^\infty_{+} + b^\infty_{-}$ the forward and backward scattering components supported in $\mathbb{S}^{d-1}_{\pm}$ respectively, then from \eqref{ugamma-estimate} it holds that
\begin{align}\label{t0}
Q^{+}_{\Phi_{l},b^{\infty}_+}(f,f)(v)\leq 2^{\frac{\gamma}{2r'}}\,\langle v \rangle^{\frac{\gamma}{r'}}\,Q^{+}_{\tilde\Phi,b^{\infty}_+}(\langle \cdot \rangle^{\gamma}\,f,\langle \cdot \rangle^{\frac{\gamma}{r}}\,f)(v) \nonumber\\
Q^{+}_{\Phi_{l},b^{\infty}_-}(f,f)(v)\leq 2^{\frac{\gamma}{2r'}}\,\langle v \rangle^{\frac{\gamma}{r'}}\, Q^{+}_{\tilde\Phi,b^{\infty}_-}(\langle \cdot \rangle^{\frac{\gamma}{r}}\,f,\langle \cdot \rangle^{\gamma}\,f)(v)
\end{align}
Thus, the corresponding Dirichlet form for the first term of \eqref{t123} is estimated as follows:  using Theorem \ref{gainmaxwell} inequality \eqref{sp1} and standard Lebesgue's interpolation one has
\begin{equation}\label{t1}
\begin{aligned}
\int_{\mathbb{R}^{d}}Q^{+}_{\Phi_{l},b^{\infty}_+}(f,f)(v)&\,f(v)^{r-1}  dv 
\leq 2^{\frac{\gamma}{2r'}}\|Q^{+}_{\tilde{\Phi},b^{\infty}_+}\big( \langle \cdot \rangle^{\gamma}\,f,\langle\cdot\rangle^{\frac{\gamma}{r} }f \big)\|_{L^{r}(\mathbb{R}^{d})}\| f \|^{r-1}_{L^{r}_{\frac{\gamma}{r} }(\mathbb{R}^{d})}\\
&\leq C_{R,b^{\infty}_+}\| f \|_{L^{1}_{\gamma}(\mathbb{R}^{d})}\| f \|_{L^{p}_{\frac{\gamma}{r}}(\mathbb{R}^{d})}\| f \|^{r-1}_{L^{r}_{{\frac{\gamma}{r} }}(\mathbb{R}^{d})}\\
&\leq C_{R,b^{\infty}_+}\| f \|_{L^{1}_{\gamma}(\mathbb{R}^{d})}\| f \|^{1-\theta_r}_{L^{1}_{\frac{\gamma}{r} }(\mathbb{R}^{d})}\| f \|^{r+\theta_r - 1}_{L^{r}_{\frac{\gamma}{r} }(\mathbb{R}^{d})}\,,\qquad \theta_r\in(0,1)\,.
\end{aligned}
\end{equation}
Similarly, Theorem \ref{gainmaxwell} inequality \eqref{sp2} renders 
\begin{equation}\label{t2}
\begin{aligned}
\int_{\mathbb{R}^{d}}Q^{+}_{\Phi_{l},b^{\infty}_-}(f,f)(v)&\,f(v)^{r-1}  dv 
\leq 2^{\frac{\gamma}{2r'}}\|Q^{+}_{\tilde{\Phi},b^{\infty}_-}\big( \langle \cdot \rangle^{\frac{\gamma}{r}}\,f,\langle\cdot\rangle^{\gamma }f \big)\|_{L^{r}(\mathbb{R}^{d})}\| f \|^{r-1}_{L^{r}_{\frac{\gamma}{r} }(\mathbb{R}^{d})}\\
&\leq C_{R,b^{\infty}_-}\| f \|_{L^{1}_{\gamma}(\mathbb{R}^{d})}\| f \|_{L^{p}_{\frac{\gamma}{r}}(\mathbb{R}^{d})}\| f \|^{r-1}_{L^{r}_{{\frac{\gamma}{r} }}(\mathbb{R}^{d})}\\
&\leq C_{R,b^{\infty}_-}\| f \|_{L^{1}_{\gamma}(\mathbb{R}^{d})}\| f \|^{1-\theta_r}_{L^{1}_{\frac{\gamma}{r} }(\mathbb{R}^{d})}\| f \|^{r+\theta_r - 1}_{L^{r}_{\frac{\gamma}{r} }(\mathbb{R}^{d})}\,,\qquad \theta_r\in(0,1)\,.
\end{aligned}
\end{equation}
Here the constants $C_{R,b^{\infty}_\pm} = C\, R^{ \frac{d-1}{\max\{r,r'\}} }\;\| b^{\infty}_\pm \|_{ L^{\infty} }$ blows up as $R$ and $\tilde{a}$ increase.   Writing
\begin{equation*}
Q^{+}_{\Phi_{l},b^{\infty}}(f,f)(v) = Q^{+}_{\Phi_{l},b^{\infty}_+}(f,f)(v) + Q^{+}_{\Phi_{l},b^{\infty}_-}(f,f)(v)\,,
\end{equation*}
one uses estimates \eqref{t1} and \eqref{t2} to conclude that
\begin{equation}\label{t3}
\int_{\mathbb{R}^{d}}Q^{+}_{\Phi_{l},b^{\infty}}(f,f)(v)\,f(v)^{r-1}  dv \leq C_{R,b^{\infty}}\| f \|_{L^{1}_{\gamma}(\mathbb{R}^{d})}\| f \|^{1-\theta_r}_{L^{1}_{\frac{\gamma}{r} }(\mathbb{R}^{d})}\| f \|^{r+\theta_r - 1}_{L^{r}_{\frac{\gamma}{r} }(\mathbb{R}^{d})}\,,\qquad \theta_r\in(0,1)\,.
\end{equation}
The second term in \eqref{t123} is estimated in a similar manner using Theorem \ref{Young} instead of Theorem \ref{gainmaxwell}.  Writing, as before, $b^1 = b^1_+ + b^1_{-}$ the forward and backward scattering components supported in $\mathbb{S}^{d-1}_\pm$ respectively, it follows that
\begin{equation}\label{t4}
\begin{aligned}
\int_{\mathbb{R}^{d}}Q^{+}_{\Phi,b^{1}_+}(f,f)(v)f(v)^{r-1}dv &\leq 2^{\frac{\gamma}{2r'}}\| Q^{+}_{b^{1}_+}\big(\langle \cdot \rangle^{\gamma} f,\langle\cdot\rangle^{\frac{\gamma}{r}}f\big)\|_{L^{r}(\mathbb{R}^{d})}\| f \|^{r-1}_{L^{r}_{\frac{\gamma}{r}}(\mathbb{R}^{d})} \\
&\leq C\|b^{1}_+\|_{L^{1}(\mathbb{S}^{d-1})} \,\|  f \|_{L^{1}_{\gamma}(\mathbb{R}^{d})}\| f \|^{r}_{L^{r}_{\frac{\gamma}{r}}(\mathbb{R}^{d})}\\
&\leq C\,\tilde{a}^{-1}\|b\|_{L^{1}(\mathbb{S}^{d-1})}\| f \|_{L^{1}_{\gamma}(\mathbb{R}^{d})}\| f \|^{r}_{L^{r}_{\frac{\gamma}{r}}(\mathbb{R}^{d})}\,.
\end{aligned}
\end{equation}
Similarly,
\begin{equation}\label{t5}
\begin{aligned}
\int_{\mathbb{R}^{d}}Q^{+}_{\Phi,b^{1}_-}(f,f)(v)f(v)^{r-1}dv &\leq 2^{\frac{\gamma}{2r'}}\| Q^{+}_{b^{1}_-}\big(\langle \cdot \rangle^{\frac{\gamma}{r}} f,\langle\cdot\rangle^{\gamma}f\big)\|_{L^{r}(\mathbb{R}^{d})}\| f \|^{r-1}_{L^{r}_{\frac{\gamma}{r}}(\mathbb{R}^{d})} \\
&\leq C\|b^{1}_-\|_{L^{1}(\mathbb{S}^{d-1})} \,\|  f \|_{L^{1}_{\gamma}(\mathbb{R}^{d})}\| f \|^{r}_{L^{r}_{\frac{\gamma}{r}}(\mathbb{R}^{d})}\\
&\leq C\,\tilde{a}^{-1}\|b\|_{L^{1}(\mathbb{S}^{d-1})}\| f \|_{L^{1}_{\gamma}(\mathbb{R}^{d})}\| f \|^{r}_{L^{r}_{\frac{\gamma}{r}}(\mathbb{R}^{d})}\,.
\end{aligned}
\end{equation}
Consequently, from inequalities \eqref{t4} and \eqref{t5} and since
\begin{equation*}
Q^{+}_{\Phi,b^{1}}(f,f)(v)=Q^{+}_{\Phi,b^{1}_+}(f,f)(v)+Q^{+}_{\Phi,b^{1}_-}(f,f)(v)\,,
\end{equation*}
it holds that
\begin{equation}\label{t6}
\int_{\mathbb{R}^{d}}Q^{+}_{\Phi,b^{1}}(f,f)(v)f(v)^{r-1}dv \leq C\,\tilde{a}^{-1}\|b\|_{L^{1}(\mathbb{S}^{d-1})}\| f \|_{L^{1}_{\gamma}(\mathbb{R}^{d})}\| f \|^{r}_{L^{r}_{\frac{\gamma}{r}}(\mathbb{R}^{d})}\,.
\end{equation}
For the latter term in \eqref{t123} simply note, since $\Phi_{s}(u)=|u|^{\gamma}\text{1}_{\{ |u|<1/R \} }\leq R^{-\gamma}$, that
\begin{equation*}
Q^{+}_{\Phi_{s},b^{\infty}}(f,f)\leq R^{-\gamma}Q^{+}_{b^{\infty}}(f,f)\,.
\end{equation*}
Therefore, by exploiting the quadratic nature of the estimate, Theorem \ref{Young} used in $b^{\infty}_{\pm}$ readily gives
\begin{equation}\label{t7}
\int_{\mathbb{R}^{d}}Q^{+}_{\Phi_{s},b^{\infty}}(f,f)(v) f(v)^{r-1}dv \leq C\,R^{-\gamma}\,\|b\|_{L^{1}(\mathbb{S}^{d-1})}\,\|f\|_{L^{1}(\mathbb{R}^{d})}\|f\|^{r}_{L^{r}(\mathbb{R}^{d})}\,.
\end{equation}
Gathering estimates \eqref{t3}, \eqref{t6}, and \eqref{t7}, one has
\begin{align*}
\int_{\mathbb{R}^{d}}Q^{+}(f,f)(v) f(v)^{r-1}dv \leq C\| f \|_{L^{1}_{\gamma}(\mathbb{R}^{d})}\Big( \big( \tilde{a}^{-1} + R^{-\gamma}\big)&\|b\|_{L^{1}(\mathbb{S}^{d-1})}\| f \|^{r}_{L^{r}_{\frac{\gamma}{r}}(\mathbb{R}^{d})} \\
& + C_{R,b^{\infty}}\| f \|^{1-\theta_r}_{L^{1}_{\frac{\gamma}{r} }(\mathbb{R}^{d})}\| f \|^{r+\theta_r - 1}_{L^{r}_{\frac{\gamma}{r} }(\mathbb{R}^{d})}\Big)\,.
\end{align*}
Adding the weight $\langle\cdot\rangle^{r\,\ell}$ in the Dirichlet form is not a problem after using the inequality $$\langle v \rangle^{r\,\ell}=\langle v \rangle^{\ell}\langle v \rangle^{(r-1)\,\ell}\leq\langle v'_{*} \rangle^{\ell}\langle v' \rangle^{\ell}\langle v \rangle^{(r-1)\,\ell}\,.$$ 
Consequently,
\begin{align*}
\int_{\mathbb{R}^{d}}&Q^{+}(f,f)(v) f(v)^{r-1}\langle v \rangle^{r\ell}dv \leq \int_{\mathbb{R}^{d}}Q^{+}\big(\langle\cdot\rangle^{\ell}f,\langle\cdot\rangle^{\ell}f\big)(v) \big( f(v)\langle v \rangle^{\ell} \big)^{r-1}dv \\
&\leq C\| f \|_{L^{1}_{\gamma+\ell}(\mathbb{R}^{d})}\Big( \big( \tilde{a}^{-1} + R^{-\gamma}\big)\|b\|_{L^{1}(\mathbb{S}^{d-1})}\| f \|^{r}_{L^{r}_{\frac{\gamma}{r}+\ell}(\mathbb{R}^{d})} + C_{R,b^{\infty}}\| f \|^{1-\theta_r}_{L^{1}_{\frac{\gamma}{r} +\ell }(\mathbb{R}^{d})}\| f \|^{r+\theta_r - 1}_{L^{r}_{\frac{\gamma}{r}  +\ell}(\mathbb{R}^{d})}\Big)\,.
\end{align*}
Furthermore, if $f$ satisfies the condition of Lemma \ref{lblemma} one has that
\begin{align*}
\int_{\mathbb{R}^{d}}Q^{-}(f,f)(v)\langle v \rangle^{r\,\ell}f(v)^{r-1}dv &= \|b\|_{L^{1}(\mathbb{S}^{d-1})}\int_{\mathbb{R}^{d}}\big(f(v)\langle v \rangle^{k}\big)^{r} \big(f\ast|\cdot|^{\gamma}\big)(v)dv \\
&\geq c_{lb}\|b\|_{L^{1}(\mathbb{S}^{d-1})}\int_{\real^{d}}\big(f(v)\langle v \rangle^{\frac{\gamma}{r} + \ell}\big)^{r}dv = c_{lb}\|b\|_{L^{1}(\mathbb{S}^{d-1})}\| f \|^{r}_{L^{r}_{\frac{\gamma}{2} + \ell}(\real^{d})}\,.
\end{align*}
As a consequence, adding these latter two estimates we obtain, for any $R>0$, $\tilde{a}>0$ and $r\in(1,\infty)$, that
\begin{align}\label{qcontrol}
\begin{split}
\int_{\mathbb{R}^{d}} Q(f,f)(v) \langle v \rangle^{r\,\ell}f(v)^{r-1}dv  &\leq C\,\| f \|_{L^{1}_{\gamma+\ell}(\mathbb{R}^{d})}\Big( \big( \tilde{a}^{-1} + R^{-\gamma}\big)\|b\|_{L^{1}(\mathbb{S}^{d-1})}\| f \|^{r}_{L^{r}_{\frac{\gamma}{r}+\ell}(\mathbb{R}^{d})} \\
&\qquad + C_{R,b^\infty}\| f \|^{1-\theta_r}_{L^{1}_{\frac{\gamma}{r} +\ell }(\mathbb{R}^{d})}\| f \|^{r+\theta_r - 1}_{L^{r}_{\frac{\gamma}{r}  +\ell}(\mathbb{R}^{d})}\Big) - c_{lb}\|b\|_{L^{1}(\mathbb{S}^{d-1})}\| f \|^{r}_{L^{r}_{\frac{\gamma}{2} + \ell}(\real^{d})}\,.
\end{split}
\end{align}
Before entering in the specifics of the statement of the following theorem, based on previous discussion, we mention the constant
\begin{equation}\label{C_ell_alpha}
C_{\ell,\alpha}:=\sup_{t\geq0}\| f(t) \|_{L^{1}_{\alpha+\ell}(\mathbb{R}^{d})}(t) =  \sup_{t\geq0} m_{\frac{\alpha+\ell}2}(t) \le\max\Big\{ m_{\frac{\alpha+\ell}2}[f _0],\,     \mathbb{E}_{\frac{\alpha+\ell}2}\Big\}<\infty\,,\quad\alpha\ge0\,,\quad \ell\geq 0\,,
\end{equation}
with the  $(\alpha+\ell)/2$-moments  globally controlled  by estimates \eqref{RateE} and \eqref{mom-prop}  developed in Theorem~\ref{propagation-generation}.
\begin{thm}[Hard potentials]\label{Dirichlet-r}
Let $b\in L^{1}(\mathbb{S}^{d-1})$, $\gamma\in(0,1]$, $\ell\geq0$, and $r\in(1,\infty)$.   Assume that $f$ satisfies the hypothesis of the Lower Bound Lemma \ref{lblemma} and $m_{\frac{\gamma+\ell}2}[f _0]<\infty$.  Then, there exists a constant $\bar C_\gamma:=\bar C_\gamma(\| b\|_1 ,C_{\ell,\gamma})$, estimated below in \eqref{C1-def0} and \eqref{C1gamma-def}, such that
\begin{align}\label{Dirichlet-form}
\int_{\mathbb{R}^{d}} Q(f,f)(v) \langle v \rangle^{r\,\ell}f(v)^{r-1}dv  &\leq \bar C_\gamma - \tfrac{c_{lb}}{2}\|b\|_{L^{1}(\mathbb{S}^{d-1})}\| f \|^{r}_{L^{r}_{\frac{\gamma}{2} + \ell}(\real^{d})} \,, \qquad r>1,
\end{align}
with $C_{\ell, \gamma}$ defined in \eqref{C_ell_alpha} controlling the propagation of moments estimates \eqref{RateE}and \eqref{mom-prop} developed in  Theorem~\ref{propagation-generation} and $c_{lb}$ was given in Lemma \ref{lblemma}.
\end{thm}
\begin{proof}
Using the hypothesis, estimate \eqref{qcontrol} reduces to
\begin{align*}
\int_{\mathbb{R}^{d}} Q(f,f)(v) \langle v \rangle^{r\,\ell}f(v)^{r-1}dv  &\leq C_{\ell,\gamma}\Big( \big( \tilde{a}^{-1} + R^{-\gamma}\big)\|b\|_{L^{1}(\mathbb{S}^{d-1})}\| f \|^{r}_{L^{r}_{\frac{\gamma}{r}+\ell}(\mathbb{R}^{d})} \\
&\qquad + C_{R,b^\infty}\,C^{1-\theta_r}_{\ell,\gamma}\| f \|^{r+\theta_r - 1}_{L^{r}_{\frac{\gamma}{r}  +\ell}(\mathbb{R}^{d})}\Big) - c_{lb}\|b\|_{L^{1}(\mathbb{S}^{d-1})}\| f \|^{r}_{L^{r}_{\frac{\gamma}{2} + \ell}(\real^{d})}\\
& \leq C_{R,b^\infty}\,C^{1-\theta_r}_{\ell,\gamma}\| f \|^{r+\theta_r - 1}_{L^{r}_{\frac{\gamma}{r}  +\ell}(\mathbb{R}^{d})} - \tfrac{3c_{lb}}{4}\|b\|_{L^{1}(\mathbb{S}^{d-1})}\| f \|^{r}_{L^{r}_{\frac{\gamma}{2} + \ell}(\real^{d})}\,,
\end{align*}
where, in the last inequality,  the parameters  $R$ and $\tilde{a}$ are chosen sufficiently large so that $C_{\ell,\gamma}\big( \tilde{a}^{-1} + R^{-\gamma}\big)\leq \frac{c_{lb}}{4}$.  For instance, we can choose \begin{equation}\label{C1-def0}
\tilde{a}=\frac{8\,C_{\ell,\gamma}}{c_{lb}} \,,\quad \text{and then} \quad R=\tilde{a}^{1/\gamma} = \Big(\frac{8\,C_{\ell,\gamma}}{c_{lb}}\Big)^{1/\gamma}\,.
\end{equation}
Furthermore, since $\theta_{r}\in(0,1)$, we can use Young's inequality  controlling 
$$C_{R,b^\infty}\,C^{2-\theta_r}_{\ell,\gamma}\| f \|^{r+\theta_r - 1}_{L^{r}_{\frac{\gamma}{r}  +\ell}}\le \Big(\frac{4}{c_{lb}\| b \|_{L^{1}(\mathbb{S}^{d-1})}}\Big)^{\frac{r+\theta_{r}-1}{1-\theta_{r}}}\big(C_{R,b^\infty}\,C^{2-\theta_r}_{\ell,\gamma}\big)^{\frac{r}{1-\theta_{r}}}  +   \tfrac{c_{lb}}{4}\|b\|_{L^{1}(\mathbb{S}^{d-1})} \| f \|^{r}_{L^{r}_{\frac{\gamma}{2} + \ell}},$$
to obtain that
\begin{multline*}
C_{R,b^\infty}\,C^{2-\theta_r}_{\ell,\gamma}\| f \|^{r+\theta_r - 1}_{L^{r}_{\frac{\gamma}{r}  +\ell}(\mathbb{R}^{d})} - \tfrac{3c_{lb}}{4}\|b\|_{L^{1}(\mathbb{S}^{d-1})}\| f \|^{r}_{L^{r}_{\frac{\gamma}{2} + \ell}(\real^{d})} \\ \leq \Big(\frac{4}{c_{lb}\| b \|_{L^{1}(\mathbb{S}^{d-1})}}\Big)^{\frac{r+\theta_{r}-1}{1-\theta_{r}}}\big(C_{R,b^\infty}\,C^{2-\theta_r}_{\ell,\gamma}\big)^{\frac{r}{1-\theta_{r}}} - \tfrac{c_{lb}}{4}\|b\|_{L^{1}(\mathbb{S}^{d-1})}\| f \|^{r}_{L^{r}_{\frac{\gamma}{2} + \ell}(\real^{d})}\,.
\end{multline*}
Recalling that $C_{R,b^\infty} = C\, R^{ \frac{d-1}{\max\{r,r'\}} }\;\| b^{\infty} \|_{ L^{\infty} }$, we can rewrite the first term in the right hand side as
\begin{equation}\label{C1gamma-def}
\bar C_{\gamma} = \bar C_{\gamma}(b, r , C_{\ell,\gamma}) := \| b \|^{1-\frac{r}{1-\theta_{r}}}_{L^{1}(\mathbb{S}^{d-1})}\| b^{\infty} \|^{\frac{r}{1-\theta_{r}}}_{ L^{\infty} }\Big(\frac{4}{c_{lb} }\Big)^{ \frac{r+\theta_{r}-1}{1-\theta_{r}}}\Big(C\,R^{ \frac{d-1}{\max\{r,r'\}} }\,C^{2-\theta_r}_{\ell,\gamma}\Big)^{\frac{r}{1-\theta_{r}}}\,,
\end{equation}
which proves the result.
\end{proof}
\subsubsection{The Maxwell molecules case}  Note that in order to control $Q^{+}_{\Phi_{s},b^{\infty}}$, the small relative velocities, the kinetic potential $\Phi(u)$ needs to vanish as $|u|\rightarrow0$.  This is not the case for the Maxwell Molecules case.  Indeed, Maxwell molecules is a critical case for uniform propagation of $L^{r}$-integrability for general initial data in $L^{r}$ in the sense that as soon as $\gamma<0$, i.e. soft potential case, the uniform propagation of moments is lost for general initial data as discussed in \cite{Carlen} let alone the $L^{r}_{\ell}$-norms \footnote{For the non-cutoff Boltzmann problem this is not the case and the analysis is quite different.}.  In order to compensate for this criticality issue, we will use the additional compactness gained through propagation of entropy after noticing that we implicitly have the additional requirement on the initial data
\begin{equation*}
\int_{\mathbb{R}^{d}}f_0(v)\ln\big(f_0(v)\big)dv<\infty\,.
\end{equation*}
In our context $f_0\in \big(L^{1}_{2}\cap L^{r}\big)(\mathbb{R}^{d})$ for $r>1$.  As a consequence, the initial entropy is finite since for any $1<r<\frac{d}{d-1}\leq 2$
\begin{align*}
\bigg|\int_{\mathbb{R}^{d}} &f_0(v)\ln\big(f_0(v)\big)dv \bigg|  \leq \int_{\{f_0\leq1\}} \big|f_0(v)\ln\big(f_0(v)\big)\big|dv  +  \int_{\{f_0\geq1\}} |f_0(v)\ln\big(f_0(v)\big)| dv\\
& \leq C_{r}\int_{\{f_0\leq1\}} \big|f_0(v)\big|^{\frac{1}{r}} dv + \int_{\{f_0\geq1\}} |f_0(v)|^{r} dv\\
& \leq C_{r,d}\bigg(\int_{\{f_0\leq1\}} f_0(v)\langle v \rangle^{r} dv\bigg)^{\frac{1}{r}} + \int_{\{f_0\geq1\}} |f_0(v)|^{r} dv =: C\Big(d,\|f_0\|_{(L^{1}_{2}\cap L^{r})(\mathbb{R}^{d})}\Big)<+\infty.
\end{align*}
In the second estimate we used that $|x\ln(x)| \leq C_{p}\,x^{\frac{1}{r}}$, for any $r>1$ and $x\in[0,1]$.  Furthermore, it is well known that using energy conservation and entropy propagation it follows for the solution of the homogeneous Boltzmann equation $f(v,t)$, see \cite[page 329]{DL} or more recently \cite[Lemma A.1]{ALods}, that
\begin{equation}\label{entropycontrol}
\sup_{t\geq0}\int_{\mathbb{R}^{d}}f(v,t)\big|\ln\big(f(v,t)\big)\big| dv \leq C\bigg(\int f_0\ln(f_0)dv,\int f_0| v |^{2}dv\bigg)=:C_{\mathcal{M}}\,.
\end{equation}
\begin{thm}[Maxwell Molecules]\label{DirichletMM}
Let $b\in L^{1}(\mathbb{S}^{d-1})$.  Fix $\ell\geq0$ and assume that
\begin{equation*}
\|  f \|_{L^{1}_{\ell}(\mathbb{R}^{d})} = C_{\ell,0} < \infty\,, \qquad \int_{\mathbb{R}^{d}}f(v)\big|\ln\big(f(v)\big)\big| dv = C_{\mathcal{M}} < \infty\,.
\end{equation*}
Then, there exists a constant $\bar C_{0}:=\bar C_{0}(\| b\|_1,C_{\ell,0},C_{\mathcal{M}})$, estimated below in \eqref{CMax-def0},\eqref{CMax-def1} and \eqref{CMax-def2}, such that
\begin{align}\label{Dirichlet-form}
\int_{\mathbb{R}^{d}} Q_0(f,f)(v) \langle v \rangle^{r\,\ell}f(v)^{r-1}dv  &\leq \bar{C}_{0} - \tfrac{1}{2}\|b\|_{L^1(\mathbb{S}^{d-1})}\| f \|^{r}_{L^{r}_{\ell}(\real^{d})}\,,\quad \forall\,r\in[1,\infty)\,.
\end{align}
\end{thm}
\begin{proof}
Use the decomposition $b(s)=b(s)\text{1}_{\{|s|<1-\varepsilon\}} + b(s)\text{1}_{\{|s|\geq1-\varepsilon\}}=:b_{1}^{\varepsilon}+b_{2}^{\varepsilon}$.  Since $\| b_{2}^{\varepsilon}\|_{L^{1}(\mathbb{S}^{d-1})}$ decreases monotonically to zero with $\varepsilon\in(0,1)$, we can choose $\varepsilon:=\varepsilon(R)$ sufficiently small such that
\begin{equation}\label{CMax-def0}
\| b_{2}^{\varepsilon}\|_{L^{1}(\mathbb{S}^{d-1})}\leq R^{-1}\| b \|_{L^{1}(\mathbb{S}^{d-1})}\,,\qquad R>0\,.
\end{equation}
Write $Q^{+}_{0,b}(f,f)= Q^{+}_{0,b^{\varepsilon}_{1}}(f,f) + Q^{+}_{0,b^{\varepsilon}_{2}}(f,f)$.

\smallskip
\nn On the one hand invoking Theorem \ref{Young}, used in $b^{\varepsilon}_{2,\pm}$ the upper and lower scattering of $b^{\varepsilon}_{2}$ exploiting the quadratic nature of the estimate, and the fact that $\langle v \rangle\leq\langle v'_* \rangle\langle v' \rangle $ we find that
\begin{align*}
\| Q^{+}_{0,b^{\varepsilon}_{2}}(f,f) \|_{L^{r}_{k}(\mathbb{R}^{d})} &\leq \| Q^{+}_{0,b^{\varepsilon}_{2}}(\langle \cdot \rangle^{\ell}f,\langle \cdot \rangle^{\ell}f) \|_{L^{r}(\mathbb{R}^{d})}  \leq  2^{\frac{d}{2r'}+1}\| b^{\varepsilon}_{2} \|_{L^{1}(\mathbb{S}^{d-1})} \| f \|_{L^{1}_{\ell}(\mathbb{R}^{d})}\| f \|_{L^{r}_{\ell}(\mathbb{R}^{d})}\\
&\leq 2^{\frac{d}{2r'}+1}R^{-1}\| b \|_{L^{1}(\mathbb{S}^{d-1})} \| f \|_{L^{1}_{\ell}(\mathbb{R}^{d})}\| f \|_{L^{r}_{\ell}(\mathbb{R}^{d})}\,.
\end{align*}
 On the other hand, the inequality $\langle v \rangle \leq \langle v'_* \rangle + \langle v' \rangle$ lead us to
\begin{equation*}
Q^{+}_{0,b^{\varepsilon}_{1}}(f,f)\,\langle v \rangle^{\ell} \leq 2^{\ell-1}Q^{+}_{0,b^{\varepsilon}_{1}}(f,\langle \cdot \rangle^\ell f) + 2^{k-1}Q^{+}_{0,b^{\varepsilon}_{1}}(\langle\cdot\rangle^\ell f,f)\,.
\end{equation*} 
Let us start with the first term in the right.  Note that for $K>1$
\begin{align*}
Q^{+}_{0,b^{\varepsilon}_{1}}(f,\langle \cdot \rangle^{\ell} f) &= Q^{+}_{0,b^{\varepsilon}_{1}}(f\,\text{1}_{\{f\leq K\}},\langle \cdot \rangle^{\ell}f) + Q^{+}_{0,b^{\varepsilon}_{1}}(f\,\text{1}_{\{f > K\}},\langle \cdot \rangle^{\ell}f)\\
&\leq Q^{+}_{0,b^{\varepsilon}_{1}}(f\,\text{1}_{\{f\leq K\}},\langle \cdot \rangle^{\ell}f) + \ln(K)^{-1}\,Q^{+}_{0,b^{\varepsilon}_{1}}(f\,\ln[f],\langle \cdot \rangle^{\ell}f)\,.
\end{align*}
Using again Theorem \ref{Young}, one concludes that
\begin{align*}
\| Q^{+}_{0,b^{\varepsilon}_{1}}(f\,\ln[f],\langle \cdot \rangle^{\ell}f) \|_{L^{r}(\mathbb{R}^{d})} &\leq \varepsilon^{-\frac{d}{2r'}}\,2^{\frac{d}{2r'}} \| b \|_{L^{1}(\mathbb{S}^{d-1})}\|f\,\ln[f]\|_{L^{1}(\mathbb{R}^{d})}\| f \|_{L^{r}_{\ell}(\mathbb{R}^{d})} \,,\\
\text{and}\qquad \| Q^{+}_{0,b^{\varepsilon}_{1}}(f\,\text{1}_{\{f\leq K\}},\langle \cdot \rangle^{\ell}f) \|_{L^{r}(\mathbb{R}^{d})} &\leq \varepsilon^{-\frac{d}{2r'}}\,2^{\frac{d}{2r'}} \| b \|_{L^{1}(\mathbb{S}^{d-1})}\|f\,\text{1}_{\{f\leq K\}}\|_{L^{\frac{2r}{r+1}}(\mathbb{R}^{d})}\|\langle \cdot \rangle^{\ell} f\|_{L^{\frac{2r}{r+1}}(\mathbb{R}^{d})}\\
&\leq \varepsilon^{-\frac{d}{2r'}}\,K^{\frac{1}{2r'}}\,2^{\frac{d}{2r'}} \| b \|_{L^{1}(\mathbb{S}^{d-1})}\|f\|^{\frac{r+1}{2r}}_{L^{1}(\mathbb{R}^{d})}\| f\|_{L^{\frac{2r}{r+1}}_{\ell}(\mathbb{R}^{d})}\,.
\end{align*}
Use, in the latter inequality, the interpolation
\begin{equation*}
\| f \|_{L^{\frac{2r}{r+1}}_{\ell}(\mathbb{R}^{d})}\leq\| f \|^{\frac{1}{2}}_{L^{1}_{\ell}(\mathbb{R}^{d})}\| f \|^{\frac{1}{2}}_{L^{r}_{\ell}(\mathbb{R}^{d})}\,,
\end{equation*}
and gather the estimates to conclude that
\begin{align*} 
\| Q^{+}_{0,b^{\varepsilon}_{1}}(f,\langle \cdot  \rangle^\ell f) \|_{L^{r}(\mathbb{R}^{d})} \leq \varepsilon^{-\frac{d}{2r'}}\,2^{\frac{d}{2r'}}\,\| b \|_{L^{1}(\mathbb{S}^{d-1})}\Big(\ln(K)^{-1} & \| f \ln[f] \|_{L^{1}(\mathbb{R}^{d}) } \| f \|_{L^{r}_{\ell}(\mathbb{R}^{d})} \\
&+ K^{\frac{1}{2r'}}\,\| f \|^{1+\frac{1}{2r}}_{L^{1}_{\ell}(\mathbb{R}^{d}) }\| f \|^{\frac{1}{2}}_{L^{r}_{\ell}(\mathbb{R}^{d})}\Big)\,.
\end{align*}
Since $Q^{+}_{0,b^{\varepsilon}_{1}}(\langle \cdot \rangle^{\ell}f,f)=Q^{+}_{0,\mathcal{R}b^{\varepsilon}_{1}}(f,\langle \cdot \rangle^{\ell}f)$ an identical estimate follows for the second term
\begin{align*} 
Q^{+}_{0,b^{\varepsilon}_{1}}(\langle \cdot \rangle^\ell f, f) \leq \varepsilon^{-\frac{d}{2r'}}\,2^{\frac{d}{2r'}}\, \| b \|_{L^{1}(\mathbb{S}^{d-1})}\Big(\ln(K)^{-1} \|f \ln[f] \|_{L^{1}(\mathbb{R}^{d}) } & \| f \|_{L^{r}_{\ell}(\mathbb{R}^{d})} \\
&+ \,K^{\frac{1}{2r'}}\,\| f \|^{1+\frac{1}{2r}}_{L^{1}_{\ell}(\mathbb{R}^{d}) }\| f \|^{\frac{1}{2}}_{L^{r}_{\ell}(\mathbb{R}^{d})}\Big)\,.
\end{align*}
Overall, we are led to the estimate
\begin{align}\label{Q+entro}
\begin{split}
\|Q^{+}_{0,b}(f,f)\|_{ L^{r}_{\ell}(\mathbb{R}^{d}) }  
&\leq 2^{\frac{d}{2r'}}\| b \|_{L^{1}(\mathbb{S}^{d-1})}\Big( \ln(K)^{-1}2^{\ell-1}\varepsilon^{-\frac{d}{2r'}}\|f \ln[f] \|_{L^{1}(\mathbb{R}^{d}) } \| f \|_{L^{r}_{\ell}(\mathbb{R}^{d})} \\
&\qquad+ 2^{\ell-1}\varepsilon^{-\frac{d}{2r'}}\,K^{\frac{1}{2r'}}\,\| f \|^{1+\frac{1}{2r}}_{L^{1}_{\ell}(\mathbb{R}^{d}) }\| f \|^{\frac{1}{2}}_{L^{r}_{\ell}(\mathbb{R}^{d})} + R^{-1}\| f \|_{L^{1}_{\ell}(\mathbb{R}^{d}) } \| f \|_{L^{r}_{\ell}(\mathbb{R}^{d})}\Big)\\
&= 2^{\frac{d}{2r'}}\| b \|_{L^{1}(\mathbb{S}^{d-1})}\Big(2^{\ell-1}\varepsilon^{-\frac{d}{2r'}}\,K^{\frac{1}{2r'}}\,C^{1+\frac{1}{2r}}_{\ell}\| f \|^{\frac{1}{2}}_{L^{r}_{\ell}(\mathbb{R}^{d})} \\
&\hspace{2.5cm} + \big(2^{\ell-1}\varepsilon^{-\frac{d}{2r'}}\,C_{\mathcal{M}}\,\ln(K)^{-1} + C_{\ell,0}R^{-1}\big)\| f \|_{L^{r}_\ell (\mathbb{R}^{d})}\Big)\,.
\end{split}
\end{align}
With this estimate at hand, the upper control goes as follows.  Using H\"{o}lder's inequality one notices that
\begin{align}\label{DMax1}
\begin{split}
\int_{\mathbb{R}^{d}}Q_0(f,f)(v)\langle v \rangle^{r\,\ell}&f^{r-1}(v) dv \\
&= \int_{\mathbb{R}^{d}}Q^{+}_0(f,f)(v)\langle v \rangle^{\ell}\big( \langle \cdot \rangle^{\ell}f \big)^{r-1}(v) dv - \|b\|_{L^{1}(\mathbb{S}^{d-1})}\| f \|^{r}_{L^{r}_{\ell}(\mathbb{R}^{d})} \\
& \leq \|Q^{+}_0(f,f) \|_{L^{r}_{\ell}(\mathbb{R}^{d})}\| f\|^{r-1}_{L^{r}_{\ell}(\mathbb{R}^{d})} - \|b\|_{L^{1}(\mathbb{S}^{d-1})}\| f \|^{r}_{L^{r}_{\ell}(\mathbb{R}^{d})}\,.
\end{split}
\end{align}
    Taking first $R>0$ and then $K>1$ sufficiently large in \eqref{Q+entro} such that 
    \begin{equation*}
    2^{\frac{d}{2r'}}\, 2^{\ell-1} \, \varepsilon(R)^{-\frac{d}{2r'}} \, C_{\mathcal{M}}\, \ln(K)^{-1} + 2^{\frac{d}{2r'}}\, C_{\ell,0}\, R^{-1}\le \frac14\, , 
    \end{equation*}
for example one can take
\begin{equation}\label{CMax-def1}
R:= 2^{\frac{d}{2r'}}\,8\,C_{\ell,0}\,,\quad\text{and then}\quad K:=\exp\Big(\frac{2^{\frac{d}{2r'}}\,2^{\ell-1}\,8\,C_{\mathcal{M}}}{\varepsilon(R)^{\frac{d}{2r'}}}\Big)\,,
\end{equation}
and plugging into \eqref{DMax1} we readily obtain, thanks to Young's inequality, that
\begin{align*}
\int_{\mathbb{R}^{d}}Q(f,f)(v)\langle v \rangle^{r\,\ell}f^{r-1}(v)dv &\leq \|b\|_{L^{1}(\mathbb{S}^{d-1})}\Big(2^{\ell-1}\,2^{\frac{d}{2r'}}\,\varepsilon^{-\frac{d}{2r'}}\,K^{\frac{1}{2r'}}\,C^{1+\frac{1}{2r}}_{\ell, 0}\|f\|^{r-\frac{1}{2}}_{L^{r}(\mathbb{R}^{d})} - \tfrac{3}{4}\| f \|^{r}_{L^{r}_{\ell}(\mathbb{R}^{d})}\Big)\\
&\leq \|b\|_{L^{1}(\mathbb{S}^{d-1})}\Big(\Big(4^{1-\frac{1}{2r}}\,2^{\ell-1}\,2^{\frac{d}{2r'}}\,\varepsilon^{-\frac{d}{2r'}}\,K^{\frac{1}{2r'}}\,C_{\ell,0}^{1+\frac{1}{2r}}\Big)^{2r} - \tfrac{1}{2}\| f \|^{r}_{L^{r}_{\ell}(\mathbb{R}^{d})}\Big)\,.
\end{align*}
We can choose
\begin{equation}\label{CMax-def2}
\bar C_{0}:=\|b\|_{L^{1}(\mathbb{S}^{d-1})}\Big(4^{1-\frac{1}{2r}}\,2^{\ell-1}\,2^{\frac{d}{2r'}}\,\varepsilon(R)^{-\frac{d}{2r'}}\,K^{\frac{1}{2r'}}\,C_{\ell,0}^{1+\frac{1}{2r}}\Big)^{2r}
\end{equation}
which proves the result.
\end{proof}
\subsection{Propagation of $L^{r}_{\ell}$ norms.} The conditions of Theorem \ref{Dirichlet-r} and Theorem \ref{DirichletMM} hold for any solution $f(t)$ of the Boltzmann equation with finite mass and energy.  Additionally, the solutions' entropy is uniformly bounded as explained in \eqref{entropycontrol} which is an important fact when addressing the case of Maxwell molecules.  Thus, multiplying the Boltzmann equation \eqref{collision3} by $\langle v \rangle^{r\,\ell}f(v)^{r-1}$, integrating in velocity, and using Theorem \ref{Dirichlet-r} for the case of hard potentials, after bounding $\| f(t) \|_{L^{r}_{ \frac{\gamma}{2}+{\ell}}}\ge \| f(t) \|_{L^{r}_{\ell}}$; or Theorem \ref{DirichletMM} for the case of Maxwell molecules, it follows the linear ordinary differential inequality 
\begin{equation}\label{Lrell ODE}
\frac{d X}{dt} \leq \bar C_{\gamma} - \tfrac{c_{lb}}{2}\|b\|_{L^{1}(\mathbb{S}^{d-1})}\, X\,,  \qquad\quad \gamma\ge 0\,,\quad X(t):=\| f(t) \|^{r}_{L^{r}_{ {\ell}}(\mathbb{R}^{d})}\,.
\end{equation}
The constant $\bar C_{\gamma}>0$,  defined in  \eqref{C1gamma-def}  and \eqref{CMax-def2},  has been estimated in such theorems as functions of  $C_{\ell, \gamma}$ for $\gamma>0$ or  $C_{\ell,0}$ following the definition \eqref{C_ell_alpha}. 

To this end, integrating the ordinary  differential inequality  \eqref{Lrell ODE} proves the following theorem.
\begin{thm}\label{Theo-prop-r}
Let $b\in L^{1}(\mathbb{S}^{d-1})$, $\gamma\in[0,1]$, $\ell\geq0$, and $r\in(1,\infty)$.   Assume that $f_0$ satisfies that
\begin{equation*}
m_{\ell_0}[f_0] <\infty\,, \quad\text{for}\quad \ell_0:=\max\Big\{ 1 + \epsilon\,,\,\tfrac{\gamma+\ell}{2} \Big\}\,,
\end{equation*}
for $\epsilon>0$ as small as desired, and that $\| f_0 \|_{L^{r}_{\ell}(\mathbb{R}^{d})}< +\infty$.  Then, the solution $f(t)$ of the Boltzmann equation \eqref{collision3} with initial datum $f_0$ satisfies
\begin{equation*}
\sup_{t\geq0}\| f(t) \|_{L^{r}_{\ell}(\mathbb{R}^{d})} \leq \max\Big\{\| f_{0} \|_{L^{r}_{\ell}(\mathbb{R}^{d})}, \Big(\frac{2\,\bar{C}_{\gamma}}{c_{lb} \|b\|_{L^{1}(\mathbb{S}^{d-1})}} \Big)^{\frac{1}{r}}\Big\}\,.
\end{equation*}
The constant $\bar{C}_{\gamma}>0$ is estimated in Theorem \ref{Dirichlet-r} equation \eqref{C1gamma-def} for hard potentials ($\gamma>0$) and Theorem \ref{DirichletMM} equation \eqref{CMax-def2} for Maxwell molecules ($\gamma>0$) .  In particular, for the case of Maxwell molecules the argument requires that the entropy must be propagated. 

These constants depend on the upper bounds  to solutions propagating moments of order $2\ell$ for the potential  rate $\alpha$, characterized in  \eqref{mom-prop},  associated to the transition probability  rate of the Boltzmann flow solved in Theorem~\ref{CauchyProblem},  were the  Lower Bound Lemma~\ref{lblemma} calculates the positive constant $c_{lb}$,  valid for solutions with $2+\epsilon$ moments,   and defines  the coercive constant $\A_\kc =c_{lb} 2^{\gamma/2}A_\kc$, with ${A}_{s\kc} =2^{-\frac{\gamma}{2}} c_{\Phi}  ( \|b\|_{L^1(\mathbb{S}^{d-1})}-\mu_{s\kc})$. In the case   Maxwell type of interactions, one can simply take $c_{lb}=1$.
\end{thm} 
\subsection{Propagation of $L^{\infty}(\mathbb{R}^d)$-norm}  The main difficulty to overcome in the case $r=\infty$ is the fact that the constant $\bar{C}_{\gamma}$ in Theorem \ref{Theo-prop-r} degenerate in the limit $r\rightarrow\infty$.  Indeed, recall that Riesz--Thorin interpolation played and important role in the gain of integrability result, in particular, the end-point case $(p,q,r)=(\infty,1,\infty)$.  This issue is resolved using the following corollary.
\begin{cor}[\bf{Upper bound}]\label{APq+} Write $b=b_{S}+b_{R}$ with $b_{S}$ vanishing in the vicinity of $\{ -1 , 1\}$.  Further, write $b_{R}(s)= b_{R}(s)1_{\{s\geq0\}}+b_{R}(s)1_{\{s<0\}}=:b^{+}_{R}(s)+b^{-}_{R}(s)$ for the forward and backward scattering components supported in $\mathbb{S}^{d-1}_{\pm}$ respectively.  Then, for any $\gamma\geq0$ it follows that
\begin{align}\label{upper}
\begin{split}
&\| Q^{+}_{|\cdot|^{\gamma},b_{S} }\big( g , f \big) \|_{L^{\infty}(\mathbb{R}^{d})} \leq C_{b_{S}}(2,2) \| g \|_{L^{2}_{\gamma}(\mathbb{R}^{d})}\| f  \|_{L^{2}_{\gamma}(\mathbb{R}^{d})}\,,\\[3pt]
&Q^{+}_{|\cdot|^{\gamma},b^{+}_{R} }\big(g , f \big)(v) \leq C\|b^{+}_{R}\|_{L^{1}(\mathbb{S}^{d-1})}\| f \|_{L^{\infty}(\mathbb{R}^{d})}\| g  \|_{L^{1}_{\gamma}(\mathbb{R}^{d})}\langle v \rangle^{\gamma}\,,\quad\text{and}\\[3pt]
&Q^{+}_{|\cdot|^{\gamma},b^{-}_{R} }\big(g , f \big)(v) \leq C\|b^{-}_{R}\|_{L^{1}(\mathbb{S}^{d-1})}\| g \|_{L^{\infty}(\mathbb{R}^{d})}\| f  \|_{L^{1}_{\gamma}(\mathbb{R}^{d})}\langle v \rangle^{\gamma}\,.
\end{split}
\end{align}
\end{cor}
\begin{proof}
The first inequality is a particular case of Theorem \ref{Young} with $(p,q,r) = (2,2,\infty)$.  For the second just notice, since $b^{+}_{R}$ is supported in $[0,1]$, that $|u|\leq \sqrt{2}\,|v'_{*} - v| \leq \sqrt{2}\, \langle v'_{*} \rangle \langle v \rangle$.  Therefore,
\begin{equation*}
Q^{+}_{|\cdot|^{\gamma},b^{+}_{R} }\big(g , f \big)(v) \leq 2^{\frac{\gamma}{2}}Q^{+}_{o,b^{+}_{R} }\big(g\langle\cdot\rangle^{\gamma},f \big)(v)\,\langle v \rangle^{\gamma}\leq C_{b^{+}_{R}}(\infty,1)\|f\|_{L^{\infty}(\mathbb{R}^{d})}\| g \|_{L^{1}_{\gamma}(\mathbb{R}^{d})}\langle v \rangle^{\gamma}\,,
\end{equation*}
where we applied Theorem \ref{Young} with $(p,q,r) = (\infty,1,\infty)$ in the last inequality.  In addition, clearly $C_{b^{+}_{R}}(\infty,1)\sim \|b^{+}_{R}\|_{L^{1}(\mathbb{S}^{d-1})}$.  Similarly, for the term related to $b^{-}_{R}$, supported in $[-1,0)$, it follows that $|u|\leq \sqrt{2}\,|v' - v| \leq \sqrt{2}\,\langle v' \rangle \langle v \rangle$.  Thus,
\begin{equation*}
Q^{+}_{|\cdot|^{\gamma},b^{-}_{R} }\big(g , f \big)(v) \leq 2^{\frac{\gamma}{2}}Q^{+}_{o,b^{-}_{R} }\big(g,f \langle\cdot\rangle^{\gamma}  \big)(v)\langle v \rangle^{\gamma}\leq C_{b^{-}_{R}}(1,\infty)\| f \|_{L^{1}_{\gamma}(\mathbb{R}^{d})}\| g \|_{L^{\infty}(\mathbb{R}^{d})}\langle v \rangle^{\gamma}\,.
\end{equation*}
Again, Theorem \ref{Young} with $(p,q,r) = (1, \infty,\infty)$ was applied in the last inequality.  The constant is estimated as $C_{b^{-}_{R}}(1,\infty)\sim \|b^{-}_{R}\|_{L^{1}(\mathbb{S}^{d-1})}$. 
\end{proof}

A direct application of Corollary \ref{APq+} gives an elementary proof of the $L^{\infty}$-norm propagation for the homogeneous Boltzmann equation.  As mentioned earlier, a first proof of this fact was given in \cite{Arkeryd-Linfty} under more stringent assumptions.  For a discussion with addition of exponential weights refer to \cite{AlonsoGThar}.
\begin{thm}[\bf{Propagation of} $L^{\infty}_{\ell}$-norms]\label{Tinfty}
Let $b\in L^{1}(\mathbb{S}^{d-1})$, $\gamma\in[0,1]$, $\ell\geq0$, and assume that for $\epsilon>0$ (as small as desired) 
\begin{equation*}
m_{\ell_0}[f_0] + \| f_{0} \|_{L^{\infty}_{\ell}(\mathbb{R}^{d})} =: \bar C_{\ell}<\infty\,,\quad\text{for}\quad \ell_0:=\max\Big\{1+\epsilon,{\ell+\gamma}\Big\}\,.
\end{equation*} 
Then, there exists a constant $C^{\infty}_{\ell}:=C^{\infty}(\| b\|_1, \bar C_{\ell},\gamma,\ell)>0$ such that
\begin{equation*}
\sup_{t\geq0}\| f(t,\cdot) \|_{L^{\infty}_{\ell}(\mathbb{R}^{d})} \leq C^{\infty}_{\ell}<\infty\,,
\end{equation*}
for the solution $f(v,t)$ of the Boltzmann equation.
\end{thm}
\begin{proof}
Let us consider the case $\ell=0$ first.  Due to the propagation of mass/energy for the Boltzmann equation,
\begin{equation*}
\sup_{t\geq0} m_{\gamma}[f(t)] < \infty\,.
\end{equation*}
Consequently $\| f_0 \|_{L^{2}_{\gamma}(\mathbb{R}^{d})}\leq \| f_0 \|_{L^{1}_{2\gamma}(\mathbb{R}^{d})}^{\frac{1}{2}}\| f_0 \|_{L^{\infty}(\mathbb{R}^{d})}^{\frac{1}{2}}<\infty$, and Theorem \ref{Theo-prop-r} applied with $r=2$ and $\ell=\gamma$ implies that
\begin{equation*}
\sup_{t\geq0}\| f(t,\cdot) \|_{L^{2}_{\gamma}(\mathbb{R}^{d})} < \infty\,.
\end{equation*}
Now, Corollary \ref{APq+} (with $g=f$) and Lemma \ref{lblemma} give us that
\begin{equation}\label{Tinfty-e3}
\begin{aligned}
\partial_{t}& f (v) = Q^{+}_{|\cdot|^{\gamma},b}(f,f)(v) - f(v) \big(f\ast|\cdot|^{\gamma}\big)(v)\\[2pt]
& = Q^{+}_{|\cdot|^{\gamma},b_{S}}(f,f)(v) + Q^{+}_{|\cdot|^{\gamma},b_{R}}(f,f)(v) - f(v)\,\|b\|_{L^{1}(\mathbb{S}^{d-1})}\,\big(f\ast|\cdot|^{\gamma}\big)(v)\\[2pt]
&\hspace{-.5cm}\leq C_{b_{S}}(2,2) \| f \|^{2}_{L^{2}_{\gamma}(\mathbb{R}^{d})} + C\|b_{R}\|_{L^{1}(\mathbb{S}^{d-1})}\| f \|_{L^{\infty}(\mathbb{R}^{d})}\| f  \|_{L^{1}_{\gamma}(\mathbb{R}^{d})}\langle v \rangle ^{\gamma} - c_{lb}\,\|b\|_{L^{1}(\mathbb{S}^{d-1})}\,f(v)\langle v \rangle ^{\gamma}\,.
\end{aligned}
\end{equation}
Since $\| f \|_{L^{1}_{\gamma}(\mathbb{R}^{d})} = m_{\frac\gamma2}[f]\leq m_{1}[f_0]<\infty$, we can choose $b_R>0$ such that
\begin{equation*}
C\| b_{R}\|_{L^{1}(\mathbb{S}^{d-1})}\| f  \|_{L^{1}_{\gamma}(\mathbb{R}^{d})}\leq \frac{c_{lb}}{2}\,\|b\|_{L^{1}(\mathbb{S}^{d-1})}
\end{equation*}
in \eqref{Tinfty-e3} to obtain that
\begin{align}\label{Tinfty-e4}
\partial_{t}f(v,t) &\leq C(b_{S},f_0) + \|b\|_{L^{1}(\mathbb{S}^{d-1})}\Big(\tfrac{c_{lb}}{2}\|f(t)\|_{L^{\infty}(\mathbb{R}^{d})}\langle v \rangle^{\gamma} - c_{lb}\,f(v,t)\langle v \rangle^{\gamma}\Big)\,,
\end{align}
where
$$
C(b_{S},f_0):=C_{b_{S}}(2,2)\sup_{t\geq0}\| f \|^{2}_{L^{2}_{\gamma}(\mathbb{R}^{d})}\leq C_{b_{S}}(2,2)\max\Big\{\| f_{0} \|_{L^{2}_{\gamma}(\mathbb{R}^{d})}, \Big(\frac{2\,\bar{C}_{\gamma}}{c_{lb} \|b\|_{L^{1}(\mathbb{S}^{d-1})}} \Big)^{\frac{1}{r}}\Big\}\,.
$$
We can integrate estimate \eqref{Tinfty-e4} in $[0,t]$ to obtain
\begin{align*}
f(v,t) &\leq f_{0}(v)e^{-c_{lb}\|b\|_{L^{1}(\mathbb{S}^{d-1})} \langle v \rangle^{\gamma} t} \\
&\qquad + \int^{t}_{0}e^{-c_{lb}\|b\|_{L^{1}(\mathbb{S}^{d-1})}\langle v \rangle^{\gamma}(t-s)}\Big( C(b_{S},f_0) + \tfrac{c_{lb}}{2}\|b\|_{L^{1}(\mathbb{S}^{d-1})}\|f(s)\|_{L^{\infty}(\mathbb{R}^{d})}\Big)\langle v \rangle^{\gamma} ds\\
&\leq \|f_{0}\|_{L^{\infty}(\mathbb{R}^{d})} + \Big(C(b_{S},f_0) + \tfrac{c_{lb}}{2}\,\|b\|_{L^{1}}\,\sup_{0\leq s \leq t}\|f(s)\|_{L^{\infty}(\mathbb{R}^{d})}\Big)\langle v \rangle^{\gamma}\int^{t}_{0}e^{-c_{lb}\|b\|_{L^{1}}\langle v \rangle^{\gamma}(t-s)} ds\\
&\leq \|f_0\|_{L^{\infty}(\mathbb{R}^{d})} + \frac{C(b_{S},f_0)}{c_{lb}\|b\|_{L^{1}(\mathbb{S}^{d-1})}}+\tfrac{1}{2}\sup_{0\leq s \leq T}\|f(s)\|_{L^{\infty}(\mathbb{R}^{d})}\,,\quad \text{ a.e. in } v\in\mathbb{R}^{d}\quad\text{and}\quad 0\leq t \leq T\,.
\end{align*}
The proof for this case is complete after computing the essential supremum in $v\in\mathbb{R}^{d}$ and, then, the supremum in $t\in[0,T]$.  Overall,
\begin{equation*}
\sup_{s\geq0}\|f(s)\|_{L^{\infty}(\mathbb{R}^{d})} \leq 2\|f_{0}\|_{L^{\infty}(\mathbb{R}^{d})} + \frac{2\,C(b_{S},f_0)}{c_{lb}\|b\|_{L^{1}(\mathbb{S}^{d-1})}}\,.
\end{equation*}
Now, for the case $\ell\geq0$ notice that 
\begin{equation*}
Q(f,f)(v)\langle v \rangle^{\ell} \leq Q^{+}(g,g)(v) - g(v)\,\|b\|_{L^{1}(\mathbb{S}^{d-1})}\,\big(f\ast|\cdot|^{\gamma}\big)(v)\,,\qquad g(v):=f(v)\langle v \rangle^{\ell}\,. 
\end{equation*} 
With this observation one can redo the previous proof with $g$ instead of $f$.
\end{proof}

\section{Fine properties of the collision operator}\label{fine properties}
The following computations are reminiscent of \cite{Bo76}.  Recalling the symmetrized weak formulation of the collisional integral, that is, for a suitably regular test function $\psi(v)$, the weak form of the collision integral is given by
\begin{equation*}
\int_{\mathbb{R}^d} Q(g, f) \psi(v) dv = \int_{\mathbb{R}^d}\int_{\mathbb{R}^d}\int_{\mathbb{S}^{d-1}} g(v) f(w)B(u|, \widehat{u}\cdot \sigma) \Big(\psi(v') - \psi(v)\Big) d\sigma dw dv \, .
\end{equation*}
Let us test in particular with the classical Fourier multiplier,
$ \psi(v) = e^{-i\zeta \cdot v }$, where $\zeta$ is the Fourier variable, we get the Fourier Transform of the collision integral through its weak form
\begin{align*}
\widehat{Q(g,f)}(\zeta) &=\int_{\mathbb{R}^d} Q(g, f) e^{-i \zeta \cdot v} dv \nonumber \\
&= \int_{\mathbb{R}^d}\int_{\mathbb{R}^d}\int_{\mathbb{S}^{d-1}} g(v) f(w)  \frac{B(|u|, \widehat{u}\cdot \sigma)}{(\sqrt{2\pi})^d}\Big(e^{-i\zeta \cdot v'} - e^{-i \zeta \cdot v}\Big) d\sigma dw dv \, .
\end{align*}
Let us use $\;\widehat{\cdot}\;$ to denote the Fourier transform.  It should not be a matter of confusion with unit vectors.  Assuming a collision kernel $B(|u|,\widehat{u} \cdot \sigma) = \Phi(u)\,b(\widehat{u}\cdot \sigma)$, one has after trivial algebra that
\begin{equation}\label{weakQHat3}
\widehat{Q(g,f)}(\zeta) =\int_{\mathbb{R}^d}\int_{\mathbb{R}^d} g(v)\, f(v - u)\, \Phi(u)\bigg(\int_ {\mathbb{S}^{d-1}} b(\widehat{u}\cdot \sigma) e^{-i \zeta \cdot v}\Big(e^{-i \frac{1}{2}\zeta \cdot (|u|\sigma - u)} - 1\Big) d\sigma\bigg) du dv\,.
\end{equation}
Using, in the $\sigma$-integration, a reflection that interchanges $\widehat{u}$ and $\widehat{\zeta}$, it follow that  
\begin{eqnarray*}
\int_{\mathbb{S}^{d-1}}b(\widehat{u}\cdot \sigma)\Big(e^{-i\frac{1}{2}\zeta\cdot ( |
u |\sigma-u)}-1\Big)d\sigma =\int_{\mathbb{S}^{d-1}}b(\widehat{\zeta}\cdot \sigma)\Big( e^{ i u\cdot\zeta^{-} }-1 \Big) d\sigma\,.
\end{eqnarray*}
As a consequence,
\begin{equation*}
\widehat{Q(g,f)}(\zeta) =\int_ {\mathbb{S}^{d-1}} b(\widehat{\zeta}\cdot \sigma) \bigg(\int_{\mathbb{R}^d}g(v)e^{-i \zeta \cdot v}\,\int_{\mathbb{R}^d} f(v - u)\, \Phi(u) \Big(e^{ i u \cdot \zeta^{-}} - 1\Big)  du dv\bigg) d\sigma\,,
\end{equation*}
which, in the case of Maxwell molecules, reduces simply to
\begin{equation}\label{weakQHat5}
\widehat{Q_{o}(g,f)}(\zeta) =\int_ {\mathbb{S}^{d-1}} b(\widehat{\zeta}\cdot \sigma) \Big(\widehat{g}(\zeta^{+})\widehat{f}(\zeta^{-}) - g(\zeta)f(0)\Big) d\sigma.
\end{equation}
Equation \eqref{weakQHat5} is regarded as Bobylev's formula and has far reaching consequences in the fine analysis of cutoff and non-cutoff collision Boltzmann operators. 

\subsection{Smoothing effects of gain operator}
Assume that $b(\cdot)$ is smooth and vanishes in the vicinity of $\{-1,1\}$.  Additionally, assume $\Phi(u)$ is smooth vanishing near the origin and at infinity.  For any two \textit{different} unitary vectors $\widehat{u},\, \widehat{\xi}$, define $\varphi$ as $\cos(\varphi)=\widehat{u}\,\cdot\,\widehat{\xi}$.  Here we argue in three dimension for the technical simplicity of exposition, however, these ideas can be extended to $\mathbb{R}^{d}$, see \cite{Hormander,Stein}.  In particular see \cite{LionsIII} for the Boltzmann context.  In polar coordinates, we parameterize $\sigma\in\mathbb{S}^{2}$ as
\begin{equation*}
\sigma = \cos(\theta)\widehat{u} + \sin(\theta)\sin(\phi)\widehat{w}_{1} + \sin(\theta)\cos(\phi)\widehat{w}_{2}\,.
\end{equation*}
Here we choose $\widehat{w}_{1}$ in the plane generated by $\{\widehat{u},\widehat{\xi}\}$, so, $\widehat{w}_{2}$ is perpendicular to both $\widehat{u}$, and $\widehat{\xi}$.   Then,

\begin{align*}
\int_{\mathbb{S}^{2}}b(\widehat{u}\cdot\sigma)e^{-ix\widehat{\xi}\cdot\sigma}d\sigma = \int^{\pi}_{0}b(\cos(\theta))e^{ix\cos(\theta)\sin(\varphi)}\sin(\theta)\bigg(\int^{2\pi}_{0}e^{-ix\sin(\theta)\sin(\varphi)\sin(\phi)}d\phi\bigg)d\theta\,.
\end{align*}
The integral inside the parenthesis has stationary phase for $\phi=\{\pi/2,3\pi/2\}$.   As a consequence, for $x\gg1$ it has the asymptotic behavior
\begin{align*}
\int_{\mathbb{S}^{2}}b(\widehat{u}\cdot\sigma)e^{-ix\widehat{\xi}\cdot\sigma}d\sigma &\approx \sqrt{\frac{2}{x\,\pi}}e^{\pm i \pi/4}\int^{\pi}_{0}\frac{b(\cos(\theta))}{\sqrt{\sin(\theta)\sin(\varphi)}}e^{-ix(\cos(\theta)\sin(\varphi)\pm \sin(\theta)\sin(\varphi))}\sin(\theta)d\theta\\
&= \sqrt{\frac{2}{x\,\pi}}e^{\pm i \pi/4}\int^{\pi}_{0}\frac{b(\cos(\theta))}{\sqrt{\sin(\theta)\sin(\varphi)}}e^{-ix\cos(\theta\pm\varphi)}\sin(\theta)d\theta\,.
\end{align*}
This latter integral has stationary phase for $\theta=\{\varphi,-\varphi +\pi \}$, therefore, we conclude that
\begin{equation*}
\int_{\mathbb{S}^{2}}b(\widehat{u}\cdot\sigma)e^{-ix\widehat{\xi}\cdot\sigma}d\sigma \approx \frac{2}{x\pi}b(\pm\cos(\varphi))e^{\pm i\pi/4}e^{\pm i x}\,,\quad x\gg1\,. 
\end{equation*}
Now, using this asymptotic expression in formula \eqref{weakQHat3} and recalling that $\Phi(\cdot)$ vanishes near zero, it follows that for $|\zeta|\gg1$,

\begin{align*}
\widehat{Q^{+}(g,f)}(\zeta) &=  \int_{\mathbb{R}^3} \bigg(\int_{\mathbb{R}^{3}}g(v)e^{-i \zeta \cdot v/2} f(v - u)e^{-i\zeta\cdot(v-u)/2}dv\bigg)\bigg(\int_{\mathbb{S}^{2}}b(\widehat{u}\cdot \sigma) \,\,e^{-i \frac{1}{2}(|\zeta||u|\widehat{\zeta}\cdot\sigma)} d\sigma\bigg)\Phi(u) du\\
&\approx \frac{2}{\pi}e^{\pm i \pi/4}\int_{\mathbb{R}^3} \bigg(\int_{\mathbb{R}^{3}}g(v)e^{-i \zeta \cdot v/2} f(v - u)e^{-i\zeta\cdot(v-u)/2}dv\bigg)\frac{b(\pm\widehat{u}\cdot\widehat{\zeta})}{|\zeta||u|}\Phi(u)e^{\pm i \frac{1}{2}|\zeta||u|} d u\\
&= \frac{2}{|\zeta|\pi}e^{\pm i \pi/4}\int_{\mathbb{R}^{3}} g(v)e^{-i\zeta\cdot v}\, \mathcal{P}_{osc}(\tau_{-v}\,\mathcal{R} f)(\xi)dv\,,
\end{align*}
where $\mathcal{R}\phi(x)=\phi(-x)$ is the reflexion operator and
\begin{equation*}
\mathcal{P}_{osc}(\varphi)(\xi) = \int_{\mathbb{R}^{3}} \varphi(u)\frac{\Phi(u)}{|u|}\, b(\pm\widehat{u}\cdot\widehat{\zeta}) e^{i \frac{1}{2}(\zeta\cdot u\pm |\zeta||u|)}\, du\,.
\end{equation*}
The operator $\mathcal{P}_{osc}$ is an oscillatory integral and it is a bounded operator in $L^{2}(\mathbb{R}^{3})$, see \cite{Hormander},  as long as the phase $S(u,\xi)=u\cdot\zeta \pm |u||\zeta|$ satisfies $$\text{Det}\bigg(\frac{\partial^{2}S}{\partial u_{i}\partial \zeta_{j}}(u,\zeta)\bigg) = \text{Det}\big(\textbf{1}_{3\times 3} \pm \widehat{u}\,\widehat{\zeta}^{T}\big) = 1 \pm \widehat{u}\cdot\widehat{\zeta} \neq 0$$ in the support of $\Phi(u)\,b(\pm\widehat{u}\cdot\widehat{\zeta})$.   The fact that $b(\cdot)$ vanishes near $\{-1,1\}$ ensures this condition.  Thus, Minkowski's inequality implies that
\begin{align*}
\bigg(\int_{|\zeta|\gg1}\big|\widehat{Q^{+}(g,f)}&(\zeta)|\zeta|\big|^{2}d\zeta\bigg)^{1/2}\leq \frac{2}{\pi}\bigg(\int_{|\zeta|\gg1}\Big(\int_{\mathbb{R}^{3}}\big|\mathcal{P}_{osc}(\tau_{-v}\mathcal{R}f)(\zeta)\big|\,|g(v)|dv\Big)^{2}d\zeta\bigg)^{1/2}\\
&\leq \frac{2}{\pi}\int_{\mathbb{R}^{3}}\Big(\int_{|\zeta|\gg1}\big|\mathcal{P}_{osc}(\tau_{-v}\mathcal{R}f)(\zeta)\big|^{2}d\zeta\Big)^{1/2}\,|g(v)|dv\\
&\hspace{-1cm}\leq \frac{2}{\pi}\|\mathcal{P}_{osc}\|_{L^{2}}\int_{\mathbb{R}^{3}}\|\tau_{-v}\mathcal{R}f)\|_{L^{2}(\mathbb{R}^{3})}\,|g(v)|dv \leq \frac{2}{\pi}\|\mathcal{P}_{osc}\|_{L^{2}}\|f\|_{L^{2}(\mathbb{R}^{3})}\|g\|_{L^{1}(\mathbb{R}^{3})}\,.
\end{align*}
We know from Theorem \ref{Young} that in the region $|\zeta|\leq R$, for any $R>0$, a similar bound holds.  This leads to the estimate
\begin{equation*}
\|Q^{+}(g,f)\|_{H^{1}(\mathbb{R}^{3})}\leq C_{\Phi,b}\|g\|_{L^{1}(\mathbb{R}^{3})}\|f\|_{L^{2}(\mathbb{R}^{3})}\,,
\end{equation*}
as long as the aforementioned conditions on $\Phi$ and $b$ hold.  This estimate was first noticed in \cite{LionsIII} and, later, proved by generalised Radon transform methods in \cite{WennbergRT, MV04}.  Such conditions can be relaxed a bit, but, such regularisation will not hold for $b\in L^{1}(\mathbb{S}^{d-1})$.

\smallskip
\nn This procedure can be generalized to any dimensions and obtain the following theorem.
\begin{thm}\label{smoothingq+}
Let $\Phi$ and $b$ be smooth with $\Phi$ vanishing in a vicinity of zero and infinity, and $b$ vanishing in a vicinity of $\{-1,1\}$.  Then,
\begin{equation*}
\begin{aligned}
\|Q^{+}(g,f)\|_{H^{\frac{d-1}{2}}(\mathbb{R}^{d})}&\leq C^{1}_{\Phi,b}\|g\|_{L^{1}(\mathbb{R}^{d})}\|f\|_{L^{2}(\mathbb{R}^{d})}\,,\\
\|Q^{+}(g,f)\|_{H^{\frac{d-1}{2}}(\mathbb{R}^{d})}&\leq C^{2}_{\Phi,b}\|g\|_{L^{2}(\mathbb{R}^{d})}\|f\|_{L^{1}(\mathbb{R}^{d})}\,.
\end{aligned}
\end{equation*} 
\end{thm}
We end with a remark.  Since $H^{\frac{d-1}{2}}(\mathbb{R}^{d})\subset L^{2d}(\mathbb{R}^{d})$, then
\begin{equation*}
\| Q^{+}(g,f) \|_{ L^{2d}(\mathbb{R}^{d}) } \leq C \|Q^{+}(g,f)\|_{H^{\frac{d-1}{2}}(\mathbb{R}^{d})} \leq C^{1}_{\Phi,b}\|g\|_{L^{1}(\mathbb{R}^{d})}\|f\|_{L^{2}(\mathbb{R}^{d})}\,,
\end{equation*}
which provides an interpolation point to prove an analog to Theorem \ref{gainmaxwell} under more stringent assumptions in the angular kernel $b$.

\smallskip
\nn Fourier methods have extensive applications in the study of fine properties of the collision operator.  These methods have being used to show propagation of Sobolev norms for the homogeneous Boltzmann equation in \cite{MV04} under the restriction $b\in L^{2}(\mathbb{S}^{d-1})$, based in a result of \cite{BD}.  In the work \cite{AlonsoGThar} Fourier methods has being used in full potential to prove propagation of Sobolev regularity for general $b\in L^{1}(\mathbb{S}^{d-1})$.  \\

\section{Appendix}\label{appendix}

The proof of Theorem~\ref{Theorem_ODE}\label{proofTheorem_ODE}, 
follows the same lines of the argument  to solve ODEs in Banach spaces proposed in \cite{bressan}.  We include it in this manuscript for completeness. The proof is divided into three steps:

\smallskip
\noindent
\textbf{Step 1.}
First note that since $\Omega$ is bounded, there exists a uniform bound $K_Q$ of $Q(u)$, for all $u$ in $\Omega$. 
Next,  let $u$ be in $\Omega$, then there exists $h_{u}>0$,  such that 
the intersection $B(u+h Q(u),\epsilon)\cap {\Omega}\backslash\{u+hQ(u)\}$ is non-empty,  for $0<h<h_{u}$ and all $\epsilon>0$ sufficiently small.

\smallskip
\noindent
In addition, one can estimate $\|Q(u)-Q(v)\|\leq \frac{\epsilon}{4}$,  if $\|u-v\|\leq (K_Q+1)h$.  
Hence, take $w$ to be a point inside $B(u+hQ(u),\epsilon)\cap {\Omega}\backslash\{u+hQ(u)\} $ satisfying 
$$\|w-u-hQ(u)\|\leq \frac{\epsilon h}{4}.$$ 
Next, consider the linear map
$$s\mapsto \rho(s)=u+\frac{s(w-u)}{h},~~~~s\in[0,h], $$
and see that  $\rho(s)\in \Omega$ for all $s$ in $[0,h]$ by the convexity of $\Omega$. 
Moreover, since $\dot{\rho}(s)=\frac{w-u}{h}$, $$\|\dot{\rho}(s)-Q(u)\| \leq \frac{\epsilon}{4}.$$
Now, we can see that $$\|\rho(s)-u\|=\left\|\frac{s(w-u)}{h}\right\|\leq \|w-u\|\leq h\|Q(u)\|+\frac{\epsilon h}{4}<\overline{n}(K_Q+1)h,$$
which implies 
$$\|Q({\rho}(s))-Q(u)\|\leq \frac{\epsilon}{4},~~\forall\, s\in[0,h].$$
Therefore,
\begin{equation}\label{Theorem:ODE:E1}
\|\dot{\rho}(s)-Q({\rho}(s))\|\leq \frac{\epsilon}{2} \leq \epsilon, ~~~\forall\, s\in[0,h]\, , 
\end{equation}
and so,
\begin{equation*}
\|\dot{\rho}(s)\|\leq 1+\,K_Q ~~~\forall\, s\in[0,h] \ \text{and} \epsilon<1 \, . 
\end{equation*}

\smallskip  
\noindent
\textbf{Step 2.} From Step 1, we have proved the existence of solution $\rho$ to the equation \eqref{Theorem:ODE:E1} on an interval $[0,h]$. From this solution, we carry on the following process.
\begin{itemize}
\item [(1)] We start with the solution $\rho$, defined on $[0,h]$ of \eqref{Theorem:ODE:E1}.
\item [(2)]  Suppose such solution $\rho$ of \eqref{Theorem:ODE:E1} is constructed on $[0,\tau]$. Since $\rho(\tau)\in\Omega$, by the same process as in Step 1, the solution $\rho$ could be extended to  $[\tau,\tau+h_\tau]$. 
\item [(3)] Next, let  $\rho$ be constructed by  from \eqref{Theorem:ODE:E1} on a series of intervals $[0,\tau_1]$, $[\tau_1,\tau_2]$, $\cdots$, $[\tau_n,\tau_{n+1}]$, $\cdots$. Moreover, suppose the increasing sequence $\{\tau_n\}$ is bounded. 
Then, setting $$\tau=\lim_{n\to\infty}\tau_n\, ,$$
then, since $G({\rho})$ is bounded by $C_G$ on $[\tau_n,\tau_{n+1}]$ for all $n\in\mathbb{N},$ $\dot{\rho}$ is bounded by $\epsilon+C_G$ on $[0,\tau)$. Therefore, we can define $\rho(\tau)$ satisfying $$\rho(\tau)=\lim_{n\to\infty}\rho(\tau_n),\quad \dot{\rho}(\tau)=\lim_{n\to\infty}\dot{\rho}(\tau_n),$$
which implies that $\rho$ is a solution of \eqref{Theorem:ODE:E1} on $[0,\tau]$. 
\end{itemize}
By (3) of this process, we can see that if the solution $\rho$, constructed as above, is defined on $[0,T)$, it could be extended to $[0,T]$. Suppose that $[0,T]$ is the maximal  closed interval that $\rho$ could be constructed, by Step 2 of the process, $\rho$ could be extended to a larger interval $[T,T+T_h]$, which means that $\rho$ can be constructed on the whole interval $[0,\infty)$.

\smallskip
\noindent
\textbf{Step 3.} Let us now consider two sequences of approximate solutions $u^\epsilon$, $w^\epsilon$, where $\epsilon$ tends to $0$.  From Step 1 and Step 2, one can see that the time interval $[0,T]$ can be decomposed into $$\left(\bigcup_{\gamma}I_\gamma\right)\bigcup \mathfrak{N},$$ where $I_\gamma$ are countably many open intervals and $\mathfrak{N}$ is of measure $0$. 

\smallskip
\noindent
Taking the derivative of the difference $\|u^\epsilon(t)-w^\epsilon(t)\|$ gives
\begin{eqnarray*}
\frac{d}{dt}\|u^\epsilon(t)-w^\epsilon(t)\|&=&\Big[u^\epsilon-w^\epsilon,\dot{u}^\epsilon(t)-\dot{w}^\epsilon(t)\Big]_{-}\\
&\le& \Big[u^\epsilon-w^\epsilon,\dot{u}^\epsilon(t)-\dot{w}^\epsilon(t)\Big]_{-} + \epsilon\\
&\le & L\|u^\epsilon(t)-w^\epsilon(t)\|  +\epsilon,
\end{eqnarray*}
which yields $$\|u^\epsilon(t)-w^\epsilon(t)\| \to 0 \quad \text{as}\quad\epsilon\to0\,,$$
and we have the convergence $u^\epsilon\to u$ uniformly on $[0,T]$. The function $u$ is, then, a solution of our equation. \qed

\ \\

\section*{Acknowledgments}
The authors would like to thank  Ioakeim Ampatzoglou,  Erica De La Canal, Antonio Farah,  Milana Pavi\'c-Coli\'c and Maja Taskovic for very the careful reading and offering many suggestions that significantly improved the presentation of this manuscript.   The authors also  thank and gratefully acknowledge the hospitality and support from the Oden Institute of Computational Engineering and Sciences and the University of Texas Austin.
The authors were partially supported by the funding from  Bolsa de Produtividade em Pesquisa CNPq (303325/2019-4), 
NSF DMS: 2009736 and  DOE DE-SC0016283 project \emph{Simulation Center for Runaway Electron Avoidance and Mitigation.}

\ \\

\ \\

\ \\

\end{document}